\pdfoutput=1
\documentclass[]{article}
\usepackage[utf8]{inputenc}
\usepackage{graphicx} %
\usepackage{macros}
\usepackage{xcolor}
\usepackage{accents}

\usepackage{microtype}

\usepackage[margin=1in]{geometry}

\usepackage[style=alphabetic, 
backref, 
url=false,
doi=true,
isbn=false,
useprefix=true,
maxbibnames=99,
maxcitenames=5,
mincitenames=3,
maxalphanames=5,
minalphanames=3]{biblatex}
\usepackage[titletoc,title]{appendix}

\newcommand{\de}{\ensuremath{\delta}}
\newcommand{\pit}{\ensuremath{\mathsf{PIT}}}
\newcommand{\nexp}{\ensuremath{\mathsf{NEXP}}}
\newcommand{\ntime}{\ensuremath{\mathsf{NTIME}}}

\newcommand{\sft}{\ensuremath{\mathsf{sf}}}
\newcommand{\ubar}[1]{\underaccent{\bar}{#1}}
\newcounter{tagg}[section]

\newcommand{\eqtag}[1]{\stepcounter{tagg}\tag{\thesection.\thetagg} \label{#1}}
\newcommand{\Nin}{\ensuremath{\calN^{\sf in}}}

\title{Polynomial-Time PIT from (Almost) Necessary Assumptions}
\author{Robert Andrews\thanks{Cheriton School of Computer Science, University of Waterloo. Email: \texttt{randrews@uwaterloo.ca.}} \and Deepanshu Kush\thanks{Department of Computer Science, University of Toronto. Email: {\tt deepkush@cs.toronto.edu.}} \and Roei Tell\thanks{Department of Computer Science, University of Toronto. Email: {\tt roei@cs.toronto.edu.}}}
\date{April 8, 2025}

\begin{document}
	
	\pagenumbering{gobble}
	
	\maketitle
	
	\begin{abstract}
		The celebrated result of Kabanets and Impagliazzo (Computational Complexity, 2004) showed that $\mathsf{PIT}$ algorithms imply circuit lower bounds, and vice versa. Since then it has been a major challenge to understand the precise connections between $\mathsf{PIT}$ and lower bounds. In particular, a main goal has been to understand which lower bounds suffice to obtain efficient $\mathsf{PIT}$ algorithms, and how close are they to lower bounds that are necessary for the conclusion.
		
		We construct polynomial-time $\mathsf{PIT}$ algorithms from lower bounds that are, up to relatively minor remaining gaps, necessary for the existence of such algorithms. That is, we prove that these lower bounds are, up to the mentioned minor gaps, both sufficient and necessary for polynomial-time $\mathsf{PIT}$, over fields of characteristic zero. Over sufficiently large finite fields, we show a similar result wherein the $\mathsf{PIT}$ algorithm runs in time $n^{\log^{(c)}(n)}$, i.e. a power of $c$-iterated log for an arbitrarily large constant $c>1$.
		
		The key to these improvements is studying $\mathsf{PIT}$ versus lower bounds in the uniform setting, in which we focus on proving lower bounds for uniform arithmetic circuits and their variants (and on deducing algorithms from such lower bounds). Indeed, by working in this setting we obtain results that are significantly tighter than previously known results concerning polynomial-time $\mathsf{PIT}$ vs lower bounds, and are in fact also tighter than known hardness-vs-randomness connections in the Boolean setting.
		
		Our results are obtained by combining recent techniques from Boolean hardness vs randomness, and in particular the generator of Chen and Tell (FOCS 2021), with the algebraic hitting-set generator of Guo, Kumar, Saptharishi, and Solomon (SIAM J. Computing 2022) along with the bootstrapping ideas of Agrawal, Ghosh, and Saxena (STOC 2018) and of Kumar, Saptharishi, and Tengse (SODA 2019).
		
	\end{abstract}
	
	\newpage
	
	\pagenumbering{roman}
	
	\setcounter{tocdepth}{2}
	\tableofcontents
	
	\newpage
	\pagenumbering{arabic}
	
	\newpage 
	
	\section{Introduction} \label{sec:int}
	
	Polynomial identity testing ($\pit$) is a central problem in algebraic complexity. In this paper we focus on what is arguably the most well-studied variant: We are given as input a description of an arithmetic circuit $C(x_1,...,x_n)$ over $\F$ of size $\poly(n)$ computing a nonzero polynomial of degree $\poly(n)$, and our goal is to find $\vec{\alpha}$ such that $C(\vec{\alpha})\ne0$.\footnote{Indeed, this is the {\sf search} version of $\pit$, whereas the decision version calls for deciding whether or not $C$ is indeed nonzero. The distinction between the two will not be crucial in our work (see~\Cref{rem:int:formulation}).} Of course, a trivial randomized algorithm can just pick $\vec{\alpha}$ at random (assuming $\F$ is large enough, or taking an extension field $\K\supseteq\F$), and when we refer to solving $\pit$ we mean solving it by \emph{deterministic} algorithms (i.e., derandomizing the trivial algorithm).
	
	\subsection{The question of $\pit$ versus lower bounds} \label{sec:int:bg}
	
	One reason for the importance of $\pit$ is its connection to lower bounds. The study of this connection was initiated by Kabanets and Impagliazzo~\cite{KI04}, who proved that $\pit$ algorithms imply circuit lower bounds, and vice versa. Specifically, they showed that a $\pit$ algorithm running in polynomial time (or even in sub-exponential time) implies that \emph{either} there is a problem in $\ntime[2^{\poly(n)}]$ that is hard for polynomial-sized Boolean circuits, \emph{or} the permanent is hard for polynomial-sized arithmetic circuits. They also showed partial converses; for example, if there is an exponential-time computable polynomial over $\Z$ that does not have sub-exponential sized arithmetic circuits, we can solve $\pit$ over $\Z$ in quasipolynomial time. 
    
	The result of~\cite{KI04} revealed a connection between $\pit$ and lower bounds, but it is far from the end of story: The lower bounds that are \emph{implied} by $\pit$ algorithms in this result are significantly weaker than lower bounds that \emph{suffice} for $\pit$ algorithms (due to the disjunctive conclusion, and to the fact that one of the lower bounds in the disjunction is only in $\nexp$). Thus, the result put forward a major open question:
	
	\begin{problem} \label{prob:int:main}
		What are the connections between $\pit$ and lower bounds? Specifically, what lower bounds \emph{suffice} to obtain efficient $\pit$ algorithms, and how close are they to lower bounds that are \emph{necessary} for the conclusion?
	\end{problem}
	
	Let us briefly survey some of the progress that has been made on~\Cref{prob:int:main} (for a recent survey, see~\textcite{KS19}). One line of works made progress by studying circuit lower bounds in non-standard models of computation. Specifically, Jansen and Santhanam~\cite{JS12} defined a hybrid model combining arithmetic complexity and Boolean complexity, dubbed $a\cdot\mathcal{C}$ for a Boolean class $\calC$ (see also~\cite{KP11}). They showed an equivalence between sub-exponential time $\pit$ algorithms and lower bounds for arithmetic circuits against a problem computable in this model with linearly many non-uniform advice bits. \textcite{CIKK15} showed another equivalence, between subexponential-time $\pit$ algorithms that work on average, and average-case circuit lower bounds in the model from~\cite{JS12} (without advice), for a notion of average-case complexity called ``robustly often containment''.
	
	A different line of attack on~\Cref{prob:int:main} focused on the direction ``hardness implies $\pit$'', trying to obtain a stronger conclusion; namely, the goal in these works is to obtain a faster $\pit$ algorithm, compared to the quasipolynomial time (or sub-exponential time) algorithms considered in the results above. For example, Kumar, Saptharishi, and Tengse~\cite{KST19} (following Agrawal, Ghosh, and Saxena~\cite{AGS18}) showed how to bootstrap relatively slow hitting-set generators (i.e., ``black-box'' $\pit$ algorithms) for circuits over few variables into hitting-set generators for circuits over many variables that yield $\pit$ in almost polynomial time, i.e. time $n^{2^{2^{O(\log^*(n))}}}$ for $n$-sized circuits. And in another influential work, Guo, Kumar, Saptharishi, and Solomon~\cite{GKSS22} deduced a $\pit$ algorithm that indeed runs in polynomial time, from a strong hardness assumption. (Jumping ahead, our results will crucially build on their work.)
	
	This line of work can be viewed as the arithmetic analogue of ``hardness vs randomness'' results from the Boolean setting. Moreover, many of the underlying techniques for proving the arithmetic results above are inspired by the techniques for proving the analogous Boolean results. However, and in contrast to what one may expect, the hardness-vs-randomness connections in arithmetic complexity so far have been \emph{less tight} than the Boolean results. (For example, even the equivalences in~\cite{JS12,CIKK15} concern sub-exponential time algorithms, need non-uniform advice bits, and/or refer to average-case $\pit$ and circuit lower bounds.) This counters our expectation that progress on arithmetic complexity should be easier.
	
	\paragraph{Our contributions, in a gist.}
	In this work, we close the gaps above almost entirely. That is, we prove that certain natural lower bounds are both \emph{sufficient and necessary for solving $\pit$ in polynomial time and in the worst case}, up to relatively minor remaining gaps. Compared to the results above, we refer to polynomial-time $\pit$, rather than sub-exponential time; we refer to the standard notion of solving $\pit$ in the worst-case; and the hard problem is computable without non-uniform advice. The lower bounds that we study are for models of computation that are relatively standard and well-studied.
	
	The key to these improvements is studying \emph{$\pit$ vs lower bounds in the uniform setting}, and in particular focusing on lower bounds for uniform arithmetic circuits (i.e., for circuit families that can be printed by a Turing machine). We stress that our algorithms still solve $\pit$ on any given circuit (i.e., in the standard worst-case sense), and it is only the lower bounds that are different from the past, and need to hold only against uniform circuits. It is easy to show that the relevant lower bounds are necessary for $\pit$, and the more surprising direction of our results is constructing $\pit$ algorithms from these lower bounds. 
	
	This means that proving lower bounds for uniform families is a potential path towards obtaining $\pit$ algorithms. In addition, this is an approach for tackling~\Cref{prob:int:main} (i.e., connecting $\pit$ to lower bounds for uniform circuits) that yields very tight connections. The connections we show are significantly tighter than previously known connections for the arithmetic setting, and moreover, they are also tighter than the connections known in the Boolean setting (thus closing another long-standing gap, i.e. between hardness-vs-randomness in the arithmetic setting and in the Boolean one).

	Technically, analogous to the fact that the results of~\cite{KI04,JS12,CIKK15} were inspired by developments in Boolean hardness-vs-randomness machinery in the late 1990's (i.e., by~\cite{NisanW94,IKW02}), our results build on very recent developments in Boolean hardness-vs-randomness (see, e.g.,~\cite{CT23b} for a survey), while also leveraging  the state-of-the-art hitting-set generators for arithmetic circuits (i.e., leveraging ideas and techniques from~\cite{AGS18,KST19,GKSS22}). We elaborate on our techniques in~\Cref{sec:tech}.

	\subsection{Lower bounds for uniform arithmetic circuits} \label{sec:int:uni}
	
	Part of our contribution is to carefully delineate the notions that we believe (and demonstrate) are useful for studying $\pit$ vs uniform lower bounds. To set the stage for presenting our results, let us spell out some of these notions in advance, while providing pointers to precise definitions.
	
	\paragraph{Uniform randomized circuits.}

	A uniform circuit family is a sequence of circuits $\{C_n\in\F[x_1,...,x_n]\}_{n\in\N}$ such that there is an efficient Turing machine that gets input $1^n$ and prints $C_n$. The notion of ``efficient'' varies across definitions, as well as the question of printing field elements (see~\Cref{sec:pre:uniform}). The study of uniform circuits has a long history in the Boolean setting (see, e.g., the survey by Allender~\cite{Alle89}, and recent results such as~\cite{CK12,SW13,san23,DPT24}). In arithmetic complexity, suitable definitions were put forward by von zur Gathen~\cite{vzGathen_survey} and by Eberly~\cite{Ebe89} in the 1980's, and more recently, lower bounds for uniform arithmetic circuits of bounded depth were proved by~\textcite{KP09}, following Allender~\cite{Allender99} (see also strengthenings in~\cite{JS13,CKK14}).

	The lower bounds that turn out to be essentially equivalent to $\pit$ are for uniform \emph{randomized} circuits (in fact, for randomized networks -- see below).\footnote{This is analogous to the results of~\cite{CT21}, who showed bidirectional connections between derandomization (i.e., $pr\BPP=pr\P$) and lower bounds for randomized algorithms. More generally, this follows the blueprint suggested by Kozen~\cite{koz80}, wherein simulation of a class turns out to be equivalent to lower bounds for the class (see also~\cite{nir03,nir06}).} It is not a-priori obvious how to define randomized arithmetic circuits in general. We bypass this question by focusing on a nice subclass of randomized circuits: namely, deterministic circuits equipped with ``$\pit$ gates'' (i.e., gates that receive a description of an arithmetic circuit, and solve $\pit$ for that circuit). This subclass represents a limited use of randomness by the circuit -- i.e., only for solving $\pit$ -- and we prove that lower bounds against this subclass suffice for constructing $\pit$ algorithms.
	
	\paragraph{Arithmetic networks.}
	Our results focus on a stronger variant of arithmetic circuits, called \emph{arithmetic networks}; these were first formally defined by von zur Gathen~\cite{vzGathen_survey}. Networks get arithmetic inputs and produce arithmetic outputs. Their computation consists of standard arithmetic operations $\set{+,\times}$, and they can also ``route'' the results of the arithmetic computation, using a limited set of Boolean gates. Specifically, the Boolean gates can test whether an arithmetic element (i.e., the value of a gate, computed on the input) is zero, and depending on the outcome of the computation, they choose which arithmetic sub-circuit the arithmetic computation will proceed on. We stress that the Boolean gates have no access to the content of arithmetic values, except for zero-testing; intuitively, all they can do is ``route'' arithmetic values around the circuit according to zero-tests. See~\Cref{def:arith-network} for details.
	
	While the definition may seem unusual at first, arithmetic networks have actually been widely studied (implicitly and explicitly) in arithmetic complexity. For example, many classical ``arithmetic algorithms'' actually work in this model \cite{BvzGH82,vzGathen84,vzGathen86}, including the Euclidean algorithm; and even very recently, the constant-depth GCD algorithm of Andrews and Wigderson~\cite{AW24} is, in fact, precisely an arithmetic network.\footnote{There are, in general, good reasons for this. First, networks do not provide significantly more power for computing polynomials (see~\Cref{sec:pre:networks}). Secondly, standard arithmetic circuits can only compute continuous functions of their inputs, whereas many interesting algebraic problems are discontinuous (e.g., computing GCD of polynomials, or solving $\pit$).} Lower bounds for arithmetic networks have also been studied: the works of \cite{MP93, MMP96, GV17} prove lower bounds for arithmetic networks that decide membership in a (semi-)algebraic set using techniques similar to those used to prove lower bounds for algebraic decision trees \cite{Yao97}. %
	
	\paragraph{Lower bounds on all but finitely many inputs.}
	\label{todo:ra2}
	Lastly, lower bounds for uniform circuits and networks can differ greatly from lower bounds for their non-uniform counterparts. As one important example, for every $\F$, there is (unconditionally!) a $\P$-uniform family of arithmetic circuits $\set{C_n}_{n\in\N}$ of size $n^{O(k)}$ over $\F$ such that for every $n^k$-time-uniform circuit family $\{C'_n\}_{n \in \N}$ of size $n^k$ over $\F$, %
	and for all but at most finitely many $n\in\N$, it holds that $C_n(\vec{x})\ne C'_n(\vec{x})$ for \emph{every $\vec{x}\in\F^n$} (see~\Cref{fact:circuits:hierarchy}). That is, $\{C_n\}_{n\in\N}$ is hard for $n^k$-time-uniform circuits of size $n^k$ on \emph{each and every possible input} (except, at most, finitely many). 
	
	Lower bounds on all but finitely many inputs turn out to be closely related to $\pit$ algorithms, as we shall show. We stress that the feasibility of this almost-all-inputs hardness relies crucially on the fact that we consider hardness with respect to \emph{uniform} families of arithmetic circuits and networks. (Indeed, a non-uniform circuit can always have the value of the function at (say) $0^n$ hard-wired.) 
	
	\subsection{Result statements} \label{sec:int:results}
	
	At a high level, we show that $\pit$ algorithms are essentially equivalent to lower bounds for uniform randomized arithmetic networks on all but finitely many inputs.
    
    Our first result focuses on fields of characteristic zero. We show that if for some $k$ there is a  family of (deterministic) arithmetic circuits of size $n^{k^2}$ that is hard on all inputs for uniform arithmetic networks with $\pit$ gates of size $n^{k}$, then $\pit$ can be solved in \emph{polynomial time} over $\F$.
	
	\begin{theorem} [hardness on all inputs implies $\pit$; informal, see~\Cref{thm:main:zero}] \label{thm:int:main:zero}
		Let $\F$ be a field of characteristic zero. Assume that for some sufficiently large $k>1$ there is a strongly uniform\footnote{The meaning of ``strongly uniform'' refers to a stricter notion of uniformity, in which given the indices of gates in $C_n$, we can decide in polynomial time in the input length whether these gates are connected (indeed, the running time is polylogarithmic in the size of $C_n$); see~\Cref{def:suf-unif} for a precise definition. This is analogous to well-studied notions in Boolean complexity.} family of arithmetic circuits of size and degree $n^{k^2}$ over $\F$ that is hard on all but finitely many inputs for uniform arithmetic networks with $\pit$ gates of size $n^{k}$  over $\F$. Then, there is a uniform family of arithmetic networks of polynomial size solving $\pit$ over $\F$.
	\end{theorem}
	
	Our second result is a tight converse for~\Cref{thm:int:main:zero}. To contextualize this, recall that previously known results reflect a trade-off: When considering suboptimal $\pit$ algorithms (e.g., running in super-polynomial time, or working only on average), known results rely on hardness assumptions that are necessary or close to being so (see, e.g.,~\cite{KI04,CIKK15}); whereas, when considering polynomial-time $\pit$ algorithms, known results rely on hardness assumptions that are natural, but not known to be close to necessary for $\pit$ (say, as in~\cite{GKSS22}). This trade-off between the running time of the $\pit$ algorithm and the tightness of the hardness assumption may seem reasonable, but our second result shows that it is, essentially, \emph{entirely avoidable}: The assumptions in~\Cref{thm:int:main:zero} are remarkably close to being necessary for the conclusion. %
	
	\begin{theorem} [$\pit$ implies hardness on all inputs; informal, see~\Cref{prop:pit:necessary}] \label{thm:int:pit:necessary}
		Let $\F$ be a field of characteristic zero. Assume that there is a uniform family of arithmetic networks of polynomial size solving $\pit$ over $\F$. Then, for every $k\in\N$, there is a family of uniform arithmetic networks over $\F$ of polynomial size that is hard on all but finitely many inputs for $n^k$-time-uniform arithmetic networks with $\pit$ gates of size and degree $n^k$ over $\F$.
	\end{theorem}
	
	For a visual depiction of the necessary and sufficient assumptions that we prove, focusing for simplicity on the case of characteristic zero, see Figure~\eqref{fig:main}. We stress that neither of the results above refers to circuit lower bounds. Indeed, our algorithms are \emph{not} hitting-set generators, but rather \emph{targeted} hitting-set generators: The $\pit$ algorithm makes essential use of the description of the arithmetic circuit given to it.

	\begin{figure}[t]
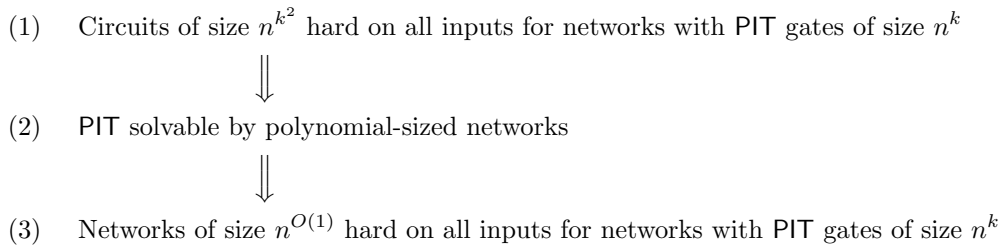

		\mm{
			(1)\;\;\;\;\;&\text{Circuits of size $n^{k^2}$ hard on all inputs for networks with $\pit$ gates of size $n^k$}\\
			&\;\;\;\;\;\;\;\;\;\;\;\;\;\;\;\;\;\;\;\;\;\;\;\;\Big\Downarrow\\
			(2)\;\;\;\;\;&\text{$\pit$ solvable by polynomial-sized networks}\\
			&\;\;\;\;\;\;\;\;\;\;\;\;\;\;\;\;\;\;\;\;\;\;\;\;\Big\Downarrow\\
			(3)\;\;\;\;\;&\text{Networks of size $n^{O(1)}$ hard on all inputs for networks with $\pit$ gates of size $n^k$}
		}	
		\caption{A visual depiction of~\Cref{thm:int:main:zero,thm:int:pit:necessary}. For simplicity of presentation, we omitted the repeated qualifier ``on all {\it but finitely many} inputs'', as well as the precise notions of uniformity and the degree bounds (the latter two cause further gaps between the necessary assumption in $(1)$ and the sufficient one in $(3)$).} %
	\label{fig:main}
\end{figure}

\paragraph{A variation for finite fields.}
We also show very similar results over sufficiently large finite fields. In this setting, the lower bounds that we focus on are for a more general model of randomized arithmetic networks, wherein the network receives random elements as inputs, and needs to err with low probability whenever these elements are sampled from any large enough set (see~\Cref{sec:pre:randomized:inputs} for details).

Beyond that, there are two minor differences in the result. On the one hand, we allow the hard family to be computed in arbitrary polynomial size, rather than size $n^{k^2}$ (e.g., it can be computed in size $n^{2^{2^k}}$), so the hardness assumption is slightly more relaxed. On the other hand, the conclusion is slightly weaker: Instead of getting a strictly polynomial-time $\pit$ algorithm, we get an algorithm running in time $n^{\log^{(c)}(n)}$, where $\log^{(c)}$ is the $c$-iterated $\log$ function and $c>1$ can be an arbitrarily large constant. 

\begin{theorem} [hardness on all inputs implies $\pit$; informal, see~\Cref{thm:main:finite}] \label{thm:int:main}
	Let $\F=\set{\F_n}_{n\in\N}$ be a sequence of finite fields, where $n^{\omega(1)}\le |\F_n|\le2^{\poly(n)}$. Assume  that for a sufficiently large $k>1$ there is a strongly uniform family of arithmetic circuits of polynomial size and degree over $\F$ that is hard on all but finitely many inputs for uniform randomized arithmetic networks of size $n^k$ over $\F$. Then, for every $c>1$ there is a uniform family of arithmetic networks of size $n^{\log^{(c)}(n)}$ solving $\pit$ over $\F$.
\end{theorem}

Again,~\Cref{thm:int:main} is coupled with a converse direction, analogous to~\Cref{thm:int:pit:necessary}, but the converse direction only yields lower bounds for networks with $\pit$ gates (rather than for the more general model of randomized networks); see~\Cref{prop:pit:necessary} for details. We also show an alternative proof for a special case of~\Cref{thm:int:main}, wherein the field's characteristic is low, in Appendix~\ref{apdx:lowchar}.

\paragraph{On lower bounds for networks, and the role of degree.}
The statements of~\Cref{thm:int:main:zero,thm:int:main} do not mention any bound on the degree of networks for which we assume lower bounds. We believe that even in this form, both results are interesting. For context, the original results of Kabanets and Impagliazzo~\cite{KI04} do not spell out any degree bound in the hardness assumption (and over finite fields, the needed degree in their assumption is indeed high; see~\cite{And20} for an improvement). Moreover, the meaning of the notion of ``degree'' for networks is less straightforward than it is for circuits (see~\Cref{sec:pre:networks}).

Nevertheless, our technical results only rely on lower bounds for networks of relatively low degree. In~\Cref{thm:int:main:zero}, the degree is quasipolynomial in the input length $n$ (see~\Cref{thm:main:zero}). In~\Cref{thm:int:main}, the degree is quasipolynomial in $n$ and in the field size $q$; lower bounds for arithmetic circuits or networks of degree more than $q$ are less well studied, and we observe that over fields of small characteristic (say, at most $n$), such lower bounds imply lower bounds for low-depth Boolean circuits (see Appendix~\ref{apdx:lowchar}). As far as we are aware, such implication is not known for the general case studied in~\Cref{thm:int:main}.

\begin{remark} [on the formulation of $\pit$ algorithms] \label{rem:int:formulation}
	As the reader might have noticed, all of our results refer to solving $\pit$ by uniform arithmetic networks; that is, intuitively, we consider \emph{arithmetic} algorithms for $\pit$ (rather than, say, arbitrary Boolean algorithms for $\pit$). Also, recall that we presented the formulation of $\pit$ as a search problem (i.e., finding a non-root). This choice is immaterial for our results: Our $\pit$ algorithms in~\Cref{thm:int:main:zero,thm:int:main} solve the search version, whereas the lower bound in~\Cref{thm:int:pit:necessary} follows from an algorithm solving the decision version.
\end{remark}

\paragraph{Organization of the rest of the paper:} In~\Cref{sec:tech}, we present a high-level overview of the proofs. \Cref{sec:pre} contains several standard facts in algebraic complexity, as well as the formal definitions and properties of  uniform, randomized, and universal arithmetic networks and circuits, which we make use of repeatedly. \Cref{sec:decomposition} presents an auxiliary technical notion that will be common to our proofs, namely, polynomial decompositions (\`a la~\cite{CT21}). \Cref{sec:gkss} and \Cref{sec:ki} then discuss the construction and correctness of our GKSS-based and KI-based targeted hitting set generators, respectively. The proofs of the main results appear in \Cref{sec:pit}. We conclude by mentioning several open problems in \Cref{sec:open}.

\section{Technical overview} \label{sec:tech}

Recall that we are given as input the description of an arithmetic circuit $C(x_1,\ldots,x_n)$ of size $\poly(n)$ computing a nonzero polynomial of degree $d \le \poly(n)$, and we want to find a point $\vec{\alpha} \in \F^n$ such that $C(\vec{\alpha}) \neq 0$.
The trivial deterministic algorithm for $\pit$ evaluates the circuit $C$ on an $n$-dimensional grid of side length $d+1$; the Schwartz--Zippel lemma guarantees that $C$ will evaluate to a nonzero value at some point on this grid.
This leads to an algorithm for $\pit$ that runs in time $(d+1)^n \cdot \poly(n)$.

A standard approach to obtain a faster algorithm for $\pit$ is to reduce the given $n$-variate instance $C(x_1,\ldots,x_n)$ to an $\ell$-variate instance $C'(y_1,\ldots,y_\ell)$ of degree $d' \le \poly(n)$.
If this reduction is efficient, we improve on the previous running time of $n^{O(n)}$ to $n^{O(\ell)}$.
Such a reduction follows from an explicit construction of a \emph{hitting-set generator}, which is a polynomial map $\calG : \F^\ell \to \F^n$ that satisfies $(C \circ \calG)(y_1,\ldots,y_\ell) \neq 0$ whenever $C$ is sufficiently small circuit. Many prior works on hardness-versus-randomness in algebraic complexity give explicit constructions of hitting set generators to derandomize $\pit$ (see~\cite{KS19}).%

As mentioned in~\Cref{sec:int:bg}, our $\pit$ algorithms are inspired by recent developments in Boolean hardness-vs-randomness (see~\cite{CT23b}). Following Goldreich~\cite{gol11} and Chen and Tell~\cite{CT21}, our $\pit$ algorithms will use a \emph{targeted hitting-set generator}, which (in the current arithmetic setting) is a procedure $\calG$ mapping a description of $C$ to a polynomial map $\calG_C(y_1,...,y_{\ell})$ such that $C'(y_1,...,y_{\ell})=C\circ \calG_C$ is nonzero. Note that in contrast to standard hitting-set generators, which are polynomial maps $\calG$ that work for all small circuits, here $\calG_C$ depends on $C$. Targeted hitting-set generators are weaker objects than standard hitting-set generators, and the advantage in working with them is that we will only need lower bounds against \emph{uniform} models (whereas standard hitting-set generators necessitate lower bounds for non-uniform circuits).

\paragraph{The motivating idea.}
The starting point of our work is the observation that the arithmetic setting shows promise for a way to bypass what is one of the main current obstacles in the Boolean setting. Specifically, recall that~\cite{CT21} proved the following result: If there is $f\colon\zo^*\rightarrow\zo^*$ computable by logspace-uniform circuits of size $\poly(n)$ and depth $n^2$, but hard on all but finitely many inputs for probabilistic time $n^c$ (for a sufficiently large $c \in \N$), then $pr\BPP=pr\P$. Note the restriction of $n^2$ on the depth of circuits for $f$: This particular restriction is the main gap between the foregoing assumption and a necessary assumption for $pr\BPP=pr\P$, and it is conjectured to be redundant (see, e.g.,~\cite{CTW23}).

The key point is that in the arithmetic setting we should not have this gap, due to the depth reduction procedure of \textcite{VSBR83}, a foundational result in arithmetic circuit complexity.
They showed that any degree-$d$ polynomial computed by a size-$s$ circuit can also be computed by a circuit of size $\poly(s,d)$ and depth $O(\log(s) \log(d))$.
In other words, when we restrict attention to low-degree polynomials---as is commonly done in algebraic complexity---small circuits can be assumed to be low-depth without loss of generality.\footnote{Of course, since we work with uniform circuits, one needs to verify that this depth reduction can be performed in a sufficiently uniform manner. This is indeed the case, as the standard proof of the depth reduction can be carried out in a uniform manner. The same concern regarding uniformity applies any time we invoke a result from arithmetic circuit complexity. For all results we use, known proofs suffice to obtain sufficiently uniform implementations.}
This suggests that by working in the arithmetic setting, we can close the gap between necessary and sufficient conditions for derandomization of $\pit$.

This intuition turns out to be, essentially, correct, and we are able to remove the depth upper bound in the sufficient condition. However, unfortunately, most of the other parts of the Boolean constructions do not migrate cleanly to the arithmetic setting, in which we want a targeted generator implementable as an arithmetic procedure, and we want to base its correctness on hardness for arithmetic networks. 

Thus, to prove our results we introduce new constructions of arithmetic targeted hitting-set generators, which rely on specific arithmetic tools (e.g., the hitting-set generators of \textcite{KI04,GKSS22} and bootstrapping ideas of~\cite{AGS18,KST19}) as well as on refined analyses of these tools.

\paragraph{Organization.}
We first start with the easier case of fields of characteristic zero, described in~\Cref{sec:tech:gkss}. This part should help introduce readers who are less familiar with the recent constructions in Boolean hardness vs randomness to the high-level outline of the arguments. Then we move to the more challenging case of finite fields, described in~\Cref{sec:tech:ki}.%

\subsection{Warming up: Fields with characteristic zero} \label{sec:tech:gkss}

For the sake of simplicity, let us now work over the field $\Q$ of rational numbers.

\subsubsection{Hardness on all but finitely many inputs is necessary}

The proof of \Cref{thm:int:pit:necessary} is a straightforward diagonalization argument.
Recall that we assume that $\pit$ can be solved by a  $\P$-uniform family of arithmetic networks of polynomial size, and our goal is to find a $\P$-uniform family of uniform arithmetic networks that is hard on all but finitely many inputs for $n^k$-time-uniform arithmetic networks with $\pit$ gates.
To do this, on input length $n$ we diagonalize against the first $n$ families of $n^k$-time-uniform arithmetic networks with $\pit$ gates.

On input $\vec{\Lambda} \in \Q^n$, we want to compute an output $\vec{\sigma} \in \Q^{n}$ such that $\sigma_i$ disagrees with the $i\ts{th}$ output of the $i\ts{th}$ family of uniform arithmetic networks with $\pit$ gates on $\vec{\Lambda}$, for all $i\in[n]$. 
To do this, we simulate each of the $n$ machines for time $n^k$, and evaluate each of the $n$ corresponding arithmetic networks with $\pit$ gates on $\vec{\Lambda}$, using the assumption that we can solve $\pit$ deterministically in order to simulate the $\pit$ gates. We then obtain $\sigma$ by taking $\sigma_i$ to be some trivial modification of the $i\ts{th}$ output element of the $i\ts{th}$ network (e.g., adding $1$ to it). 
This diagonalization procedure yields a uniform family of arithmetic networks whose outputs cannot be computed on \emph{any} input $\vec{\Lambda}$ (except, at most, finitely many) by smaller uniform arithmetic networks with $\pit$ gates.\footnote{The diagonalization implementation is an arithmetic network (rather than an arithmetic circuit) since the assumed $\pit$ algorithm is a uniform arithmetic network, and not a uniform arithmetic circuit.}%

\subsubsection{Hardness on all but finitely many inputs is sufficient}

The proof of \Cref{thm:int:main:zero} is more intricate and, together with \Cref{thm:int:main}, constitutes the technical bulk of our work.
By assumption, we have a strongly uniform family of arithmetic circuits $\set{C_n}_{n \in \N}$ of size and degree $n^{k^2}$ that is hard on all but finitely many inputs for uniform arithmetic networks with $\pit$ gates of size $n^k$, and our goal is to use this hard family of circuits to design a uniform family of deterministic arithmetic networks that solve $\pit$.
Thanks to depth reduction for arithmetic circuits \cite{VSBR83}, we may assume that the circuits $\set{C_n}$ are of depth $\Delta \coloneqq O(\log^2 n)$.

Let $D$ be the $n$-variate arithmetic circuit for which we want to solve $\pit$.
By using the universal arithmetic circuits of \textcite{Raz10}, we can regard the description of $D$ as a point $\vec{\Lambda} \in \Q^{n}$.
Following~\cite{CT21}, we will construct a generator $\calG_{\vec{\Lambda}} : \Q^{1000} \to \Q^n$ so that if $D \circ \calG_{\vec{\Lambda}} = 0$, then we can compute the hard function $C_n$ at the point $\vec{\Lambda}\in\Q^n$ with a uniform arithmetic network with $\pit$ gates of bounded polynomial size.
Because $C_n$ is hard for such networks on all but finitely many inputs $\vec{\Lambda}$, our generator succeeds in hitting all but finitely many circuits $D$ (i.e., for all but finitely many $\vec{\Lambda}$, the generator $\calG_{\vec{\Lambda}}$ hits the circuit represented by $\vec{\Lambda}$).

\paragraph{The generator.}
The generator $\calG_{\vec{\Lambda}}$ is built from a \emph{polynomial decomposition} of the computation of an arithmetic circuit $C_n$ on input $\vec{\Lambda}$. Evaluating $C_n$ at $\vec{\Lambda}$ assigns each gate of the circuit a rational number. If we view $C_n$ as a layered circuit, then the $i\ts{th}$ layer of $C_n$ corresponds to a vector $\vec{v}_i \in \Q^{n^{k^2}}$, which we think of as specifying a ``hard'' polynomial $f_i$. Indeed, we can simply interpolate $f_i$ from these points, or treat them as the coefficients of $f_i$. However, we would like the resulting sequence of polynomials $\set{f_1,...,f_{\Delta}}$ to have the additional property that they are {\sf downward self-reducible}: We can evaluate $f_i$ quickly (i.e., in time much smaller than $n^{k^2}$) if we have access to $f_{i-1}$. 
Using an idea inspired by the proof system of Goldwasser, Kalai, and Rothblum~\cite{GKR15} (which, in turn, uses the sumcheck protocol), we can encode each layer $\vec{v_i}$ of $C_n(\vec{\Lambda})$ by $\ell\le\log(n)$ of polynomials $f_{i,1}, \ldots, f_{i,\ell}$ such that the resulting sequence of polynomials $\set{f_{i,j}}_{i\in[\Delta],j\in[\ell]}$ is, indeed, downward reducible: the value of $f_{i,j}$ at any point can be learned by querying $f_{i,j-1}$ (or $f_{i-1,\ell}$, if $j=0$) at $O(n^{\eps})$ many related points, where $\eps > 0$ is a small enough constant. Moreover, in this sequence $\set{f_{i,j}}$ the polynomial computing the input layer is (unconditionally) efficiently computable. The polynomial decomposition relies on $\set{C_n}$ being strongly uniform, and its description appears in~\Cref{sec:ki:decomposition}.

The targeted generator gets input $\vec{\Lambda}$, and computes the polynomial decomposition of $C_n(\vec{\Lambda})$. It the instantiates the (standard) hitting-set generator $G_{\sf GKSS}$ of \textcite{GKSS22} with each of the $f_{i,j}$'s as the hard polynomial. The final generator construction is
\mm{
\calG(\vec{\Lambda},y_1,y_2,\vec{z}) = \sum_{i,j}L_{i,j}(y_1,y_2)\cdot G_{\sf GKSS}^{f_{i,j}}(\vec{z}) \;\;, \tag{2.1}\label{eq:tech:tarhsg}
}
where the $L_{i,j}$'s are standard Lagrange interpolation polynomials computing the indicator function of $(i,j)$ in $[\Delta]\times[\ell]$. Note that the notation $G_{\sf GKSS}^{f_{i,j}}(\vec{z})$ suppresses $\vec{\Lambda}$, but $f_{i,j}$ depends on the gate-values of $C_n(\vec{\Lambda})$.

We instantiate the polynomial decomposition so that the hard polynomials $f_{i,j}$ are $m$-variate polynomials $\Q^{m}\rightarrow\Q$ where $m=O(1)$, in which case the seed length of $G_{\sf GKSS}$ is constant; that is, we have $G_{\sf GKSS}\colon\Q^{O(1)}\rightarrow \Q$. The $\pit$ algorithm evaluates $D$ (i.e., the $n$-variate arithmetic circuit represented by $\vec{\Lambda}$) on a grid of constant dimension and polynomial side length.

\paragraph{Correctness: An arithmetic reconstruction procedure.}

To prove the correctness of $G_{\vec{\Lambda}}$, we need to show that if a small circuit $D$ with description $\vec{\Lambda}$ satisfies $D \circ G_{\vec{\Lambda}} = 0$, then we can compute the mapping $\vec{\Lambda}\mapsto C_n(\vec{\Lambda})$ by a small uniform arithmetic network with $\pit$ gates.

Since $G_{\vec{\Lambda}}$ is the combination of the GKSS generator instantiated with the polynomials $f_{i,j}$, we have that $D \circ G_{\mathsf{GKSS}}^{ f_{i,j}} = 0$ for each $f_{i,j}$. The GKSS generator is a reconstructive generator, meaning that using the distinguisher $D$ and a bounded number of queries to $f_{i,j}$, one can reconstruct a circuit for the polynomial $f_{i,j}$ (see more on this below). As shown in \cite{GKSS22}, the complexity of this reconstructed circuit largely depends on the complexity of the distinguisher $D$. In particular, if $D$ is a small arithmetic circuit, then we obtain a correspondingly-small arithmetic circuit for $f_{i,j}$. This suggests an obvious roadmap: as we have a small circuit $D$ that satisfies $D \circ G_{\mathsf{GKSS}}^{f_{i,j}} = 0$ for each $f_{i,j}$, apply the GKSS reconstruction algorithm to iteratively reconstruct circuits for each of the polynomials $f_{i,j}$, for $(i,j)=(1,1),...,(1,\ell),(2,1),...,(\Delta,\ell)$; each time we need to answer queries of the GKSS reconstruction to $f_{i,j}$, use the downward self-reducibility as well as our circuit for $f_{i,j-1}$ (or for $f_{i-1,\ell}$). If the circuit $D$ has size $n'$ and each round of reconstruction can be implemented in time $(n')^c$ for some fixed constant $c$, then by iteratively reconstructing circuits for the polynomials $f_{i,j}$, we obtain in time $(n')^c\cdot\Delta$ a circuit for $f_{\Delta, O(1)}$, whose evaluations (at certain points determined by the polynomial decomposition) corresponds to the output of the circuit $C_n$ on input $\vec{\Lambda}$.

We stress that downward self-reducibility is not used as part of the circuit for $f_{i,j}$, since that would cause an exponential blow-up in the circuit size as we iteratively construct circuits for increasing $(i,j)$'s. Instead, downward self-reducibility is only used to answer queries of the GKSS reconstruction to $f_{i,j}$, and the GKSS reconstruction outputs a circuit of fixed polynomial size (depending on $D$) for $f_{i,j}$. 

\medskip\noindent\underline{Technical details for the reconstruction procedures.} Now, to be more precise, we need a uniform arithmetic network with $\pit$ gates that iteratively performs the reconstruction argument of GKSS, producing a description of a circuit for $f_{i,j}$ at each step. In particular, we need to implement the reconstruction procedure for the GKSS generator by a uniform arithmetic network with $\pit$ gates (that prints an arithmetic circuit).

Let us point out key parts of this implementation, deferring %
the full details to~\Cref{sec:gkss-recon-step}. 
In the original proof of \cite{GKSS22}, the reconstruction algorithm for the generator $G_{\textsf{GKSS}}^{f}$ first constructs a circuit for the partial derivatives of order up to $n$ of the homogeneous components of $f$ up to degree $n$.
Because $f$ is a polynomial on a constant number of variables, there are $n^{O(1)}$ such polynomials to compute, each of which is a sum of $n^{O(1)}$ monomials. 
After this base case, the reconstruction algorithm of \cite{GKSS22} uses the distinguisher circuit $D$ to compute all partial derivatives of order up to $n$ of the homogeneous components of $f$, proceeding iteratively one degree at a time.

To obtain the (partial derivatives of) the degree-$d$ homogeneous component of $f$, the reconstruction algorithm of \cite{GKSS22} first obtains evaluations of the order-$n$ partial derivatives of this component.
The algorithm then interpolates the order-$n$ derivatives from these evaluations, and then repeatedly applies Euler's formula (\Cref{fact:euler}) to reconstruct the degree-$d$ homogeneous component of $f$.
Crucially, we cannot evaluate the order-$n$ partials of the degree-$d$ component of $f$ on any points of our choosing, but only on points where a certain polynomial does not vanish.
Very roughly, these evaluations are obtained by solving a univariate polynomial equation to first order.
If $g(x)$ is a univariate polynomial, then we can Taylor expand $g$ at $a$ as
\[
g(x) = g(a) + (x-a) g'(a) + O((x-a)^2).
\]
If we want to find a point where $g(x) = 0$, then we can obtain an approximate root by solving
\[
0 = g(x) = g(a) + (x-a) g'(a) + O((x-a)^2).
\]
To first order, we have
\[
x = a - \frac{g(a)}{g'(a)}.
\]
For the solution $x$ to be well-defined, we need to ensure that $a$ is chosen so that $g'(a) \neq 0$.

We obtain the order-$n$ partials of the degree-$d$ component of $f$ in a similar manner.
To ensure that the evaluations we obtain can be used to reconstruct the partial derivatives themselves, we need to find a set of points where (i) a first-order partial derivative of some related polynomial does not vanish, and (ii) these points are an interpolating set for degree-$d$ polynomials in $n$ variables.
We can write down an explicit arithmetic circuit that computes a nonzero polynomial that vanishes on sets of points that do not satisfy both of these conditions.
By using $\pit$ gates within the reconstruction, we can explicitly construct an interpolating set that allows us to recover the partial derivatives of higher-degree homogeneous components of $f$.
With a good interpolating set in hand, this iterative part of the reconstruction can be implemented by a uniform network in a fairly straightforward manner.

Perhaps surprisingly, the starting point of each execution of the GKSS reconstruction is more difficult to implement than the iterative part sketched above.
In the setting of \cite{GKSS22}, they simply hard-wire the partial derivatives up to order $n$ of the relevant homogeneous components to the reconstructed circuit (since they do not care about uniformity). In our setting, we could compute these by polynomial interpolation (i.e., we could query $f_{i,j}$ to learn the coefficients of its low-degree homogeneous components, using downward self-reducibility), but interpolating the entire polynomial $f_{i,j}$ would be too expensive. The key idea to avoid this is to construct a circuit $C^{\sf temp}_{i,j}$ for $f_{i,j}$, which uses downward-self-reducibility and the circuit for $f_{i,j-1}$, and then homogenize $C^{\sf temp}_{i,j}$ up to degree $n$. Indeed, the circuit $C^{\sf temp}_{i,j}$ is too big for us to simply use it as the circuit for $f_{i,j}$ in the next iteration (recall that, as explained above, doing so in each iteration would cause an exponential blow-up), but we will only be using $C^{\sf temp}_{i,j}$ temporarily to extract information---the coefficients of its low-degree homogeneous components---when implementing the GKSS reconstruction, which outputs a circuit of fixed polynomial size for $f_{i,j}$.

Finally, the number of partial derivatives is bounded, but it is unfortunately not independent of the number of points used to define $f_{i,j}$; in other words, their number does depend on the size of the circuit $C_n$. This is why in our GKSS-based reconstruction, we reconstruct circuits of size $n^{K}$ by arithmetic networks with $\pit$ gates of size $n^{\sqrt{K}}$ (indeed, the result was stated in~\Cref{sec:int} using $K=k^2$). For further details, see~\Cref{sec:gkss}.

\subsection{$\pit$ over finite fields} \label{sec:tech:ki}

\newcommand{\ki}{\ensuremath{\mathsf{KI}}}

Turning to the case of finite fields of super-polynomial size, we now use the $\ki$ generator~\cite{KI04} instead of the GKSS generator. That is, for each polynomial $f_{i,j}$ encoding a row in the polynomial decomposition, we instantiate the $\ki$ generator with $f_{i,j}$. The final generator will combine the $f_i$'s similarly to Eq.~\ref{eq:tech:tarhsg}, and we will implement a reconstruction procedure by repeatedly invoking a uniform implementation of the reconstruction procedure of $\ki$ (i.e., analogously to~\Cref{sec:tech:gkss}, but working with a uniform reconstruction of $\ki$ instead of GKSS, and making the necessary adjustments).

However, there is a problem with this idea. For now, let us focus even just on $\ki$ with a single polynomial $f=f_{i,j}$. In fact, let us even ignore the fact that we added $\Delta\cdot\ell$ auxiliary polynomials, and just pretend that each polynomial $f$ corresponds to a row of gate-values in $C_n(\vec{\Lambda})$. 

An inherent limitation of $\ki$ is that it only yields $\pit$ algorithms running in \emph{quasipolynomial} time, whereas we are trying to get in $\pit$ algorithms running in polynomial time (or close to it). One reason why this limitation arises is because the $\ki$ generator needs to output $n$ elements (for the $n$-input circuit on which we want to solve $\pit$). Since it uses combinatorial designs \`{a} la~\cite{NisanW94}, it must rely on a hard polynomial over $\ell=O(\log n)$ variables, in which case the number of input elements to the generator $\ki$ is at least $O(\ell)$. Evaluating the given circuit of degree $d=\poly(n)$ on a grid of side length $d+1$ and dimension $O(\ell)$, we get a quasipolynomial-time $\pit$. (In fact, there is another reason that using $\ki$ only yields a quasipolynomial-time $\pit$ algorithm, which we explain below; for now, let us focus on the current challenge.)

Bootstrapping techniques from~\cite{AGS18,KST19} suggest a way out of this conundrum, but at a cost of a stronger hardness assumption. Specifically, recall that we started with a circuit $D_0$ that has $n_0=n$ inputs, and used an $(\ell_0=\ell)$-variate polynomial $f^{(0)}=f$ to reduce the effective number of inputs of $D_0$ to be $O(\ell_0)$; that is, we constructed a nonzero $O(\ell_0)$-variate circuit $D_1=D_0\circ \ki^{f^{(0)}}$ of polynomial size and degree. Of course, if we would have more hard polynomials with different parameters, we would be able to repeat this procedure: At each iteration $i\ge1$ we use a hard polynomial over $\ell_i=O(\log(\ell_{i-1}))$ variables, and obtain $D_i=D_{i-1}\circ \ki^{f^{(i)}}$ over $O(\ell_i)$ variables, terminating after the number of variables is sufficiently small. The problem is that this bootstrapping procedure requires a \emph{sequence of hard polynomials}, each over fewer and fewer variables; we do not wish to rely on this seemingly stronger assumption.

\paragraph{A free lunch setting for bootstrapping.}
Our key idea for avoiding this involves a change of perspective. Instead of thinking of $f$ as a hard polynomial, let us just think of the gate-values of $C_n(\vec{\Lambda})$ at layer $i$ as a \emph{hard sequence of field elements}. That is, this sequence should be hard to compute quickly, otherwise (i.e., if all rows in the polynomial decomposition are ``easy'') a reconstruction procedure as in~\Cref{sec:tech:gkss} computes $C_n(\vec{\Lambda})$ too quickly. The advantage in this perspective is that now we are allowed to define a hard polynomial from this sequence using \emph{an arithmetization of our choosing}; that is, the string specifies a sequence of values, and we are allowed to interpolate an $\ell$-variate polynomial from this sequence for a value of $\ell$ of our choosing.

We stress that arithmetizing each layer differently does not affect the hardness assumption. The assumption is still that $C_n(\vec{\Lambda})$ is hard to compute quickly. It just so happens that when we arithmetize a layer with a value of $\ell$ that we choose, it specifies an $\ell$-variate polynomial. The ability to choose $\ell$ is thus, effectively, free of cost in terms of our hardness assumption (in contrast to the setting of~\cite{AGS18,KST19}). (One subtlety is that the way we arithmetize must conform to the way the polynomial decomposition is defined. Let us ignore this complication for simplicity of presentation; it does not affect the argument.)

This suggests a natural approach. We compute the hard circuit $C_n(\vec{\Lambda})$ and use the $\ki$ generator to obtain a description $\vec{\Lambda}_1$ of $D_1$, which is defined as above, with a standard arithmetization of the layers (i.e., as $(O(\ell_0=\log n))$-variate polynomials). 
Now, at each iteration $i$, we arithmetize each layer as an $\ell_i$-variate polynomial to obtain $D_i$, which is the composition of $D_{i-1}$ with the resulting targeted hitting-set generator (obtained by using $\ki$ with the $\ell_i$-variate polynomials given by the layers). 
Thus, the $\pit$ algorithm involves a recursive sequence of computations of $C_n$, each time on a new circuit $\vec{\Lambda}_i$ representing $D_i$, and after each computation $i$ of $C_n$, we arithmetize the layers differently to obtain the targeted hitting-set generator. 
Since $C_n$ is hard to compute on \emph{all} inputs, each of the applications of the generator (i.e., on each $\vec{\Lambda}_i$) will result in a nonzero polynomial. After $c=O(1)$ iterations we obtain a polynomial with $\log^{(c)}(n)$ variables, which we can evaluate over a grid of the appropriate side length.

(We intentionally leave the precise details of how polynomials are arithmetized vague, since they depend on the implementation details of the polynomial decomposition. For full details, see~\Cref{sec:decomposition,sec:ki}.)

\paragraph{A subtle challenge when recursing: The targeted generator is not an arithmetic circuit.}
The reader might have noticed one subtle point in the description above: At each iteration $i$, we need $D_i=D_{i-1}\circ\calG_{\ki}$ to remain an arithmetic circuit of polynomial size and degree, where $\calG_{\ki}$ is our $\ki$-based targeted generator. This is because we will be feeding $D_i$ to $\calG_{\ki}$ again as the starting point of iteration $i+1$, and we want the next iteration maintain the property that $D_{i+1}$ is nonzero.\footnote{Specifically, if $D_i$ is not an arithmetic circuit of polynomial size, then our reconstruction procedure will not be an arithmetic network of polynomial size computing the hard problem. Thus, we will not be guaranteed that $\calG_{\ki}$ will work in iteration $i+1$ (i.e., that $D_{i+1}=D_i\circ\calG_{\ki}$ will be nonzero). This problem is not just syntactic: Recall that if $D_i$ is an arithmetic network, it may not even compute a polynomial, and thus designing a hitting-set generator for it seems to be a problem of a completely different nature.} This might seem like a triviality, given that we expect our generator to be computable arithmetically. However, unfortunately, our targeted generator is \emph{not an arithmetic circuit}; it is inherently an arithmetic network.

Surprisingly, this issue is manageable. To see how, recall that in iteration $i$ we are composing $D_{i-1}$ with $\calG_{\ki}$, when the latter targeted generator is applied to input $\vec{\Lambda}_{i-1}$ that is the description of $D_{i-1}$. The key observation is that we are not actually composing $D_{i-1}$ with a procedure that maps $\vec{\Lambda}_{i-1}$ to $\calG_{\ki}$ (this would be the arithmetic network); we are composing  $D_{i-1}$ with a procedure $\calG_{\ki,\vec{\Lambda}_{i-1}}$ that has $\vec{\Lambda}_{i-1}$ fixed and ``hard-wired'', and computes the outputs of $\calG_{\ki,\vec{\Lambda}_{i-1}}(\vec{s})=\calG_{\ki}(\vec{\Lambda}_{i-1},\vec{s})$ as a function of the new input variables $\vec{s}$ (i.e., of the seed of the targeted generator). In particular, we can construct a circuit for $\calG_{\ki,\vec{\Lambda}_{i-1}}$ that has the hard-to-compute values hard-wired, and only computes the (easy-to-compute) mapping of seed $\vec{s}$ to output $\calG_{\ki}(\vec{\Lambda}_{i-1},\vec{s})$. To see the implementation details, see~\Cref{sec:pit:sufficient:small}.

Lastly, since we are interested in a $\pit$ algorithm implementable by an arithmetic network, we show how to implement all of the above as an arithmetic network. This network gets $\vec{\Lambda}$, and iteratively produces descriptions of the circuits $D_i=D_{i-1}\circ\calG_{\vec{\Lambda}_{i-1}}$, each time using $C_n$ to compute the values of $\calG_{\ki}$ that it needs. The construction for a single recursive iteration (which is the crux of the proof) appears in~\Cref{sec:ki:generator}, and the recursive application appears in~\Cref{sec:pit:sufficient:small}.

\paragraph{The problematic first iteration.}
Finally, we return to the second obstacle involved with using the $\ki$ generator, which was mentioned above. The problem is that the reconstruction procedure of $\ki$ requires quasipolynomial circuit size when it works with an $O(\log n)$-variate polynomial (as it does in the first iteration). This is a problem, since if the reconstruction is of quasipolynomial size, our hard circuit must also be of (larger) quasipolynomial size, in which case we only obtain a quasipolynomial $\pit$ algorithm.

To see where the inefficiency comes from, recall that the reconstruction procedure involves constructing partial truth-tables of the hard polynomial $f$, by interpolating $f$ on a grid of bounded side length $d_f$ and dimension $\eps\cdot\log(n)$, where $d_f$ is the individual degree of $f$.\footnote{That is, the reconstruction is an arbitrarily small constant fraction of the size of sets in the combinatorial design (corresponding to the size of pairwise intersections of sets in the design). When working with an $\ell$-variate polynomial, the designs sets are of size $\ell$, and hence the pairwise intersections can be made of size $\eps\cdot\ell$.} The issue is that the individual degree of $f$ is not small enough, and in particular, in the first iteration we have $d_f\approx \polylog(n)$ and hence $d_f^{\eps\cdot\log(n)}=n^{O(\mathrm{loglog} n)}$. (In subsequent iterations we can handle this using an appropriate arithmetization, but in the first iteration we must use an $O(\log n)$-variate polynomial.)

To resolve this issue we change the polynomial decomposition, only for the first iteration of the $\ki$-based targeted generator, so that $f$ has constant individual degree (and thus $d_f^{\eps\cdot\log(n)}=n^{O(\eps)}$). In a gist, recall that the polynomials are defined based on ideas from the sumcheck protocol, and we note that they involve computing a certain low-degree formula corresponding to the circuit-structure function of $C_n$. We use an idea of Kalai, Lombardi, and Vaikuntanathan~\cite{KalaiLV23}, who apply a Cook-Levin-style trick to reduce the individual degree of this formula to be constant, while increasing the round complexity of the relevant sumcheck protocol. In our setting this will result in having a $\ki$ generator with $\polylog(n)$ input variables in the first iteration (rather than $O(\log n)$), but this overhead becomes immaterial after subsequent recursions.

\section{Preliminaries} \label{sec:pre} 

We use $\F$ to denote a field.
Often, we will need to work with a countably-infinite family of fields, such as the family $\set{\F_{2^{n}} : n \in \N}$ of finite fields of characteristic $2$.
To simplify notation, we will write such a family of fields as $\set{\F_n : n \in \N}$ with the understanding that \emph{$\F_n$ is not necessarily a finite field of order $n$}.
In some instances, we will only need to work over a single field, such as $\Q$; for notational consistency, we still write this as a family of fields $\set{\F_n : n \in \N}$ where $\F_n = \Q$ for all $n \in \N$. 

Whenever considering a field sequence $\F=\set{\F_n}$, we implicitly also associate with it a way of representing field elements as bit-strings.\footnote{For fields such as $\C$, we can either represent only a subset of field elements as bit-strings (e.g., the rational ones), or work in a more generalized real-word-RAM-like model. Our results are not sensitive to this choice: see~\Cref{rem:pre:uni,rem:pre:nequni}.} We say that $\F=\set{\F_n}$ is {\sf feasible} if there is a uniform Turing machine that gets as input $(1^n,1^m)$ where $m\le|\F_n|$, runs in time $\poly(n+m)$, and outputs $m$ distinct elements in $\F_n$. Most natural field sequences are feasible (when considering the straightforward way of representing field elements as bit-strings). For example, if $\F_n=\F_{p(n)}$ is a finite field of prime order, a uniform machine can simply print the first $m$ integers (i.e., without having to know the prime $p(n)$); and similarly, if $\F_n=\F_{p^r}$, then a uniform machine can print $m$ distinct vectors in $[p]^r$ (i.e., without having to know an irreducible of degree $r$ over $\F_p$). When $\F_n$ is of characteristic zero, the machine can also just print $m$ integers.

We abbreviate a tuple $(x_1,\ldots,x_n)$ as $\vec{x} \coloneqq (x_1, \ldots, x_n)$; in all instances, the length of the tuple will either be explicitly stated or will be clear from context.
For an exponent vector $\vec{e} = (e_1,\ldots, e_n) \in \Z_{\geq 0}^n$, we write $|\vec{e}| \coloneqq e_1 + \cdots + e_n$ for the $1$-norm of $\vec{e}$ and use $\vec{x}^{\vec{e}}$ to denote the monomial $x_1^{e_1}\cdots x_n^{e_n}$.
We write $\F[x_1,\ldots,x_n]$ and $\F(x_1,\ldots,x_n)$ for the polynomial ring and field of rational functions, respectively, over the variables $x_1,\ldots,x_n$ with coefficients in $\F$.
Given a polynomial $P(\vec{x}) \in \F[x_1,\ldots,x_n]$, we write $\deg(P)$ and $\ideg(P)$ for the degree and individual degree of $P$, respectively.
We use $\partial_{\vec{x}^{\vec{e}}}(P(\vec{x}))$ to abbreviate the partial derivative $\frac{\partial^{\vec{e}}}{\partial \vec{x}^{\vec{e}}}(P(x_1,\ldots,x_n))$.
We use $\langle \vec{x} \rangle^i$ to denote the ideal in the polynomial ring $\F[x_1,\ldots, x_n]$ that is generated by all degree-$i$ monomials in the variables $\vec{x}$.

We say that a function $f\colon\mathcal{D}\rightarrow\mathcal{R}$ \emphdef{represents} a string $x_f\in\mathcal{R}^{\mathcal{D}}$ if $f(i)=(x_f)_i$ for all $i\in\mathcal{D}$.

\subsection{Standard results}

Here, we quote the Schwartz--Zippel Lemma and Gauss's Lemma, two standard results that we make use of throughout our work.
We begin with the Schwartz--Zippel Lemma.

\begin{lem}[Schwartz--Zippel Lemma, total degree \cite{Schwartz80}]\label{lem:sz}
Let $\F$ be a field and let $f \in \F[\vec{x}]$ be a nonzero polynomial.
Then for any finite set $S \subseteq \F$, we have
\[
\Pr_{\vec{\alpha} \gets S^n} \sbr{f(\vec{\alpha}) = 0} \le \frac{\deg(f)}{|S|}.
\]
\end{lem}

\begin{lem}[Schwartz--Zippel Lemma, individual degree \cite{Zippel79}]\label{lem:sz-ideg}
Let $\F$ be a field and let $f \in \F[x_1,\ldots,x_n]$ be a nonzero polynomial.
Then for any finite set $S \subseteq \F$, we have
\[
\Pr_{\vec{\alpha} \gets S^n} \sbr{f(\vec{\alpha}) = 0} \le 1 - \del{1 - \frac{\ideg(f)}{|S|}}^n.
\] 
Let $f$ be a nonzero $n$-variate polynomial of individual degree at most $d$ over a field $\F$. Then for any set $S \subseteq \F$ with $|S| > d$, there is a point $\vec{a} \in S^n$ such that $f(\vec{a}) \neq 0$.
\end{lem}

Next, we recall the statement of Gauss's Lemma, which allows us to relate the factorizations of a polynomial $f \in \F[\vec{x}, y]$ over the two rings $\F(\vec{x})[y]$ and $\F[\vec{x}][y]$.

\begin{lem}[{Gauss's Lemma \cite[Corollary 6.10]{vzGG13}}]\label{lem:gauss}
Let $f \in \F[\vec{x}, y]$ be monic in $y$.
Then $f$ is irreducible in $\F(\vec{x})[y]$ if and only if $f$ is irreducible in $\F[\vec{x}][y] \cong \F[\vec{x}, y]$.
\end{lem}

Gauss's Lemma implies the following corollary, which we frequently use.

\begin{cor} \label{cor:gauss}
Let $f \in \F[\vec{x}, y]$ and let $g \in \F[\vec{x}]$.
Suppose that $f(\vec{x}, g(\vec{x})) = 0$.
Then $y - g(\vec{x})$ is an irreducible factor of $f(\vec{x}, y)$ in $\F[\vec{x}, y]$.
\end{cor}

\subsection{Arithmetic circuits and networks} \label{sec:pre:networks}

In this subsection, we recall the definitions of arithmetic circuits and arithmetic networks.
We start with arithmetic circuits.

\begin{definition} \label{def:algebraic circuit}
Let $\F$ be a field.
An \emphdef{arithmetic circuit} over $\F$ is a directed acyclic graph whose internal gates are labeled as addition or multiplication gates, whose input gates are labeled with variables $x_i$ or constants from $\F$, and which computes a polynomial in the natural way.
We say that an arithmetic circuit is \emphdef{homogeneous} if every gate of the circuit computes a homogeneous polynomial.
Additionally, we assume that gates labeled by field constants appear only in the bottom-most layer of an arithmetic circuit.\footnote{This simplifying assumption does not meaningfully lose generality, since constants in the bottom layer can be used in upper layers by propagating a constant element $\sigma$ using the operations $\sigma\times 1$ and $\sigma+0$.}
\end{definition}

We assume throughout the paper that all arithmetic circuits are alternating and have fan-in two. 
In addition to arithmetic circuits, we also need the notion of an \emph{arithmetic network}, first defined by \textcite{vzGathen_survey}.
Arithmetic networks endow arithmetic circuits with the ability to test if a computed value is zero and branch accordingly.

\begin{definition} [arithmetic networks] \label{def:arith-network}
Let $\F$ be a field.
An \emphdef{arithmetic network} over $\F$ is a directed acyclic graph together with the following additional data.
\begin{itemize}
	\item
	Internal gates of the network are labeled with either an arithmetic operation from $\set{+, -, \times, \div}$, a Boolean operation from $\set{\land, \lor, \neg}$, a test operation $\testequals 0$, or a selection operation $\select$.
	\item
	The input gates are labeled by variables $x_i$ that take values in $\F$, variables $y_j$ that take Boolean values in $\set{0,1}$, and by field constants $\alpha \in \F$.
	\item
	The gates of the circuit compute as follows.
	\begin{itemize}
		\item
		Arithmetic gates $\set{+,-,\times,\div}$ receive two arithmetic values as input and output the corresponding function of the inputs.
		If division by zero occurs, the output of the gate is undefined.
		\item
		Boolean gates $\set{\land, \lor, \neg}$ receive two Boolean values as input and output the corresponding function of the inputs.
		\item
		The test gate $\testequals 0$ receives an arithmetic value $\alpha \in \F$ as input and outputs a Boolean value that indicates whether $\alpha$ is zero or nonzero.
		\item
		The selection gate $\select$ receives one Boolean input $b$ and two arithmetic inputs $\alpha, \beta \in \F$ and computes the function
		\[
		\select(b, \alpha, \beta) \coloneqq
		\begin{cases}
			\alpha & \text{if $b = 0$} \\
			\beta & \text{if $b = 1$.}
		\end{cases}
		\]
	\end{itemize}
\end{itemize}
An arithmetic network computes a function $\F^n \times \set{0,1}^m \to \F^\nu \times \set{0,1}^\mu$.
The \emphdef{size} of a network is the number of gates in the network.
The \emphdef{depth} of a network is the length of the longest path from an input gate to the output gate.

When the Boolean inputs to the network and the outputs of the $\select$ gates are fixed, the arithmetic gates of the network compute a rational function $f / g \in \F(x_1,\ldots,x_n)$ of the arithmetic inputs.
The degree of this rational function $f/g$ is given by $\deg(f/g) \coloneqq \deg(f) + \deg(g)$.
The \emphdef{degree} of an arithmetic network is the maximum degree of any rational function computed by an arithmetic gate, where the maximum is taken over all fixings of the Boolean inputs and the outputs of the $\select$ gates in the network.
\end{definition}

As is typical in arithmetic circuit complexity, we will restrict our attention to networks of low degree, typically degree $\poly(n)$ where $n$ is the number of arithmetic inputs to the network.

Arithmetic networks clearly generalize arithmetic circuits, as networks can compute strictly more functions than circuits.
For polynomial functions, which are computable by both circuits and networks, the following remark shows that networks in fact provide no advantage over circuits, at least in the non-uniform setting.

\begin{remark} \label{rem:network vs circuit}
Let $\F$ be an algebraically closed field and let $f(\vec{x}) \in \F[\vec{x}]$ be a polynomial.
If there is an arithmetic network of size $s$ that computes the function $\F^n \to \F$ given by $\vec{\alpha} \mapsto f(\vec{\alpha})$, then there is an arithmetic \emph{circuit} of size $s$ that computes $f(\vec{x})$.
To see this, take the generic path through the network: assume all comparisons with zero return ``not equal to zero,'' unless the value being compared is identically zero.
This fixes all Boolean values appearing in the network; under this fixing, the network behaves as an arithmetic circuit instead.
The resulting arithmetic circuit computes a polynomial $g(\vec{x})$ that agrees with $f$ on a nonempty Zariski-open subset $U \subseteq \F^n$.
It is a basic fact of algebraic geometry (see, e.g., \cite[Chapter 1, Section 3.2]{Shafarevich13}) that if two polynomials agree as functions on a nonempty Zariski-open subset of $\F^n$, then they are the same polynomial.
This implies that the circuit obtained from the generic path in the network correctly computes the polynomial $f(\vec{x})$ on all inputs.
\end{remark}

Unfortunately, this method of eliminating the Boolean part of an arithmetic network will not apply to our reconstruction algorithms.
In that setting, we compute a polynomial using a uniform randomized arithmetic network (a model of computation whose definition appears in \Cref{sec:pre:randomized}), and it is less clear if a randomized arithmetic network can be efficiently converted into an equivalent arithmetic circuit.

\subsection{Uniform arithmetic circuits and networks} \label{sec:pre:uniform}

In this section we define our notions of \emph{uniform} arithmetic circuits and networks.
We first explain some general considerations, then formally define two notions of uniformity, and then show that there are functions that are hard on all but finitely many inputs for uniform circuits.

\subsubsection{Background} \label{sec:pre:uni:bg}

The data describing an arithmetic network can be divided into a Boolean part, which encompasses the structure and labeling (without the field constants) of the directed acyclic graph, and an arithmetic part, corresponding to the particular field constants used in the arithmetic network.
To impose a uniformity constraint on arithmetic networks, we must specify a notion of uniformity for both the Boolean part and the arithmetic part of the network.

It is fairly clear how this impacts the Boolean part of an arithmetic network: the underlying dag and its labeling should be computable by a Turing machine within some resource bound. It is less clear how uniformity interacts with the arithmetic part of the network description.
There are several notions of uniformity for the arithmetic part of the network considered in the literature.
\begin{enumerate}
\item The most restrictive form of uniformity requires the network to be constant-free, i.e., only the constant $1 \in \F$ is used by the network.
\item A more relaxed notion is what \textcite{vzGathen_survey} calls $\neq$-uniformity: the network can make use of a list of pairwise distinct field elements $(\alpha_1 = 1, \alpha_2, \alpha_3, \ldots)$, but the network is required to compute the desired function regardless of the precise values of $\alpha_2, \alpha_3, \ldots$.
For example, one can implement polynomial interpolation in a $\neq$-uniform manner (since many distinct field elements are needed, but their precise values are not important).
\item For feasible fields, a more general notion is the simplest one: Constants are allowed in the circuit, but the Turing machine is required to print them in the circuit's description. Indeed, this calls for considering some natural notion of representing field elements.
\item A different generalization of $\neq$-uniformity suggested by~\textcite{vzGathen_survey} allows one to specify polynomial equations or inequations that must be satisfied by the constants appearing in the network, where the (in)equations specifying the $i$-th constant $\beta_i$ are computed by a uniform arithmetic circuit that is allowed to use the first $i-1$ constants $\beta_1, \ldots, \beta_{i-1}$ in its description.\footnote{Formally, the numerical values of the constants $\beta_1, \ldots, \beta_{i-1}$ are not printed as part of the circuit's description.
	Instead, they are treated as symbolic variables, with the intended meaning that the variable $\beta_j$ corresponds to the $j$-th constant specified thus far.}
	\end{enumerate}
	
	The notion that our definitions below will use is essentially the third one above (i.e., the machine needs to print the constants labeling the circuit), but many of our results also hold if we allow the more relaxed notion of $\neq$-uniformity (see~\Cref{rem:pre:uni,rem:pre:nequni}).
	
	\subsubsection{$\P$-uniform circuits and networks} \label{sec:uni:p-uni}
	
	The first and more relaxed notion of uniformity that we consider is $\P$-uniformity, which essentially says that an arithmetic circuit or network can be printed in time that is polynomial in its size.
	
	\begin{definition} [$\P$-uniform circuits and networks] \label{def:p-uniform networks} \label{def:pre:uni}
Let $\F = \set{\F_n}_{n \in \N}$ be a sequence of fields and let $\calN = \set{\calN_n}_{n \in \N}$ be a family of arithmetic networks where $\calN_n$ is defined over $\F_n$. We say that $\calN$ is \emphdef{$\P$-uniform} if there is Turing machine $M$ that on input $1^n$ runs in time polynomial in the size of $\calN_n$ and outputs a description of $\calN_n$, including the graph structure of the network and the constants labeling gates in the bottom layer.
\end{definition}

If the network family $\calN = \set{\calN_n}_{n \in \N}$ in \Cref{def:p-uniform networks} is in fact a family of arithmetic circuits, then we refer to the family $\calN$ as a \emphdef{$\P$-uniform family of arithmetic circuits}. We also extend~\Cref{def:p-uniform networks} in the straightforward way to circuit families that can be printed in some \emph{fixed} polynomial time. That is, a circuit family is \emphdef{$n^c$-time-uniform} if it satisfies~\Cref{def:p-uniform networks} with a machine $M$ that runs in time at most $n^c$.

\begin{remark} \label{rem:pre:uni}
Our results are not sensitive to the choice of how to represent field elements in~\Cref{def:p-uniform networks}. For example, over fields of characteristic zero, our results hold if we only allow constants that can be represented as bit-strings; but they also hold if we consider a real-RAM-like model in which machines can store field elements in registers, perform operations on them, and print them. The crucial point in~\Cref{def:p-uniform networks} is uniformity, i.e. that the networks $\set{\calN_n}$ are printable by a uniform machine of constant description size (in whichever model of descriptions we are working with).
\end{remark}

\begin{remark} \label{rem:pre:nequni}
All of our results in the ``hardness $\Rightarrow$ randomness'' direction hold even if the required hardness is only for the more restrictive notion of $\neq$-uniform circuits (recall that in this notion a machine does not need to print field elements, so we can avoid the question of representing field elements). Specifically, recall that when the field sequence $\F=\set{\F_n}$ is feasible, $\P$-uniformity (as in~\Cref{def:pre:uni}) is a more relaxed notion than $\neq$-uniformity. In~\Cref{thm:main:zero,thm:main:finite} we only use the notion of $\P$-uniformity to model the networks against which we need hardness, and we can indeed replace this notion by $\neq$-uniformity in these results, yielding a weaker hardness assumption. (We can do this because the reconstruction procedures in~\Cref{thm:gkss-tarhsg,thm:ki:tarhsg} can be implemented by $\neq$-uniform networks.) The reason we define $\P$-uniformity as our model of choice for the paper, rather than $\neq$-uniformity, is only since the upper bound obtained via diagonalization in~\Cref{prop:pit:necessary} (and~\Cref{fact:circuits:hierarchy}) does not seem to be $\neq$-uniform, but rather only $\P$-uniform.
\end{remark}

\subsubsection{$\log^c$-uniform circuits and networks} \label{sec:pre:uniform:logc}

Our generators will rely on hard functions that satisfy a stricter notion of uniformity. Loosely speaking, the circuit-structure function, which receives names of gates and computes whether or not these gates are connected in the circuit, should be computable in time that is polynomial in the size of the gates' names (i.e., polylogarithmic in the circuit size).

\begin{definition}[$\log^c$-uniform circuit families]\label{def:suf-unif}

Let $\F=\{\F_{n}\}_{n\in \N}$ be a sequence of fields and let $c \in \N$.
We say that an arithmetic circuit family $\set{C_n}$ over $\F$ of size $s(n)$ and depth $\Delta(n)$ is \emphdef{$\log^c$-uniform} if the following holds.
\begin{enumerate}
	\item (Circuit-structure function.) For each $n\in\N$, let $\Phi_n$ be the function that gets as input $i\in [\Delta(n)]$ and $u,v,w\in \zo^{\log(s(n))}$, returns $1$ if gate $u$ in layer $i$ of $C_n$ is fed by gates $v$ and $w$ in layer $i-1$, and returns $0$ otherwise.
	Then the family of functions $\Phi = \set{\Phi_n}_{n \in \N}$ is computable by a $\P$-uniform family of Boolean formulas of size $(\log (s(n)))^c$.\footnote{For concreteness, we consider Boolean formulas with NAND gates of fan-in two (analogously to closely related definitions in~\cite{GKR15,Gol18,CT21}), and the precise choice of gate-basis does not matter for our results.}
	
	\item (Circuit-constants function.) For each $n\in\N$, let $\Psi_n$ be the function that gets as input $\vec{\Lambda}\in\F_{n}^n$ and $w\in \{0,1\}^{\log(s(n))}$ and outputs the value of the $w^{\text{th}}$ gate in the input layer of $C_n(\vec{\Lambda})$.\footnote{Recall that arithmetic circuits have gates labeled with constant field elements only in their bottom layer (see~\Cref{def:algebraic circuit}). That is, for a size-$s$ circuit, the bottom layer consists of $n$ input gates and at most $s-n$ gates labeled with various constants in the field. Accordingly, if $w\le n$ then $\Psi_n(\vec{\Lambda},w)=\vec{\Lambda}_w$, and otherwise $\Psi_n(\vec{\Lambda},w)$ is the constant that labels gate $w$ in the input layer.}
	Then, there is a $\P$-uniform family of arithmetic circuits $\{\Psi'_n\colon\F_{n}^n\times\F_{n}^{\log(s)}\rightarrow\F_{n}\}_{n \in \N}$ of size $n \cdot (\log(s(n)))^c$ and degree $n \cdot (\log(s(n)))^c$ that computes $\Psi=\{\Psi_n\}$.
	That is, $\Psi'_n$ agrees with $\Psi_n$ on all inputs in the set $\F_{n}^n\times\{0,1\}^{\log(s)}$.
\end{enumerate}
\end{definition}

We note that the assumption on the computability of the circuit-constants function above is quite ad hoc. We allowed arithmetic circuits for convenience, and since these can handle arbitrary field elements, but we could also replace this by requiring Boolean formulas akin to the circuit-structure function. The crucial points are only that the circuits/formulas will be of size $\polylog(s)$ and that they are arithmetizable as low-degree polynomials (i.e., if these are Boolean models, we need them to be formulas). 

\subsubsection{Lower bounds for uniform circuits on all but finitely many inputs} 

As explained in~\Cref{sec:int}, when considering uniform families of circuits with multiple output elements, there are functions that are hard for such circuits on \emph{all but finitely many inputs}. This fact follows from a simple diagonalization argument.

\begin{fact} [almost-all-inputs hierarchy for uniform arithmetic circuits] \label{fact:circuits:hierarchy}
For every integer $k\in\N$ and every sequence $\F = \set{\F_n}_{n \in \N}$ of fields, where the fields may be finite or infinite (e.g., we may have $\F_n = \C$ for all $n \in \N$), there is a $\P$-uniform family of arithmetic circuits $\set{C_n}_{n\in\N}$ with $n$ input gates and $n$ output gates such that the following holds. For every $n^k$-time-uniform family $\set{C'_n}$ of arithmetic circuits of size $n^k$ and all but finitely many $n\in\N$, for every $\vec{\alpha}=(\alpha_1,\ldots,\alpha_n) \in \F_n^n$ it holds that $C_n(\vec{\alpha})\ne C'_n(\vec{\alpha})$.
\end{fact}

\bproof
The proof follows by simple diagonalization. We define the Turing machine that prints $C=\set{C_n}$ as follows. 
Given $1^n$, the machine enumerates the first $n$ Turing machines (according to some efficient enumeration), simulates each of them for $n^k$ steps, and obtains $n$ arithmetic circuits over $\F_n$ of size $n^k$ with $n$ input gates and $n$ output gates, denoted $A_n^{(1)},\ldots,A_n^{(n)}$. 
(If one of the $n$ Turing machines does not print such an arithmetic circuit in $n^k$ steps, we define the corresponding arithmetic circuit in some trivial way.)

For each $i\in[n]$, the machine modifies $A_n^{(i)}$ to a circuit $B_n^{(i)}$ that only computes the $i^{\text{th}}$ output gate of $A_n^{(i)}$. 
The machine prints a circuit $C_n$ that gets input $\vec{x}$ and outputs the $n$ elements
\[
B_n^{(1)}(\vec{x})+1, B_n^{(2)}(\vec{x})+1, \ldots , B_n^{(n)}(\vec{x})+1.
\]

Note that $C_n$ is indeed of polynomial size (i.e., size $O(n^{k+1})$). Also, for every $n^k$-time-uniform circuit family $\set{C'_n}$ printed by the $i^{\text{th}}$ machine, for all inputs $\vec{\alpha}$ of length at least $i$, we have $C'_n(\vec{\alpha})\ne C_n(\vec{\alpha})$; this is the case since otherwise we would have $B^{(i)}_n(\vec{\alpha})=A^{(i)}_n(\vec{\alpha})_i=C_n(\vec{\alpha})_i=B^{(i)}_n(\vec{\alpha})+1$, a contradiction. 
\eproof

\subsection{Randomized arithmetic networks} \label{sec:pre:randomized}

In this section we define two models of \emph{randomized} arithmetic networks. As explained in~\Cref{sec:int}, the first and simpler model will consist of networks with $\pit$ gates (see~\Cref{sec:pre:randomized:pit}) and the second and more general model consists of networks with random inputs (see~\Cref{sec:pre:randomized:inputs}).

\subsubsection{Arithmetic networks with $\pit$ gates} \label{sec:pre:randomized:pit}

In the following definition of networks with $\pit$ gates, the gates receive a description of an arithmetic circuit $C$, and compute $\pit$ (as a Boolean-valued function) on $C$. That is:

\begin{definition}[arithmetic networks with $\pit$ gates] \label{def:arith network PIT gates}
An \emphdef{arithmetic network with $\pit$ gates} $\calC$ is an arithmetic network that, in addition to  arithmetic, Boolean, test, and selection gates described in \Cref{def:arith-network}, has \emph{$\pit$} gates that have unbounded fan-in and implement the following functionality: each $\PIT$ gate receives as input a universal encoding\footnote{For a formal definition of a universal encoding, see~\Cref{subsec:props-arith-networks}.} $\vec{\alpha}$ of an arithmetic circuit $C$ and outputs a Boolean value $v$ such that $v = 1$ if and only if $C$ computes the zero polynomial.
\end{definition}

We remark that the choice of using universal encodings in the definition above over standard encodings (see \Cref{def:standard encoding}) is not arbitrary; the former guarantees that an encoding of length $n$ can only describe circuits computing polynomials of degree at most $n$, %
whereas this may not be the case for the latter.%

The computation of a network with $\pit$ gates on any given input is defined in the straightforward way (i.e., since the value of a $\pit$ gate on an input encoding of a circuit $C$ is well-defined). 
And indeed, this definition  will be, conceptually, our primary focus.
However, in order to design algorithms that meet \Cref{def:arith network PIT gates}, it will be convenient for us to work with auxiliary subroutines that compute \emph{relations}, i.e.,  we consider these subroutines to be correct if they print any output in some predetermined set of valid outputs on this input (i.e., rather than insisting on a single canonical output). %

Specifically, in the following definition, we consider a relation (representing a search problem) and define the network as correct on input $(\vec{\alpha},\vec{\beta})$ if it outputs any element in the relation.

\begin{definition} [computation of arithmetic relations via networks with $\pit$ gates] \label{def:arith-network-pit-gates-relations}
Let $C(\vec{x},\vec{y})$ be an arithmetic network with $\pit$ gates and $R\subseteq (\F^n\times\{0,1\}^m)\times(\F^{\nu}\times\zo^{\mu})$ be a relation. For any $(\vec{\alpha},\vec{\beta})\in \F^n\times\{0,1\}^m$, define $R(\vec{\alpha},\vec{\beta}) \coloneqq \set{ (\vec{\sigma},\vec{\sigma}')\in \F^{\nu}\times\zo^{\mu} : ((\vec{\alpha},\vec{\beta}),(\vec{\sigma},\vec{\sigma}')) \in R}$.
We say that \emphdef{$C$ computes $R$ at input $(\vec{\alpha},\vec{\beta})$} if $C(\vec{\alpha},\vec{\beta})\in R(\vec{\alpha},\vec{\beta})$, and we say that \emphdef{$C$ computes $R$} if this property holds for every input $(\vec{\alpha},\vec{\beta})$.
\end{definition}

We stress again that networks with $\pit$ gates are a relatively restrictive model of randomized networks (which means that hardness assumptions for it are relatively weak hardness assumptions). In particular, if $\pit$ can be solved by standard arithmetic networks, then arithmetic networks with $\pit$ gates can be simulated by standard arithmetic networks (see, e.g., the proof of~\Cref{prop:pit:necessary}).

\subsubsection{Arithmetic networks with random inputs} \label{sec:pre:randomized:inputs}

In this subsection, we define a randomized network model that is a more general notion than \Cref{def:arith network PIT gates} that equips arithmetic networks with an additional designated set of ``random inputs'' (that come from a prescribed finite set in the underlying field) to perform its computation, while still morally being algebraic in nature.
Intuitively, the idea is that when compared to \Cref{def:arith network PIT gates}, it is no longer restricted to using its randomness {\it only} to solve $\pit$ for auxiliary circuits---instead, it may now use it to perform {\it other} auxiliary computation.
However, it is still subject to the constraint that it eventually computes the correct value with high probability---and in particular, with error probability no more than the corresponding error bound guaranteed by \Cref{lem:sz} (for the randomized procedure that solves $\pit$ by simply evaluating its input polynomial at a random point). %

This slightly more general notion of randomized arithmetic networks is what is referenced in \Cref{thm:int:main}, as well as its corresponding formal statement in \Cref{thm:main:finite}.
Jumping ahead, the reason for assuming hardness against this subtly-generalized model there---instead of simply working with \Cref{def:arith network PIT gates}---is because the proof of correctness of our corresponding ($\ki$-based) targeted hitting-set generator relies on a reconstruction procedure that uses Kaltofen's \cite{Kaltofen1989FactorizationOP} classical factoring algorithm as a subroutine.
Indeed, this subroutine (described in detail in \Cref{subsec:kaltofen-arith-net}) uses its input randomness in a manner that does not seem to be easily simulated using $\pit$ gates (see also \Cref{prob:int:rand} in \Cref{sec:open}).

\begin{definition}[randomized arithmetic networks] \label{def:randomized arith network}
A \emphdef{randomized arithmetic network} $C$ is an arithmetic network that, in addition to arithmetic inputs $\vec{x}$ and Boolean inputs $\vec{y}$ (now referred to as ``primary" inputs), receives a supplementary set $\vec{r}$ of arithmetic inputs (which we term the ``random" set of arithmetic inputs).
\end{definition}

\Cref{def:randomized arith network} defines only the syntax of randomized arithmetic networks.
We next define the semantics of randomized arithmetic networks. %

\begin{definition} [computation by randomized arithmetic networks] \label{def:rand-arith-network}
Let $C((\vec{x},\vec{y}),\vec{r})$ be an arithmetic network with primary inputs $\vec{x},\vec{y}$ and supplementary (``random'') inputs $\vec{r}$. 
We say that \emphdef{$C$ computes value $\vec{\gamma}$ at primary input $(\vec{\alpha},\vec{\beta})$ with error degree $e$} if for every finite set $S\subset \F$, when each random input $r_i$ is chosen independently from $S$, with probability at least $1-e/|S|$ we have $C((\vec{\alpha},\vec{\beta}),\vec{r})=\vec{\gamma}$.
When $\F$ is finite, we may drop the qualifier ``error degree $e$'', in which case we consider $S=\F$ and require probability at least $2/3$.
\end{definition}

Analogous to \Cref{def:arith-network-pit-gates-relations}, we next define the notion of computing a relation by a randomized arithmetic network.

\begin{definition} [computation of arithmetic relations via randomized arithmetic networks] \label{def:rand-arith-network:relations}
Let $C((\vec{x},\vec{y}),\vec{r})$ be an arithmetic network with primary inputs $\vec{x},\vec{y}$ and supplementary (``random'') inputs $\vec{r}$.
Let $R\subseteq (\F^n\times\{0,1\}^m)\times(\F^{\nu}\times\zo^{\mu})$ be a relation, and define $R(\vec{\alpha},\vec{\beta}) \coloneqq \set{ (\vec{\sigma},\vec{\sigma}')\in \F^{\nu}\times\zo^{\mu} : ((\vec{\alpha},\vec{\beta}),(\vec{\sigma},\vec{\sigma}')) \in R}$ for every $\vec{\alpha},\vec{\beta}$.

We say that \emphdef{$C$ computes $R$ at primary input $(\vec{\alpha},\vec{\beta})$ with error degree $e$} if for every finite set $S \subseteq \F$, when each random input $r_i$ is chosen independently from $S$, with probability at least $1-e/|S|$, we have $C((\vec{\alpha},\vec{\beta}),\vec{r})\in R(\vec{\alpha},\vec{\beta})$. 
We say that $C$ computes $R$ with \emphdef{error degree} $e$ if this property holds for every primary input. 
When $\F$ is finite, we may drop the qualifier ``error degree $e$'', in which case we consider $S=\F$ and require probability at least $2/3$. 
\end{definition}

\subsubsection{Concatenation and composition of randomized arithmetic networks}

In this section we define notions of concatenating randomized arithmetic networks and of composing randomized arithmetic networks, and state some basic properties of these operations. The two notions are relevant both for~\Cref{def:arith network PIT gates} and for~\Cref{def:randomized arith network}. However, since the operations and their properties are obvious for the model presented in the former definition, we will only present the operations and their properties formally while referring to the latter definition.

\begin{definition}[Concatenation of randomized arithmetic networks]\label{def:concat-arith-networks}
Let $\calR$ and $\calS$ be two randomized arithmetic networks that compute relations $R_\calR\subseteq(\F^{n_{\calR}}\times \zo^{m_\calR})\times (\F^{\nu_\calR}\times \zo^{\mu_\calR})$ and $R_\calS\subseteq (\F^{n_\calS}\times \zo^{m_\calS})\times (\F^{\nu_\calS}\times \zo^{\mu_\calS})$, and take $r_\calR, r_\calS$ many random arithmetic inputs, respectively. 
The \emphdef{concatenated arithmetic network $(\calR,\calS)$} computes a relation 
\[
(R_\calR, R_\calS)\subseteq (\F^{n_\calR}\times \zo^{m_\calR})\times (\F^{n_\calS}\times \zo^{m_\calS})\times (\F^{\nu_\calR}\times \zo^{\mu_\calR})\times (\F^{\nu_\calS}\times \zo^{\mu_\calS})
\]
as follows. 
The network $(\calR, \calS)$ takes as primary arithmetic input a vector $(\vec{\alpha}_1,\vec{\alpha}_2) \in \F^{n_\calR}\times \F^{n_\calS}$, a Boolean vector $(\vec{b}_1,\vec{b}_2)\in \zo^{m_\calR}\times \zo^{m_\calS}$, and random arithmetic inputs  $\vec{\gamma}_1\in \F^{r_\calR}$ and $\vec{\gamma}_2\in\F^{r_\calS}$, and on input $(\vec{\alpha}_1, \vec{b}_2,\vec{\alpha}_2, \vec{b}_2)$ it outputs  $(\calR(\vec{\alpha}_1, \vec{b}_2,\vec{\gamma}_1),\calS(\vec{\alpha}_2, \vec{b}_2,\vec{\gamma}_2))$.
\end{definition}

\begin{proposition}[Degree and error degree under concatenation]\label{prop:concat-arith-networks}
Let $\calR$ and $\calS$ be randomized arithmetic networks as in \Cref{def:concat-arith-networks}. 
Let $d_\calR$ and $d_\calS$ denote the degrees of $\calR$ and $\calS$, respectively; let $e_\calR$ and $e_\calS$ denote the error degrees of $\calR$ and $\calS$, respectively.
Then the degree of the concatenated arithmetic network $(\calR,\calS)$ is bounded by $\max\{d_\calR, d_\calS\}$ and the error degree of $(\calR, \calS)$ is bounded by $e_\calR + e_\calS$. 
\end{proposition}

\bproof[Proof Sketch]
The bound for the degree is straightforward.
The bound for the error degree follows from a simple union bound.
\eproof

We also need a notion of composition for randomized arithmetic networks.
We formalize this in the following definition, and then subsequently bound the degree and error degree of the composition.

\begin{definition}[Composition of randomized arithmetic networks]\label{def:comp-arith-network}
Let $\calR$ and $\calS$ be two randomized arithmetic networks that compute relations $R_\calR\subseteq(\F^{n_{\calR}}\times \zo^{m_\calR})\times (\F^{\nu_\calR}\times \zo^{\mu_\calR})$ and $R_\calS\subseteq (\F^{n_\calS}\times \zo^{m_\calS})\times (\F^{\nu_\calS}\times \zo^{\mu_\calS})$, and take $r_\calR, r_\calS$ many random arithmetic inputs, respectively. 
Then, we say that the $\calR$ and $\calS$ (as an \emph{ordered} pair) are \emphdef{composable} if $\nu_\calS = n_\calR$ and $\mu_\calS = m_\calR$. 

If $\calR$ and $\calS$ are composable, we define the relation
\begin{align*}
	R_\calR \circ R_\calS \subseteq (\F^{n_\calS} \times \zo^{m_\calS}) \times (\F^{\nu_\calR} \times \zo^{\mu_\calR})
\end{align*}
to be the set of all tuples $((\vec{\alpha}_1, \vec{b}_1), (\vec{\alpha}_3, \vec{b}_3))$ such that there exists $(\vec{\alpha}_2, \vec{b}_2)$ for which both $R_{\calR}((\vec{\alpha}_1, \vec{b}_1), (\vec{\alpha}_2, \vec{b}_2))$ and $R_{\calS}((\vec{\alpha}_2, \vec{b}_2), (\vec{\alpha}_3, \vec{b}_3))$ are satisfied.
The \emphdef{composition $\calR\circ\calS$} computes the relation $R_\calR\circ R_\calS$ as follows: it takes as primary arithmetic input a vector $\vec{\alpha} \in \F^{n_\calS}$, a Boolean vector $\vec{b}\in \zo^{m_\calS}$, and random arithmetic inputs $\vec{\gamma}_1\in \F^{r_\calR}$ and $\vec{\gamma}_2\in\F^{r_\calS}$, and outputs $\calR(\calS(\vec{\alpha},\vec{b},\vec{\gamma}_2),\vec{\gamma}_1)$.
\end{definition}

\begin{proposition}[Degree and error degree under composition]\label{prop:error-deg-comp}
Let $\calR$ and $\calS$ be composable randomized arithmetic networks as in \Cref{def:comp-arith-network}. Furthermore, let $d_\calR$ and $d_\calS$ denote the degrees of $\calR$ and $\calS$, respectively; let $e_\calR$ and $e_\calS$ denote the error degrees of $\calR$ and $\calS$, respectively.
Then the degree of the composed arithmetic network $\calR\circ \calS$ is bounded by the product $d_\calR\cdot d_\calS$ and the error degree of $\calR \circ \calS$ is bounded by $e_\calR + e_\calS$.
\end{proposition}

\bproof
The bound for the degree is straightforward. To show the bound on the error degree of the composition, from \Cref{def:rand-arith-network:relations}, it suffices to show that for every input $(\vec{\alpha}, \vec{b}) \in \F^{n_\calS} \times \set{0,1}^{m_\calS}$ and any finite set $S\subset\F$,
\begin{equation}
\Pr_{\substack{\vec{\gamma}_1 \gets S^{r_\calR},\\ \vec{\gamma}_2 \gets S^{r_\calS}}}\sbr{\calR\circ\calS(\vec{\alpha}, \vec{b}, (\vec{\gamma}_1,\vec{\gamma}_2)) \not\in R_\calR\circ R_\calS(\vec{\alpha}, \vec{b})} \le \frac{e_\calR + e_\calS}{|S|}.
\end{equation}
Consider any fixed input $(\vec{\alpha}, \vec{b}) \in \F^{n_\calS} \times \set{0,1}^{m_\calS}$. %
Let $E_1$ denote the event that $\calR\circ\calS$ outputs in the relation $R_{\calR}\circ R_{\calS}$ on input $(\vec{\alpha},\vec{b})$, and $E_2$ denote the event that  $\calS$ outputs in the relation $R_{\calS}(\vec{\alpha},\vec{b})$ (i.e., computes ``correctly''). 
Note that again from \Cref{def:rand-arith-network:relations}, we have that $\Pr[E_1|E_2]\geq 1 - e_\calR/|S|$ and $\Pr[E_2]\geq 1 - e_\calS/|S|$. 
Therefore, we have
\[
\Pr[E_1] \geq \Pr[E_1\cap E_2] = \Pr[E_1|E_2]\cdot \Pr[E_2] \geq \left(1 - \frac{e_\calR}{|S|}\right)\left(1 - \frac{e_\calS}{|S|}\right) \geq 1 - \frac{e_\calR}{|S|} - \frac{e_\calS}{|S|}.
\]
Since this is true for an arbitrary choice of a primary input to $\calR\circ\calS$, we conclude that its error degree is at most $e_\calR + e_\calS$.
\eproof

\subsection{The relation to BSS machines} \label{sec:pre:bss}

The $\pit$ algorithms in~\Cref{thm:int:main:zero,thm:int:main} work assuming lower bounds for $\P$-uniform randomized networks, and in fact even for $\neq$-uniform randomized networks as mentioned in~\Cref{rem:pre:nequni}. However, our proofs do not heavily capitalize on the specific features of this computational model, and we suspect that the $\pit$ algorithms also work assuming lower bounds for other reasonable models of uniform randomized arithmetic computation. The key point is that the reconstruction procedures we present in~\Cref{sec:gkss,sec:ki} can also be implemented in other models of uniform randomized arithmetic computation.\footnote{The main non-trivial operation that the reconstruction procedures use is branching according to zero-tests; that is, an operation of the form ``if $v=0$ continue with arithmetic procedure $A$, and if $v\ne0$ continue with arithmetic procedure $B$'', where $v$ is the result of an intermediary arithmetic computation.}

In particular, we believe that our $\pit$ algorithms should work assuming lower bounds for the more standard model of (randomized) constant-free Blum--Shub--Smale machines (BSS machines)~\cite{BSS89}. This is because constant-free BSS machines are seemingly at least as strong as $\neq$-uniform networks (and ditto if we allow randomness or $\pit$ gates/oracle in both models), and thus hardness for the former seems to be a stronger assumption than hardness for the latter.

Uniform networks and BSS machines are very similar in spirit (see below). The reason we work with networks in this paper is the conceptually clearer comparison of the upper-bounds to the lower-bounds (i.e., both refer to circuits or networks), as well as the clean converse direction in~\Cref{thm:int:pit:necessary}. An additional (mostly cosmetic) benefit in networks is a clearer separation between auxiliary Boolean inputs and arithmetic inputs, especially in intermediate auxiliary computations in the reconstruction procedures.\footnote{ In contrast to networks, the inputs to a BSS machine are all from the underlying ring. Thus, if we want to allow for subroutines that operate only on Boolean data, it becomes necessary to speak of \emph{partial} functions computed by BSS machines.}

\paragraph{More detail about BSS machines.}
Recall that a BSS machine receives a vector of elements from a ring $R$ as input, performs ring operations on its input and storage (which consists of ring elements) and then outputs another vector of elements from $R$. In addition to ring operations, a BSS machine may also perform equality tests with zero (similarly to a network) to decide which ring operations to perform next; over ordered rings inequality tests with zero are often also allowed (i.e., tests of the form ``$v\ge0$?'', which are not allowed in networks). The control flow of the machine is specified by a finite directed graph, just as a Turing machine is specified by a finite collection of states and transitions. The machine may also be supplied with a finite list of constants from $R$.

Randomized BSS machines are slightly more subtle, and we need to restrict the constants allowed in the machine's description (otherwise $\pit$ can be solved in polynomial time, by using ``unreasonable'' real constants to represent non-uniform advice; see~\cite[Chapter 17, Theorem 7]{BCSS98}). As mentioned above, we can simply consider constant-free BSS machines, which already seem strong enough to simulate $\neq$-uniform networks. To augment the BSS machine with randomness, we either allow oracle calls to $\pit$ (analogously to networks with $\pit$ gates) or allow random field elements (analogously to randomized networks).

For more details on BSS machines, we refer the reader to~\textcite{BCSS98,mm97}, and in particular to~\cite[Chapter 17]{BCSS98} for the randomized variant.

\subsection{Universal arithmetic circuits and networks}\label{subsec:props-arith-networks}

In this subsection, we define two notions of encodings for  arithmetic circuits and observe that it is possible to uniformly construct arithmetic networks that can evaluate any circuit on a given point when given its encoding as input. These notions are repeatedly used by the reconstruction networks of our main results.

\begin{definition}[Standard encoding of an arithmetic circuit] \label{def:standard encoding}
Let $\F$ be a field and let $\Phi$ be an algebraic circuit over $\F$.
We say that a tuple $(\vec{\alpha}, \vec{b}) \in \F^n \times \zo^m$ is a \emphdef{standard encoding} or \emphdef{standard description} of $\Phi$ if $\vec{b}$ is a bit-string that describes the wiring of $\Phi$ and the type of each gate and $\vec{\alpha}$ is a vector of field elements where $\alpha_i$ is the $i$-th field constant used by the circuit $\Phi$.
\end{definition}

\begin{proposition}\label{prop:short-desc}
If $C$ is an arithmetic circuit of size $s$, then $C$ admits a standard encoding of size $O(s \log s)$.
\end{proposition}

\bproof
For each gate $v$ in $C$, the type of $v$ can be specified with $O(1)$ bits, so the types of all gates can be specified with $O(s)$ bits.
Each wire in $C$ can be described by $O(\log s)$ bits.
As we assume our circuits to have fan-in two, a circuit of size $s$ necessarily has $O(s)$ wires, so we can describe the wires of $C$ using $O(s \log s)$ bits.
Finally, a circuit of size $s$ uses at most $s$ field constants.
Thus there is a standard description $(\vec{\alpha}, \vec{b}) \in \F^n \times \zo^m$ of $C$ of length $n + m \le O(s \log s)$.
\eproof

\begin{proposition}[Arithmetic networks can evaluate arithmetic circuits] \label{prop:universal-arith-network}
Let $\F$ be any field and let $d : \N \to \N$ be a time-constructible function.
There is a $\P$-uniform family of arithmetic networks of size $\poly(n,s,d(n))$ and degree $\poly(n,d(n))$ over $\F$ that takes as input (i) a standard encoding of an $n$-variate arithmetic circuit $C$ over $\F$ of size $s$ and formal degree $d(n)$ and (ii) a vector $\vec{\alpha}\in \F^n$ (as arithmetic data), and outputs (as arithmetic data in $\F$) the evaluation $C(\vec{\alpha})$.
\end{proposition}

\bproof
The network proceeds by evaluating the circuit $C$ gate-by-gate in topological order.
For each $i \in [s]$, we compute the value of the $i$-th intermediate gate in the circuit $C$ on input $\vec{\alpha}$.
The network does this by computing all pairwise sums of the values it has computed so far, along with all pairwise products subject to the constraint that the resulting product has formal degree at most $d(n)$.
These computations can be carried out using $O(i^2) \le O(s^2)$ many arithmetic gates.
The network then reads the Boolean part of the description of the circuit $C$ to determine the type of the $i$-th gate in $C$, along with the names of the children of the $i$-th gate.
Feeding this into $\select$ gates, the network can correctly compute the value of the $i$-th gate in $C$.

It follows from the description above this network has size $\poly(n,s,d(n))$ and degree $\poly(n,d(n))$.
Moreover, the wiring of this network can be printed by a Turing machine running in time $\poly(n,s,d(n))$.
\eproof

\begin{proposition}\label{prop:comp-arith-network}
There is a $\P$-uniform arithmetic network of size $\poly(n,m,s,s')$ and degree $1$ over any field $\F$ that takes as input (i) a standard description of an arithmetic circuit $C$ of size $s$ over $\F$ computing an $n$-variate polynomial and (ii) a standard description of a multi-output arithmetic circuit $C'$ of size $s'$ over $\F$ having $n$ output gates where each output gate computes an $m$-variate polynomial, and outputs a standard description of the circuit computing their composition i.e., $C \circ C'$.
\end{proposition}

\bproof
The arithmetic network computes the Boolean part of a standard encoding of $C \circ C'$ by replacing the input gates of $C$ labeled by variables with the output gates of $C'$.
This rewiring can be done by a $\P$-uniform family of Boolean circuits.
The arithmetic part of a standard encoding of $C \circ C'$ corresponds to the concatenation of the arithmetic parts of standard encodings of $C$ and $C'$, which can clearly by computed by a $\P$-uniform arithmetic network of size $s + s'$ and degree $1$.
\eproof

\begin{definition}[Universal circuits]\label{def:universal circuit} 
Let $\F$ be a field and let $n, d, s \in \N$.
We say that a polynomial $C(\vec{x},\vec{y})$ is \emphdef{universal for $n$-variate, degree-$d$, size-$s$ computation over $\F$} if for every polynomial $f(\vec{x})$ where $\deg(f) \le d$ and $f$ is computable by an algebraic circuit of size $s$, there is a point $\vec{\beta} \in \F^m$ such that $C(\vec{x},\vec{\beta}) = f(\vec{x})$.
An algebraic circuit is \emphdef{universal for $n$-variate, degree-$d$, size-$s$ computation} if it computes a polynomial that is universal for $n$-variate, degree-$d$, size-$s$ computation.
\end{definition}

\begin{definition}[Universal encoding] \label{def:universal encoding}
Let $\F$ be a field and fix a family $\mathcal{C} = \set{C_n}_{n \in \N}$ of $\P$-uniform arithmetic circuits such that $C_n$ is universal for $n$-variate, degree-$d$, size-$s$ computation.
Let $\Phi$ be an $n$-variate algebraic circuit of size $s$ over $\F$ that computes a polynomial of degree $d$.
We say that a point $\vec{\alpha} \in \F^t$ is a \emphdef{$\mathcal{C}$-universal encoding} of $\Phi$ if for $m = \max(n,s,d)$, the circuit $C_m$ satisfies $C_m(\vec{x}, \vec{\alpha}) = \Phi(\vec{x})$.
\end{definition}

Throughout this work, we make use of the universal encoding with respect to the family of universal circuits $\set{\Psi_n}_{n \in \N}$ constructed by \textcite{Raz10}.
As this is the only universal circuit we make use of, we abbreviate the term ``$\Psi$-universal encoding'' as ``universal encoding.''

\begin{theorem}[{\cite[Proposition 2.8]{Raz10}}] \label{thm:raz universal circuit}
Let $\F$ be a field.
For every $n, d, s \in \N$, there is a constant-free circuit $\Psi$ that is universal for $n$-variate, degree-$d$, size-$s$ computation over $\F$.
The size of $\Psi$ is bounded by $O(s d^4)$.
Moreover, there is a polynomial-time algorithm that receives $(1^n, 1^s, 1^d)$ as input and outputs a description of the circuit $\Psi$.
\end{theorem}

We conclude this subsection by extending some of the aforementioned notions to arithmetic networks with $\pit$ gates. These extensions will be useful in establishing the randomness $\Rightarrow$ hardness direction of our results, i.e., \Cref{thm:int:pit:necessary} (or more precisely, its formalization, \Cref{prop:pit:necessary} in \Cref{sec:pit:necessary}).

\begin{definition}[Standard encoding of an arithmetic network with $\pit$ gates] \label{def:standard encoding networks}
Let $\F$ be a field and let $\calC$ be an arithmetic network with $\pit$ gates over $\F$.
We say that a tuple $(\vec{\alpha}, \vec{b}) \in \F^n \times \zo^m$ is a \emphdef{standard encoding} or \emphdef{standard description} of $\calC$ if $\vec{b}$ is a bit-string that describes the wiring of $\calC$ and the type of each gate and $\vec{\alpha}$ is a vector of field elements where $\alpha_i$ is the $i$-th field constant used by the network $\calC$.
\end{definition}

The following statement has a proof similar to that of \Cref{prop:short-desc}.
\begin{proposition}\label{prop:short-desc-network}
If $C$ is an arithmetic network with $\pit$ gates of size $s$, then $C$ admits a standard encoding of size $O(s \log s)$.
\end{proposition}

\begin{proposition}[Arithmetic networks can evaluate arithmetic networks] \label{prop:universal-arith-network-for-networks}
Let $\F$ be any field and let $d : \N \to \N$ be a time-constructible function.
There is a $\P$-uniform family of arithmetic networks of size $\poly(n,s,d(n))$ and degree $\poly(n,d(n))$ over $\F$ that takes as input (i) a standard encoding of an  arithmetic network $C$ over $\F$ of size $s$ and degree $d(n)$ and (ii) a vector $(\vec{\alpha},\vec{\beta})\in \F^n\times \zo^n$, and outputs the evaluation $C(\vec{\alpha},\vec{\beta})$.
\end{proposition}

\bproof[Proof Sketch.] The proof is similar to that \Cref{prop:universal-arith-network}. 
The network mimics the one in the proof of \Cref{prop:universal-arith-network}, with the only difference now being that it computes the value of the $i$-th intermediate gate in the network $C$ on input $(\vec{\alpha},\vec{\beta})$ by not just accounting for all the pairwise sums and products (of degree at most $d(n)$) of the arithmetic values it has computed so far, but also for the corresponding Boolean operations performed on the Boolean values computed so far, along with the test gates values $\testequals$, and the $\select$ gate values of their appropriate combinations.
\eproof

\section{Polynomial decompositions for arithmetic circuits} \label{sec:decomposition}

In this section, we recall the notion of \emph{polynomial decompositions of a circuit}, following \textcite{CT21}.
The definition provided in this section adapts \cite[Definition 4.6]{CT21} to algebraic circuits. This notion is an encoding of a circuit's computation (on a fixed input $\vec{\Lambda}\in\F^n$) as a sequence of downward self-reducible polynomials. 
After presenting the abstract notion of polynomial decompositions, we will show that any $\log^c$-uniform 
family of algebraic circuits has efficient polynomial decompositions.

\subsection{Variable-extended formulations}

Before introducing polynomial decompositions of circuits, we need the definition of a variable-extended formulation of a boolean function. This is a technical notion that, while not strictly necessary for the definition of polynomial decompositions, leads to an improvement in the parameters of a polynomial decomposition. In particular, the use of variable-extended formulations allows us to guarantee that the individual degrees of the hard polynomials are bounded by a constant.

\begin{definition}[Variable-extended formulation]\label{def:var-ext}
Let $f\colon S^m\rightarrow\{0,1\}$ where $S$ is an arbitrary set. We say that $g(x_1,\ldots,x_m,y_1,\ldots, y_r)$ is a \emphdef{variable-extended formulation} of $f$ if the following holds for every $\vec{x} \in S^m$.
\begin{enumerate}
	\item If $f(\vec{x}) = 1$, then there exists a unique $\vec{y}(\vec{x}) \in S^r$ (``unique witness") such that $g(\vec{x}, \vec{y}(\vec{x})) = f(\vec{x}) = 1$, and $g(\vec{x}, \vec{y}) = 0$ for all $\vec{y} \neq \vec{y}(\vec{x})$.
	\item If $f(\vec{x}) = 0$, then $g(\vec{x}, \vec{y}) = 0$ for all $\vec{y} \in S^r$.
\end{enumerate}
\end{definition}

\begin{obs}\label{obs:var-ext}
Any variable-extended formulation $g$ of $f$ satisfies
\[
f(\vec{x}) = \sum_{\vec{y}\in S^r} g(\vec{x},\vec{y})
\]
for every $\vec{x} \in S^m$.
\end{obs}

The following lemma asserts that any small Boolean formula has a variable-extended formulation that can be computed by a small arithmetic formula of constant individual degree. The proof follows \cite{KalaiLV23} and is based on a Cook--Levin-style trick.

\begin{lem}[{Modification of \cite[Lemma 5.3]{KalaiLV23}}]\label{lem:var-ext}
Let $f:\zo^m\to \zo$ be a NAND-boolean formula of size $s$ and let $\F$ be any field. Then there exists a variable-extended formulation $g:\zo^{m+s}\to\zo$ of $f$ that can be computed by an arithmetic formula\footnote{By this, we mean that the formula computes a polynomial whose evaluations on $\zo^{m + s}$ agree with a variable-extended formulation $g : \zo^{m + s} \to \zo$.} of size $O(s\cdot m)$ over $\F$, which computes a polynomial of total degree $O(s\cdot m)$ and individual degree $2$. 
Moreover, there is a $\P$-uniform family of arithmetic networks of size $\poly(s\cdot m)$ that, when given as input a standard description\footnote{Since $f$ is a Boolean formula, by a standard description, we simply mean a description that describes the wiring of $f$.} of the NAND-formula for $f$, outputs a standard description of this arithmetic formula for $g$.
\end{lem}

\bproof
We introduce $s$ new variables $y_1,\ldots,y_s$, one for each wire of the formula computing $f$.
For convenience of notation, we assume that there is a wire coming out of the output gate that has the label $y_s$.

We identify each gate $\Phi=\Phi_{i}$ with the wire $i\in[s]$ going out of $\Phi$. 
Let 
\[
\Phi_i(\vec{y}) \coloneqq y_i + y_jy_k - 2y_iy_jy_k,
\]
where $j,k\in[s]$ are the indices of the wires going into $\Phi_i$. 
For a zero-one-valued input $\vec{y}\in\zo^{s}$, we have 
\[
\Phi_i(\vec{y}) = \begin{cases}
1 & \text{if $y_i=\mathsf{NAND}(y_j,y_k)$,} \\
0 & \text{otherwise.}
\end{cases}
\]
For every index $(j,\sigma)\in[m]\times\zo$ of an input literal (i.e., $x_j$ or $\lnot x_j$), we define 
\[
h_{j,\sigma}(\vec{x},\vec{y}) \coloneqq \ell_{j,\sigma}\prod_{i\in S_{j,\sigma}}y_i + (1-\ell_{j,\sigma})\prod_{i\in S_{j,\sigma}}(1-y_i),
\]
where 
\[
\ell_{j,\sigma} \coloneqq 
\begin{cases} 
x_j & \text{if $\sigma = 0$,} \\
1 - x_j & \text{if $\sigma = 1$,}
\end{cases}
\]
and $S_{j,\sigma}$ denotes the set of leaf wires corresponding to $\ell_{j,\sigma}$.
Observe that for every $\vec{x},\vec{y}\in\zo^{m+s}$ we have that $h_{j,\sigma}(\vec{x},\vec{y}) = 1$ if and only if for all $i \in S_{j, \sigma}$, the variable $y_i$ is set to $\ell_{j, \sigma}$.
Otherwise, we have $h_{j, \sigma}(\vec{x}, \vec{y}) = 0$.
Then we define
\[
g(\vec{x},\vec{y}) \coloneqq y_s \prod_{i = 1}^s \Phi_i(\vec{y}) \prod_{j = 1}^m\left(\prod_{\sigma\in\zo} h_{j,\sigma}(\vec{x},\vec{y})\right).
\]

The claims about the formula size and total degree follow immediately from the definition of $g$. 
It follows from the definitions of the $\Phi_i$'s and $h_{j, \sigma}$'s that the restriction of $g(\vec{x}, \vec{y})$ to $\zo^{m+s}$ corresponds to a Boolean function that is a variable-extended formulation of $f$.

Finally, note that $g(\vec{x},\vec{y})$ has individual degree $2$, since (i) the variable $y_s$ appears only twice, (ii) each intermediate (non-output, non-leaf) variable only appears twice because the corresponding wire has exactly two endpoints, and (iii) the variables $\{x_j , y_i|j\in [m],i\in S_{j,\sigma}\}$ have degree at most $2$ (they occur at most once in the first product, while the second product is multilinear). %
The efficient computation of a standard description of $g$ (see \Cref{def:standard encoding}) can be done via a $\P$-uniform family of arithmetic networks $\{\calN_m\}_{m\in\N}$ as follows. %
First, from the given description of $f$, the arithmetic network $\calN_m$ reads off the set of wires $S_{j,\sigma}$ adjacent to each input literal.
Given $S_{j,\sigma}$, a standard description of $h_{j,\sigma}$ can be output using a constant overhead using its definition above.
From the given wiring of $f$, the network is able to compute a description of each $\Phi_i$ of constant size.
Putting everything together, the network $\calN_m$ outputs a standard description of the arithmetic formula for $g$.
\eproof

\subsection{Polynomial decompositions} \label{sec:ki:decomposition}

We recall that the overall idea in the definition below is to arithmetize
the circuit structure function layer-by-layer over an appropriate arithmetic setting, and “in between layers”, we
add additional polynomials that represent a suitable sumcheck protocol, which allows reducing
claims about each layer to claims about the preceding layer. We allow some parameters to be arbitrary to allow for greater functionality in the reconstruction procedures later, in particular the size of an interpolating set $H\subseteq\F$ and an integer parameter $r$. 
Jumping ahead, we will instantiate the notion either with $H=\zo$ and $r>0$, or with a larger $H$ (e.g., $|H|=n^{1/\log\log(n)}$) and $r=0$.

\begin{definition}[Polynomial decompositions of a circuit] \label{def:poly-dec}
Let $C$ be an algebraic circuit over a field $\F$ that has $n$ inputs, fan-in two, size $s$, and depth $\Delta$.
For any $\vec{\Lambda} \in \F^n$, we call a collection of polynomials $\{f_{i,j}\}$ an \emphdef{$(h,m,r)$-polynomial decomposition of $C(\vec{\Lambda})$} if it satisfies the following requirements.
\begin{enumerate}
	\item
	(Setting.)
	The polynomials $f_{i,j}$ are defined over the field $\F$.
	For some integer $h \le |\F|$ such that $h$ is a power of $2$, let $H \subseteq \F$ be of size $h$, let $m$ be the minimum positive integer such that $h^m \ge s$, and let $r\geq 0$ be an integer.
	\item
	(Circuit structure polynomial.)
	For each $i \in [\Delta]$, let $\hat{\Phi}_i : H^{3m} \to \set{0,1}$ be the function such that $\hat{\Phi}_i(\vec{\alpha}, \vec{\beta}, \vec{\gamma}) = 1$ if and only if the gate in layer $i$ of $C$ indexed by $\vec{\alpha}$ has as children the gates in layer $i-1$ indexed by $\vec{\beta}$ and $\vec{\gamma}$.\footnote{The function $\hat{\Phi}_i$ is defined to be zero if any of the three input strings is not a valid indexing of a gate.} Let $\hat{\Phi}_i^e : H^{3m + r} \to \set{0,1}$ be a variable-extended formulation of $\hat{\Phi}_i$ (so if $r = 0$, $\hat{\Phi}_i^e = \hat{\Phi}_i$).
	Let the polynomial $\Phi_i(w_1,\ldots,w_m, u_1, \ldots, u_m, v_1, \ldots, v_m, y_1,\ldots, y_r) \in \F[\vec{w},\vec{u},\vec{v}, \vec{y}]$ be any polynomial which agrees with the function $\hat{\Phi}_i^e$ on the set $H^{3m+r}$. 
	\item\label{item:f0}
	(Input-layer function.)
	Let $\hat{f}_0 : H^m \to \F$ represent\footnote{\label{fn:lex}See \Cref{sec:pre} for a definition. Also, note that here, the elements of $H^m$ are assumed to be listed in the lexicographic order.} the input layer string $\vec{\Lambda}\vec{K}\in\F^s$ (interpreted as the concatenation of vectors $\vec{\Lambda}$ and $\vec{K}$), where $\vec{K}$ is the list of constants that $C$ uses. 
	Let $f_0(w_1,\ldots,w_m)$ be any polynomial over $\F$ that agrees with $\hat{f}_0$ on all inputs in $H^m$.%
	
	\item\label{item:layer-polys}
	(Layer polynomials.)
	For each $i \in [\Delta]$, let $\hat{f}_i : H^m \to \F$ represent\footnote{See \Cref{fn:lex}.} %
	the values of the gates at the $i$\textsuperscript{th} layer of $C$ in the computation of $C(\vec{\Lambda})$, with zeroes in locations that do not index valid gates.
	Define $f_i \in \F[w_1,\ldots,w_m]$ inductively by 
	\[
	f_i(\vec{w}) \coloneqq 
	\begin{cases}
		\sum_{\vec{\beta}, \vec{\gamma} \in H^m, \vec{\delta}\in H^r} \Phi_i(\vec{w}, \vec{\beta}, \vec{\gamma},\vec{\delta}) \del{f_{i-1}(\vec{\beta}) + f_{i-1}(\vec{\gamma})}, &\text{if $i$ is odd} \\
		\sum_{\vec{\beta}, \vec{\gamma} \in H^m, \vec{\delta}\in H^r} \Phi_i(\vec{w}, \vec{\beta}, \vec{\gamma},\vec{\delta}) \del{f_{i-1}(\vec{\beta}) \cdot f_{i-1}(\vec{\gamma})}, &\text{if $i$ is even.} \\
	\end{cases}
	\]
	\item\label{item:sumcheck}
	(Sumcheck polynomials.)
	Denote $t = 2m+r$. For each $i \in [\Delta]$, let $f_{i,0}(\vec{w},\sigma_1,\ldots,\sigma_{t}) \in \F[\vec{w}, \vec{\sigma}]$ be the polynomial
	\begin{align*}
		f_{i,0}(\vec{w},\vec{\sigma}) &\coloneqq 
		\begin{cases}
			\Phi_i(\vec{w}, \sigma_{1},\ldots,\sigma_{t}) \del{f_{i-1}(\sigma_{1},\ldots,\sigma_m) + f_{i-1}(\sigma_{m+1},\ldots,\sigma_{2m})}, &\text{if $i$ is odd}\\
			\Phi_i(\vec{w}, \sigma_{1},\ldots,\sigma_{t}) \del{f_{i-1}(\sigma_{1},\ldots,\sigma_m) \cdot f_{i-1}(\sigma_{m+1},\ldots,\sigma_{2m})}, &\text{if $i$ is even.}\\
		\end{cases}
	\end{align*}
	For each $j \in [t-1]$, let $f_{i,j}(\vec{w},\sigma_1,\ldots,\sigma_{t-j}) \in \F[\vec{w},\vec{\sigma}]$ be the polynomial
	\[
	f_{i,j}(\vec{w},\sigma_{1},\ldots,\sigma_{t-j}) \coloneqq 
	\begin{cases}
		\displaystyle\sum_{\sigma_{t-j+1},\ldots,\sigma_{t} \in H^j} \Phi_i(\vec{w}, \sigma_{1},\ldots,\sigma_{t}) \del{f_{i-1}(\sigma_{1},\ldots,\sigma_{m}) + f_{i-1}(\sigma_{m+1},\ldots,\sigma_{2m})}, &\text{if $i$ is odd}\\
		\displaystyle\sum_{\sigma_{t-j+1},\ldots,\sigma_{t} \in H^j} \Phi_i(\vec{w},\sigma_{1},\ldots,\sigma_{t}) \del{f_{i-1}(\sigma_{1},\ldots,\sigma_{m}) \cdot f_{i-1}(\sigma_{m+1},\ldots,\sigma_{2m})}, &\text{if $i$ is even}\\
	\end{cases} 
	\]
	Finally, we define $f_{i,t}(\vec{w}) \coloneqq f_i(\vec{w})$.
\end{enumerate}
\end{definition}

Towards constructing efficient polynomial decompositions, it will be useful---particularly for our targeted hitting set generator based on the Kabanets--Impagliazzo generator---to define an auxiliary notion of an \emph{enumeration} of a set of given size within feasible field families (see \Cref{sec:pre} for a refresher on its definition). 

\begin{definition}[enumeration of a set of prescribed size in feasible fields]\label{def:enumeration}
Let $\F = \{\F_n\}_{n\in \N}$ be a feasible sequence of fields, and let $h(n)\leq |\F_n|$ be a given time-constructible function. 
We say that a family of sets $\{H_n\}$ is an \emphdef{enumeration of $h(n)$ with respect to $\{\F_n\}$} (or simply that $H$ is an \emphdef{enumeration of $h$ in $\F$} if the family is clear from context) if there is a Turing machine that takes as input $1^n$ and outputs the elements of a set $H_n\subset \F_n$ such that $|H_n| = h(n)$ and whose running time is bounded by $\poly(n,h(n))$. 
We call a function family $\{\phi_n: [h(n)]\rightarrow \F_n\}$ an \emphdef{enumerator of $h(n)$ with respect to $\{\F_n\}$} if $\phi_n$ is injective, its image is $H_n$, and it is computable in  $\poly(n,h(n))$ time.\footnote{Note that the defintion of feasible field families implies that such enumerators exist.} 
\end{definition}

We now show that a polynomial decomposition exists for any $\log^c$-uniform family $\{C_n\}$ of arithmetic circuits (see \Cref{def:suf-unif}). Moreover, we show that there is a polynomial-sized $\P$-uniform family of arithmetic networks that when given an input $\vec{\Lambda}$, outputs small arithmetic circuits for the layer polynomials as well as for the sumcheck polynomials. These small circuits ensure the efficiency of our targeted hitting set generators (defined later in \Cref{cons:generator} and \Cref{cons:GKSS-based-tarHSG}). Furthermore, we record crucial downward self-reducibility properties of these polynomials, which will be essential in the reconstruction argument.

The proposition below actually constructs two polynomial decompositions, parameterized differently. The first one is parameterized with the set $H=\{0,1\}$ and with an integer $r=\polylog(s)$ (where $s$ is the size of $\{C_n\}$), and the second one is parameterized with an arbitrary set $H$ and with $r=0$. We crucially use both instantiations for our targeted hitting set generator based on the Kabanets--Impagliazzo generator (which is discussed in \Cref{sec:ki}), whereas we only use the latter for our construction based on the hitting set generator of Guo--Kumar--Saptharishi--Solomon (which is discussed in \Cref{sec:gkss}).

\begin{proposition} [Polynomial decompositions for arithmetic circuits] \label{prop:poly-dec-cons}
Let $\{\F_n\}_{n\in \N}$ be a feasible sequence of fields where $\F_n$ has order $q(n) \le \exp(n^{O(1)})$.
Let $\set{C_n}_{n \in \N}$ be a $\log^c$-uniform family of arithmetic circuits correspondingly over $\{\F_n\}_{n\in \N}$ of size $s(n)$ and depth $\Delta(n)$. Let $h=h(n)\ge2$ be a time-constructible function whose range is powers of two. Then, for every $n\in\N$ and every $\vec{\Lambda} \in \F_n^n$, 
\begin{enumerate}
	\item {\bf ($r=\polylog(s)$ instantiation.)} There is a \emph{$(h=2,m=\lceil\log(s)\rceil,r=\log(s)^c)$-polynomial decomposition} $\{f_{i,j} \mid i\in[\Delta(n)], j\in \{0,..., t\}  \}$ of $C_n(\vec{\Lambda})$, where $t \coloneqq 2m+r$.
	\item {\bf ($r=0$ instantiation.)} There is an \emph{$(h,m,0)$-polynomial decomposition} $\{f_{i,j} \mid i\in[\Delta(n)], j\in \{0,..., t \} \}$ of $C_n(\vec{\Lambda})$, where $m$ is the minimal integer such that $h^m\ge s$ and $t \coloneqq 2m$.
\end{enumerate}
where each of the foregoing decompositions satisfies the following properties.
\begin{enumerate}
	\item\label{item:faithful}
	(Faithful representation.)
	Let $H$ be any enumeration of $h$ in $\F_n$ and $\phi:[h]\to H$ be its corresponding enumerator. 
	Then, for every $i \in [\Delta(n)]$ and $\vec{\alpha} \in H^m$ representing a gate in the $i$\textsuperscript{th} layer, %
	it holds that\footnote{Recall from \cref{item:sumcheck} of \Cref{def:poly-dec} that $f_i$ is defined to be the layer polynomial $f_{i,t}$.} $f_i(\vec{\alpha})$ is the value of the gate $\vec{\alpha}$ in $C_n(\vec{\Lambda})$, i.e., $f_i(\vec{\alpha}) = \hat{f}_i(\vec{\alpha})$.
	\item\label{item:layer-poly-comp}(Layer Polynomial Computation.) It is possible to compute the layer polynomials efficiently by a uniform family of arithmetic networks as follows:
	\begin{enumerate}
		\item\label{item:input-layer-comp} (Input Layer.)
		There is a $\P$-uniform family of arithmetic networks of size $\poly(n,m,h,\log(s))$ 
		that takes as input $\vec{\Lambda}$ and outputs the description of an arithmetic circuit of size $O(n\cdot m\cdot h\cdot\polylog(s(n)))$ and degree $n\cdot h\cdot\polylog(s)$ that computes the polynomial $f_0(\vec{w})$.
		\item(Other Layers.)
		There is a $\P$-uniform family of arithmetic networks of size $\poly(s(n))$ %
		that takes as input $\vec{\Lambda}$ (as arithmetic data), and integers $i\in [\Delta(n)]$ and $j\in \{0,..., t\}$ (as boolean data) and outputs the description of an arithmetic circuit of size $\poly(s(n))$ and degree $h\cdot \polylog(s(n))$
		that computes $f_{i,j}(\vec{w},\vec{\sigma})$.%
	\end{enumerate}
	
	\item\label{item:dsr}
	(Layer downward self-reducibility.)
	There is a $\P$-uniform arithmetic network family of size $\poly(h,\log(s(n)))$ 
	that takes as input $\vec{\Lambda}$ (as arithmetic data) and an integer $i\in[\Delta(n)]$ (as boolean data), and outputs the description of an oracle arithmetic circuit of size and degree at most $h\cdot\polylog(s(n))$ that computes $f_{i,0}(\vec{w}, \vec{\sigma})$. 
	The only oracle access that the output circuit requires is for querying $f_{i-1}$ \emph{twice} on inputs in $H^m$. 
	\item\label{item:prop-sumcheck}
	(Sumcheck downward self-reducibility.)
	There is a $\P$-uniform arithmetic network family of size $\poly(h)$ 
	and degree $O(1)$ 
	that takes as input $\vec{\Lambda}$ (as arithmetic data) and integers $i\in[\Delta(n)]$ and $j\in\{0,..., t\}$(as boolean data), and outputs the description of an oracle arithmetic circuit of size $O(h)$ and degree $1$ that computes $f_{i,j}(\vec{w}, \vec{\sigma})$. 
	The only oracle access that the output circuit requires is for querying $f_{i,j-1}$ on inputs in $H^{m+t-j+1}$. 
	
	\item\label{item:ideg}
	(Low degree.)
	In the instantiation with $r>0$, the maximum individual degree of $f_{i,j}$ for each $(i,j)\in[\Delta]\times\{0,...,t\}$ is at most $6(h-1)$. Moreover, in either instantiation ($r>0$ or $r=0$), the total degree $\deg(f_{i,j})$ of each $f_{i,j}$ is bounded by $h\cdot\polylog(s)$.
	
\end{enumerate}
\end{proposition}

\bproof
As $\set{C_n}_{n \in \N}$ is a $\log^c$-uniform family of circuits, we know that its circuit structure function $\hat{\Phi}$ (see Definition \ref{def:suf-unif}) is computable by a $\P$-uniform family of formulas of size $(\log(s(n)))^c$. Note that the domain of $\hat{\Phi}$ is $[\Delta(n)]\times [s(n)]^3$. Let the $i$-th ``slice" be denoted by $\hat{\Phi}_i(\cdot)$ (that is, $\hat{\Phi}_i(w,u,v)=\hat{\Phi}(i,w,u,v)$). In the rest of the proof, we use the shorthand $s = s(n)$, $\Delta = \Delta(n)$, and $\F = \F_n$.

We construct two instantiations of polynomial decompositions (as described in \Cref{def:poly-dec}): one corresponding to $r>0$ and $H=\{0,1\}$, and the other corresponding to $r=0$. From the definition of the $f_i$'s in \cref{item:layer-polys} of \Cref{def:poly-dec}, together with \Cref{obs:var-ext}, in order to establish the claim about faithful representation, it suffices to find an appropriate arithmetization of the $\hat{\Phi}_i$’s. 
In the first case, we construct the arithmetization of an appropriate variable-extended formulation (see \Cref{def:var-ext}) of $\hat{\Phi}_i$ for some $0 < r = \polylog(s(n))$ to be specified soon, whereas in the second case, we do so for $\hat{\Phi}_i$ itself. Thus, to complete the description of the polynomial decomposition and conclude the proof, we now specify the two arithmetizations (correspondingly in the two subclaims below), and subsequently prove our claims regarding downward
self-reducibility and efficient computability (which is common to both cases).

\begin{subclaim} [$r>0,H=\{0,1\}$]\label{clm:arithmetization}
Let $H=\{0,1\}$ and let $r =\lfloor (\log s)^c \rfloor$. For each $i\in [\Delta]$, there is a polynomial $\Phi_i\in \F[z_1, \ldots, z_{3m+r}]$ that satisfies the following:
\begin{enumerate}
	\item For every $(\vec{w},\vec{u},\vec{v})\in \{0,1\}^{3m}$ (so that each of these vectors denotes the label of a gate in layer $i$ of $C_n$), we have that 
	\[
	\hat{\Phi}_i(\vec{w},\vec{u},\vec{v}) = \sum_{\vec{y}\in \{0,1\}^r}\Phi_i(\vec{w},\vec{u},\vec{v}, \vec{y}).
	\]
	\item The maximum individual degree of $\Phi_i$ is $2$ and the total degree is bounded by $\polylog(s)$.
	\item There is an arithmetic circuit that computes $\Phi_i$ that has size at most $\polylog(s)$.
\end{enumerate}
Moreover, there exists a $\P$-uniform family of arithmetic networks of size $\polylog(s)$ %
that takes as input an integer $i\in [\Delta]$ and outputs the description of the foregoing arithmetic circuit.
\end{subclaim}

\begin{subproof}
From \Cref{def:suf-unif}, we have a $\P$-uniform family of NAND-formulas of size $r' = \lfloor (\log s)^c \rfloor$ that compute $\hat{\Phi}_i$. For any given $n$, we assume that in the corresponding boolean NAND-formula, the domain of $\hat{\Phi}_i$ is $\zo^{3m}$. 

Next, we apply \Cref{lem:var-ext} to this boolean formula to obtain an arithmetic formula $\Phi'$ of size $5(\log s)^c$ over $\F$ that computes a variable-extended formulation of $\hat{\Phi}_i$. We define $\Phi_i$ to be the polynomial computed by this arithmetic formula. By~\Cref{lem:var-ext}, for every $(\vec{w},\vec{u},\vec{v})\in\{0,1\}^{3m}$ we have $\sum_{\vec{y}\in\{0,1\}^r}\Phi_i(\vec{w},\vec{u},\vec{v},\vec{y})=\hat{\Phi}_i(\vec{w},\vec{u},\vec{v})$, and the bounds on the individual degree, total degree, and size of $\Phi_i$ also follow immediately from the lemma's statement (note that $s_f=r'=\polylog(s)$ and $m=\log(s)$).

As for the $\P$-uniform family of networks that get $i$ and print a circuit for $\Phi_i$, the network has the descriptions of $\hat{\Phi}_1,\ldots,\hat{\Phi}_{\Delta}$ hard-wired.
This can be done in a $\P$-uniform manner, as each of $\hat{\Phi}_1, \ldots, \hat{\Phi}_{\Delta}$ can be printed in time $\polylog(s)$, since the circuit family $\{C_n\}$ is $\log^c$-uniform.
Given $i$, the network uses the network from~\Cref{lem:var-ext} to transform the description of $\hat{\Phi}_i$ into a description of $\Phi_i$.
\end{subproof}

\begin{subclaim} [$r=0$, $H$ arbitrary]\label{clm:arithmetization'}
For each $i\in [\Delta]$, there is a polynomial $\Phi_i\in \F[z_1, \ldots, z_{3m}]$ that satisfies the following:
\begin{enumerate}
	\item For every $(\vec{w},\vec{u},\vec{v})\in H^{3m}$ (so that each of these vectors denotes the label of a gate in layer $i$ of $C_n$), we have that $\hat{\Phi}_i(\vec{w},\vec{u},\vec{v}) = \Phi_i(\vec{w},\vec{u},\vec{v}).$
	\item The degree of $\Phi_i$ is bounded by $h\cdot \polylog(s)$.
	\item There is an arithmetic circuit that computes $\Phi_i$ that has size at most $h\cdot \polylog(s)$.
\end{enumerate}
Moreover, there exists a $\P$-uniform family of arithmetic networks of size 
$h\cdot \polylog(s)$ that takes as input an integer $i\in [\Delta]$ and outputs the description of the foregoing arithmetic circuit.

\end{subclaim}

\begin{subproof}
We think of $\hat{\Phi}_i$ as a function $\zo^{3\log s}\to \zo$, and note that it is computable by an arithmetic formula (over $\zo$), whose structure mimics the one of the original Boolean formula (i.e., each gate $\hat{w}$ in the original formula computes the NAND of two sub-formulas $\hat{u}$ and $\hat{v}$, and thus $\hat{w}$ can be replaced by the expression ${w} = 1 -{u}\cdot {v}$). The same arithmetic formula (now interpreted to be over $\F$) can be used to compute a polynomial ${\Phi}_i': \F^{3\log s} \to \F$ of size at most $3\cdot \size(\hat{\Phi}) = O((\log s)^c)$ and degree bounded by $O((\log s)^c)$. The latter follows from the simple observation that the degree of a polynomial computed by an arithmetic formula is always bounded by its size\footnote{Here is a simple proof by structural induction: the statement is clearly true for leaves, regardless of whether the leaf is labelled by a constant or a variable. Suppose $u$ is an  addition gate with children $v$ and $w$; then by the inductive hypothesis for $v$ and $w$, $\deg(u)\leq \max(\deg(v),\deg(w))\leq \deg(v) + \deg(w)\leq \size(v) + \size(w) \leq \size(u)$. If $u$ is a multiplication gate, then $\deg(u) = \deg(v) + \deg(w)\leq \size(v) + \size(w) \leq \size(u)$.}. 
We conclude that $\Phi_i'$ is a polynomial of degree $\polylog(s(n))$ over $\F$ computable by an arithmetic  formula of size $\polylog(s(n))$ that agrees with $\hat{\Phi}_i$ on $\zo^{3\log s}$.  

Next, let $\ell = \log h$, and for every $j \in [\ell]$ consider the function $\hat{\pi}_j : H \to \zo$ such that
$\hat{\pi}_j(a)$ is the $j$-th bit in the binary representation of the integer $\phi(a)$, where $\phi^{-1}$ is an enumerator for $H$ (see \Cref{def:enumeration}). 
Note that there is a polynomial $\pi_j : \F \to \F$ of degree at most $h$ that agrees with $\hat{\pi}_j$ on $H$. Finally, let $\Phi_i : \F^{3m} \to \F$ be the polynomial
\[
\Phi_i(z_1,\ldots,z_{3m}) \coloneqq \Phi_i'(\pi_1(z_1),\ldots,\pi_\ell(z_1),\ldots,\pi_1(z_{3m}),\ldots,\pi_\ell(z_{3m})).
\]
By definition, when ${\Phi}_i$ is given input $(\vec{w},\vec{u},\vec{v})\in H^{3m}$, the $\pi_j$’s project each of the three inputs
to a bit-string that is the index of a gate, and the arithmetic formula $\Phi_i'$ computes $\hat{\Phi}_i$ on the corresponding gates. Also, the size and degree of $\Phi_i$ are bounded by $h\cdot \size(\Phi') = h\cdot (\log (s(n)))^c = h\cdot \polylog(s(n))$. 
Finally, to see the ``moreover" part of the claim, we note that the arithmetic formulas for $\Phi_i$ are defined by composing $\Phi_i'$ (whose formula descriptions are already output by the $\P$-uniform family of arithmetic networks in \Cref{lem:var-ext}) and $\pi_j$'s (which have small $O(h)$-sized $\P$-uniform formulas as $H$ can be enumerated efficiently).
Applying \Cref{prop:comp-arith-network} then yields the result. %
\end{subproof}

To specify a polynomial decomposition, it suffices to specify $H,m,r,\Phi$, and $f_0$. Thus, given~\Cref{clm:arithmetization} and~\Cref{clm:arithmetization'}, all that is left in order to fully specify the two polynomial decompositions is $f_0$. 

Let $\vec{\Lambda}\vec{K}\in\F^s$ be the input layer in $C_n(\vec{\Lambda})$. Since $\{C_n\}_{n\in\N}$ is $\log^c$-uniform, there is a $\P$-uniform family $\{\Psi'_n\colon\F_n^n\times\F_n^{\log(s)}\rightarrow\F_n\}$ of arithmetic circuits such that $\Psi'_n(\vec{\Lambda},w)=(\vec{\Lambda}\vec{K})_w$ for all $w\in\{0,1\}^{\log(s)}$. Consider an arithmetic circuit $C_{\vec{\Lambda}}\colon\F^m\rightarrow\F$ that has $\vec{\Lambda}$ hard-wired, gets input $\vec{w}=(w_1,...,w_m)\in\F^m$, computes $z=(\pi_1(w_1),...,\pi_{\ell}(w_1),...,\pi_1(w_m),...,\pi_{\ell}(w_m))\in\zo^{\ell\cdot m=\log(s)}$ (where $\ell=\log(h)$ and the $\pi_i$'s are as in the proof of~\Cref{clm:arithmetization'}) and outputs $\Psi'_n(\vec{\Lambda},z)$. The circuit $C_{\vec{\Lambda}}$ is of size $h\cdot m\cdot n\cdot\polylog(s)$ and of degree $n\cdot h\cdot\polylog(s)$, and for every $\vec{w}\in H^m$ it outputs the $w^{th}$ entry of $\vec{\Lambda}\vec{K}$. Also, there is a $\P$-uniform arithmetic network that gets $\vec{\Lambda}$ and outputs the description of $C_{\vec{\Lambda}}$ in size $\poly(n,m,h,\log(s))$. We thus define $f_0$ to be the polynomial that $C_{\vec{\Lambda}}$ computes.

\paragraph{Layer Polynomial Computation:} 
We already proved the claim about $f_0$ above, so let us focus on other layers. Given $\vec{\Lambda}$ and $i\in[\Delta]$, consider the following arithmetic circuit for $f_i(\vec{w})$.

Consider all pairs of gates in the $(i-1)$-th layer of $C_n$ evaluated on the input $\vec{\Lambda}$ (represented by a pair of index vectors $(\vec{\beta},\vec{\gamma})$), apply the suitable arithmetic operation (addition if $i$ is odd, and multiplication if $i$ is even), multiply the result with $\sum_{\vec{\delta}\in H^r}\Phi_i(\vec{w},\vec{\beta},\vec{\gamma},\vec{\delta})$ if we are working with the $r>0$ instantiation (whose circuit is obtained from \Cref{clm:arithmetization} above) or $\Phi_i(\vec{w},\vec{\beta},\vec{\gamma})$ if we are working with the $r=0$ instantiation (whose circuit is obtained from \Cref{clm:arithmetization'} above), and finally add the result for all such pairs of gates in the $(i-1)$-th layer of $C_n$ evaluated on the input $\vec{\Lambda}$ (using \Cref{prop:universal-arith-network}). 

As the circuit size of $C_n$ is bounded by $s(n)$, the size of the sub-circuit that computes the evaluation of any gate in the $(i-1)$-th layer on input $\vec{\Lambda}$ is also bounded by $s(n)$. As the computation described above can also be performed by a sub-circuit of $\poly(s(n))$ size, it follows that the size of the circuit for $f_i(\vec{w})$ constructed above is also bounded by $\poly(s(n))$. 

Moreover, given as input $j\in \{0,...,t\}$, we can construct an arithmetic  circuit that computes $f_{i,j}$ by simply using the formula defining $f_{i,j}(\vec{w}, \vec{\sigma})$ in \cref{item:sumcheck} of \Cref{def:poly-dec}: enumerate over the $h^j$ different values for $\sigma_{t-j+1},\ldots,\sigma_{t}$ and use the foregoing circuits for $\Phi_i$ and $f_{i-1}$. As $h^j\leq h^m = O(s)$, the size of the overall circuit for $f_{i,j}$ is still $\poly(s(n))$. 

Finally, note that the implementation of the above using a $\P$-uniform family of arithmetic networks follows from a suitably repeated application of \Cref{prop:short-desc}, \Cref{prop:universal-arith-network}, and \Cref{prop:comp-arith-network}.%

\paragraph{Layer Downward Self-reducibility:} 
Given as input $\vec{\Lambda}$, an integer $i\in[\Delta(n)]$, and oracle access to $f_{i-1}$, in order to compute (the description of) an oracle arithmetic circuit for $f_{i,0}$, the arithmetic network does the following: It takes (the description of) the circuit for $\Phi_i$ as constructed above (accordingly in \Cref{clm:arithmetization} if $r>0$ and \Cref{clm:arithmetization'} if $r=0$), and computes the description of a circuit that is obtained by performing the two oracle calls to $f_{i-1}$ (as described in \cref{item:layer-polys} of \Cref{def:poly-dec}) and the appropriate arithmetic operation on the two values (depending on the parity of $i$), and multiplying this result with the circuit for $\Phi_i$.
The correctness of this circuit computing the desired polynomial $f_{i,0}$ again follows from \cref{item:sumcheck} of \Cref{def:poly-dec}. The size and degree bounds follow
from  \Cref{clm:arithmetization} and~\Cref{clm:arithmetization'}. 
The claim about the the polynomial identity $f_{i,t}(\vec{w}) = f_i(\vec{w})$ is by \cref{item:sumcheck} of \Cref{def:poly-dec}.

\paragraph{Sumcheck downward self-reducibility:} From \cref{item:sumcheck} of \Cref{def:poly-dec}, we have 
\mm{
f_{i,j}(\vec{\alpha},\beta_{1},\ldots,\beta_{t-j}) = \sum_{\beta_{t-j+1}\in H} f_{i,j-1}(\vec{\alpha},\beta_{1},\ldots,\beta_{t-j}, \beta_{t-j+1}),
}
and this expression can be computed by an oracle formula (with access to $f_{i,j-1}$) of size $O(h)$ and degree $1$. Therefore, given an integer $j\in\{0,..., t \}$ and oracle access to $f_{i,j-1}$, the arithmetic network computes its description simply using this equation (which uses the oracle to the circuit for $f_{i,j}$ $h$ times).
The claims about the $\P$-uniformity %
and the size and degree of the arithmetic network that computes this description are straightforward in this case. 

\paragraph{Low degree:}
The maximum individual degree of $f_{i}$ for each $i\in[\Delta]$ is at most that of $\Phi_{i}$ by \cref{item:layer-polys} of \Cref{def:poly-dec}. Therefore, it satisfies 
\[
\ideg(f_i) \le 
\begin{cases}
2(h-1) & \text{if $r > 0$} \\
h \cdot \polylog(s) & \text{if $r = 0$.}
\end{cases}
\]
In the $r>0$ instantiation, for each $(i,j)\in[\Delta]\times\{0,...,t\}$, we have that
\mm{
\ideg(f_{i,j}) \le \ideg(\Phi_i) + 2\ideg(f_{i-1})\leq \ideg(\Phi_i) + 2\ideg(\Phi_{i-1}) \leq 6(h-1) \;\;,
}
from \Cref{clm:arithmetization}. Moreover, in either instantiation, the total degree $\deg(f_{i,j})$ of each $f_{i,j}$ is bounded by $h\cdot\polylog(s)$.
\eproof

\section{A targeted version of the Guo--Kumar--Saptharishi--Solomon generator}\label{sec:gkss}

Our goal in this section is to construct an arithmetic reconstructive targeted hitting-set generator that works over fields of characteristic zero, or over sequences of finite fields of large characteristic. The generator is based on a family $\set{C_n}$ of circuits computing a hard problem, and is computable by a family of arithmetic networks. Whenever a distinguisher represented by a string $\vec{\Lambda}$ breaks the generator, another family $\set{R_n}$ of arithmetic networks computes the hard problem at input $\vec{\Lambda}$.

\begin{theorem} [algebraic GKSS-based targeted hitting-set generator] \label{thm:gkss-tarhsg}
Let $d,\nbar:\N\to\N$ be time-constructible functions where $d$ is polynomially-bounded, and let $\eps>0$ be a constant.
Let $\set{C_n}_{n \in \N}$ be a $\log^c$-uniform family of arithmetic circuits of size $s(n)$ and degree at most $s(n)$ having $n$ output gates, defined over a field family $\set{\F_n}_{n\in\N}$ where either (i) each $\F_n$ is a fixed field $\F$ of characteristic zero, or (ii) for each $n\in \N$, $\ch(\F_n) > s^{O(\eps/c)}\cdot \polylog(s) > \nbar$.
Then, there are two $\P$-uniform families of arithmetic networks $\set{G_n}$ and $\set{R_n}$ satisfying the following.
\begin{enumerate}
    \item {\bf (Generator.)} \label{item:gkss-tarhsg-generator} Networks in $\set{G_n}$ are of size $\poly(s)$.
    On input $\vec{\Lambda}\in\F_n^n$, the network $G_n$ outputs the description of an arithmetic circuit of size $\poly(s,\nbar)$ that computes a polynomial map
    \[
        \calH_{\vec{\Lambda}}\colon\F_n^{O(1/\eps)} \rightarrow \F_n^{\nbar}
    \]
    of degree $\deg(\calH_{\vec{\Lambda}})\le n\cdot s^{\eps}$.
    
    \item {\bf (Reconstruction.)} \label{item:gkss-tarhsg-reconstruction} Networks in $\set{R_n}$ are networks with $\pit$ gates and are of size $\poly(n,\nbar^{1/\eps},s^{\eps},d(\nbar))$ and degree $\poly(d(\nbar),s^{\eps},\nbar)^{\polylog(s)}$. %
    When $\vec{\Lambda}\in\F_n^n$ is a universal encoding of a degree-$d(\nbar)$ arithmetic circuit with $\nbar$ input gates such that $\vec{\Lambda}\circ \calH_{\vec{\Lambda}} \equiv 0$, then $R_n$ computes $C_n(\vec{\Lambda})$ on input $\vec{\Lambda}$. %
\end{enumerate}
\end{theorem}

We now prove \Cref{thm:gkss-tarhsg} based on claims that will be established throughout the section. We include the proof here (i.e., in advance) to assist in verification. The definitions and analyses of the generator networks $\set{G_n}$ and reconstruction networks $\set{R_n}$ appear in \Cref{sec:gkss-generator} and \Cref{sec:gkss-reconstruction}, respectively.

\begin{proof}[{\bf Proof of~\Cref{thm:gkss-tarhsg}}]
Let $c'>1$ be a sufficiently large universal constant. We use~\Cref{cons:GKSS-based-tarHSG} with the circuit family $\set{C_n}$ and with parameters $h=2^{\lceil (\eps/c)\log(s) \rceil}$ and $m=c'/\eps$ (so that indeed $h^m\ge s'$ is satisfied) and with output length $\nbar$. By~\Cref{lem:GKSS-based-gen}, the degree of $\calH_{\vec{\Lambda}}$ is indeed at most $n\cdot \tilde{O}(s^{\eps/c})<n\cdot s^{\eps}$, for every $\vec{\Lambda}\in \F_n^n$. %
The reconstruction in~\Cref{prop:GKSS-recons-arith-net}, instantiated with degree bound $d(\nbar)$, is a network of size at most %
$\poly(n,\nbar^{1/\eps},s^{\eps},d(\nbar))$ and degree at most $\poly(d(\nbar),s^{\eps},\nbar)^{\polylog(s)}$.%
\end{proof}

\subsection{The targeted hitting set generator} \label{sec:gkss-generator}

Loosely speaking, the targeted generator will get input $\vec{\Lambda}$ and apply the generator of~\textcite{GKSS22} to each of the polynomials in the polynomial decomposition of $C_n(\vec{\Lambda})$. We first recall their construction, and then describe our targeted generator.

\begin{construction}[\cite{GKSS22}] \label{cons:GKSS}
    Let $n \in \N$ and let $\F$ be a field with $\ch(\F) = 0$ or $\ch(\F) \ge n$.
    Let $P(z_1,\ldots,z_\ell) \in \F[z_1,\ldots,z_\ell]$ be an $\ell$-variate polynomial. %
    For $i \in [n]$, define $\Delta_{P,i}$ to be the $2\ell$-variate polynomial
    \[
        \Delta_{P,i}(z_1,\ldots,z_\ell,y_1,\ldots,y_\ell) \coloneqq \sum_{\substack{\vec{e} \in \Z_{\geq 0}^\ell \\ |\vec{e}| = i}} \frac{\vec{y}^{\vec{e}}}{\vec{e}!} \cdot \frac{\partial^{\vec{e}}}{\partial \vec{z}^{\vec{e}}}(P(z_1,\ldots,z_\ell)).
    \]
    Equivalently, the polynomial $\Delta_{P,i}$ is the component of $P(y_1 + z_1, \ldots, y_\ell + z_\ell)$ that is homogeneous of degree $i$ with respect to the $\vec{y}$ variables.
    The \emphdef{$n$-output GKSS generator with respect to $P$} is the polynomial map $\mathcal{H}_{\text{GKSS}, P, n} : \F^{2\ell} \to \F^{n}$ given by
    \[
        \mathcal{H}_{\text{GKSS}, P, n}(z_1, \ldots, z_\ell, y_1,\ldots,y_\ell) = (\Delta_{P,0}(\vec{z}, \vec{y}), \Delta_{P,1}(\vec{z},\vec{y}), \ldots, \Delta_{P, n-1}(\vec{z}, \vec{y})).
    \]
\end{construction}

\paragraph{Preliminary observations.} We will later use some basic observations about the generator of \Cref{cons:GKSS}, and a well-known fact, which we state here for convenience.

\begin{obs}[{\cite[Observation 2.11]{GKSS22}}] \label{obs:shift-invariance-gkss}
    Let $P \in \F[z_1,\ldots,z_\ell]$ and let $\vec{a} \in \F^\ell$.
    Letting $P'(\vec{z}) \coloneqq P(\vec{z} + \vec{a})$, we have
    \[
        \Delta_{P', i}(\vec{z}, \vec{y}) = \Delta_{P,i}(\vec{z} + \vec{a}, \vec{y})
    \]
    for all $i \in \N$.
    As a consequence, we have 
    \[
        \mathcal{H}_{\text{GKSS}, P', n}(\vec{z}, \vec{y}) = \mathcal{H}_{\text{GKSS}, P, n}(\vec{z} + \vec{a}, \vec{y})
    \]
    for all $n \in \N$.
\end{obs}

\begin{lem}\label{lem:z=0-gkss}
    Let $P \in \F[z_1,\ldots,z_\ell]$.
    For $i \in \N$, let $P_i$ denote the degree-$i$ homogeneous component of $P$.
    Then with $\Delta_{P,i}(\vec{z}, \vec{y})$ as in \Cref{cons:GKSS}, we have $\Delta_{P,i}(\vec{0}, \vec{y}) = P_i(\vec{y})$.
    As a consequence, we have $\mathcal{H}_{\text{GKSS}, P, n}(\vec{0}, \vec{y}) = (P_0(\vec{y}), \ldots, P_{n-1}(\vec{y}))$.
\end{lem}

\bproof
    Recall that $\Delta_{P,i}(\vec{z}, \vec{y})$ is, by definition, the component of $P(\vec{z} + \vec{y})$ that is homogeneous of degree $i$ with respect to the $y$ variables.
    Applying the evaluation $\vec{z} \mapsto \vec{0}$, we see that $\Delta_{P,i}(\vec{0}, \vec{y})$ is the degree-$i$ homogeneous of component of $P(\vec{0} + \vec{y}) = P(\vec{y})$, which is precisely $P_i(\vec{y})$.
\eproof

\begin{fact}[Euler’s formula for differentiation of homogeneous polynomials]\label{fact:euler} 
    If $A \in \F[x_1, \ldots, x_\ell]$ is a homogeneous polynomial of degree $t$, then 
    \[
        \sum_{i=1}^\ell x_i \cdot \frac{\partial}{\partial x_i} (A(\vec{x})) = t \cdot A(\vec{x}).
    \]
\end{fact}

\paragraph{The targeted generator.}
We now state the construction of our targeted hitting set generator that is based on~\Cref{cons:GKSS}.
The generator networks $\set{G_n}_{n \in \N}$ of \Cref{thm:gkss-tarhsg} will output arithmetic circuits that compute this targeted hitting set generator.

\begin{construction}[GKSS-based targeted HSG]\label{cons:GKSS-based-tarHSG}
    Let $\set{C_n}_{n \in \N}$ be a $\log^c$-uniform family of arithmetic circuits of size $s(n)$ and degree at most $s(n)$, defined over a field family $\{\F_n\}_{n\in\N}$ satisfying certain conditions specified in the next paragraph. %
    Let $\{C'_n\}$ be the $\log^{c+O(1)}$-uniform family of arithmetic circuits of size $s'(n)=\poly(s(n))$ and depth $\Delta(n)=O(\log(s(n))^2)$ obtained by applying \Cref{prop:uniform-depth-red} to $\set{C_n}$. 
	
    \medskip\noindent {\bf Parameters.} Let $h,m,\nbar : \N \to \N$ be time-constructible functions such that $h(n)$ is a power of $2$, $m(n)$ is the minimal integer for which $h(n)^{m(n)} \geq s'(n)$. %
    The field family $\{\F_n\}_{n\in\N}$ is such that either (i) each $\F_n$ is a fixed field $\F$ of characteristic zero, or (ii) for each $n\in \N$, $\ch(\F_n) > \max(\nbar,h(n)\cdot \polylog(s(n)))$. 

    \medskip\noindent {\bf Generator.} Given $\vec{\Lambda} \in \F_n^n$, let $\set{f_{i,j}}$ be the $(h,m,0)$-polynomial decomposition of $C'_n$ on input $\vec{\Lambda}\in \F_n^n$ obtained via \Cref{prop:poly-dec-cons} (i.e., the $r=0$ instantiation).
    We view all polynomials as having $\ell \coloneqq 3m$ variables, padding them with extra variables as necessary.
    The \emphdef{$\nbar$-output GKSS-based targeted HSG for $\vec{\Lambda}$} is the polynomial map $\calH_{\vec{\Lambda}}^{(h,m,\nbar)}:\F_n^{2\ell+2}\to \F_n^{\nbar}$ given by
    \[
        \mathcal{H}_{\vec{\Lambda}}^{(h,m,\nbar)}(w_1,w_2,\vec{z},\vec{y}) \coloneqq \sum_{i = 1}^{\Delta(n)} \sum_{j = 0}^{2m} L_{i,j}(w_1,w_2) \cdot \mathcal{H}_{\text{GKSS},f_{i,j},\nbar}(\vec{z},\vec{y}), 
        \eqtag{eq:GKSS-tarhsg-def}
    \]
    where $L_{i,j}(w_1,w_2)$ is the Lagrange interpolation polynomial that satisfies\footnote{\label{fn:lagrange-char}Note that the parameter settings imply that both $m$ and $\Delta(n)$ are bounded by $\polylog(s(n))$, so the assumption on the field characteristic implies that $[\Delta(n)]$ and $[2m]$ can be regarded as subsets of $\F_n$.}
    \[
		L_{i,j}(w_1,w_2) = \begin{cases}
			1 & \text{if $w_1 = i$ and $w_2 = j$,} \\
			0 & \text{if $(w_1,w_2) \in ([\Delta(n)]\times \set{0,...,2m}) \setminus \set{(i,j)}$.}
		\end{cases}
    \]		
\end{construction}

\begin{remark}\label{rem:char-GKSS}
    We remark here that the characteristic assumption $\ch(\F_n)> h(n)\cdot \polylog(s(n))$ above is made primarily because it allows us to assume that $\ch(\F_n) > \deg(f_{i,j})$ (see \cref{item:ideg} of \Cref{prop:poly-dec-cons}) for each pair $(i,j)$.
    In particular, this allows us to apply \Cref{fact:euler} to compute lower-order derivatives of $f_{i,j}$ from its higher-order derivatives, an operation that is used crucially later in the reconstruction procedure.
    For a more precise explanation of the reasons behind this assumption, see \Cref{rem:char-GKSS-final}.
\end{remark}

We now show that there is a $\P$-uniform family of arithmetic networks that take $\vec{\Lambda} \in \F_n^n$ as input and output a description of an arithmetic circuit that computes the corresponding GKSS-based targeted hitting set generator $\mathcal{H}_{\Lambda}$.
This corresponds to \cref{item:gkss-tarhsg-generator} in \Cref{thm:gkss-tarhsg}.

\begin{lem}[The complexity of the targeted generator] \label{lem:GKSS-based-gen}
    Let $\set{C_n}$, $h$, $m$, and $\nbar$ be as in the statement of \Cref{cons:GKSS-based-tarHSG}.
    Then, the following statements hold.
    \begin{enumerate}
    	\item There is a $\P$-uniform family of arithmetic networks $\set{G_n}_{n \in \N}$ where the network $G_n$ has size $\poly(s(n))$ and on input $\vec{\Lambda} \in \F_n^n$ outputs the description of an arithmetic circuit that computes $\calH_{\vec{\Lambda}}^{(h,m,\nbar)}$, the $\nbar$-output GKSS-based targeted hitting set generator for $\vec{\Lambda}$. %
    	\item The degree of $\calH_{\vec{\Lambda}}^{(h,m,\nbar)}$ is at most $n\cdot h(n)\cdot \polylog(s(n))$. %
    \end{enumerate}
\end{lem}

\bproof
    We stress that there are three algorithms involved in the claim. 
    Our goal is to construct a Turing machine $M$ that on input $1^n$ prints an arithmetic network $N$.
    The arithmetic network $N$ receives as input $\vec{\Lambda} \in \F_n^n$ and computes the description of an arithmetic circuit $H_{\vec{\Lambda}}$ that computes the generator $\calH_{\vec{\Lambda}}^{(h, m, \nbar)}$. We note that we have the basic chain of inequalities $s\geq n\geq \nbar$; the former is because $s$ is the size of an $n$-input circuit, and the latter is because $n$ is the length of the encoding of a circuit computing an $\nbar$-variate polynomial.

    We first describe the arithmetic circuit $H_{\vec{\Lambda}}$.
    We then explain how $N$ computes the description of $H_{\vec{\Lambda}}$ given $\vec{\Lambda}$ as input.
    Finally, we describe how $M$ prints a description of $N$ when given input $1^n$.
    
    \paragraph{The circuit $H_{\vec{\Lambda}}$.}
    We build the circuit $H_{\vec{\Lambda}}$ for $\calH_{\vec{\Lambda}}^{(h,m,\nbar)}$ by implementing \Cref{eq:GKSS-tarhsg-def}.
    To do this, we need small circuits for both the Lagrange interpolation polynomials $L_{i,j}$ and the generators $\calH_{\text{GKSS}, f_{i,j}, \nbar}$.
    
    The Lagrange interpolation polynomial $L_{i,j}(w_1, w_2)$ can be written explicitly as 
    \[
        L_{i,j}(w_1, w_2) \coloneqq \del{\prod_{k \in [\Delta]} \frac{w_1 - k}{i - k}} \cdot \del{\prod_{\ell \in \set{0,1,\ldots,2m}} \frac{w_2 - \ell}{j - \ell}}.
    \]
    Note that from the given characteristic assumption $\ch(\F_n) \ge h\cdot \polylog(s)$, it follows that $[\Delta],\{0,\ldots,2m\}\subseteq \F_n$ as we have $\ch(\F_n) \geq \max\set{2m, \Delta}$ because $\Delta$ is assumed to be bounded by $\polylog(s)$.
    It is clear from this expression that $\deg(L_{i,j}) \le O(\Delta m)$ and that there is an arithmetic formula of size $O(\Delta m) \le \polylog(s)$ that computes $L_{i,j}(w_1,w_2)$.

    To compute the generators $\calH_{\text{GKSS}, f_{i,j}, \nbar}$, we need arithmetic circuits that compute the layer polynomials $f_{i,j}$.
    By \cref{item:layer-poly-comp} of \Cref{prop:poly-dec-cons}, there are arithmetic circuits of size $\poly(s)$ %
    and degree at most $n \cdot h \cdot \polylog(s)$ that compute the polynomials $f_0$ and $\set{f_{i,j} : i \in [\Delta], j \in \set{0,\ldots,2m}}$.
    Recall that the outputs of the generator $\calH_{\text{GKSS}, f_{i,j}, \nbar}$ correspond to the components of $f_{i,j}(\vec{z} + \vec{y})$ that are homogeneous with respect to the $\vec{y}$ variables and are of degree less than $\nbar$.
    Applying \Cref{lem:network algo homogeneous components} to extract the $\vec{y}$-homogeneous components of $f_{i,j}(\vec{z} + \vec{y})$ %
    we obtain a multi-output arithmetic circuit of size $O(s \nbar^2)$ that computes $\calH_{\text{GKSS}, f_{i,j}, \nbar}$.

    Combining the Lagrange interpolation polynomials $L_{i,j}$ with the generators $\calH_{\text{GKSS}, f_{i,j}, \nbar}$ can be done using additional $O(\Delta m) \le \polylog(s)$ gates.
    Overall, we obtain a circuit $H_{\vec{\Lambda}}$ of size $\poly(s, \nbar) \le \poly(s)$ %
    and degree $n \cdot h \cdot \polylog(s)$ that computes $\calH_{\vec{\Lambda}}^{(h,m,\nbar)}$.

    \paragraph{The network $N$.}
    Recall that $N$ is an arithmetic network that receives $\vec{\Lambda} \in \F_n^n$ as input and outputs a description of the circuit $H_{\vec{\Lambda}}$.
    The outer summation in \Cref{eq:GKSS-tarhsg-def} and the Lagrange interpolation polynomials $L_{i,j}$ do not depend on the value of $\vec{\Lambda}$, so the descriptions of these components can be hard-wired into the network $N$.
    It remains to describe an arithmetic network that computes descriptions of circuits for the generators $\calH_{\text{GKSS}, f_{i,j}, \nbar}$.

    To do this, we invoke \cref{item:layer-poly-comp} of \Cref{prop:poly-dec-cons}, which provides us with a $\P$-uniform family of arithmetic networks of size $\poly(s)$ that take as input $\vec{\Lambda}$ and print descriptions of arithmetic circuits of size $\poly(s)$ and degree at most $n\cdot h\cdot\polylog(s)$ for the polynomials $f_{i,j}$.
    From a description of a circuit that computes $f_{i,j}(\vec{y})$, we can easily obtain a description of a circuit that computes $f_{i,j}(\vec{z} + \vec{y})$.
    We then apply \Cref{lem:network algo homogeneous components} %
    to the circuit for $f_{i,j}(\vec{z} + \vec{y})$, which yields the description of circuits that compute the components of $f_{i,j}(\vec{z} + \vec{y})$ that are homogeneous of degree less than $\nbar$ with respect to the $\vec{y}$-variables, which are precisely the outputs of the generator $\calH_{\text{GKSS}, f_{i,j}, \nbar}$.

    Thus, for each $(i,j) \in [\Delta] \times \set{0,\ldots,2m}$, we have an arithmetic network of size $\poly(s)$ that takes as input $\vec{\Lambda}$ and outputs a description of an arithmetic circuit for the component generator $\calH_{\text{GKSS}, f_{i,j}, \nbar}$.

    \paragraph{The Turing machine $M$.}
    Finally, we describe a polynomial-time Turing machine $M$ that takes $1^n$ as input and outputs a description of the arithmetic network $N$.
    The network $N$ has both the outer summation in \Cref{eq:GKSS-tarhsg-def} and the Lagrange interpolation polynomials hardwired as part of the network.
    These components can be computed in polynomial time by $M$.
    
    The part of $N$ that takes as input $\vec{\Lambda}$ and produces circuits for the component generators $\calH_{\vec{\Lambda}}^{(h,m,\nbar)}$ is likewise computable in polynomial time: 
    the reason is that \Cref{item:layer-poly-comp} of \Cref{prop:poly-dec-cons} provides a polynomial-time algorithm that takes $1^n$ as input and computes the description of a network that maps $\vec{\Lambda}$ to descriptions of circuits for the layer polynomials $f_{i,j}$.
    The subsequent homogenization of these circuits can be done by a $\P$-uniform arithmetic network using \Cref{lem:network algo homogeneous components}.
\eproof

\subsection{Reconstruction procedure}\label{sec:gkss-reconstruction}

Before we describe the reconstruction procedure for our GKSS-based targeted HSG, we need the following definition.

\begin{definition}[{Interpolating Sets, \cite{GKSS22}}]\label{def:interpolating-set}
    Let $\F$ be a field and let $\calP(\ell,d)$ denote the set of all $\ell$-variate degree-$d$ polynomials over $\F$.
    Let $M_{\ell,d} \coloneqq \binom{\ell + d}{d}$ be the number of $\ell$-variate monomials of total degree at most $d$.
    We say that a set of points $\set{\vec{a}_1, \ldots, \vec{a}_r} \subseteq \F^{\ell}$ is an \emphdef{interpolating set for $\calP(\ell, d)$} if the set of vectors
    \[
        \set{(\vec{a}_i^{\vec{e}} : \vec{e} \in \Z_{\ge 0}^\ell, |\vec{e}| \leq d) : i \in [r]} \subseteq \F^{M_{\ell,d}}
    \]
    spans $\F^{M_{\ell,d}}$. 
    In other words, for every $\vec{e} \in \Z_{\ge 0}^\ell$ satisfying $|\vec{e}| \le d$, there are field constants $\beta_1, \ldots, \beta_r \in \F$ such that for all $P(\vec{y}) \in \mathcal{P}(\ell, d)$, we have 
    \[
        \coeff_{\vec{y}^{\vec{e}}}(P(\vec{y})) = \sum_{i = 1}^r \beta_i \cdot P(\vec{a}_i).
    \]
\end{definition}

We now describe the reconstruction procedure for our targeted hitting set generator based on the GKSS generator.
This corresponds to \cref{item:gkss-tarhsg-reconstruction} in \Cref{thm:gkss-tarhsg}.

\begin{proposition}\label{prop:GKSS-recons-arith-net}
    Let $\nbar:\N \to \N$ be a time-constructible function and $\set{C_n}_{n \in \N}$ be a $\log^c$-uniform family of $n$-output arithmetic circuits of size $s(n)$ and degree at most $s(n)$, defined over a field family $\{\F_n\}_{n\in\N}$ where either (i) each $\F_n$ is a fixed field $\F$ of characteristic zero, or (ii) for each $n\in \N$, $\ch(\F_n) > h(n)\cdot \polylog(s)$ and for the same poly-logarithmically bounded function, we also have $h(n)\cdot \polylog(s)> \nbar(n)$. Let $d : \N \to \N$ be a polynomially-bounded time-constructible function. Then there is a $\P$-uniform family of arithmetic networks with $\pit$ gates $R = \set{R_n}_{n \in \N}$ that satisfies the following.
   
    \begin{enumerate}
        \item 
            $R_n$ receives as input a vector $\vec{\Lambda} \in \F_n^n$, where the vector $\vec{\Lambda}$ is promised to be the universal encoding of an arithmetic circuit $D$ that computes a nonzero $\nbar$-variate polynomial of degree at most $d = d(\nbar)$ %
            that satisfies $D\circ \calH_{\vec{\Lambda}}^{(h,m,\nbar)} = 0$, where $\calH_{\vec{\Lambda}}^{(h, m, \nbar)}$ is the generator described in \Cref{cons:GKSS-based-tarHSG}.
        \item
            The size\footnote{In the statement of this proposition, the notation $\poly(T)$ indicates a bound of $T^{cK}$, where $K$ is a universal constant and $c$ is the given constant that parameterizes the uniformity of $\{C_n\}$.} of the network $R_n$ is bounded by $\poly(n, d, h, \log s, \nbar^{m})$.
        \item 
            The degree of $R_n$ is bounded by $\poly(d,h,\nbar,\log s)^{\polylog(s)}$.
        \item
            $R_n$ computes the vector $C_n(\vec{\Lambda}) \in \F_n^n$ at input $\vec{\Lambda}\in \F_n^n$. %
    \end{enumerate}   
\end{proposition}

The rest of this subsection is devoted to the proof of \Cref{prop:GKSS-recons-arith-net}, which follows a similar structure to that of \Cref{prop:recons-arith-net}.
Recall that the generator $\calH_{\vec{\Lambda}}^{(h,m,\nbar)}$, defined in \Cref{cons:GKSS-based-tarHSG}, is obtained by instantiating copies of the GKSS generator (\Cref{cons:GKSS}) with polynomials coming from the polynomial decomposition of a circuit family $\set{C_n'}$ obtained by applying depth reduction to the circuit family $\set{C_n}$.
On input $\vec{\Lambda}$, the network $R_n$ will output the value of $C_n'(\vec{\Lambda})$, which equals $C_n(\vec{\Lambda})$ by correctness of the depth reduction.
The depth reduction ensures that the size of $C_n'$ is bounded by $\poly(s)$ and the depth of $C_n'$ is bounded by $\Delta \le O(\log^2 s)$.
For notational convenience, we will drop the prime symbol in $C_n'$ and instead refer to these circuits as $C_n$.

Before we describe the construction formally, we provide a brief sketch of its various components.
On input $\vec{\Lambda}$, the network efficiently constructs a description of a circuit computing the polynomial $f_\Delta = f_{\Delta,2m}$.
As this polynomial faithfully represents the output gate values of $C_n(\vec{\Lambda})$ (see \cref{item:faithful} of \Cref{prop:poly-dec-cons}), the network eventually computes $C_n(\vec{\Lambda})$ by simply evaluating $f_{\Delta}$ over the first $n$ many vectors (in lexicographical order) in $H^m$.
To construct the description of a circuit that computes the polynomial $f_\Delta$, the network $R_n$ iteratively constructs a small arithmetic circuit that computes $f_i$ for $i = 0, 1, \ldots, \Delta$. The circuit for $f_i$ is obtained by using query access to a circuit for $f_{i-1}$ to iteratively construct circuits for $f_{i,j-1}$ for $j = 1, \ldots, 2m$.
	
We will now describe the arithmetic network $R_n$, accounting for both the complexity of implementing each iteration (i.e., the complexity of the uniform arithmetic network) and the complexity of the arithmetic circuits whose descriptions are produced by each iteration.
Since $\vec{\Lambda}$ is the universal encoding of $D$, it follows from \Cref{thm:raz universal circuit} that $\size(D)\leq n$. 
To begin, $R_n$ employs the $\P$-uniform family of arithmetic networks from \Cref{lem:network-algo-uni-std} to convert $\vec{\Lambda}$ to a standard description of $D$ of $\poly(n,d)$ size. In the rest of this section, we abuse notation in the following way: whenever we refer to $\vec{\Lambda}$, we mean this standard description of $D$ of $\poly(n,d)$ size (unless specified otherwise). %

We start the iterative procedure using \Cref{item:input-layer-comp} of \Cref{prop:poly-dec-cons}, which allows the network $R_n$ to produce an arithmetic circuit of size $O(n m h \cdot \polylog(s))$ and degree $n h \cdot \polylog(s)$ that computes $f_0$.
Then the arithmetic network $R_n$ successively uses the networks given by the following lemma to construct circuits for the subsequent layer polynomials $f_{i,j}$.

\begin{lem}[One iteration: moving from a circuit for $f_{i,j-1}$ to a circuit for $f_{i,j}$] \label{lem:gkss-ind-recons}
    There is a $\P$-uniform family $\calN=\{\calN_n\}$ of  arithmetic networks with $\pit$ gates that when given as input $\vec{\Lambda}$ and a standard description of a size $s_{i,j-1}$, degree $d_{i,j-1}$ arithmetic circuit computing $f_{i,j-1}$, outputs a standard description of an arithmetic circuit computing $f_{i,j}$ of size $O(n\cdot d \cdot h^3\cdot \polylog(s)\cdot \nbar^{30m}) = \poly(n,d,h,\log s,\nbar^m)$ and degree $h\cdot \polylog(s)$. 
    The size of $\calN_n$ is $\poly(n, s_{i,j-1}, d_{i,j-1}, d, h, \log s, \nbar^{m})$ and its degree is at most $\poly( d_{i,j-1},d,h,\nbar,\log s)$. %
    
    Furthermore, there is another family $\calN'$ of arithmetic networks with $\pit$ gates with the same complexity as $\calN$ (i.e., $\P$-uniformity, size, and degree) that when given as input $\vec{\Lambda}$ and a standard description of a size $s_{i-1,2m}$, degree $d_{i-1,2m}$ arithmetic circuit computing $f_{i-1,2m}$, outputs the description of an arithmetic circuit computing $f_{i,0}$ of size $\poly(n,d,h,\log s,\nbar^m)$ and degree $h\cdot \polylog(s)$.
\end{lem}

We defer the proof of \Cref{lem:gkss-ind-recons} to \Cref{sec:gkss-recon-step}.
The remainder of this subsection completes the proof of \Cref{prop:GKSS-recons-arith-net} assuming \Cref{lem:gkss-ind-recons}.

The network $R_n$ makes repeated use of the networks $\calN_n$ and $\calN_n'$ from \Cref{lem:gkss-ind-recons}.
If we have computed the description of a circuit for $f_{i-1}$, then applying the network $\calN'_n$ of \Cref{lem:gkss-ind-recons} to the description of the circuit for $f_{i-1}$ yields a description of a circuit for $f_{i,0}$.
We then use $2m$ sequential applications of the network $\calN_n$ of \Cref{lem:gkss-ind-recons} to compute descriptions of circuits for $f_{i,1}$, $f_{i,2}, \ldots, f_{i,2m} = f_i$.
After having computed the description of a circuit that computes $f_{\Delta} = f_{\Delta, 2m}$, we use \Cref{prop:universal-arith-network} to query $f_{\Delta}$ on the $n$ lexicographically first elements of $H^m$.
Finally, the network outputs the values of $f_{\Delta}$ on these $n$ points, which by \cref{item:faithful} of \Cref{prop:poly-dec-cons}, equals $C_n(\vec{\Lambda})$. %

The network $R_n$ is obtained through one application of \Cref{item:input-layer-comp} of \Cref{prop:poly-dec-cons}, $2m\Delta$ applications of \Cref{lem:gkss-ind-recons}, and $n$ applications of \Cref{prop:universal-arith-network}.
For each application of \Cref{lem:gkss-ind-recons}, the size of the resulting circuit for $f_{i,j}$ is independent of the circuit for the preceding layer polynomial $f_{i,j-1}$.
In other words, letting $s_{i,j}$ denote the size of the circuit obtained for $f_{i,j}$, we have
\[
    s_{i,j} \le \poly(n, d, h, \log s, \nbar^{m}).
\]
In addition, the degree of $f_{i,j}$ is bounded by $h \cdot \polylog(s)$.
From the bound on the size of the circuit computing $f_{i,j}$, we see that the $2 m \Delta \le \polylog(s)$ applications of networks $\calN_n$ and $\calN_n'$ increase the size of $R_n$ by a total of $\poly(n, d, h, \log s, \nbar^{m})$.
The $n$ applications of \Cref{prop:universal-arith-network} correspond to a cost of 
\[
    \poly(n, s_{\Delta,2m}, \deg(f_{\Delta,2m})) \le \poly(n, d, h, \log s, \nbar^{m})
\]
in the size of $R_n$.
In total, the size of the network $R_n$ can be bounded by $\poly(n, d, h, \log s, \nbar^{m})$ as claimed.

To see the degree bound, first note that from each application of \Cref{lem:gkss-ind-recons}, the degree of the resulting circuit for $f_{i,j}$ is independent of that of the circuit for the preceding layer polynomial $f_{i,j-1}$. In fact, we have the bound $d_{i,j} = h\cdot \polylog(s)$ for all $i$ and $j$. 
Next, the network produced by \Cref{item:input-layer-comp} of \Cref{prop:poly-dec-cons} that produces a description of $f_0$ has degree $n\cdot h\cdot \polylog(s)$. By \Cref{prop:error-deg-comp}, we can bound the degree of $R_n$ by the product of the degrees of its constituent networks. 
To this end, note that each of the $2m\Delta \le \polylog(s)$ invocations of $\calN_n$ or $\calN_n'$ multiplies the degree of $R_n$ by $\poly(d, h, \nbar, \log s)$. Therefore, overall, we conclude that the degree of $R_n$ is bounded by $\poly(d, h, \nbar, \log s)^{\polylog(s)}$.

\subsubsection{A single iteration: Proof of \Cref{lem:gkss-ind-recons}} \label{sec:gkss-recon-step}

We first prove the main part of the statement, corresponding to the family $\set{\calN_n}_{n \in \N}$ of arithmetic networks with $\pit$ gates.
We then explain how to obtain the family $\set{\calN_n'}_{n \in \N}$, relying on essentially the same argument as in the construction of $\set{\calN_n}_{n \in \N}$.
	
Relabel the indices $i$ and $j$ in the lemma statement as $i'$ and $j'$, respectively, so that $i$ and $j$ are now free to use for other purposes in the remainder of the proof of this lemma.
From now on, we assume the indices $i'$ and $j'$ are fixed, so we refer to $f_{i',j'}$ as $f$ for simplicity.
We also use $s'$ and $d'$, respectively, to refer to the size and degree of the circuit for $f_{i', j'-1}$ that is given to $\calN_n$ as input. Let $d_f$ be the degree of $f$ (which by \cref{item:ideg} of \Cref{prop:poly-dec-cons} is bounded by $h\cdot \polylog (s)$).
	
Recall that in \Cref{cons:GKSS-based-tarHSG} we added dummy variables to all the $f_{i',j'}$'s so that they are all $3m$-variate polynomials.
We will repeatedly use the fact that given the circuit for $f_{i', j'-1}$, the network can evaluate $f$ at any given point in $\F_n^{3m}$ by using the downwards self-reducibility of the decomposition $\{f_{i,j}\}$.
Let us state this formally.
	
\begin{subclaim} \label{clm:gkss-ind-recons:dsr}
    There is a $\P$-uniform network of size $h\cdot \poly(s')$ and degree $O(1)$ that gets as input $(i',j')\in[\Delta]\times\{0,...,2m\}$ and a standard description of an arithmetic circuit of size $s'$ and degree $d'$ computing $f_{i',j'-1}$, and outputs a standard description of an arithmetic circuit of size $h\cdot\poly(s')$ and degree $d'$ that computes $f$.%
\end{subclaim}
	
\begin{subproof}
    The network relies on \Cref{prop:universal-arith-network} (to evaluate the circuit for $f_{i',j'-1}$) and on the downward self-reducibility of the polynomial decomposition $\set{f_{i,j}}$ (i.e., \Cref{item:prop-sumcheck} of \Cref{prop:poly-dec-cons}). 
\end{subproof}
	
\paragraph{High-level overview.} For a better exposition, we first describe the ideas involved in the construction of the arithmetic network $\calN_n$ in a list below and subsequently discuss the formal details of their implementation via a $\P$-uniform family of arithmetic networks with $\pit$ gates.
The reconstruction closely follows the argument laid out in \cite{GKSS22}, with additional ideas needed to ensure the uniformity of the reconstruction.

\begin{itemize}
    \item 
        From the equation
        \[
            D \circ \calH_{\vec{\Lambda}}^{(h,m,\nbar)}(w_1,w_2,\vec{z},\vec{y}) = 0,
        \]
        applying the partial evaluation $(w_1, w_2) \mapsto (i',j')$ yields
        \[
            D \circ \mathcal{H}_{\text{GKSS},f_{i',j'},\nbar}(\vec{z},\vec{y}) = 0.
        \]
        That is, the distinguisher $D$ for $\calH_{\vec{\Lambda}}^{(h, m, \nbar)}$ is necessarily a distinguisher for $\calH_{\text{GKSS}, f_{i',j'}, \nbar}$. Let the output of the generator $\calH_{\text{GKSS}, f_{i',j'}, \nbar}(\vec{z},\vec{y})$ (or just $\mathcal{H}_{\text{GKSS},f,\nbar}(\vec{z},\vec{y})$ in our simplified notation) be denoted by $(g_0(\vec{z}, \vec{y}),\ldots,g_\nbar(\vec{z}, \vec{y}))$.

    \item 
        {\bf (Distinguisher to ``predictor'' transformation.)} 
        For $a \in \set{0,1,\ldots,\nbar}$, define the $a^\text{th}$ hybrid polynomial 
        \[
            D_a(\vec{z},\vec{y},\vec{x}) = D(g_0(\vec{z},\vec{y}),\ldots,g_a(\vec{z},\vec{y}),x_{a+1},\ldots,x_\nbar).
        \]
        Because $D\circ \mathcal{H}_{\text{GKSS},f,\nbar}(\vec{z},\vec{y}) = 0$ and $D\neq 0$, it follows that there exists some $a \in \set{0,1,\ldots,\nbar-1}$ for which $D_a(\vec{z},\vec{y},\vec{x}) \neq 0$ but $D_{a+1}(\vec{z},\vec{y},\vec{x}) = 0$.
        Since we do not know which index $a$ satisfies this condition, the arithmetic network $\calN_n$ will simultaneously try all possible values of $a$ in $\nbar$ different branches.
        In what follows, we work within one such branch of the arithmetic network with a fixed value of $a$.
        
    \item
        {\bf (Guessing multiplicity of the generator as a root.)} 
        Suppose we have correctly guessed the value of $a$ for which $D_a(\vec{z}, \vec{y}, \vec{x}) \neq 0$ but $D_{a+1}(\vec{z}, \vec{y}, \vec{x}) = 0$.
        This implies that the distinguisher $D$ depends non-trivially on the variable $x_{a+1}$.
        Furthermore, since $D_a(\vec{z},\vec{y},\vec{x})$ vanishes on the substitution $x_{a+1} \mapsto g_{a+1}(\vec{z},\vec{y})$, it follows from \Cref{lem:gauss} that $(x_{a+1} - g_{a+1}(\vec{z}, \vec{y}))$ divides $D_a(\vec{z}, \vec{y}, \vec{x})$.
        Let $r$ be the highest power of $(x_{a+1} - g_{a+1})$ that divides $D_a$.
        Then we have
        \begin{align*}
            (\partial_{x_{a+1}}^{r-1} D_a) \circ \calH_{\text{GKSS},f,\nbar}(\vec{z},\vec{y}) &= 0 \\
            (\partial_{x_{a+1}}^{r} D_a) \circ \calH_{\text{GKSS},f,\nbar}(\vec{z},\vec{y}) &\neq 0.
        \end{align*}
        Note that we must have $r \le \deg(D_a) \le d$.
        We would like to work with this $r$, but since we do not know its value, the network $\calN_n$ simultaneously tries using all values of $b \in [d]$ as the value of $r$.
        Just as in the previous bullet, we now work within one such branch of the arithmetic network with a fixed value of $b$.
        
    \item 
        {\bf (Within branch $(a,b)$.)}
        We have now fixed guesses for the values of $a$ and $b$.
       From \Cref{lem:network algo partial derivative}, we can compute a description of a circuit of size $O(\size(D)\cdot d) = O(nd)$ for the partial derivative $\partial_{x_{a+1}}^b D_a$ using a $\P$-uniform family of arithmetic networks. 
        We make use of this circuit in the subsequent steps.
        \begin{itemize}
        
            \item{\bf (Testing for viability of the sub-branch.)}
            First, we verify, using $\pit$ gates, whether our choice of $(a, b)$ will allow us to reconstruct a circuit for the polynomial $f$. 
            For the reconstruction to succeed in this sub-branch, it is sufficient that the polynomial
            \[\Phi^{(a,b)}(\vec{z}, \vec{y}) \coloneqq (\partial^b_{x_{a+1}} D_a) \circ \calH_{\text{GKSS}, f, \nbar}(\vec{z}, \vec{y}) = (\partial^b_{x_{a+1}} D_a) \circ (\Delta_{f,0}(\vec{z}, \vec{y}), \Delta_{f,1}(\vec{z},\vec{y}), \ldots, \Delta_{f, \nbar-1}(\vec{z}, \vec{y}))\]
            is nonzero.\footnote{We do not justify this statement here; rather, this is an artifact of the reconstruction procedure of \cite{GKSS22} and becomes evident after a complete reading of the present high-level overview.}
            We first use the downward self-reducibility relation (\Cref{clm:gkss-ind-recons:dsr}) to compute the description of an arithmetic circuit $F$ of size $h \cdot s'$ and degree $d'$ that computes $f$.
            Using this, we can generate the description of a circuit $G$ of size $h \cdot s' + 2\ell = O(h\cdot s')$ and the same degree $d'$ that computes the $2\ell$-variate polynomial $f(\vec{y} + \vec{z})$.
            We then use the arithmetic network in \Cref{lem:network algo homogeneous components} to extract the $\vec{y}$-homogeneous components of degree at most $\nbar$ of this circuit, i.e., the polynomials $\Delta_{f,0}(\vec{z}, \vec{y}), \Delta_{f,1}(\vec{z},\vec{y}), \ldots, \Delta_{f, \nbar-1}(\vec{z}, \vec{y})$.
            Let $G_i$ denote the circuit of size $O(h s' i^2)$ that computes $\Delta_{f,i}$. From the aforementioned circuit description for the partial derivative $\partial_{x_{a+1}}^b D_a$ along with these descriptions of $G_i$'s, and applying \Cref{prop:comp-arith-network}, it follows that we can compute a description of a circuit for $\Phi^{(a,b)}(\vec{y},\vec{z})$ of size $O(nd\cdot hs'\nbar^2)$.
            
            We feed this description\footnote{Strictly speaking, $\pit$ gates expect \emph{universal} encodings as inputs but we note that it is possible to convert a standard encoding to a universal one (and vice-versa) without significant overhead (see \Cref{lem:network-algo-std-uni} and \Cref{lem:network-algo-uni-std}), and ignore this point about encoding format in the high-level overview.} to a $\pit$ gate and only proceed with subsequent steps in this sub-branch if the $\pit$ gate outputs that $\Phi^{(a,b)}$ is nonzero.
            Otherwise, the current branch detects its failure to reconstruct $f$ and aborts, returning (a description of) the zero circuit.

        \item{\bf (Ensuring nonzeroness of an auxiliary polynomial.)}
            In the rest of the high-level overview, we work with the assumption that the polynomial $\Phi^{(a,b)}(\vec{z}, \vec{y}) \coloneqq (\partial^b_{x_{a+1}} D_a) \circ \calH_{\text{GKSS}, f, \nbar}(\vec{z}, \vec{y})$ is nonzero. %
            In fact, for the reconstruction to succeed in this sub-branch, it is sufficient that the polynomial
            \[
                \Psi_f^{(a,b)}(\vec{y}) \coloneqq (\partial^b_{x_{a+1}} D_a) \circ \calH_{\text{GKSS}, f, \nbar}(\vec{0}, \vec{y})
            \]
            is nonzero.
            We would like to work with the assumption that the simpler polynomial $\Psi_f^{(a,b)}(\vec{y})$ is nonzero; however, it is possible that even in a sub-branch corresponding to ``correct'' guesses for $a$ and $b$ (i.e., when $\Phi^{(a,b)}$ is nonzero), the polynomial $\Psi_f^{(a,b)}(\vec{y})$ simplifies to the zero polynomial.
            In such a situation, since $\Phi^{(a,b)}(\vec{z},\vec{y})$ is a nonzero $\ell$-variate polynomial of degree at most $d\cdot d_f$ in the $\vec{z}$ variables over $\F_n(\vec{y})$, \Cref{lem:sz} implies that $\Phi^{(a,b)}(\vec{c}, \vec{y})$ is a nonzero polynomial for a random choice of $\vec{c} \in [2d \cdot d_f]^{\ell}$.
            By \Cref{obs:shift-invariance-gkss}, we have
            \[
                \Phi_{(a,b)}(\vec{z} + \vec{c},\vec{y}) = (\partial^b_{x_{a+1}} D_a) \circ \calH_{\text{GKSS}, f, \nbar}(\vec{z} + \vec{c}, \vec{y}) = (\partial^b_{x_{a+1}} D_a) \circ \calH_{\text{GKSS}, f', \nbar}(\vec{z}, \vec{y}),
            \]
            where $f'$ denotes the polynomial $f'(\vec{y}) \coloneqq f(\vec{y}+\vec{c})$ obtained by shifting the point $\vec{c}$ to the origin.
            In particular, $\Phi^{(a,b)}(\vec{c},\vec{y}) = (\partial^b_{x_{a+1}} D_a) \circ \calH_{\text{GKSS}, f', \nbar}(\vec{0}, \vec{y}) = \Psi_{f'}^{(a,b)}(\vec{y})$ is a nonzero polynomial for a random choice of $\vec{c} \in [2d \cdot d_f]^{\ell}$.
            
            Therefore, the next goal for our reconstruction network is to repeatedly make use of $\pit$ gates in order to find such a shift $\vec{c}$. 
            The reconstruction network attempts to find such a shift by fixing one coordinate of $\vec{c}$ at a time, using the standard search-to-decision reduction for $\pit$.
            Fix a constant $c_1\in [2d\cdot d_f]$. 
            From the aforementioned description of a circuit for $\Phi_{g}^{(a,b)}(\vec{y},\vec{z})$ of size $O(nd\cdot hs'\nbar^2)$, we generate the description of a circuit for the $(2\ell - 1)$-variate polynomial $\Phi_{g}^{(a,b)}(\vec{y},c_1,z_2,\ldots, z_\ell)$.
            We feed the resulting description to a $\pit$ gate. If for this choice of $c_1$, the polynomial is zero, then we try the next value for $c_1$ in $[2d\cdot d_f]$, and continue in this manner. If we find a choice of $c_1$ for which the polynomial is nonzero, we fix it and move on to checking the nonzeroness of the choices for $z_2$ upon substituting constants $c_2\in[2d\cdot d_f]$, and so on. 
            Note that since $\Phi^{(a,b)}(\vec{z}, \vec{y})$ is nonzero, from \Cref{lem:sz}, we are guaranteed to find a shift $\vec{c}$ using this procedure such that $\Psi_{f'}^{(a,b)}(\vec{y})$ is nonzero. 

            In the rest of this high-level overview, the reconstruction is truly for $f'$, not $f$; however, we remark that once we manage to reconstruct a small circuit for $f'$ at the end of this sub-branch, we can easily obtain one for $f$ by simply subtracting $\vec{c}$ from the inputs.
            For the sake of simpler notation, we ignore this point about translation by a random vector $\vec{c}$ in the rest of the high-level overview, and work with the assumption that we are in a sub-branch where $\Phi^{(a,b)}$ is nonzero, and furthermore, so is $\Psi(\vec{y}) \coloneqq (\partial^b_{x_{a+1}} D_a) \circ \calH_{\text{GKSS}, f, \nbar}(\vec{0}, \vec{y})$.
            From now on, we will drop the subscript $f$ in $\Psi$, as well as the superscript ${(a,b)}$, as we are working in a fixed sub-branch of the reconstruction.
         \item{\bf (Constructing a constrained interpolating set.)} The next step is to construct an explicit interpolating set for $\calP(\ell,\nbar)$ with the crucial property that $\Psi$ is nonzero at every point in this set (the reason for this demand will be explained momentarily). 
        Consider the following $M_{\ell,\nbar}\times M_{\ell,\nbar}$ variable matrix $A(\vec{y}_1,\ldots, \vec{y}_{M_{\ell,\nbar}})$:  the rows of $A$ are indexed by $[M_{\ell,\nbar}]$ and the columns indexed by $\ell$-variate monomials $\vec{x}^{\vec{e}}$ of degree at most $\nbar$, and the entry at $(i,\vec{x}^{\vec{e}})$ is $\vec{y}_i^{\vec{e}}$ (the monomial $\vec{x}^{\vec{e}}$ evaluated at $\vec{y}_i$). 
        The condition that a collection of vectors $\{\vec{a}_1,\ldots, \vec{a}_{M_{\ell,\nbar}}\}$ (where each $\vec{a}_i\in \F_n^\ell$) is an interpolating set is characterized by the nonzeroness of the determinant of $A(\vec{a}_1,\ldots, \vec{a}_{M_{\ell,\nbar}})$. 
        Therefore, it suffices to find vectors $\vec{a}_1,\ldots, \vec{a}_{M_{\ell,\nbar}}$ such that the evaluation of the polynomial $Q(\vec{y}_1,\ldots, \vec{y}_{M_{\ell,\nbar}})\coloneqq \det(A(\vec{y}_1,\ldots, \vec{y}_{M_{\ell,\nbar}}))\cdot \prod_{i=1}^{M_{\ell,\nbar}}\Psi (\vec{y}_i)$ at $(\vec{a}_1,\ldots, \vec{a}_{M_{\ell,\nbar}})$ does not vanish. %
        
        To do this, the network first uses a standard description of the $\P$-uniform family of arithmetic circuits of $\poly(M_{\ell,\nbar})$ size and degree from {\cite{Berkowitz84}} to compute $\det(A)$, and the description of the circuit for $\Psi(\vec{y_i}) = \Phi^{(a,b)}(\vec{0},\vec{y_i})$ computed in a previous bullet point to compute a description of $Q$. 
        Note that our assumption about $\Psi$ implies that $Q$ is a nonzero polynomial, and furthermore, note that the determinant of $A$ as a polynomial over $\ell\cdot M_{\ell,\nbar}$ many variables has the property that each of its rows is indexed by a disjoint set of variables. As a consequence, the individual degree of $Q$ is bounded by $\ideg(A)\cdot \deg(\Psi)\leq  \nbar + d\nbar$ and therefore, by \Cref{lem:sz-ideg}, we have that when the
        vectors $\vec{a}_1,\ldots, \vec{a}_{M_{\ell,\nbar}}$ are chosen uniformly at random from $[2\nbar+2d\nbar]^\ell$, then with high probability, they form an interpolating set for $\calP(\ell,\nbar)$ and $\Psi(\vec{a}_i) \neq 0$ for each $i\in [M_{\ell,\nbar}]$.
        
        The network finds such a non-vanishing point $(\vec{a}_1,\ldots, \vec{a}_{M_{\ell,\nbar}})$ for $Q$ by fixing one coordinate (out of $\ell\cdot M_{\ell,\nbar}$ in all) at a time, in a process similar to that described earlier for finding a good shift $\vec{c}$ described previously. We omit the details but note that this leads to at most $(2\nbar+2d\nbar)\cdot \ell\cdot M_{\ell,\nbar}$  calls to $\pit$ gates, where each such $\pit$ gate is fed a description of an arithmetic circuit of size bounded by $O(nd\cdot hs'\nbar^2)\cdot \poly(M_{\ell,\nbar}) = \poly(n,d,h,s',\nbar^m)$.
         \item{\bf (Induction setup.)}
            We now describe the setup of the inductive reconstruction of a small circuit for $f$. 
            The induction is on a parameter $j$ which takes values from $0$ up to $d_f-\nbar$. Note that this is where we use the given bound on $\nbar$ in the statement of \Cref{prop:GKSS-recons-arith-net}.
            At the end of the $j$-th step, we will have a circuit that computes all partial derivatives of order at most $\nbar$ of all homogeneous components of $f$ of degree up to $j+\nbar$.
            We now describe the steps of the induction argument more formally.
        \item 
            \textbf{(Base case $(j = 0)$.)}
            For each $\vec{e} \in \Z_{\ge 0}^\ell$ of weight $|\vec{e}| \le \nbar$ and each $t \le \nbar$, we can write the polynomial $\partial_{\vec{z}}^{\vec{e}} f_t$---where $f_t$ is the degree-$t$ homogeneous component of $f$---explicitly as a sum of at most $M_{\ell,\nbar} \coloneqq \binom{\nbar + \ell}{\ell}$ monomials,\footnote{This is overloaded notation, as $f_t$ has already been defined as a layer polynomial in \Cref{def:poly-dec}. However, since we are working with a fixed layer polynomial $f$ inside the proof of \Cref{lem:gkss-ind-recons}, there is no ambiguity, as we will never use $f_t$ to refer to a layer polynomial throughout the proof.} where recall from \Cref{cons:GKSS-based-tarHSG} that $\ell = 3m$ is a bound on the number of variables in $f$.
            Hence there is a multi-output circuit $B_0$ of size $s_0 = M_{\ell,\nbar}^2$ that computes the polynomials $\{\partial_{\vec{z}}^{\vec{e}}f_t: 0 \leq t \leq \nbar , |e| \leq \nbar\}$.

            To construct the circuit $B_0$, we first use the downward self-reducibility relation (\Cref{clm:gkss-ind-recons:dsr}) to compute the description of an arithmetic circuit $F$ of size $h \cdot s'$ and degree $d'$ that computes $f$.
            We use the arithmetic network in \Cref{lem:network algo homogeneous components} to extract the homogeneous components of degree at most $\nbar$ of this circuit.
            Let $F_i$ denote the circuit of size $O(h s' i^2)$ that computes $f_i$, the degree-$i$ homogeneous component of $f$.
            For $i \le \nbar$, the polynomial $f_i$ is an $\ell$-variate polynomial of degree at most $\nbar$.
            We can interpolate $f_i$ using a circuit of size at most $O(\nbar^{\ell+1} \ell)$ via
            \begin{equation}\label{eq:gkss-f-interpolation}
                f_i(\vec{z}) \coloneqq \sum_{\vec{\beta} \in T^{\ell}} f_i(\vec{\beta}) \prod_{k\in [\ell]} \prod_{\gamma \in T\setminus \set{\beta_k}} \frac{z_k - \gamma}{\beta_k - \gamma},
            \end{equation}
            where $T \subseteq \F_n$ is any finite set\footnote{\label{fn:T-large-enough}For example, one can take $T = [\nbar]$, as the characteristic is large enough (see \Cref{cons:GKSS-based-tarHSG}).} of size at least $\nbar + 1$.
            We can compute the value of $f_i(\vec{\beta})$ by evaluating the circuit $F_i$ at $\vec{\beta}$ using \Cref{prop:universal-arith-network}.

            Likewise, for any $\vec{e} \in \Z_{\ge 0}^\ell$ and any $t \le \nbar$, we can interpolate a circuit for $\partial^{\vec{e}}_{\vec{z}} f_t(\vec{z})$ of the same size.
            This requires a description of a circuit that computes $\partial^{\vec{e}}_{\vec{z}} f_t(\vec{z})$.
            We can obtain such a circuit of size $O(s \prod_{i=1}^\ell e_i)$ by repeatedly applying \Cref{lem:network algo partial derivative} to the circuit $F$ to compute the appropriate partial derivative of $f$.
       
        \item \textbf{(Induction hypothesis.)}
            We have the description of a circuit $B_{j-1}(\vec{z})$ that has size at most $s_{j-1}$.
            The circuit $B_{j-1}$ computes the $M_{\ell,\nbar} \cdot (\nbar + j - 1)$ polynomials  $\partial_{\vec{z}}^{\vec{e}}f_t$ for $|\vec{e}| \leq \nbar$ and $t \leq \nbar + j - 1$.
        \item
            \textbf{(Induction step.)} The goal by the end of this inductive step is to construct a circuit $B_j(\vec{z})$ of size at most $s_j$ (to be defined shortly) that computes $\partial_{\vec{z}}^{\vec{e}}f_t$ for $|\vec{e}| \leq \nbar$ and $t \leq \nbar + j$. 
   \begin{itemize}

   \item For a point $\vec{a} \in \F_n^\ell$, define 
   \[
       \Gamma_{j-1,\vec{a}} \coloneqq (\Delta_{f_{\leq \nbar + j -1},0}(\vec{z},\vec{a}),\ldots, \Delta_{f_{\leq \nbar + j -1},\nbar}(\vec{z},\vec{a})),
   \]
   where $f_{\leq \nbar + j -1}$ denotes the sum of the first $\nbar + j -1$ homogeneous components of $f$.
   The crucial observation in \cite{GKSS22} is that $\Delta_{f_{n+j},n}$ can be computed using the evaluation of our distinguisher $D$ on the points $\Gamma_{j-1,\vec{a}}$ as $\vec{a}$ ranges over the interpolating set we have constructed. 
   We state their main technical lemma below.\footnote{Strictly speaking, this lemma is only stated in \cite{GKSS22} for $a = \nbar - 1$ and $b = 0$ in our notation; but it is easy to verify that it also holds in this slightly more general form.}
   \begin{lem}[{\cite[Lemma 3.2]{GKSS22}}]\label{lem:gkss-technical}
       Let $\vec{a}\in\F_n^\ell$ be such that $\Psi(\vec{a})\neq 0$. Then, %
       \[
       \left(\frac{-1}{\Psi(\vec{a})}\right) (\partial_{x_{a+1}}^{b+1}D_a)(\Gamma_{j-1,\vec{a}}) = \Delta_{f_{\nbar+j},\nbar}(\vec{z},\vec{a}) \pmod{\langle \vec{z}\rangle^{j+1}}.
       \]
   \end{lem}
   \item We begin the reconstruction of $f_{\nbar + j}$ with the circuit $B_{j-1}(\vec{z})$ that computes every $\partial_{\vec{z}}^{\vec{e}}f_t$ for $|\vec{e}| \leq \nbar$ and $t \leq \nbar + j -1$. 
   By taking suitable linear combinations of the output gates, $\calN$ can create a new circuit $B$ of size at most $s_{j-1}+M_{\ell,\nbar}^5$ that computes the points $\{\Gamma_{j-1,\vec{a}_r} : r \in [M_{\ell,\nbar}]\}$.
   Using \Cref{lem:gkss-technical} for each $\vec{a}_i$, it then obtains a circuit of size $s_{j-1} +M_{\ell,\nbar}^5 +n\cdot d\cdot M_{\ell,\nbar}$ that computes $\{\Delta_{f_{\nbar+j},\nbar}(\vec{z},\vec{a}_r) : r \in [M_{\ell,\nbar}]\}$ modulo the ideal $\langle \vec{z}\rangle^{j+1}$.
    \item 
    We can regard $\Delta_{f_{\nbar+j},\nbar}$ as an $\ell$-variate degree-$\nbar$ polynomial in $\F_n(\vec{z})[\vec{y}]$.
    By definition, we can write $\Delta_{f_{\nbar+j},\nbar}(\vec{z},\vec{a})$ as a linear combination of the $\nbar$-th order partial derivatives of  $f_{\nbar+j}(\vec{z})$.
    As $\{\vec{a}_1,\ldots,\vec{a}_{M_{\ell,\nbar}}\}$ is chosen to be an interpolating set for $\calP(\ell,\nbar)$, each $\partial_{\vec{z}}^{\vec{e}}f_{\nbar + j}$ with $|\vec{e}| = \nbar$ can be written as a suitable linear combination of $\{\Delta_{f_{\nbar+j},\nbar}(\vec{z},\vec{a}_r) : r \in [M_{\ell,\nbar}]\}$. %
 More precisely, from \Cref{cons:GKSS}, including the definition of $\Delta_{P,i}$, and the ``in other words'' part of \Cref{def:interpolating-set}, we have
 \begin{equation}\label{eq:Delta-lc}
     \left(\frac{1}{\vec{e}}\right)\partial_{\vec{z}}^{\vec{e}}f_{\nbar + j}(\vec{z})=\coeff_{\vec{y}^{\vec{e}}}(\Delta_{f_{\nbar+j},\nbar}(\vec{z},\vec{y})) = \sum_{i=1}^{M_{\ell,\nbar}} \beta_i \cdot \Delta_{f_{\nbar+j},\nbar}(\vec{z},\vec{a}_i) = \sum_{i=1}^{M_{\ell,\nbar}}\beta_i\cdot \sum_{\substack{\vec{\alpha} \in {\Z_{\geq 0}^\ell} \\ |\vec{\alpha}| = \nbar}} \frac{\vec{a}_i^{\vec{\alpha}}}{\vec{\alpha}!} \cdot \partial_{\vec{z}}^{\vec{\alpha}}(f_{\nbar+j}(\vec{z}))
 \end{equation}
 for some $\set{\beta_1,\ldots,\beta_N} \subseteq \F_n(\vec{z})$.
 Note that this amounts to solving the linear system $L\vec{\beta} = \vec{\chi}_{\vec{e}}/\vec{e}!$.
 Here, $L$ is an $M_{\ell,\nbar} \times M_{\ell,\nbar}$ matrix whose rows are indexed by exponent vectors $\set{\vec{\alpha} \in \Z_{\ge 0}^\ell : |\vec{\alpha}| \le \nbar}$, whose columns are indexed by $[M_{\ell,\nbar}]$, and whose $(\vec{\alpha}, i)$ entry is $\vec{a}_i^{\vec{\alpha}}/\vec{\alpha}!$.
 The vector $\vec{\beta} = (\beta_1,\ldots,\beta_N)$ is the vector of indeterminates we want to solve for, and the right-hand side $\vec{\chi}_{\vec{e}}$ is the indicator vector for $\vec{e}$.
 
 Because the set $\set{\vec{a}_1, \ldots, \vec{a}_N}$ is an interpolating set for $\calP(\ell,\nbar)$, the matrix $L$ is invertible.
 Hence we can solve this system of equations for $\vec{\beta}$ by inverting $L$ (which can be done with a $\P$-uniform arithmetic circuit of polynomial size \cite{Berkowitz84}) and computing the matrix-vector product $L^{-1} \vec{\chi}_{\vec{e}}$.
 For $\vec{e} \in \Z_{\ge 0}^\ell$ with $|\vec{e}| = \nbar$, this lets us write the partial derivative $\partial^{\vec{e}}_{\vec{z}}f_{\nbar + j}(\vec{z})$ as a linear combination of the evaluations $\set{\Delta_{f_{\nbar + j}, \nbar}(\vec{z}, \vec{a}_r) : r \in [M_{\ell,\nbar}]}$, which we have already computed in the previous step.

\begin{remark}\label{rem:Delta-lc}
    We use the field characteristic assumption here to represent the derivative $\partial_{\vec{z}}^{\vec{e}}f_{\nbar + j}(\vec{z})$ as a linear combination of the evaluations $\set{\Delta_{f_{\nbar + j}, \nbar}(\vec{z}, \vec{a}_r) : r \in [M_{\ell,\nbar}]}$ (using the formula described in \Cref{cons:GKSS} in \cref{eq:Delta-lc} and in the subsequent formulation of the linear system). %
    The assumption that $\ch(\F) = 0$ or $\ch(\F_n) > d_f = h \cdot \polylog(s)$ %
    ensures that the linear system $L \vec{\beta} = \vec{\chi}_{\vec{e}} / \vec{e}!$ is well-defined.
\end{remark}

 \item 
 We have now computed all $\nbar$-th order partial derivatives of $f_{\nbar+j}$.
 Because $f_{\nbar+j}$ is a homogeneous polynomial, an arithmetic network can also compute all lower order derivatives via repeated applications of Euler’s formula (\Cref{fact:euler}).

 \begin{remark}\label{rem:Euler-char}
    It is crucial in these applications of \Cref{fact:euler} that $\deg(f_{\nbar + j})$ is nonzero in $\F$.
    The value of $\deg(f_{\nbar + j})$ can be as large as $d_f = h \cdot \polylog(s)$.
    Because we assume that $\ch(\F) = 0$ or $\ch(\F_n) > h \cdot \polylog(s)$, we are guaranteed that $\deg(f_{\nbar+j})$ is nonzero in $\F$.
 \end{remark}

 Combined with the outputs of $B_{j-1}(\vec{z})$, we have a circuit $B'_j(\vec{z})$ of size $s_{j-1}+ M_{\ell,\nbar}^{10} +n\cdot d\cdot M_{\ell,\nbar}$ that computes the polynomials
 \[
     \{\partial_{\vec{z}}^{\vec{e}}f_t: |\vec{e}| \leq \nbar, t \leq \nbar + j - 1\} \cup \{\partial_{\vec{z}}^{\vec{e}}\tilde{f}_{\nbar+j}: |\vec{e}| \leq \nbar\},
 \]
 where $\partial_{\vec{z}}^{\vec{e}}\tilde{f}_{\nbar+j} = \partial_{\vec{z}}^{\vec{e}}{f}_{\nbar+j} \pmod{\langle z\rangle^{\nbar+j-|\vec{e}|+1}}$ for every $|\vec{e}| \leq \nbar$. 
 \item 
 The last remaining task is to fix the errors present in the polynomials $\partial^{\vec{e}}_{\vec{z}}\tilde{f}_{\nbar + j}$ and obtain circuits that correctly compute the desired partial derivatives $\partial^{\vec{e}}_{\vec{z}} f_{\nbar + j}$.
 The circuit $B'_j$ is a composition of a circuit of size $M_{\ell,\nbar}^{10} + n\cdot d\cdot M_{\ell,\nbar}$ with the homogeneous circuit $B_{j-1}$ of size $s_{j-1}$. 
 Using \Cref{lem:network algo partial hom}, 
 we extract the lowest degree homogeneous parts of the outputs of $B'_j \circ B_{j-1}$ by constructing an equivalent homogeneous circuit (which we call $B_j$) of size at most $s_{j-1} + ((M_{\ell,\nbar}^{10} + n\cdot d\cdot M_{\ell,\nbar}) \cdot d_f^2)$ that computes $\{\partial_{\vec{z}}^{\vec{e}}f_t:|\vec{e}| \leq \nbar, t \leq \nbar + j \}$. 
 This completes the induction step. 
\end{itemize}

\item Unraveling the induction, for $d_f -\nbar$ steps, we eventually obtain a circuit of size
 at most $s_{d_f-\nbar} = O(n\cdot d\cdot d_f^3 \cdot M_{\ell,\nbar}^{10}) = O(n\cdot d \cdot h^3\cdot \polylog(s)\cdot \nbar^{30m})$ that computes all the partial
 derivatives of order at most $\nbar$ of $f_0,\ldots,f_{d_f}$.
 The zeroth order partial derivatives are precisely $f_0,\ldots,f_{d_f}$, and these polynomials can be summed to produce a circuit for $f$ of size $O(n \cdot d \cdot h^3 \cdot \polylog(s) \cdot \nbar^{30m})$.\footnote{It is at this point that the network undoes the initial translation by $\vec{c}$ to recover a circuit of the same asymptotic size for the original polynomial.}

\item The network outputs the description of a candidate circuit produced by any sub-branch of the computation which does not abort after testing viability, breaking ties arbitrarily.

\end{itemize}
\end{itemize}

\paragraph{Details of implementation.} We now formalize the details of the implementation of this reconstruction procedure using a sequence of claims below.

\begin{subclaim} [evaluating a single sub-branch] \label{clm:gkss-recon-branch}
		There is a $\P$-uniform family $\Nin=\{\Nin_n\}$ of arithmetic networks with $\pit$ gates with the following properties.
        \begin{itemize}
            \item It takes as input $\vec{\Lambda}\in\F_n^n$, integers $a\in[\nbar-1]$ and $b\in [d]$, and a standard description  of a size $s'$, degree $d'$ arithmetic circuit computing $f_{i',j'-1}$.
            \item It outputs a Boolean value $v_{a,b}$ along with a standard description of a circuit $\calC_{a,b}$ such that:
            \begin{itemize}
                \item if $(\partial^b_{x_{a+1}} D_a) \circ \calH_{\text{GKSS}, f, \nbar}(\vec{z}, \vec{y})$ is nonzero, then $v_{a,b} = 1$ and $\calC_{a,b}$ has size $O(n\cdot d \cdot h^3\cdot \polylog(s)\cdot \nbar^{30m})$, degree $d_f$, and computes $f$, and 
                \item if $(\partial^b_{x_{a+1}} D_a) \circ \calH_{\text{GKSS}, f, \nbar}(\vec{z}, \vec{y}) = 0$, then $v_{a,b} = 0$, and $\calC_{a,b}$ is a description of the zero circuit.
            \end{itemize}
            \item Moreover, $\Nin_n$ has size $\poly(n, d, h, \log s, \nbar^{m}, s',d')$ and degree at most $\poly(d,d',h,\nbar,\log s)$. %
        \end{itemize} 
\end{subclaim}

    \begin{subproof}
        We verify that within any sub-branch $(a,b)$ where $(\partial^b_{x_{a+1}} D_a) \circ \calH_{\text{GKSS}, f, \nbar}(\vec{z}, \vec{y})$ is nonzero, the base case as well as the inductive step as laid out in the overview above can be performed using a $\P$-uniform family of arithmetic networks with $\pit$ gates, and subsequently analyze the complexity of the subroutines involved. %
        
        \paragraph{Testing for viability of $(a,b)$:}  As described in the high-level overview, $\Nin$ feeds a description of $\Phi^{(a,b)}$ to a $\pit$ gate, and reports its output as $v_{a,b}$. If $v_{a,b} = 0$, then it outputs a description of the zero circuit (and if not, then it proceeds with the implementing the rest of the high-level overview and eventually outputs a standard description of $B_{d_f-\nbar}$ that we subsequently describe).
        
        Let us now analyze the network size and degree complexity of the corresponding sub-procedure laid out in the high-level overview. First, the network that $\Nin_n$ uses from \Cref{clm:gkss-ind-recons:dsr} as a subroutine to construct a description of $G$ has size $O(h\cdot s')$ and degree $O(1)$. Next, the network from \Cref{lem:network algo homogeneous components} used to extract circuits for $\Delta_{f,i}$ (for $0\leq i \leq \nbar -1$) has size $\poly(\size(G),\deg(G)) = \poly(h,s',d')$ and degree $O(1)$.
        The network from \Cref{lem:network algo partial derivative} that is used to compute a description of the partial derivative $\partial_{x_{a+1}}^b D_a$ has size $\poly(n,d)$ and degree $O(1)$. 
        Finally, the arithmetic network from \Cref{prop:comp-arith-network} used to generate the description of a circuit for $\Phi^{(a,b)}(\vec{y},\vec{z})$ (whose corresponding universal encoding after applying \Cref{lem:network-algo-std-uni} is fed into a $\pit$ gate) has size $\poly(\nbar, \ell, \size(\partial^b_{x_{a+1}} D_a), \size(G)) = \poly(\nbar, \ell, n,d,h,s',d')$.
        \paragraph{Finding a good shift:} To find a shift $\vec{c}\in \F_n^\ell$ such that $\Psi_{f'}^{(a,b)}(\vec{y})$ is nonzero (where $f'(\vec{z})\coloneqq f(\vec{z} + \vec{c})$), $\Nin_n$ uses at most $2\ell\cdot d\cdot d_f$ $\pit$ gates, each of which are fed a description of (an appropriate partial substitution of) the circuit described above for $\Phi^{(a,b)}(\vec{y},\vec{z})$.  

        \paragraph{Computing a constrained interpolating set:}
        As described in the high-level overview, to construct an interpolating set, $\Nin_n$ makes at most at most $(2\nbar+2d\nbar)\cdot \ell\cdot M_{\ell,\nbar}$  calls to $\pit$ gates, where each such $\pit$ gate is fed a description of an arithmetic circuit for (an appropriate partial substitution of) $Q$ of size bounded by $O(nd\cdot hs'\nbar^2)\cdot \poly(M_{\ell,\nbar}) = \poly(n,d,h,s',\nbar^m)$.
         \paragraph{Base Case:} Given the description of a circuit for $f_{i',j'-1}$, the network $\Nin_n$ uses it to construct the description of the circuit $F'$ for $f'(\vec{z}) = f(\vec{z}+\vec{c})$ in \Cref{clm:gkss-ind-recons:dsr} of size $O(hs')$ and degree $d'$.
        Next, it uses the network in \Cref{lem:network algo homogeneous components}, which has size $\poly(h,s',d')$ and degree $O(1)$ to compute the descriptions of circuits for the homogeneous components $F_0',\ldots, F_\nbar'$, where each of these circuits has size bounded by $\poly(h,s',\nbar)$.
        Using \cref{eq:gkss-f-interpolation} and \Cref{prop:universal-arith-network}, the network $\Nin_n$ constructs a circuit for each $f_i'$ (for $i \leq \nbar$) of size at most $M_{\ell,\nbar}^2$.

        Using the same procedure applied to circuits that compute partial derivatives of $F_i$, the network $\Nin_n$ also constructs circuits for the derivatives $\set{\partial_{\vec{z}}^{\vec{e}}f_i' : \vec{e} \in \Z_{\ge 0}^\ell, |\vec{e}| \le \nbar}$ up to order $\nbar$ of $f_i'$.
        We can obtain a circuit for $\partial^{\vec{e}}_{\vec{z}} F_i'$ of size $O(s' \prod_{i=1}^\ell e_i) \le O(s' \nbar ^\ell) = O(s' \nbar^{m})$. %
        by repeatedly applying \Cref{lem:network algo partial derivative} to the circuit that computes $F_i'$.
        The interpolation procedure described in the previous paragraph then yields a circuit of size at most $M_{\ell,\nbar}^2$ for the corresponding partial derivative of $f_i'$. As in the high-level overview, we drop the primed notation (i.e., we continue to work with $f$ instead of $f'$) for the rest of this proof to simplify notation.
        
        \paragraph{Inductive Step:} 
        Next, for each point $\vec{a}_i$, the network computes the vector \linebreak $\Gamma_{j-1,\vec{a}} = (\Delta_{f_{\leq \nbar + j -1},0}(\vec{z},\vec{a}),\ldots, \Delta_{f_{\leq \nbar + j -1},\nbar}(\vec{z},\vec{a}))$ by using the formula for $\Delta_{P,i}$ in \Cref{cons:GKSS}.
        Note that each $\Delta_{f_{\leq \nbar + j -1},i}$ is a linear combination of the partial derivatives of $f_{\leq \nbar + j -1}(\vec{z})$ of order $i$, of which there are at most $M_{\ell,\nbar}$, and $\Nin_n$ has access to the descriptions of all derivatives of $f_{\leq \nbar + j -1}(\vec{z})$ of order up to $\nbar$ by the inductive hypothesis.
        As mentioned in the overview, the network constructs a circuit $B$ of size at most $s_{j-1}+M_{\ell,\nbar}^5$, that computes $\{\Gamma_{j-1,\vec{a}_r} : r \in [M_{\ell,\nbar}]\}$.
        Furthermore, this computation can be performed using a $\P$-uniform family of networks of the same size and degree $O(1)$ (from \Cref{prop:comp-arith-network}).

        Next, in order to apply \Cref{lem:gkss-technical}, the network needs to evaluate $\Psi(\vec{a}_i)$.
        For this, it uses the circuits obtained in the base case for $f_i$ (for $i\leq \nbar$) together with \Cref{lem:z=0-gkss} and \Cref{prop:universal-arith-network}.
        Also, from \Cref{lem:network algo partial derivative}, it can use a $\P$-uniform family of arithmetic networks of size $\poly(n,d)$ and degree $O(1)$ to compute a circuit of size $O(nd)$ for $(\partial_{x_{a+1}}^{b+1} D_a)(\vec{x})$.
        Overall, for each $i\in[M_{\ell,\nbar}]$, it computes the left-hand side of \Cref{lem:gkss-technical} using a $\P$-uniform family of arithmetic networks of size $\poly(n,d,M_{\ell,\nbar},s_{j-1})$ and degree $\poly(d,\nbar)$ (this again follows from \Cref{prop:comp-arith-network}).
        Next, it uses a $\P$-uniform circuit of size and degree $\poly(M_{\ell,\nbar})$ to solve the appropriate linear system, writing the order-$\nbar$ partial derivatives of $f_{\nbar + j}$ in terms of the evaluations obtained through \Cref{lem:gkss-technical}.
        It then uses the equation in \Cref{fact:euler} to compute circuits for the derivatives of $f_{\nbar+j}$ of all lower orders. 
        Finally, it uses the $\P$-uniform network family in \Cref{lem:network algo partial hom} of size $\poly(M_{\ell,\nbar},n,d,s_{j-1},\deg(B_j) = O(d_f))$ and degree $O(1)$ to construct the circuit $B_j$ to conclude the inductive step. 
        Note that there are  $d_f - \nbar = O(d_f) = h\cdot \polylog(s)$ iterations of this procedure.

        As argued in the overview above, the size of the circuit $B_{d_{f}-\nbar}$ produced in the final iteration of the inductive step is bounded by $s_{d_f-\nbar} = O(n\cdot d\cdot d_f^3 \cdot M_{\ell,\nbar}^{10}) = \poly(n, d, h, \log s, \nbar^{m})$.
        This implies that the bounds mentioned in this proof that rely on $s_{j-1}$ can also be bounded by this quantity. We conclude that $\Nin$ has size $\poly(n, d, h, \log s, \nbar^{m},s',d')$,  and degree $\poly(\nbar, d,d', h, \log s)$. %
        \end{subproof}
    \paragraph{Description of $\calN_n$.}Given this claim, we can now  provide a complete description of $\calN_n$ as follows. The output of $\calN_n(\vec{\Lambda})$ is provided by the output of $\calS$ (from \Cref{clm:ki:recon:ind}), which in turn, receives as input the output of the execution of $\Nin_n$ (from \Cref{clm:gkss-recon-branch}) on each of the indices $a\in[\nbar-1]$ and $b\in[d]$, i.e., the descriptions of the circuits $\{\calC_{a,b}\}$ and the corresponding Boolean values $v_{a,b}$.
    More precisely, $\calS$ receives as input the concatenation (see \Cref{def:concat-arith-networks}) of $(\nbar-1)\cdot d$ networks with $\pit$ gates ($\Nin_n$ instantiated with each $a\in[\nbar-1]$ and $b\in[d]$). %
	
	\paragraph{Complexity of $\calN_n$.} 
	Note that there are at most $\nbar d$ branches (i.e., instantiations of $\Nin_n$) in all. Hence, by \Cref{clm:gkss-recon-branch}, and \Cref{clm:ki:recon:ind}, the total size of $\calN_n$ is at most $\poly(n,s',d', d, h, \log s, \nbar^{m})$. %

    To see the degree bound, note that the degree of $\Nin_n$ (the network from \Cref{clm:ki:recon:branch}) in each of its instantiations is bounded by $\poly(\nbar, h,d,\log s,d')$, and as they are each run in parallel inside $\calN_n$, the degree of the network that computes the overall set of candidate circuit descriptions $\{\langle \calC_{a,b}\rangle|a\in[\nbar-1],b\in[d]\}$ is also bounded by $\poly(\nbar,h,d,\log(s),d')$ (by \Cref{prop:concat-arith-networks}). Finally, since $\calN_n$ is the composition of $\calS$ (which has degree $O(1)$) with these networks, the claimed overall degree bound of $\poly(\nbar,h,d,\log(s),d')$ follows from \Cref{prop:error-deg-comp}.

	\paragraph{The ``furthermore'' part.}
	The proof is essentially identical to that of the main case, the only difference is that to prove~\Cref{clm:gkss-ind-recons:dsr} we now rely on the downward self-reducibility relation between $f_{i-1,2m}$ and $f_{i,0}$ as specified in~\Cref{item:dsr} of~\Cref{prop:poly-dec-cons} (note the identity $f_{i,2m}(\vec{w}) = f_i(\vec{w})$). The size and degree of the circuit for downward self-reducibility both increase now to $h\cdot\polylog(s)$ (rather than $O(h)$ and $1$ in the main case), and hence the circuit for evaluating $f$ that~\Cref{clm:gkss-ind-recons:dsr} outputs now has size $h\cdot\poly(s',\log(s))$ and degree at most $d'\cdot h\cdot\polylog(s)$. This factors into the size and degree of the circuits from~\Cref{clm:gkss-recon-branch}, but does \emph{not} change the final asymptotic bounds.%

\begin{remark}\label{rem:char-GKSS-final}
    The assumptions on the field characteristic made in \Cref{thm:gkss-tarhsg} are used at the locations pointed to in \Cref{fn:lagrange-char}, \Cref{rem:Delta-lc}, and \Cref{rem:Euler-char}.
    The first instance can be remedied in arbitrary characteristic by working over a sufficiently large field, and the second instance can likely be addressed by working with Hasse derivatives (a notion of derivative that is better suited to positive characteristic).
    However, it is not clear to the authors how one can remove the restriction on the characteristic arising from the use of \Cref{fact:euler}.
\end{remark}

\section{A targeted version of the Kabanets--Impagliazzo generator} \label{sec:ki}

Our goal in this section is to construct an arithmetic reconstructive targeted hitting-set generator that works over sequences of finite fields, and is analogous to the generator from~\Cref{thm:gkss-tarhsg}.

A technical subtlety is that the generator can be instantiated in two parametric ``modes'', represented by the parameter $\sigma\in\{0,1\}$ below. The mode $\sigma=1$ useful when the output length is large, and the mode $\sigma=0$ is useful for when the output length is small. (Jumping ahead, our $\pit$ algorithm will compose this generator constantly many times, the first iteration with $\sigma=1$ and subsequent iterations with $\sigma=0$; see the proof of~\Cref{thm:main:finite} for details.)

We say that a field family $\F=\set{\F_n}_{n\in\N}$ is {\sf efficiently constructible} if there is a Turing machine that gets input $1^n$ and outputs a representation of $\F_n$ (i.e., a prime or a suitable irreducible polynomial).

\begin{theorem} [algebraic KI-based targeted hitting-set generator] \label{thm:ki:tarhsg}
Let $\set{C_n}_{n \in \N}$ be a $\log^c$-uniform family of algebraic circuits of size $s(n)$ and degree at most $s(n)$ having $n$ output gates, defined over a feasible sequence of finite fields $\set{\F_{n}}_{n\in\N}$, where $\F_{n}$ has order $q(n) \le \exp(n^{O(1)})$ and characteristic $p(n)$.
Let $\eps>0$ be a constant, and let $m(n)<s(n)$ and $d(n)$ be time-constructible functions.
Then, for $\sigma\in\zo$, there are two $\P$-uniform families of arithmetic networks $\set{G_n}$ and $\set{R_n}$ satisfying the following.
\begin{enumerate}
    \item {\bf (Generator.)} Networks in $\set{G_n}$ are of size $\poly(s)$. When given input $\vec{\Lambda}\in\F_n^n$, the network $G_n(\vec{\Lambda})$ outputs the description of an arithmetic circuit
    \mm{
    \calG_{\vec{\Lambda}}\colon\F_n^{\ell(n)}\times\F_n^{2} \rightarrow \F_n^{m(n)}
    }
    of degree $\deg(\calG_{\vec{\Lambda}})\le n\cdot s^{\Theta(1/\log(m))}\cdot \polylog(s(n))$, where $\ell(n)=\begin{cases}\polylog(s)&\text{if $\sigma=1$}\\O(\log(m))&\text{if $\sigma=0$}\end{cases}$.
    
    \item {\bf (Reconstruction.)} Networks in $\set{R_n}$ are randomized, and when $\vec{\Lambda}\in\F_n^n$ is a universal encoding of a degree-$d$ arithmetic circuit with $m$ input gates such that $\vec{\Lambda}\circ \calG_{\vec{\Lambda}}\equiv0$, then $R_n(\vec{\Lambda})$ computes $C_n(\vec{\Lambda})$ with error degree $\poly(n,d,\log(s))$. 
    
    Each network $R_n$ is of size $\begin{cases}\poly(n,d,m,\log(s),\log(q))&\text{if $\sigma=1$}\\ s^{\eps}\cdot \poly(n,d, m^{\log\log(s)},\log(q))&\text{if $\sigma=0$}\end{cases}$ and of degree $(n\cdot d\cdot q)^{\polylog(s)}$. Moreover, when $\set{\F_n}$ is efficiently constructible, and if we allow each $R_n$ to use \emphdef{p\ts{th} root gates}, which are gates that map a field element $\alpha \in \F_q$ to $\alpha^{1/p(n)}$, the network $R_n$ is of degree $(n\cdot d\cdot p)^{\polylog(s)}$.\footnote{To formally define the degree of a network with $p$\ts{th} root gates, we define it as the degree of the network obtained by replacing each $p$\ts{th} root gate with a dummy gate that computes the identity function $\alpha \mapsto \alpha$. In other words, for this purpose, we consider the $p$\ts{th} root gate as having degree $0$.}
\end{enumerate}
\end{theorem}

The generator networks $\set{G_n}$ and reconstruction networks $\set{R_n}$ are defined and analyzed in \Cref{sec:ki:generator,sec:ki-reconstruction}, respectively.
Our algorithms will make use of uniform versions of classical algorithms (i.e., versions implementable by uniform arithmetic circuits and networks), in particular the standard depth-reduction procedure of \textcite{VSBR83} and Kaltofen's algorithm for factorization of arithmetic circuits \cite{Kaltofen1989FactorizationOP}.
The uniform versions of these algorithms are described in Appendix~\ref{sec:uniform algos}.

\begin{remark} [the reconstruction's degree, and the role of $p$\ts{th} root gates]
The reconstruction procedure underlying~\Cref{thm:ki:tarhsg} computes $p$\ts{th} roots of certain elements during its computation. Over a finite field of order $q$ and characteristic $p$, the $p$\ts{th} root operation can be implemented via $x \mapsto x^{q/p}$, and %
this may cause a large (and generally, unavoidable) blowup in the network's degree. One alternative is to  simply pay (in degree), as spelled out at the end of the theorem's statement. (We do so in a careful manner that bounds the number of $p$\ts{th} root gates along any computational path; see~\Cref{sec:ki:recon:pth}.) The reason we also presented the reconstruction as networks using $p$\ts{th} root gates is to clarify that this is the only part that uses high-degree computation; that is, other than $p$\ts{th} root gates, the reconstruction is of quasipolynomial degree.
\end{remark}

\subsection{The generator networks} \label{sec:ki:generator}

Loosely speaking, the targeted generator will get input $\vec{\Lambda}$ and apply the construction of the Kabanets--Impagliazzo~\cite{KI04} generator to each of the polynomials in the polynomial decomposition of $C_n(\vec{\Lambda})$. We first recall the KI construction, and then describe our targeted generator.

\paragraph{A reminder: The Kabanets--Impagliazzo generator.}

The construction relies on combinatorial designs as in the work of Nisan and Wigderson~\cite{NisanW94}. Let us define these and state some standard constructions.

\begin{definition}
\label{def:nw-design}
        Given $m,n\in \N$ with $n<2^m$, a family of sets $S_1,\ldots, S_n \subseteq[\ell]$ is said to be a \emphdef{Nisan--Wigderson $(m,a)$-design} if
        \begin{itemize}
            \item for every $1\leq i\leq n$, $|S_i| = m$, and
            \item for every $1\leq i < j\leq n$, $|S_i\cap S_j| \leq a$.
        \end{itemize}
\end{definition}

\begin{lem}[see, e.g.,~{\cite[Problem 3.2]{vadhan-book}}] \label{lem:nw-design}
    For every $m,n,a\in \N$ with $n<2^{m}$ and $a<m$, there exists a Nisan--Wigderson $(m,a)$-design $S_1,...,S_n\subseteq[\ell]$ that can be constructed deterministically in time $\poly(n,m)$, where $\ell= O( n^{1/a} \cdot m^2/a )$.
\end{lem}

\begin{cor} \label{cor:nw-design}
    For every $m,n\in \N$ with $n<2^{m}$, there exists a Nisan--Wigderson $(m,\log n)$-design $S_1,...,S_n\subseteq[\ell]$ that can be constructed deterministically in time $\poly(n,m)$, where $\ell=O(m^2/\log n)$.
\end{cor}

Given a hard polynomial $g(\vec{x})\in\F[x_1,...,x_n]$ and NW-designs as above, the KI construction is as follows.

\begin{construction}[\cite{KI04}] \label{cons:KI-gen}
    Let $n$ and $m$ be integers satisfying $n < 2^{m}$.
    Let $g(\vec{x}) \in \F[x_1,\ldots,x_{m}]$.
    Let $\bar{S}=S_1, \ldots, S_n \subseteq [\ell]$ be a Nisan--Wigderson $(m,a)$-design.
    The \emphdef{Kabanets--Impagliazzo generator} $\mathcal{G}_{KI,g}(\vec{z}) : \F^\ell \to \F^n$ is the polynomial map given by
    \[
        \mathcal{G}^{\bar{S}}_{KI,g}(z_1,\ldots,z_\ell) \coloneqq (g(\vec{z}|_{S_1}), \ldots, g(\vec{z}|_{S_n})),
    \]
    where $\vec{z}|_{S_i}$ denotes the restriction of $\vec{z}$ to the variables indexed by $S_i$.
\end{construction}

\paragraph{The targeted generator.}
We now describe our construction of a targeted generator. Starting with a hard $\log^c$-uniform family of arithmetic circuits $\{C_n\}$ of size $s$, we apply uniformity-preserving depth reduction to obtain a uniform family $\{C'_n\}$ of size $\poly(s)$ and depth $\polylog(s)$, and given $\vec{\Lambda}\in\F^n$ we apply~\Cref{cons:KI-gen} to each of the polynomials in the polynomial decomposition of $C'_n(\vec{\Lambda})$.

We first present a version of the generator with general parameters, and then present two specific instantiations that will be useful for us. To facilitate parsing the  parameters, think of $n$ as the input length to the targeted generator, and of $\nbar$ as the number of pseudorandom field elements that the generator outputs. The parameters $h,m,r$ are auxiliary parameters, which will be used for instantiating the polynomial decomposition from~\Cref{sec:ki:decomposition}. We will either set the parameters to get seed length $O(\log(\nbar))$ and reconstruction time proportional to $s^{\eps}\cdot \nbar^{\log\log(s)}$, for a small constant $\eps>0$; or set the parameters to get seed length $\polylog(s)$ and reconstruction time proportional to $\poly(\nbar)$.

\begin{construction}[KI-based targeted HSG]\label{cons:generator}
	Let $\set{C_n}_{n \in \N}$ be a $\log^c$-uniform family of algebraic circuits of size $s(n)$ and degree at most $s(n)$, defined over a feasible sequence of fields $\set{\F_n}_{n\in\N}$, where $\F_n$ has order $q(n) \le \exp(n^{O(1)})$. Let $\{C'_n\}$ be the $\log^{c' = c + O(1)}$-uniform family of arithmetic circuits of size $s'(n)=\poly(s(n))$ and depth $\Delta(n)=O(\log(s(n))^2)$ obtained from~\Cref{prop:uniform-depth-red}.
	
	\medskip\noindent {\bf Parameters.} Let $h,m,r,\nbar$ be any time-constructible functions (interpreted as functions of $n$) such that $3m+r>\log(\nbar)$, and one of the following holds:
	\begin{itemize}
		\item $h=2$ and $m=\lceil \log(s') \rceil$ and $r=\log(s')^{c'}$, or
		\item $h$ is a power of two, and $m$ is the minimal integer such that $h^m\ge s'$, and $r=0$.
	\end{itemize} 
	
	\medskip\noindent {\bf Generator.} Given $\vec{\Lambda}\in\F_n^n$, let $\set{f_{i,j}}$ be the $(h,m,r)$-polynomial decomposition of $C'_n$ on input $\vec{\Lambda}$ obtained via~\Cref{prop:poly-dec-cons}, and consider all polynomials as having $3m+r$ variables (i.e., add dummy variables if necessary). Let $\bar{S}=S_1, \ldots, S_{\nbar} \subseteq [\ell=O((3m+r)^2/\log(\nbar))]$ be the Nisan--Wigderson $(3m+r,\log \nbar)$-design obtained via~\Cref{cor:nw-design}. Define the polynomial map $\calG_{\vec{\Lambda}}^{(h,m,r,\nbar)}:\F_n^{\ell+2}\to \F_n^\nbar$ (called the \emphdef{KI-based targeted HSG} for $\vec{\Lambda}$) as follows:
	\[
	\mathcal{G}_{\vec{\Lambda}}^{(h,m,r,\nbar)}(y_1,y_2,\vec{z}) \coloneqq \sum_{i = 1}^{\Delta(n)}\sum_{j = 0}^{2m+r} L_{i,j}(y_1,y_2) \cdot \mathcal{G}_{KI,f_{i,j}}^{\bar{S}}(\vec{z}), \eqtag{cons:generator:def}
	\]
	where $L_{i,j}(y_1,y_2)$ is the Lagrange interpolation polynomial, defined as\footnote{The argument for $\phi_2$ is $j+1$ instead of the more natural $j$ because the domain of enumerators as given in \Cref{def:enumeration} starts from $1$, not $0$.}
		\[
		L_{i,j}(y_1,y_2) = \begin{cases}
			1, &\text{ if } y_1 = \phi_1(i) \text{ and } y_2 = \phi_2(j+1), \\
			0, &\text{ if } (y_1,y_2) \in \phi_1([\Delta(n)])\times \phi_2([2m+r+1])\setminus\{(\phi_1(i),\phi_2(j+1))\},\\
		\end{cases}
		\]
		where $\phi_1$ and $\phi_2$ are enumerators for $\Delta(n)$ and $t=2m+r+1$, respectively.  %
\end{construction}

\newtheorem{corollary}[theorem]{Corollary}

\begin{corollary}[KI-based targeted HSG, two useful instantiations] \label{cons:generator:specific}
    Let $\set{C_n}_{n \in \N}$ and $\{C'_n\}$ be as in~\Cref{cons:generator}, and let $\nbar=\nbar(n)<n$ be time-constructible. Then,
    \begin{enumerate}
    	\item Instantiating~\Cref{cons:generator} with $h=2,m=\lceil \log(s)\rceil,r=\log(s')^{c'}$, we obtain a generator $\mathcal{G}_{\vec{\Lambda}}^{(h,m,r,\nbar)}$ with seed length $\ell=\polylog(s)$.
    	\item Instantiating~\Cref{cons:generator} with $h=2^{\log(s')/O(\log(\nbar))},m=O(\log(\nbar)),r=0$, we obtain a generator $\mathcal{G}_{\vec{\Lambda}}^{(h,m,r,\nbar)}$ with seed length $\ell=O(\log(\nbar))$.
    \end{enumerate}
\end{corollary}

The seed length in the second case above is strictly better (because $\nbar<2^m<s$), but the choice of parameters $h,m,r$ affects the reconstruction time, so the second case is not always preferable to the first.
(In the first case, the reconstruction time will be proportional to $\poly(\nbar)$, whereas in the second case the reconstruction time is proportional to $s^{\eps}\cdot \nbar^{\log\log(s)}$; see~\Cref{sec:ki-reconstruction}.)

\begin{remark}
    The notation $\calG_{\vec{\Lambda}}^{(h,m,r,\nbar)}$ suppresses the fact  that the generator also depends on the choice of a $\log^c$-uniform family of arithmetic circuits $\{C_n\}$ and the choice of a Nisan--Wigderson design.  We will drop the superscript $(h,m,r,\nbar)$ whenever the choice of these functions is clear from context.
\end{remark}

Let us now analyze the complexity of the generator. We stress that for a fixed $\vec{\Lambda}$, the generator $\mathcal{G}_{\vec{\Lambda}}(y_1,y_2,\vec{z})$ should be viewed as a \emph{polynomial} map from $\F_n^{\ell + 2}\to \F_n^\nbar$. This viewpoint is indeed necessary for repeated self-composition of this generator in order to reduce the final seed length to solve \pit. %

\begin{lem} [the complexity of the targeted generator] \label{lem:KI-based-gen}
Let $\{C_n\},h,m,r,\nbar$ be as in the statement of~\Cref{cons:generator}. Then, 
\begin{enumerate}
	\item There is a $\P$-uniform family of arithmetic networks of size $\poly(s(n))$ %
	that gets input $\vec{\Lambda}$ and outputs the description of the arithmetic circuit $\calG_{\vec{\Lambda}}^{(h,m,r,\nbar)}$.
	\item The degree of $\calG_{\vec{\Lambda}}^{(h,m,r,\nbar)}$ is at most $n\cdot h(n)\cdot \polylog(s(n))$.%
\end{enumerate}

\end{lem}

Note that the bounds on the complexity of the generator in~\Cref{thm:ki:tarhsg}, with $\sigma\in\zo$, follow from the two instantiations in~\Cref{cons:generator:specific} and from~\Cref{lem:KI-based-gen}.

\begin{proof}[Proof of~\Cref{lem:KI-based-gen}]
We stress that there are three algorithms involved in the claim. We wish to construct a machine $M$ that prints an arithmetic network $N$, where $N$ gets input $\vec{\Lambda}$ and computes the description of an arithmetic circuit $G_{\vec{\Lambda}}$ that computes $\calG_{\vec{\Lambda}}$. 

 We first describe $G_{\vec{\Lambda}}$, and then explain how $N(\vec{\Lambda})$ computes its description (and how $M(1^n)$ prints a description of $N$).
 
\paragraph{The circuit $G_{\vec{\Lambda}}$.} The circuit $G_{\vec{\Lambda}}$ has hard-wired arithmetic circuits for $f_{i,j}$ for all $i\in[\Delta]$ and $j\in\set{0,...,2m+r}$. It also has hard-wired arithmetic circuits for $L_{i,j}$, for all $i,j$. It also has hard-wired the design $\bar{S}$. By~\Cref{cons:KI-gen}, using the hard-wired $\bar{S}$ and circuits for $f_{i,j}$, the circuit $G_{\vec{\Lambda}}$ can compute the mapping $\vec{z}\mapsto \calG_{KI,f_{i,j}}^{\bar{S}}(\vec{z})$. Thus, the circuit $G_{\vec{\Lambda}}$ implements the functionality in Eq.~\eqref{cons:generator:def} in the straightforward way. The size of $G_{\vec{\Lambda}}$ is at most
\[
	O(\log(s)^2\cdot m\cdot r)\cdot \max_{i,j}\set{\size(L_{i,j})}\cdot \max_{i,j}\set{\size(f_{i,j})} + |\bar{S}|,
\]
and its degree is at most
\[
	\max_{i,j}\set{\deg(L_{i,j})\cdot \deg(f_{i,j})}.
\]

\paragraph{The polynomials $L_{i,j}$.}
Let us present an explicit expression for the polynomials $L_{i,j}$. For $i\in[\Delta]$, define the univariate 
\[
	Q_{i,1}(y) \coloneqq \prod_{k\in [\Delta] \setminus \set{i}} \frac{y - \phi_1(k)}{\phi_1(i) - \phi_1(k)},
\]
and for $j\in\{0,...,2m+r\}$, define the univariate polynomials
\[
	Q_{j,2}(y) \coloneqq \prod_{k\in \{0,...,2m+r\}\setminus \set{j}} \frac{y - \phi_2(k+1)}{\phi_2(j+1) - \phi_2(k+1)}.
\]
Then, for any $i\in [\Delta]$ and $j\in\{0,...,2m+r\}$, we define $L_{i,j}(y_1,y_2) \coloneqq Q_{i,1} (y_1) \cdot Q_{j,2}(y_2)$.
The degree of $L_{i,j}$ is $O(\Delta\cdot m\cdot r)$, and it can be implemented in size $\poly(m,r,\log(s))=\polylog(s)$.
To see that there is a $\P$-uniform arithmetic network of size $\polylog(s)$ that on any $n$-element input $\vec{\Lambda}$ computes a description of $L_{i,j}$, note that the machine that constructs the network can simply hard-wire the description of $L_{i,j}$ into the network, as this description does not depend on the input $\vec{\Lambda}$ given to the network.
The Turing machine that prints this network can print the constants $\set{\phi_1(k) : k \in [\Delta]}$ and $\set{\phi_2(k+1) : k \in \set{0, \ldots, 2m+r}}$ in $\poly(n, \Delta, 2m+r)$ time by invoking the Turing machines underlying the enumerators $\phi_1$ and $\phi_2$.

\paragraph{The circuits for $f_{i,j}$.}
By \cref{item:layer-poly-comp} of \Cref{prop:poly-dec-cons}, there is a $\P$-uniform family of arithmetic networks of size $\poly(s)$ that takes input $\vec{\Lambda}$ and prints descriptions of arithmetic circuits of size $\poly(s)$ and degree at most $n\cdot h\cdot\polylog(s)$ for $\set{f_0}\cup\set{f_{i,j}}_{i,j}$.

\paragraph{The design $\bar{S}$.}
Recall that $\bar{S}$ is obtained via~\Cref{cor:nw-design}. In particular, the machine $M$ can compute this design and hard-wire it into $N$, as the designs do not depend on the input $\vec{\Lambda}$ to $N$. As a consequence, the network $N$ can hard-wire the design into the arithmetic circuit $G_{\vec{\Lambda}}$. The size of the design's description is at most $\nbar\cdot \ell=O(m\cdot r)<\polylog(s)$.

\paragraph{Putting things together.}
Combining the constructions of the $L_{i,j}$'s, the $f_{i,j}$'s, and $\bar{S}$, we obtain a machine $M$ running in time $\poly(s)$ and printing a network of size $\poly(s)$, which maps $\vec{\Lambda}$ to $G_{\vec{\Lambda}}$ of size $\poly(s)$ and degree $n\cdot h\cdot\polylog(s)$.
\end{proof}

\subsection{The reconstruction procedure} \label{sec:ki-reconstruction}

We now describe the reconstruction procedure for our generator. We first describe a version corresponding to the ``mode'' $\sigma=1$, in which case the reconstruction time is roughly $\poly(n,d,\log(s))$. Then we explain how to obtain the version for ``mode'' $\sigma=0$, which has higher reconstruction time (see~\Cref{prop:recons-arith-net'}). In both cases, we assume that $\F$ is efficiently constructible, and we allow $p$\ts{th} root gates; in~\Cref{sec:ki:recon:pth} we explain how to obtain the claimed reconstruction networks if $\F$ is not efficiently constructible and/or when disallowing $p$\ts{th} root gates.

\subsubsection{The reconstruction in $\sigma=1$ mode}

In the statement below we keep $h,m,r$ as parameters, despite the fact that they are fixed to particular values in the $\sigma=1$ ``mode'' (i.e., to $h=2,m=\log(s),r=\polylog(s)$). The reason is to present the construction in a way that will be easier to generalize later for the $\sigma=0$ ``mode''.

\begin{proposition} [the reconstruction for~\Cref{thm:ki:tarhsg} with $\sigma=1$]\label{prop:recons-arith-net}
    Let $\set{C_n}_{n \in \N}$ be a $\log^c$-uniform family of algebraic circuits of size $s(n)$ having $n$ output gates, defined over a feasible sequence of fields $\set{\F_n}_{n \in \N}$, where each $\F_n$ has order $q(n) \le \exp(n^{O(1)})$ and characteristic $p(n)$.
		Let $d(n)$ be time-constructible, %
        let $\vec{\Lambda} \in \F^{n}_{n}$, and let $\calG_{\vec{\Lambda}}^{(h,m,r,\nbar)}$ be the generator described in \Cref{cons:generator:specific}, with the $r>0$ instantiation.
    Then, there is a $\P$-uniform family $R = \set{R_n}_{n \in \N}$ of randomized arithmetic networks that satisfies the following.
    \begin{enumerate}
        \item 
            The input to $R_n$ is a vector $\vec{\Lambda}\in\F_n^n$ that is the universal encoding of an arithmetic circuit $D$ that computes an $\nbar$-variate, degree-$d$ nonzero polynomial and satisfies $D\circ \calG_{\vec{\Lambda}}^{(h,m,r,\nbar)}(y_1, y_2, \vec{z}) = 0$.
        \item
            The size\footnote{In the statement of this proposition, the notation $\poly(T)$ indicates a bound of $T^{cK}$, where $K$ is a universal constant and $c$ is the given constant that parameterizes the $\log^c$-uniformity of $\{C_n\}$.} %
            of the network $R_n$ is bounded by $\poly(n,d,\nbar^{\log h}, \log(s),\log(q))$.
        \item The degree of $R_n$ is at most $(n\cdot d\cdot h\cdot p)^{\polylog(s)}$ 
        (allowing $p$\ts{th} root gates, and assuming that $\F$ is efficiently constructible).
        \item
            $R_n(\vec{\Lambda})$ computes the vector $C_n(\vec{\Lambda}) \in \F_n^{n}$ with error degree at most $\poly(d,h,\nbar,\log(s))$.
    \end{enumerate}
\end{proposition}

The rest of this section is devoted to the proof of \Cref{prop:recons-arith-net}.
Recall that \Cref{cons:generator} transforms $\{C_n\}$ into $\{C'_n\}$ of depth $\Delta=O((\log s)^2)$ and size $s'=\poly(s)$, and $\calG_{\vec{\Lambda}}$ is obtained by applying \Cref{cons:KI-gen} to the polynomials from the polynomial decomposition $\{f_{i,j}\}$ of $C'_n(\vec{\Lambda})$.
The network $R_n(\vec{\Lambda})$ will compute $C'_n(\vec{\Lambda})=C_n(\vec{\Lambda})$.
For notational convenience, in what follows we will drop the prime symbol in $C'_n$ and $s'$, denoting them by $C_n$ and $s$, respectively.

Before we describe the construction formally, we provide a brief sketch of its various components.
On input $\vec{\Lambda}$, the network efficiently constructs the description of a circuit computing the polynomial $f_\Delta = f_{\Delta,t}$.
As this polynomial faithfully represents the output gate values of $C_n(\vec{\Lambda})$ (see \cref{item:faithful} of \Cref{prop:poly-dec-cons}), the network  eventually computes $C_n(\vec{\Lambda})$ by simply evaluating $f_{\Delta}$ over the first $n$ many vectors (in lexicographical order)  in $H^m$.
The main idea is that $R_n$ iteratively, for $i = 0,\ldots,\Delta$, constructs a small arithmetic circuit for $f_{i}$.
At a finer level, $R_n$ constructs a circuit for $f_i$ using query access to a circuit for $f_{i-1}$ by sequentially constructing circuits for $f_{i,j}$ from $f_{i,j-1}$ for $j = 1,\ldots,t$.
	
We will now describe this arithmetic network, while accounting both for the complexity of implementing each iteration (i.e., the complexity of the uniform arithmetic network), and for the complexity of the arithmetic circuits whose descriptions are produced by each iteration.
Note that the assumption $D \circ \calG_{\vec{\Lambda}}^{(h,m,r,\nbar)}(y_1,y_2,\vec{z}) = 0$ implies that $D\circ \mathcal{G}_{KI,f_{i,j}}(\vec{z}) = 0$ for every $\vec{z}\in\F^\ell$, and for all $i\in [\Delta]$ and $j\in \{0,...,2m+r\}$.

To begin, $R_n$ employs the $\P$-uniform family of arithmetic networks from \Cref{lem:network-algo-uni-std} to convert $\vec{\Lambda}$ to a standard description of $D$ of $\poly(n,d)$ size.
In the rest of this section, we abuse notation in the following way: whenever we refer to $\vec{\Lambda}$, we mean this standard description of $D$ of $\poly(n,d)$ size.
To start the iterative procedure, by \Cref{item:input-layer-comp} in \Cref{prop:poly-dec-cons}, the arithmetic network can produce an arithmetic circuit of size $O(n\cdot m\cdot h\cdot \polylog(s))$ and degree $n\cdot h\cdot\polylog(s)$ that computes $f_0$.
Then the arithmetic network successively uses the network given by the following lemma.

\begin{lem}[One iteration: moving from a circuit for $f_{i,j-1}$ to a circuit for $f_{i,j}$] \label{lem:ind-recons}
	There is a $\P$-uniform family $\calN=\{\calN_n\}$ of randomized arithmetic networks with $p$\ts{th} root gates that when given as input $\vec{\Lambda}$, a natural number $\beta \in \N$, and a standard description of a size $s_{i,j-1}$, degree $d_{i,j-1}$ arithmetic circuit computing $f_{i,j-1}^{p^\beta}$, outputs a natural number $\gamma \in \N$ and a standard description of an arithmetic circuit computing $f_{i,j}^{p^\gamma}$ of size $\poly(n,d, \nbar^{\log h})$ and degree at most $\poly(d,h,\log \nbar)$.
	The size of $\calN_n$ is at most $\poly(n,d, \nbar^{\log h}, s_{i,j-1}, d_{i,j-1},\log(s),\log(q))$, its degree is at most $\poly(\log s,d,h,p,d_{i,j-1})$ (using $p$\ts{th} root gates) and its error degree is at most $\poly(\nbar,h,d,\log(s), d_{i,j-1})$.

	Furthermore, there is another $\P$-uniform family $\calN'=\set{\calN'_n}$ of randomized arithmetic networks with $p$\ts{th} root gates that, when given as input $\vec{\Lambda}$, a natural number $\beta \in \N$, and the standard description of a size $s_{i-1,2m+r}$, degree $d_{i-1,2m+r}$ arithmetic circuit computing $f_{i-1,2m+r}^{p^\beta}$, outputs a natural number $\gamma \in \N$ and a standard description of an arithmetic circuit computing $f_{i,0}^{p^\gamma}$ of size $\poly(n,d,\bar{n}^{\log h})$ and degree at most $\poly(d,h,\log\nbar)$. 
	The size of $\calN'_n$ is $\poly(n,d, \nbar^{\log h}, s_{i-1,2m+r}, d_{i-1,2m+r},\log(s),\log(q))$, its degree is at most $\poly(h,d,\log(s),p,d_{i-1,2m+r})$ (using $p$\ts{th} root gates) and its error degree is at most $\poly(\nbar,h,d,\log(s),d_{i-1,2m+r})$.
\end{lem}
	
We defer the proof of \Cref{lem:ind-recons} to \Cref{sec:ki:recon:step}.
The network $R_n$ makes repeated use of the networks $\calN$ and $\calN'$ from \Cref{lem:ind-recons}.
After computing the description of a circuit for $f_0$ as explained above, $R_n$ uses the network $\calN'$ (applied with $i = 1$) to output the description of a $p$\ts{th} power of $f_{1,0}$. 
This description is then fed into the network $\calN$ (applied with $i = 1$ and $j = 1$), which in turn then outputs the description of a $p$\ts{th} power of $f_{1,1}$.
This description is, in turn, fed into the network $\calN$ (now applied with $i = 1$ and $j = 2$), and so on until $j = 2m+r$, at which point $R_n$ observes that $f_{1,2m+r} = f_1$ and subsequently invokes the network $\calN'$ again (now applied with $i = 2$).
This process repeats until $R_n$ eventually computes the description of a $p$\ts{th} power of $f_\Delta = f_{\Delta,2m+r}$.
Finally, $R_n$ uses \Cref{prop:universal-arith-network} and $p$\ts{th} root gates to query $f_\Delta$ on (the first $n$ elements of) the set $H^m$, and returns these elements as a vector $\vec{\alpha}$ (which we claim is equal to $C_n(\vec{\Lambda})$ with high probability).

This results in the application of \Cref{lem:ind-recons} a total of $(2m+r)\cdot \Delta$ times, and the application of the universal arithmetic network in \Cref{prop:universal-arith-network} to evaluate $f_\Delta$ a total of $n$ times.
Furthermore, the size of the output circuit produced by each application of \Cref{lem:ind-recons} is independent of the input circuit for the layer polynomial below, i.e., it follows that each $s_{i,j}$ is bounded by $\poly(n,d,\nbar^{\log h})$ and degree $d_{i,j}$ is bounded by $\poly(d,h,\log \nbar)$.
Hence, the overall size of $R_n$ is bounded by
\[
	(2m+r)\cdot\Delta\cdot\poly(n,d,\nbar^{\log h},\log(q)) < \poly(n,d,\nbar^{\log h},\log(s),\log(q)). \eqtag{eq:ki:recon:sizefinal}
\]

Next, we analyze the error degree of $R_n$. By \Cref{def:rand-arith-network}, this amounts to showing that for some integer $e$ bounded by $\poly(\nbar,d,h \log s)$, and for every finite set $S\subset \F$, when each random input $r_i$ (collectively represented by $\vec{r}$, say) is chosen independently and uniformly at random from $S$, with probability at least $1-e/|S|$, we have $R_n(\vec{\Lambda},\vec{w},\vec{r})= C_n(\vec{\Lambda})$, i.e., the vector $\vec{\alpha}$ as defined above equals $C_n(\vec{\Lambda})$.

From the construction of $R_n$, observe that it consists of a repeated composition of the networks $\calN_n$ or $\calN_n'$, successively instantiated appropriately for the indices $i\in[\Delta]$ and $j\in \{0,\ldots, t\}$, and the (deterministic) network from \Cref{item:input-layer-comp} of \Cref{prop:poly-dec-cons} that outputs a description of $f_0$.
Note that since the network in \Cref{prop:universal-arith-network} is also deterministic, conditioned on the event that a description of a $p$\ts{th} power of $f_\Delta$ is computed correctly in the final instantiation of $\calN$ (i.e., with $i = \Delta$ and $j = t$), the vector $\vec{\alpha}$ as defined above is guaranteed to equal $C_n(\vec{\Lambda})$ by definition.
By \Cref{prop:error-deg-comp}, we can bound the error degree of $R_n$ by its constituent networks.
The network produced by \cref{item:f0} of \Cref{prop:poly-dec-cons} is deterministic, so it has error degree $0$.
Each of the $2m+r\leq \polylog(s)$ invocations of $\calN_n$ or $\calN_n'$ increase the error degree of $R_n$ by $\poly(\nbar, \log s, d, h)$.
In total, we can bound the error degree of $R_n$ by $t\Delta\cdot\poly(\nbar,d,h, \log s) = \poly(\nbar,d,h \log s)$.

Finally, we analyze the degree of $R_n$.
Using \Cref{prop:error-deg-comp}, we can again bound the degree of $R_n$ by its constituent networks.
The network produced by \cref{item:input-layer-comp} of \Cref{prop:poly-dec-cons} has degree $n\cdot h\cdot \polylog(s)$. 
Each of the $2m+r\leq \polylog(s)$ invocations of $\calN_n$ or $\calN_n'$ multiply the degree of $R_n$ by $\poly(h,d,p,\log s)$. %
In total, we can bound the degree of $R_n$ by $(h\cdot d\cdot p\cdot \log s)^{O(t\Delta)}=(n\cdot d\cdot h\cdot p)^{\polylog(s)}$.

\subsubsection{A single iteration: Proof of \Cref{lem:ind-recons}} \label{sec:ki:recon:step}

	We first prove the main part of the statement, and then explain how to deduce the ``furthermore'' part, relying on essentially the same argument.
	
	Relabel the indices $i$ and $j$ in the lemma statement as $i'$ and $j'$, respectively, so that $i$ and $j$ are now free to use for other purposes in the remainder of the proof of this lemma.
	The indices $i'$ and $j'$ are henceforth assumed to be fixed and we call $f_{i',j'}$ as simply $f$.
	We also relabel the size and degree of the circuit for $f_{i',j'-1}$  that is given to the network by replacing $s_{i',j'-1}$ and $d_{i',j'-1}$ with $s'$ and $d'$, respectively.
	
	Recall that in \Cref{cons:generator} we added dummy variables to all the $f_{i',j'}$'s, so that they are all $(3m+r)$-variate polynomials.
	We will repeatedly use the fact that the network can evaluate $f$ at any given point in $\F_n^{3m+r}$ using a circuit for $f_{i',j'-1}^{p^\beta}$ and the downward self-reducibility of the decomposition $\{f_{i,j}\}$.
	Let us state this formally.
	
	\begin{subclaim} \label{clm:ind-recons:dsr}
		There is a $\P$-uniform arithmetic network that gets as input $(i',j')\in[\Delta]\times\{0,...,2m+r\}$, a description of an arithmetic circuit of size $s$ and degree $d$ that computes $f_{i',j'-1}^{p^\beta}$ for some $\beta \in \N$, and a vector $\vec{\alpha} \in \F^{3m+r}_n$, and outputs the value of $f(\vec{\alpha})^{p^\beta}$.
        The size of this network is bounded by $h \cdot \poly(m,r,s,d)$ and its degree is bounded by $\poly(m,r,d)$.
	\end{subclaim}
	
	\bproof
		The network uses \Cref{prop:universal-arith-network} to evaluate the given circuit for $f_{i',j'-1}^{p^\beta}$ %
        and then applies the downward self-reducibility of the polynomial decomposition $\set{f_{i,j}}$ (i.e., \Cref{item:prop-sumcheck} of \Cref{prop:poly-dec-cons}) to compute $f_{i', j'}^{p^\beta}$. %
        Since downward self-reducibility of the polynomial decomposition $\set{f_{i,j}}$ is only stated and proved for the polynomials $f_{i', j'-1}$ and $f$, we need to check that downward self-reducibility remains valid when we are working with $p$\ts{th} powers of these polynomials.
        Over a field of characteristic $p > 0$, the $p$\ts{th} power map is linear.
        That is, for any $a, b \in \F$, we have $(a + b)^p = a^p + b^p$.
        Using this identity, we can distribute the $p$\ts{th} power over the sums appearing in the downward self-reducibility property.

        In more detail, recall the sumcheck relation from \Cref{item:sumcheck} of \Cref{def:poly-dec}, which states
        \[
            f(\alpha_1, \ldots, \alpha_{3m+r-j'}) = \sum_{\gamma \in H} f_{i', j'-1}(\alpha_1, \ldots, \alpha_{3m+r-j'}, \gamma).
        \]
        Taking $p^\beta$-th powers of both sides and repeatedly using the fact that $(a + b)^p = a^p + b^p$, we have 
        \[
            f(\alpha_1, \ldots, \alpha_{3m+r-j'})^{p^\beta} = \sum_{\gamma \in H} f_{i', j'-1}(\alpha_1, \ldots, \alpha_{3m+r-j'}, \gamma)^{p^\beta}.
        \]
        Thus, given a circuit that computes $f_{i', j'-1}^{p^\beta}$, we can correctly compute the value of $f^{p^\beta}$ at any point of our choosing.

        Each application of the network of \Cref{prop:universal-arith-network} increases the size of our network by an additive $\poly(m,r,s,d)$ and costs degree $\poly(m,r,d)$.
        We apply this network $h$ times in parallel and add the results together, so the size and degree of the whole network are bounded by $h \cdot \poly(m,r,s,d)$ and $\poly(m,r,d)$, respectively.
	\eproof
	
	\paragraph{High-level overview.} For a better exposition, we first describe the ideas involved in the construction of this arithmetic network in a list below and subsequently discuss the formal details of their implementation via a $\P$-uniform family of randomized arithmetic networks. 
	\begin{itemize}
		\item The input $\vec{\Lambda}$ describes a nonzero arithmetic circuit $D$ (``distinguisher'') satisfying $D\circ \mathcal{G}_{KI,f} = 0$.
		\item {\bf (Distinguisher to ``predictor'' transformation.)} For each $0\leq i \leq \nbar$, define the $i^\mathrm{th}$ \emph{hybrid} circuit
		\[
			C_i(\vec{x},\vec{z}) \coloneqq D(f(\vec{z}|_{S_1}), \ldots, f(\vec{z}|_{S_i}), x_{i+1},\ldots, x_\nbar).
		\]
		Here, both $\vec{z} = (z_1,\ldots,z_\ell)$ and $\vec{x} = (x_1,\ldots,x_\nbar)$ are tuples of indeterminates, and the sets $S_1,\ldots,S_\nbar\subseteq [\ell] $ are the Nisan--Wigderson $(3m+r,\log \nbar)$-design used in~\Cref{cons:generator}.
		
		Since $C_0 = D \neq 0$ and $C_\nbar = D\circ\calG_{KI,f} = 0$, there is $i\in \{0,...,\nbar-1\}$ for which $C_i(\vec{z},\vec{x})\neq 0$ but $C_{i+1}(\vec{z},\vec{x}) = 0$.
		Since we do not know which index $i$ satisfies this condition, we will try out each possible value of $i$ in parallel in its own branch in the arithmetic network.
		
		\item {\bf (Within branch $i$.)} In what follows, we work with one such fixed index $i$, i.e., within one such branch of the arithmetic network.
		
		\begin{itemize}
			\item {\bf (Fixing irrelevant variables.)} Call $\vec{z}|_{S_{i+1}}$ and $x_{i+1}$ the \emph{relevant} variables ($3m+r+1$ in all) and the remaining variables the \emph{irrelevant} variables. Note that the polynomial computed by $C_{i}$ depends on at most $\ell - (3m+r) + \nbar - (i+1) = O(\nbar + \ell)$ irrelevant variables. We use the given \emph{random} set of arithmetic inputs to the network in order to fix the irrelevant variables to field constants in $\F_n$.\footnote{We stress that in each independent branch of the network, we will use different random arithmetic inputs. This ensures independence across branches.} Let us denote this assignment to irrelevant variables by the vector $\vec{\alpha}$. The notation $(\vec{z},\vec{x})\circ\vec{\alpha}$
			refers to the variable vector in which the positions corresponding to the irrelevant $z$ variables (i.e., the indices \emph{not} in $S_{i+1}$), along with the irrelevant $x$ variables (i.e., all $x$ variables except for $x_{i+1}$) are specified according to $\vec{\alpha}$, while the relevant variables are left unset.
			Define
			\[
				C_i' = C_i'(\vec{z}|_{S_{i+1}},x_{i+1}) \coloneqq C_i((\vec{z},\vec{x})\circ\vec{\alpha})
			\]
			and
			\[
				C_{i+1}' = C_{i+1}'(\vec{z}|_{S_{i+1}},x_{i+1}) \coloneqq C_{i+1}((\vec{z},\vec{x})\circ\vec{\alpha}),
			\]
			i.e., the circuits $C_i'$ and $C_{i+1}'$ are obtained from $C_i$ and $C_{i+1}$, respectively, after the partial substitution specified by $\vec{\alpha}$.
			
			Note that $C_{i+1}'$ is obtained as a restriction of $C_{i+1}$ and is therefore the zero polynomial in the branch in which our guess for $i$ is correct.
			On the other hand, since $C_i$ is a nonzero polynomial, and since the assignment to the irrelevant variables comes from the random set of arithmetic inputs, with high probability
			$C_i'$ is a nonzero polynomial as well (see \Cref{clm:ki:recon:branch}).%
			
			\item {\bf (Efficiently constructing a circuit for $C'_i$.)} Since the polynomial obtained by replacing $x_{i+1}$ in $C_i'$ by the polynomial $f(\vec{z}|_{S_{i+1}})$ is the zero polynomial, it follows from \Cref{cor:gauss} that $x_{i+1}- f(\vec{z}|_{S_{i+1}})$ is a factor of the polynomial computed by the circuit $C_i'$. 
			Our goal is to construct a circuit for $C_i'$ and then use \Cref{cor:kaltofen arith network} to factor this circuit, obtaining a circuit that computes $f(\vec{z}|_{S_{i+1}})$.
			Indeed, one caveat is that \Cref{cor:kaltofen arith network} will yield a circuit that computes $f^{p^e}$ for some $e \in \N$, rather than $f$ itself.

			To construct a circuit for $C_i'$, it suffices for the network to construct a circuit for the mapping $\vec{z}|_{S_j\cap S_{i+1}}\mapsto f(\vec{z}|_{S_j} \circ \vec{\alpha})$, for each $1\leq j \leq i$.\footnote{To see this, recall that $C'_i$ gets as input the variables $\vec{z}\rest_{S_{i+1}}$ and $x_{i+1}$, and its goal is to output $D$ on $(f(\vec{z}\rest_{S_1},...,f(\vec{z}\rest_{S_i}),x_{i+1},x_{i+2},...,x_{\nbar})$, where $x_{i+2},...,x_{\nbar}$ are fixed (according to $\vec{\alpha}$, which is hard-wired into $C'_i$) and the values in $\vec{z}$ with indices not in $S_{i+1}$ are also fixed (according to $\vec{\alpha}$). Thus, it suffices for $C'_i$ to be able to compute $\set{f(\vec{z}\rest_{S_j})}_{j\le i}$ as a function of $\vec{z}\rest_{S_j\cap S_{i+1}}$ (and of the hard-wired $\vec{\alpha}$).}
			
			To do so, first note that by Item~\eqref{item:ideg} of~\Cref{prop:poly-dec-cons}, $f^{(j)}_{\vec{\alpha}}(\vec{z}|_{S_j\cap S_{i+1}})\coloneqq f(\vec{z}|_{S_j} \circ \vec{\alpha})$ 
			is a polynomial of individual degree at most $6h$ and total degree at most $h\cdot \polylog(s)$ over $|S_j\cap S_{i+1}|\leq \log \nbar$ variables. As such, it can be interpolated using at most ${(6h)}^{\log \nbar}$ evaluations of $f$ on a grid specified by a fixed, explicit set $T\subseteq\F_n$ of size  $6h$ (for which we can  use an enumerator for the function $6h$)
            More precisely, we can explicitly write $f^{(j)}_{\vec{\alpha}}$ using multivariate interpolation as
			\mm{
				f^{(j)}_{\vec{\alpha}}(\vec{z}|_{S_j\cap S_{i+1}}) \coloneqq \sum_{\vec{\beta} \in T^{S_j\cap S_{i+1}}} {f}(\vec{\beta}\circ \vec{\alpha}) \prod_{k\in S_j\cap S_{i+1}} \prod_{\beta \in T\setminus \set{\beta_k}} \frac{z_k - \beta}{\beta_k - \beta}, \eqtag{eq:f-interpolation}
			}
			where the evaluations ${f}(\vec{\beta}\circ \vec{\alpha})$ can be computed using \Cref{clm:ind-recons:dsr} and $p$\ts{th} root gates.

			Doing this for every $j\in [i]$ and plugging in the circuit $D$, we obtain a circuit for $C_i'$ whose size is bounded by
			\[
				\poly(|T|^{\log \nbar + 1}\cdot \log \nbar \cdot \nbar\cdot n) = \poly(n,\nbar, h^{\log \nbar}) \le \poly(n,\nbar^{\log h}). \eqtag{eq:ki:recon:grid}
			\]
			and degree is bounded by
			\[
				O(d |T| \log \nbar) = O(dh\log \nbar). \eqtag{eq:ki:recon:grid:deg}
			\]
			
			\item {\bf (Kaltofen's factorization.)} 
			We now invoke the network of \Cref{cor:kaltofen arith network}.
			This network takes as input a description of $C'_i$ and computes a list of (descriptions of) circuits $C_{i,1},\ldots,C_{i,t_i}$ and a list of natural numbers $\beta_{i,1}, \ldots, \beta_{i,t_i} \in \N$, where $t_i\le\deg(C'_i)$.
			We are guaranteed that for every polynomial $g$ such that $C_i'(\vec{z}|_{S_{i+1}}, g(\vec{z}|_{S_{i+1}}))$ is the identically zero polynomial, at least one of the circuits $C_{i, k}$ computes the polynomial $g^{p^{\beta_{i,k}}}$.
			In particular, one of the circuits in this list computes $f^{p^e}$ for some $e \in \N$.
			For each $k \in [t_i]$, the size of the circuit $C_{i,k}$ is bounded by
			\mm{
				\poly(|C'_i|,\deg(C'_i),d,\nbar) < \poly(n, d, \nbar^{\log h}), \tag*{(by Eqs.~\eqref{eq:ki:recon:grid},~\eqref{eq:ki:recon:grid:deg})}
			}
			and the degree of each $C_{i,k}$ is bounded by $\deg(C_i')<dh\log(\nbar)$ (by Eq.~\eqref{eq:ki:recon:grid:deg})
		\end{itemize}

		\item ({\bf Weeding candidate circuits.})

		We know that for some choice of $i\in[\nbar-1]$ and $k\in[t_i]$, the circuit $C_{i,k}$ computes $f^{p^{\beta_{i,k}}}$.
		To determine which circuit computes a $p$\ts{th} power of $f$, we iteratively test each one against evaluations of $f^{p^{\beta_{i,k}}}$ using polynomial identity testing.
		
		Specifically, we evaluate each candidate circuit $C=C_{i,k}$ on a random point, which we choose using $3m+r$ arithmetic inputs taken from the random set of arithmetic values given to the network $R_{i,j}$.
		We can evaluate $C_{i,k}^{1/p^{\beta_{i,k}}}$ on this point by using \Cref{prop:universal-arith-network} and $\beta_{i,k}$ many $p$\ts{th} root gates, and we can evaluate $f$ on this point using \Cref{clm:ind-recons:dsr}.
		We then use a $\testequals$ gate (See \Cref{def:arith-network}) to check if the two evaluations are equal.
		
		We choose the random evaluation point independently for each candidate circuit.
		If $C_{i,k}$ computes $f^{p^{\beta_{i,k}}}$, then this circuit always passes the test.
		Conversely, if $C_{i,k}$ does not compute $f^{p^{\beta_{i,k}}}$, then the Schwartz--Zippel lemma (\Cref{lem:sz}) implies that if each coordinate of the evaluation point is chosen uniformly at random from a subset $S\subseteq \F_n$ of size at least $1000 \nbar\Delta \deg(C_i')^2$, then $C_{i,k}$ passes the test with probability at most $1/1000\nbar\Delta \deg(C_i')$.
		Because there are at most $\nbar \deg(C'_i)$ candidate circuits to test, by the union bound, the probability that some circuit erroneously passes the test is at most $1/1000\Delta$.
		(We will use the iterative step in the current lemma $O(\Delta)$ times, so we require an error bound of $O(1/\Delta)$ here to ensure the overall algorithm has bounded error probability.)
		
		\item
		If one of the candidate circuits passes the test, we have a description of a circuit that computes $f^{p^e}$ for a known integer $e \in \N$.
		The network outputs the description of this circuit for $f^{p^e}$ together with a binary description of $e$. 
		If more than one candidate circuit passes the test in the previous step, we break ties arbitrarily.
		(Even for a correct choice of $i$ for the hybrid, it is possible that after choosing the fixing $\vec{\alpha}$ the circuit $C_i'$ computes the zero polynomial, or that the arithmetic network from \Cref{cor:kaltofen arith network} may err.)
	\end{itemize} 
	
	\paragraph{Details of implementation.} 
	We formalize the implementation details using a sequence of claims below. The first claim helps describe the behavior of $\calN$ inside a given branch, i.e., corresponding to a given setting of an index $i\in [\nbar-1]$.
	
	\begin{subclaim} [evaluating a single branch] \label{clm:ki:recon:branch}
		There is a $\P$-uniform family $\Nin=\{\Nin_n\}$ of randomized arithmetic networks such that $\Nin_n$ that has the following properties.
        \begin{itemize}
            \item It takes as input $\vec{\Lambda}\in\F_n^n$, an integer $i\in[\nbar-1]$, an integer $\beta \in \N$, and a standard description  of a size $s'$, degree $d'$ arithmetic circuit computing $f_{i',j'-1}^{p^\beta}$.
            \item It outputs standard descriptions of circuits $C_{i,1},\ldots,C_{i,t_i}$ and a list of numbers $\beta_{i,1}, \ldots, \beta_{i,t_i} \in \N$, where $t_i\le d'$ and each circuit is of size $\poly(n, \nbar^{\log h},d)$ and degree at most $O(d\cdot h\cdot\log(n))$.
            Moreover, if the input $i$ is so that $C_i(\vec{z},\vec{x}) \neq 0$ but $C_{i+1}(\vec{z},\vec{x}) = 0$, then there exists $k\in[t_i]$ for which $C_{i,k}$ computes a $(3m+r)$-variate polynomial $g^{p^{\beta_{i,k}}}$ such that $C_i'(\vec{z}|_{S_{i+1}},g(\vec{z}|_{S_{i+1}}))$ is the zero polynomial, where $C_i'$ is a nonzero polynomial that is obtained from $C_i$ by a partial substitution to its irrelevant variables.%
            \item Finally, $\Nin_n$ has size $\poly(n,d,s',d',\nbar^{\log h}, \log(q),\log(s))$, degree bounded by %
            $\poly(h,d,p,\log s,d')$ (using $p$\ts{th} root gates), and error degree bounded by $\poly(h,d,\log(s),d')$.
        \end{itemize}
	\end{subclaim}
	\begin{subproof}
		Recall the three steps of the computation of $\Nin_n$: fixing irrelevant variables, computing a description of $C'_i$, and evaluating the network from \Cref{cor:kaltofen arith network}.
		
		\medskip\noindent\underline{Fixing irrelevant variables.} To fix irrelevant variables to $\vec{\alpha}$, the Nisan--Wigderson design will be hard-wired into $\Nin_n$.
		Recall that \Cref{cons:generator} uses the designs from \Cref{cor:nw-design}, which are computable in time $\poly(\nbar,m,r)<\poly(n,\log(s))$.
		The Turing machine that prints $\Nin_n$ will compute these designs and hard-wire them.
		Given these designs and access to random elements, the step of fixing irrelevant variables can be implemented by a $\P$-uniform arithmetic network of size $\poly(s')$ and degree $O(1)$.
		
		\medskip\noindent\underline{Computing a description of $C'_i$.} Once $\vec{\alpha}$ is determined, the network $\Nin_n$ computes a description of $C'_i$ that takes input $(\vec{z}\rest_{S_{i+1}},x_{i+1})$, computes $\set{f(\vec{z}\rest_{S_j})}_{j\le i}$ and $x_{i+2},\ldots,x_{\nbar}$ from its inputs and from $\vec{\alpha}$, and evaluates $D$ (i.e., the circuit represented by $\vec{\Lambda}$) on the resulting sequence. In order to compute the description of each $f^{(j)}_{\vec{\alpha}}$, which is defined as the circuit mapping $\vec{z}\rest_{S_j\cap S_{i+1}}\mapsto f(\vec{z}|_{S_j} \circ \vec{\alpha})$, the network implements the interpolation described in Eq.~\eqref{eq:f-interpolation}, while computing the constants $f(\vec{\beta}\circ\vec{\alpha})$ using the network in \Cref{clm:ind-recons:dsr} together with $\beta$ applications of a $p$\ts{th} root gate. 

		We argue that the computation of (a description of) $C'_i$ can be performed by a $\P$-uniform $\Nin_n$ of size $\poly(n,\nbar^{\log h},s',d',\log(s))$ and degree $\poly(d',log(s))$, using $p$\ts{th} root gates. Recall that (as explained above) each $f^{(j)}_{\vec{\alpha}}$ is of size $\poly(n,\nbar^{\log h})$ and degree $O(h\cdot\log \nbar)$; using \Cref{prop:short-desc}, the network $\Nin_n$ can produce a description for each $f^{(j)}_{\vec{\alpha}}$ of size $\poly(n,\nbar^{\log h})$, and the network for doing so is uniform and of size 
		\mm{
		\poly(n,\nbar^{\log h})\cdot h\cdot\poly(m,r,s',d') \le \poly(n,\nbar^{\log h},s',d',\log(s))
		}
		and degree $\poly(m,r,d')=\poly(d',\log(s))$, using $p$\ts{th} root gates. (We used the fact that the field family $\F$ is feasible to argue that this network is uniform, i.e. a Turing machine can hard-code grid elements from $\F_n$ for the interpolation step.) The wiring of the functionality computing a description of $C'_i$ (from $\vec{\Lambda}=D$, $\vec{\alpha}$, and the circuits for $f^{(j)}_{\vec{\alpha}}$) does not depend on the inputs to $\Nin_n$, and can thus be computed by the uniform machine that constructs $\Nin_n$ and hard-wired into $\Nin_n$ (without increasing the degree, and with size polynomial in $|C'_i|$). Thus, overall, computing the description of $C'_i$ can be done by a uniform network of size $\poly(n,\nbar^{\log h},s',d',\log(s))$ and degree $\polylog(s)$ (relying on \Cref{prop:comp-arith-network}).

		\medskip\noindent\underline{Factoring.} The last step is applying~\Cref{cor:kaltofen arith network}. The bounds on the number $t_i$ of output circuits and on the size and degree of each output circuit were established in the description above (i.e., they follow from Eqs.~\eqref{eq:ki:recon:grid} and~\eqref{eq:ki:recon:grid:deg} and from~\Cref{cor:kaltofen arith network}). The size of the network in \Cref{cor:kaltofen arith network} is at most $\poly(|S_{i+1}| + 1, \deg(C_i'),\size(C_i'),\log(q)) = \poly(n,\nbar^{\log h},d,\log(s),\log(q))$ and its degree is at most $\poly(|S_{i+1}| + 1, \deg(C_i'),p)\le\poly(d,h,\log s,p)$, 
		where we again use the bounds on the size and degree of $C_i'$ from Eqs.~\eqref{eq:ki:recon:grid} and~\eqref{eq:ki:recon:grid:deg}.

		\medskip\noindent\underline{Error degree.} To see the error degree bound, we note that the sources of randomness in $\Nin_n$ are (i) the choice of the vector $\vec{\alpha}$ to fix the irrelevant variables in $C_i$ (which has degree at most $d\cdot \deg(f) \leq d\cdot h\cdot \polylog(s)$), and (ii) the execution of the factoring arithmetic network from \Cref{cor:kaltofen arith network}, which has error degree $\poly(d,h,\log \nbar)$.
		From \Cref{lem:sz} applied to $C_i(\vec{z},\vec{x})$ when viewed as a polynomial over the field extension $\F(\vec{z}|_{S_{i+1}},x_{i+1})$, if $C_i$ is a nonzero polynomial, then fixing the irrelevant variables of $C_i$ to uniformly random elements of a set $S \subseteq \F_n$ results in the zero polynomial with probability $\deg(C_i) / |S|$.
		It follows that the sub-network that computes the description of $C_i'$ has error degree at most $\deg(C_i) \le d \cdot h \cdot \polylog(s)$.
		Therefore, the final error degree bound follows from \Cref{prop:error-deg-comp}.
	\end{subproof}
	
	The next subclaim describes the behavior of the network $\calN$ when testing each candidate circuit $C_{i,k}$ to see if it computes a $p$\ts{th} power of $f$.
	
	\begin{subclaim}[testing candidate circuits] \label{clm:ki:recon:test}
		There is a $\P$-uniform randomized arithmetic network $\calT$ with $p$\ts{th} root gates that gets as input $\vec{\Lambda}$, natural numbers $\beta, \gamma \in \N$, the description of a size $s'$, degree $d'$ arithmetic circuit that computes $f_{i',j'-1}^{p^\beta}$, and the description of a circuit $C$ of size $s''$ and degree at most $d''$ computing a $(3m+r)$-variate polynomial $g^{p^\gamma}$, and outputs a boolean value $b$ which indicates whether $f = g$ as polynomials.
		Moreover, $\calT$ has size $\poly(s',s'',d'',h)$, degree at most $\poly(d',d'')$, and error degree at most $\max(d', d'')$.
	\end{subclaim}
	
	\begin{subproof}
		The network first generates a random point $\vec{\alpha} \in \F^{3m+r}$ using $3m+r$ values from its set of random inputs.
		The network then evaluates $f^{p^\beta}$ and $g^{p^\gamma}$ at $\vec{\alpha}$.
		The value of $g(\vec{\alpha})^{p^\gamma}$ is obtained by applying the network of \Cref{prop:universal-arith-network} to the given description of the circuit for $g^{p^\gamma}$.
		The value of $f(\vec{\alpha})^{p^\beta}$ is computed by applying the network of \Cref{clm:ind-recons:dsr} to the given circuit that computes $f_{i', j'-1}^{p^\beta}$.
		To obtain the values of $f(\vec{\alpha})$ and $g(\vec{\alpha})$, the network applies $\beta$ copies of a $p$\ts{th} root gate to the value of $f(\vec{\alpha})^{p^\beta}$, and likewise applies $\gamma$ copies of a $p$\ts{th} root gate to the value of $g(\vec{\alpha})^{p^\gamma}$.
		Finally, the network feeds the difference of these two arithmetic values into a test gate and returns the value output by this test gate.
		
		The degree and size bounds on the network, as well as the uniformity, follow from the corresponding bounds on the network of \Cref{prop:universal-arith-network}.
		The error degree bound follows from \Cref{lem:sz} and the observation that $\deg(f - g)\leq \max\{d',d''\}$.
	\end{subproof}

	Finally, the subclaim below describes how the network $\calN$ aggregates the results of the tests across different branches to choose one circuit that computes a $p$\ts{th} power of $f$.
	
	\begin{subclaim}[choosing a circuit according to indicator values] \label{clm:ki:recon:ind}
		There is a $\P$-uniform arithmetic network $\calS$ that takes as input a sequence of descriptions of circuits $C_1,\ldots, C_t$ and boolean values $b_1,\ldots,b_t$, and outputs the description of the circuit $C_i$ with the smallest index $i$ such that $b_i = 1$ (we do not care about the output, if none of the $b_i$ are $1$). Moreover, $\calS$ has linear size and degree $O(1)$.
	\end{subclaim}
	\begin{subproof}
		Note that each description consists of an arithmetic part and a boolean part. Let $(\vec{\alpha}_i, s_i)$ denote the description of $C_i$, where $\vec{\alpha}_i$ is the arithmetic part with each of its coordinates in $\F_n$, and $s_i$ is  a boolean string denoting its boolean part. Moreover, using padding by $0\in \F_n$ (respectively, the boolean value $0$) if necessary, let us assume without loss of generality that the length of the arithmetic part (respectively, boolean part) of the description of $C_i$ is the same for every $i\in[t]$.
		
		The task then is to construct an arithmetic network that outputs $(\vec{\alpha}_{i^*}, s_{i^*})$ for the smallest index $i^*\in[t]$ for which $b_i = 1.$ Because $i^*$ is unique, the task can be performed independently for the boolean and the arithmetic part. We note that the boolean part of this task is easily performed by a linear sized $\P$-uniform boolean circuit with $\land, \lor, \neg$ gates. 
		
		Next, we show that there is a simple $\P$-uniform arithmetic network gadget that can output the \textit{first} coordinate of $\vec{\alpha}_{i^*}$. This suffices for the proof of the proposition, as one can then simply repeat the gadget for all the coordinates. The idea is to repeatedly use $\select$ gates to achieve this. Notice that if $t=2$, then simply $\select(b_1, \alpha_{2,1},\alpha_{1,1})$ works, and if $t = 3$, then $\select(b_1,\select(b_2,\alpha_{3,1},\alpha_{2,1}), \alpha_{1,1})$ works. Extrapolating this pattern, we conclude that there is a $\P$-uniform linear sized arithmetic network that solves this task.
	\end{subproof}
	
	Given these claims, we can now  provide a complete description of $\calN_n$ as follows. The output of $\calN_n(\vec{\Lambda})$ is provided by the output of $\calS$, which in turn, receives as input:
	\begin{enumerate}[(i)]
		\item The output of the execution of $\Nin_n$ on each of the indices $i\in[\nbar-1]$, i.e., the descriptions of the circuits $C_{i,1},\ldots,C_{i,t_i}$, and
		\item A boolean value $b_{i,k}$ corresponding to each circuit $C_{i,k}$, which is in turn obtained as the output of $\calT$ when instantiated with the description of $C_{i,k}$.
	\end{enumerate}
	We stress that every randomized arithmetic network that is used as a subroutine in the description above uses its own block (i.e., a fresh block) of random arithmetic inputs from the original set of random arithmetic inputs that is given to $\calN_n$.
	
	\paragraph{Complexity of $\calN_n$.} 
	By~\Cref{clm:ki:recon:branch}, each circuit $C_{i,k}$ fed into the network from~\Cref{clm:ki:recon:test} is of size $s''=\poly(n,\nbar^{\log h},d)$ and degree at most $d''=O(d\cdot h\cdot\log(\nbar))$, and the number of output circuits in each of the $\nbar$ branches is at most $O(d\cdot h\cdot\log(\nbar))$. Hence, combining Claims~\ref{clm:ki:recon:branch},~\ref{clm:ki:recon:test}, and~\ref{clm:ki:recon:ind}, the total size of $\calN_n$ is at most $\poly(n, d,\nbar^{\log h}, s',d',\log(q),\log(s))$.
 
    To see the degree bound, first note that the degree of $\calT$ (the network from \Cref{clm:ki:recon:test}) in each of its instantiations is bounded by $\poly(d',d'') = \poly(h,d,\log \nbar,d')$. Also, note that since instantiation of $\calT$ that is executed in order to compute the \emph{tuple} of boolean values $(b_{i,k})_{i\in[\nbar - 1],k\in[t_i]}$ is run in parallel inside $\calN_n$, the degree of the (multi-output) sub-network that computes $(b_{i,k})_{i\in[\nbar - 1],k\in[t_i]}$ is also bounded by the degree of $\calT$ in any of its instantiations, and in particular, is also bounded by $\poly(h,d,\log \nbar,d')$ (by \Cref{prop:concat-arith-networks}). Next, the degree of $\Nin_n$ (the network from \Cref{clm:ki:recon:branch}) in each of its instantiations is bounded by $\poly(h,d,p,\log(s),d')$, and as they are each run in parallel inside $\calN_n$, the degree of the network that computes the overall set of candidate circuit descriptions $\{\langle C_{i,k}\rangle|i\in[\nbar-1],k\in[t_i]\}$ is also bounded by $\poly(h,d,p,\log(s),d')$ (by \Cref{prop:concat-arith-networks}). Finally, since $\calN_n$ is the composition of $\calS$ (which has degree $O(1)$) with these networks, the claimed overall degree bound of $\poly(h,d,p,\log(s),d')$ follows from \Cref{prop:comp-arith-network}.
	
	To see the error degree bound, using the same analysis as above, and once again using \Cref{prop:concat-arith-networks} and \Cref{prop:error-deg-comp} and the three claims above, it follows that the error degree of $\calN_n$ is bounded by $\poly(\nbar,h,d,\log(s),d')$, where we get the additional factor of $\nbar$ from an application of \Cref{prop:concat-arith-networks}, since each of the $\nbar$ instantiations of $\Nin_n$ are executed in parallel.%
	
	\paragraph{The ``furthermore'' part.}
        The proof is essentially identical to that of the basic case. 
        The only difference is that to prove \Cref{clm:ind-recons:dsr}, we now rely on the downward self-reducibility relation between $f_{i-1,2m+r}$ and $f_{i,0}$ as specified in \Cref{item:sumcheck} of \Cref{def:poly-dec}.
        Recalling that $f_{i-1,2m+r}(\vec{w}) = f_{i-1}(\vec{w})$, we have the equality
        \begin{align*}
            f_{i,0}(\vec{w},\vec{\sigma}) &\coloneqq 
                \begin{cases}
                \Phi_i(\vec{w}, \sigma_{1},\ldots,\sigma_{t}) \del{f_{i-1}(\sigma_{1},\ldots,\sigma_m) + f_{i-1}(\sigma_{m+1},\ldots,\sigma_{2m})}, &\text{if $i$ is odd}\\
                \Phi_i(\vec{w}, \sigma_{1},\ldots,\sigma_{t}) \del{f_{i-1}(\sigma_{1},\ldots,\sigma_m) \cdot f_{i-1}(\sigma_{m+1},\ldots,\sigma_{2m})}, &\text{if $i$ is even.}\\
                \end{cases}
        \end{align*}
        Taking $p^{\beta}$-th powers and using the fact that $(a + b)^p = a^p + b^p$ over a field of characteristic $p$, we have
        \begin{align*}
            f_{i,0}(\vec{w},\vec{\sigma})^{p^\beta} &=
                \begin{cases}
                \Phi_i(\vec{w}, \sigma_{1},\ldots,\sigma_{t})^{p^\beta} \del{f_{i-1}(\sigma_{1},\ldots,\sigma_m)^{p^\beta} + f_{i-1}(\sigma_{m+1},\ldots,\sigma_{2m})^{p^\beta}}, &\text{if $i$ is odd}\\
                \Phi_i(\vec{w}, \sigma_{1},\ldots,\sigma_{t})^{p^\beta} \del{f_{i-1}(\sigma_{1},\ldots,\sigma_m)^{p^\beta} \cdot f_{i-1}(\sigma_{m+1},\ldots,\sigma_{2m})^{p^\beta}}, &\text{if $i$ is even.}\\
                \end{cases}
        \end{align*}
        Thus, given circuits for $f_{i-1}^{p^\beta}$ and $\Phi_i^{p^\beta}$, we can evaluate $f_{i,0}^{p^\beta}$ at a point of our choosing using a network with similar complexity bounds to \Cref{clm:ind-recons:dsr}.
        By assumption, we have such a circuit that computes $f_{i-1}^{p^\beta}$.
        A circuit of size $\polylog(s) \cdot \beta \log(p) \le \polylog(s, d)$ for $\Phi_i^{p^\beta}$ can be obtained by using repeated squaring to take the $p^{\beta}$-th power of the circuit for $\Phi_i$ given by \Cref{clm:arithmetization}.

\subsubsection{The reconstruction in $\sigma=0$ mode}

We now describe the reconstruction procedure when it is instantiated in ``mode'' $\sigma=0$. The proof is exactly identical to that of~\Cref{prop:recons-arith-net'}, the only difference being in the parameterization: Instead of using $h=2,m=\lceil \log(s) \rceil,r=\polylog(s)$ we now use $r=0$ and arbitrary $h,m$ such that $h^m\ge \size(C'_n)$. This yields the following result:

\begin{proposition}  [the reconstruction for~\Cref{thm:ki:tarhsg} with $\sigma=0$] \label{prop:recons-arith-net'}
    Let $\set{C_n}_{n \in \N}$ be a $\log^c$-uniform family of multi-output algebraic circuits of size $s(n)$ having $n$ output gates, defined over a feasible sequence of fields $\set{\F_n}_{n \in \N}$, where each $\F_n$ has order $q(n) \le \exp(n^{O(1)})$, and let $\eps>0$ be an arbitrarily small constant. Let $\vec{\Lambda} \in \F^{n}_{q(n)}$, and let $\calG_{\vec{\Lambda}}^{(h,m,0,\nbar)}$ be the generator described in \Cref{cons:generator:specific}, with the $r=0$ instantiation.
    Then, there is a $\P$-uniform family of randomized arithmetic networks $R=\{R_n\}_{n\in\N}$ with $p$\ts{th} root gates that satisfies the following.
    \begin{enumerate}
        \item 
            The input to $R_n$ is a vector $\vec{\Lambda}\in\F_n^n$ that is the universal encoding %
            of an arithmetic circuit $D$ that computes an $\nbar$-variate, degree $d$ nonzero polynomial and satisfies $D\circ \calG_{\vec{\Lambda}}^{(h,m,r,\nbar)}(y_1, y_2, \vec{z}) = 0$, and a candidate output $w\in \F_n^n$.
        \item
            The size of the network $R_n$ is at most $s^{\eps}\cdot\poly(n,d,\nbar^{\log\log(s)}, \log(s))$.
        \item The degree of $R_n$ is at most $(n\cdot d\cdot h\cdot p)^{\polylog(s)}$ 
        (allowing $p$\ts{th} root gates, and assuming that $\F$ is efficiently constructible).
        \item
            $R_n(\vec{\Lambda})$ computes the vector $C_n(\vec{\Lambda}) \in \F_n^{n}$ with error degree at most $\poly(d,h,\nbar,\log(s))$.
    \end{enumerate}
\end{proposition}

\bproof
We instantiate~\Cref{cons:generator:specific} with $r=0$ and with $h=2^{\log(s')/(c\cdot\log(\nbar))}$ and $m=c'\cdot(\log(\nbar))$, where $c'>c>1$ are sufficiently large constants that depend on $\eps>0$. We follow the proof of~\Cref{prop:recons-arith-net}, using the new values for $m,h,r$. The changes are as follows:
\begin{enumerate}
	\item The circuit for $f_0$ is of size $O(n\cdot m\cdot h\cdot\polylog(s))=n\cdot s^{1/O(\log(\nbar))}\cdot\polylog(s)<s^{\eps}\cdot n$.
	\item The number of applications of~\Cref{lem:ind-recons} is $O(m\cdot\Delta)<\polylog(s)$.
	\item The main difference is in the complexity bounds of~\Cref{lem:ind-recons}. Going through~\Cref{sec:ki:recon:step}, when the network constructs a circuit for $C'_i$, the individual degree cannot be assumed to be $O(h)$ anymore, and thus we use the overall degree bound of $h\cdot \polylog(s)$ for $f_{i,j}$ given in Item~\eqref{item:ideg} of~\Cref{prop:poly-dec-cons}. The size of each $f^{(j)}_{\vec{\alpha}}$ thus becomes $(h\cdot\polylog(s))^{\log(\nbar)}$ and its degree becomes $h\cdot\polylog(s)\cdot\log(\nbar)$.
	
	Thus, in~\Cref{clm:ki:recon:branch}, the size of each $C_{i,k}$ becomes $\poly(n,(h\cdot\polylog(s))^{\log(\nbar)})$ and its degree is at most $\poly(h,d,\log(s))$. (And these are the size and degree bounds on the circuit that $\Nin_n$ outputs, since it just outputs one of the $C_{i,k}$'s.)	The size and degree of $\Nin_n$ increases appropriately, and thus the size of $\calN_n$ is at most $\poly(n,d, (h\cdot\polylog(s))^{\log(\nbar)}) ,s',d',\log(s),\log(q))$.
\end{enumerate}

Replacing the calculation in Eq.~\eqref{eq:ki:recon:sizefinal}, the final size of $R_n$ is now 
\[
	(2m+r)\cdot\Delta\cdot\poly(n,d,(h\cdot\polylog(s))^{\log(\nbar)}) < \poly(n,d,(h\cdot\polylog(s))^{\log(\nbar)},\log(q)),
\]
and our claimed size bound follows since, by an appropriate choice of constant $c=c(\eps)>1$ for $h=2^{\log(s')/(c\cdot\log(\nbar))}$, we have
\[
	(h\cdot\polylog(s))^{\log(\nbar)} = h^{\log(\nbar)}\cdot (\polylog(s))^{\log(\nbar)} = s^{\eps}\cdot\polylog(s)^{\log(\nbar)}.
\]

Next, to analyze the error degree of $R_n$, we first note that the degree bound on $C_i'$ that is mentioned in \eqref{eq:ki:recon:grid:deg} changes to $\poly(h,d,\log s)$ as the individual degree cannot be assumed to be $O(h)$ anymore, and instead, we use a bound of $h\cdot \polylog(s)$ for $\deg(f_{i,j})$. However, in~\Cref{clm:ki:recon:branch}, the functional degree remains $\poly(h,d,\log s)$ as argued at the end of its proof (note that randomly fixing the irrelevant coordinates still yields an upper bound of $\poly(h,d,\log s)$). Next, we observe that while the size of the output circuits of \Cref{clm:ki:recon:branch} increases to $\poly(n,(h\cdot\polylog(s))^{\log(\nbar)})$, its degree remains at most $\poly(h,d,\log(s))$. This is why subsequent applications of \Cref{clm:ki:recon:test} (the only remaining source of randomness now in $\calN_n'$ after \Cref{clm:ki:recon:branch}) provide identical asymptotic upper bounds on the error degree, as we plug in an identical bound for $d''$.

Finally, we perform a similar analysis to conclude that the degree bound remains identical in each iteration of \Cref{lem:gkss-ind-recons}, and hence remains identical overall. 
\eproof

\subsubsection{Simulating $p$\ts{th} root gates} \label{sec:ki:recon:pth}

The description of the reconstruction procedure so far assumed access to $p$\ts{th} root gates, and assumed that the field family is constructible. We now explain how the construction differs when we do not assume this. In both cases, the main difference is the degree bound on the network.

First, when the field family is not necessarily constructible, the degree bound in~\Cref{cor:kaltofen arith network} increases from $\poly(n,d,p)$ to $n^{O(1)}\cdot d^{O(\log q)}$. This increases the degree bounds in~\Cref{clm:ki:recon:branch} to $(h\cdot d\cdot \log s)^{O(\log q)}\cdot\poly(d')$, which in turn increases the degree bound in~\Cref{lem:ind-recons} to $(h\cdot d\cdot \log s)^{O(\log q)}\cdot\poly(d_{i,j-1})$, and the final degree bound for the reconstruction network becomes $O(n\cdot d\cdot h)^{\log(q)\cdot\polylog(s)}$. (The calculation is identical for the $\sigma=1$ mode and for the $\sigma=0$ mode.)

To finish the proof, we now explain how to simulate $p$\ts{th} root gates with standard arithmetic gates. Each gate can be simulate with a degree increase of at most $q$, but naively simulating all gates might increase the network's degree by a factor of $q$ raised to the network's size. We now explain how to avoid this, paying only $q^{\polylog(s)}$.

\paragraph{A single iteration.}
The arithmetic networks $\calN$ and $\calN'$ constructed in \Cref{lem:ind-recons} make use of $p$\ts{th} root gates in two places.
    One is to obtain descriptions of circuits that compute the restricted polynomials $f^{(j)}_{\vec{\alpha}}$.
    The need to take $p$\ts{th} roots here arises because we interpolate circuits for the $f^{(j)}_{\vec{\alpha}}$ using Equation \eqref{eq:f-interpolation}, but we only have access to $p$\ts{th} powers of the constants $f(\vec{\beta} \circ \vec{\alpha})$ appearing in that expression.
    The other use of $p$\ts{th} root gates is in \Cref{clm:ki:recon:test}, where we test if a candidate circuit indeed computes a $p$\ts{th} power of $f$.

    As remarked following the statement of \Cref{thm:ki:tarhsg}, we can implement a $p^{\beta}$-th root by the operation $x \mapsto x^{q/p^\beta}$, where $q$ is the size of the underlying field.
    In one invocation of \Cref{lem:ind-recons}, these $p$\ts{th} root operations can be batched into two parallel steps: one for the interpolation of circuits that compute the polynomials $f^{(j)}_{\vec{\alpha}}$, and one to perform testing of candidate circuits.
    Implementing a single $p^\beta$-th root operation increases the degree of a network by a multiplicative factor of at most $q$.
    By \Cref{prop:concat-arith-networks}, each batch of $p$\ts{th} root operations performed in parallel can be simulated with the same multiplicative increase of $q$ in the degree.
    Using \Cref{prop:error-deg-comp}, simulating two parallel batches of $p$\ts{th} root operations increases the degree of the network by a factor of $q^2$.

    \paragraph{The overall procedure.}
    The reconstruction procedure invokes \Cref{lem:ind-recons} a total of $(2m + r) \cdot \Delta \le \polylog(s)$ times in sequence, and uses $p$\ts{th} root gates one last time (in parallel) after the final iteration.     By \Cref{prop:error-deg-comp}, the degree of the reconstruction network is bounded by the product of the degrees of the networks used at each step of the reconstruction.
    We saw that simulating the $p$\ts{th} root operations in a single application of \Cref{lem:ind-recons} incurs a degree blowup of $q^2$ in each step, and there are $\polylog(s)$ steps, so the degree of the full reconstruction procedure increases by a multiplicative factor of $q^{\polylog(s)}$ when simulating $p$\ts{th} root operations with multiplications in the field.

\section{Proofs of the main results} \label{sec:pit}

Using the targeted generators from~\Cref{thm:gkss-tarhsg} and~\Cref{thm:ki:tarhsg}, we will now prove the main results. In~\Cref{sec:pit:necessary} we show that hardness for arithmetic networks with $\pit$ gates on almost all inputs is necessary for $\pit$, and in~\Cref{sec:pit:sufficient} we show that very similar hardness assumptions suffice for $\pit$.

\paragraph{Uniformity over sequences of finite fields.}
Consider an arithmetic circuit $C_n$ over a finite field $\F_n$. We %
would like to perform basic operations on $C_n$, which may increase its size or its degree, while still considering the modified circuit $C'_n$ as over $\F_n$. In particular, after the modification we may wish to perform $\pit$ on $C'_n$, or to evaluate some hard polynomial on (the description of) $C'_n$. This should be a triviality, but it runs into an artificial complication when considering uniform families of circuits computing $\pit$ or hard polynomials. Specifically, if we consider a uniform family of functions $f=\set{f_n}$ over a sequence of finite fields $\set{\F_n}$, where $f_n$ is defined over $\F_n$, then we cannot (for purely syntactic reasons) start with a description of $C_n$ over $\F_n$, modify it to $C'_n$ of size $n'$ over $\F_n$, and evaluate $f_{n'}$ (or $f_n$) on $C'_n$.

Indeed, this syntactic dependency is artificial and overly strict, and to avoid it, we will consider a more general notion that allows a polynomial gap between the index of the field and the size of the circuit. Specifically, we say that $\sft\colon\N\rightarrow\N$ is a {\sf shift function} if $\sft$ is increasing, and $\sft(n)\le n$ for all $n\in\N$, and $\sft$ is computable on input $n$ in time $\poly(n)$. We say that a shift function is {\sf polynomially bounded} if $\sft(n)=n^{\Omega(1)}$. For a sequence $\F=\set{\F_n}_{n\in\N}$ of finite fields and a shift function $\sft$, we define $\F^{\sft}=\set{\F^{\sft}_{n}}_{n\in\N}$ as the sequence of fields wherein $\F^{\sft}_{n}=\F_{\sft(n)}$. Then, when we think of applying operations on uniform circuit families over $\F$, we allow any polynomial slack, i.e. we allow circuits of size $n$ over $\F_{\sft(n)}$. For example, for $\pit$:

\begin{definition} [solving $\pit$ over finite fields]
Let $\F=\set{\F_n}$ be a sequence of finite fields. We say that a function\footnote{The notation $\F^*$ here means strings over the alphabet $\F$ of arbitrary length.} $f\colon\cup_n(\F_n)^*\rightarrow\{0,1\}$ {\sf solves $\pit$ over $\F$} if for any polynomially bounded shift function $\sft$, and every sufficiently large $n\in\N$, and every $\vec{\Lambda}\in\F_{\sft(n)}^n$ describing an arithmetic circuit of size and degree at most $n$, the output is $f(\vec{\Lambda})=0$ if and only if the circuit described by $\vec{\Lambda}$ computes the zero polynomial.
\end{definition}

Similarly, when we will consider uniform families of arithmetic circuits (computing hard polynomials), we will also consider them with respect to all polynomially bounded shifts. We stress that both for $\pit$ and for hard polynomials, we consider shifts both in the assumption and in the conclusion of all our results.

\subsection{A necessary assumption for PIT} \label{sec:pit:necessary}

The following result asserts that hardness for networks with $\pit$ gates on all but finitely many inputs is necessary for $\pit$. For convenience, we state the result simultaneously for fields of characteristic zero and for finite fields (i.e., we do not separate the cases).

\begin{theorem} [PIT requires hardness for arithmetic networks with PIT gates on almost all inputs] \label{prop:pit:necessary}
Let $\F=\set{\F_n}_{n\in\N}$ be a sequence of fields, where either the fields are all finite, or there is a field $\F$ of characteristic zero such that $\F_n=\F$ for all $n\in\N$. Then, the following conditional implication holds.
\begin{itemize}
	\item {\bf Assumption:} There is a $\P$-uniform family of arithmetic networks of polynomial size solving $\mathsf{PIT}$ for arithmetic circuits over $\F$. 
	\item {\bf Conclusion:} For every $k\in\N$ and every polynomially bounded shift function $\sft$ there is a $\P$-uniform family of arithmetic networks $\set{F_n}_{n\in\N}$ of polynomial size over $\F^{\sft}$ satisfying the following. For every $n^k$-time-uniform family  of arithmetic networks with $\pit$ gates $\set{C_n}$ of size $n^k$ and degree $n^k$ over $\F^{\sft}$, and every sufficiently large $n\in\N$, and every $\vec{\Lambda}\in \F_{\sft(n)}^n$, we have $C_n(\vec{\Lambda})\ne F_n(\vec{\Lambda})$.
\end{itemize}
\end{theorem}

\bproof
We prove the claim for the case of finite fields $\F=\set{\F_n}$. The proof for fields of characteristic zero is almost identical, except that the function $\sft$ does not matter in that case.

For $k\in\N$ and a shift function $\sft$, let $A=A_{k,\sft}$ be the following Turing machine. On input $1^n$, the machine $A$ simulates each of the first $n$ Turing machines (according to some predetermined efficient enumeration of TMs), each for $n^k$ steps. These machines print $n$ arithmetic networks with $\pit$ gates, each with $n$ input gates and $n$ output gates, denoted by $\set{C^{(i)}\in\F_{\sft(n)}[\vec{x}]}_{i\in[n]}$, where $\vec{x}=(x_1,...,x_n)$. (If a certain machine does not print an arithmetic network with $\pit$ gates over $\F_{\sft(n)}$ that has the appropriate number of gates in time $n^k$, then $A$ discards this machine's output and defines $C^{(i)}$ as a dummy network.) 

Let $\set{P_n}_{n\in\N}$ be the family of networks solving $\mathsf{PIT}$ over $\F$ with shift $\sft'$ such that $\sft'(n^k)=\sft(n)$ (note that $\sft'$ is polynomially bounded). The machine $A$ simulates the machine that prints $\set{P_n}$ on input $1^{n^k}$ to obtain a network $P_{n^k}$ solving $\mathsf{PIT}$ for arithmetic circuits of description size and degree $n^k$ over $\F_{\sft'(n^k)}=\F_{\sft(n)}$. Then $A$ prints a network $F_n$ that gets an $n$-element input $\vec{\Lambda}$ and executes as follows:
\begin{enumerate}
\item The network $F_n$ has the descriptions of $C^{(1)},...,C^{(n)}$ hard-wired, where each $\pit$ gate in each $C^{(i)}$ is replaced with a copy of $P_{n^k}$. Denote the resulting networks by $D^{(1)},...,D^{(n)}$.

\item For each $i\in[n]$, the network $F_n$ computes $D^{(i)}(\vec{\Lambda})$ (using the network in \Cref{prop:universal-arith-network-for-networks}), and letting $\sigma_i$ be the $i^{th}$ output element of this computation, it outputs 
\mm{
	F_n(\vec{\Lambda})=(\sigma_1+1,...,\sigma_{n}+1) \;\;\;.
}
\end{enumerate}

Observe that $F=\set{F_n}$ is a $\P$-uniform family of arithmetic networks of polynomial size. Assume towards a contradiction that there is an $n^k$-time-uniform family of arithmetic networks with $\pit$ gates $\set{B_n}$ of size $n^k$ and degree $n^k$ and infinitely many $n\in\N$ and $\vec{\Lambda}\in\F_{\sft(n)}^n$ such that $B_n(\vec{\Lambda})=F_n(\vec{\Lambda})$. Let $n\in\N$ be sufficiently large such that the Turing machine printing $\set{B_n}$ is one of the first $n$ machines, in which case there is $i\in[n]$ such that $C^{(i)}=B_n$. Fix $\vec{\Lambda}\in\F_{\sft(n)}^n$ such that $B_n(\vec{\Lambda})=F_n(\vec{\Lambda})$. Since $P_{n^k}$ solves $\pit$ correctly over $\F_n$, we have $C^{(i)}(\vec{\Lambda})=D^{(i)}(\vec{\Lambda})$, and hence $B_n(\vec{\Lambda})_i=F_n(\vec{\Lambda})_i=D^{(i)}(\vec{\Lambda})_i+1=B_n(\vec{\Lambda})_i+1$, a contradiction.
\eproof

\subsection{A sufficient assumption for PIT} \label{sec:pit:sufficient}

We now show that hardness for randomized arithmetic networks over all but finitely many inputs suffices to obtain polynomial-time (or nearly so) worst-case $\pit$ algorithms. In~\Cref{sec:pit:sufficient:zero} we show a result for fields with characteristic zero, and in~\Cref{sec:pit:sufficient:small} we show a result for finite fields.

\subsubsection{Fields of characteristic zero} \label{sec:pit:sufficient:zero}

The following result uses the targeted generator from~\Cref{thm:gkss-tarhsg}, which is based on the GKSS generator~\cite{GKSS22}. The description of the proof appeared in~\Cref{sec:tech:gkss}.

\begin{theorem} [hardness for networks with $\pit$ gates on all inputs yields $\pit$ algorithms; the case of characteristic zero] \label{thm:main:zero}
There is a constant $c_0>1$ such that for every sufficiently large constant $k>1$ the following holds. Let $\F$ be a field of characteristic zero.
\begin{itemize}
	\item {\bf Assumption:} There is $c>1$ and a family of $\log^c$-uniform arithmetic circuits $\set{C_n}_{n\in\N}$ of size $n^k$ and degree $n^k$ over $\F$ such that the following holds. For every $\P$-time-uniform family  of arithmetic networks with $\pit$ gates $\set{R_n}_{n\in\N}$ of size $n^{c_0\cdot\sqrt{k}}$ and degree $2^{(\log(n))^{c_0}}$ over $\F$, and every sufficiently large $n\in\N$, and every $\vec{\Lambda}\in\F^n$, it holds that $R_n(\vec{\Lambda})\ne C_n(\vec{\Lambda})$.
	\item {\bf Conclusion:} There is a $\P$-uniform family of arithmetic networks of polynomial size solving $\mathsf{PIT}$ over $\F$. 
\end{itemize}
\end{theorem}

\bproof
The network for $\mathsf{PIT}$ is given a description $\vec{\Lambda}\in\F^n$ of an arithmetic circuit $\F^n\rightarrow\F$ computing a nonzero polynomial of degree at most $n$. Using~\Cref{thm:gkss-tarhsg} with the circuit family $\set{C_n}_{n \in \N}$,  with $d(n)=n$ and $\bar{n}=n$, and with $\eps=1/\sqrt{k}$, the network computes a description of $\calH_{\vec{\Lambda}}$. It then evaluates the composition $\vec{\Lambda} \circ \calH_{\vec{\Lambda}}$ on a grid of side length $n^2\cdot (n^k)^{\eps}+1=\poly(n)$ and dimension $O(1/\eps)=O(1)$ to test if this polynomial is zero or nonzero. The precise field constants used to define this grid do not matter (e.g., we can use $1,2,...,\poly(n)$), so the network can be printed by a polynomial-time Turing machine.

Assume towards a contradiction that for infinitely many arithmetic circuits, each represented by a vector $\vec{\Lambda}$, it holds that $\vec{\Lambda}\circ\calH_{\vec{\Lambda}} = 0$. Then, the $\P$-uniform network with $\pit$ gates $R_n$ from~\Cref{thm:gkss-tarhsg} computes $R_n(\vec{\Lambda})=C_n(\vec{\Lambda})$ in size 
\[
\poly(n^{1/\eps},n^{\eps\cdot k}) \le n^{O(\sqrt{k})} 
\]
and degree 
\[
\poly(n,n^{\eps \cdot k})^{k \cdot \polylog(n)} \le 2^{\polylog(n)}\;\;\;,
\]
a contradiction.
\eproof

\subsubsection{Finite fields} \label{sec:pit:sufficient:small}

The following result uses the targeted generator from~\Cref{thm:ki:tarhsg}, which is based on the KI generator~\cite{KI04}. The description of the proof appeared in~\Cref{sec:tech:ki}.

\begin{theorem} [hardness for randomized networks on all inputs yields $\pit$ algorithms over finite fields] \label{thm:main:finite}
For every $c_0>1$ there is $k>1$ such that for every feasible sequence $\F=\set{\F_n}_{n\in\N}$ of finite fields, where $\F_n$ has order $q(n)$ and satisfies $n^{\omega(1)}\le q(n)\le 2^{n^{c_0}}$, the following conditional implication holds.
\begin{itemize}
	\item {\bf Assumption:} For every polynomially bounded shift function $\sft$, there is $c>1$ and a $\log^c$-uniform family of arithmetic circuits $\set{C_n}_{n\in\N}$ over $\F^{\sft}$ of polynomial size and polynomial degree satisfying the following. For every $\P$-uniform family of randomized arithmetic networks $\set{R_n}_{n\in\N}$ of size $n^k$ over $\F^{\sft}$, and every sufficiently large $n\in\N$, and every $\vec{\Lambda}\in \F_{\sft(n)}^{n}$, it holds that $R_n(\vec{\Lambda})$ does not compute $C_n(\vec{\Lambda})$ with error degree at most $n^k$.
	
	\item {\bf Conclusion:} For every constant $C\ge1$ there is a $\P$-uniform family of arithmetic networks of size $n^{\log^{(C)}(n)}$ solving $\mathsf{PIT}$ over $\F$.\footnote{Recall that $\log^{(i)}(n)$ is the $i^{\text{th}}$ iterated $\log$ function; that is, $\log^{(i)}(n)=\log(\log^{(i-1)}(n))$ for $i>1$ and $\log^{(1)}(n)=\log(n)$.}
\end{itemize}
\end{theorem}

\bproof
Let $k_0>1$ be a sufficiently large universal constant that bounds the degree of all the polynomially-bounded functions appearing in~\Cref{thm:ki:tarhsg}. Let $k=4c_0\cdot k_0$. Let $\set{C_n}$ be the $\log^c$-uniform family over $\F$ of size $s(n)=\poly(n)$ and degree at most $s(n)$, from our hypothesis. 

We first describe the $\mathsf{PIT}$ procedure, and later explain how to implement it as a $\P$-uniform arithmetic network. Recall that to solve $\mathsf{PIT}$, it suffices to construct an algorithm that gets as input an arithmetic circuit computing a nonzero polynomial, and finds an input on which this polynomial does not vanish. 

Fixing any polynomially bounded shift function $\sft$, we are given as input a description $\vec{\Lambda}\in\F_{\ubar{n}}^{n}$ of an arithmetic circuit computing a nonzero polynomial of degree at most $n$ over $\F_{\ubar{n}}$, where $\ubar{n}=\sft(n)$. Let $n_0=n$ and $\vec{\Lambda}_0=\vec{\Lambda}$.

\paragraph{First step: Reducing the input length to be polylogarithmic.}
Let $\set{C_n}$ be the hypothesized family of hard polynomials over $\F^{\sft}$. We use~\Cref{thm:ki:tarhsg} with $\set{C_n}$ and $\F^{\sft}$, and with parameters $m=n_0$ and $d=n_0$ and $\sigma=1$ and $\eps=1/2$. The network $G_n$ is defined over $\F_{\ubar{n}}$, and given $\vec{\Lambda}_0\in\F_{\ubar{n}}^{n_0}$ it computes a description of the circuit $\calG_{\vec{\Lambda}_0}\colon\F_{\ubar{n}}^{\polylog(n_0)}\rightarrow\F_{\ubar{n}}^{n_0}$. Denote
\mm{
\vec{\Lambda}_1 := (\Lambda\circ\calG_{\vec{\Lambda}})\in\F_{\ubar{n}}[x_1,...,x_{\polylog(n)}] \;\;\;.
}

We claim that $\vec{\Lambda}_1$ is nonzero as a polynomial. To see why this is true, assume otherwise. Then, the randomized network $R_{n_0}(\vec{\Lambda}_0)$ computes $C_{n_0}(\vec{\Lambda}_0)$ over $\F_{\ubar{n}}$ with error degree
\mm{
(n_0\cdot d \cdot \log(s))^{k_0} < n_0^{k} \;\;,
}
where the size of $R_{n_0}$ is
\mm{
(n_0 \cdot d \cdot m \cdot \log(s),\log(q))^{k_0} < n_0^{4c_0\cdot k_0} = n_0^k \;\;\;,
}
contradicting the hardness of $\set{C_n}$. 

By~\Cref{thm:ki:tarhsg}, the size of $\vec{\Lambda}_1$ is at most $n_0^{k'}$ for some constant $k'\in\N$ that depends on $\set{C_n}$, and the degree of $\vec{\Lambda}_1$ is also at most $n_0^{k'}$. We define $n_1=n_0^{k'}$ and consider $\vec{\Lambda}_1\in\F_{\ubar{n}}^{n_1}$.

\paragraph{Iterative step: Exponentially reducing the input length.}
For $i\in[C-1]$,  in iteration $i$ the network starts with a description $\vec{\Lambda}_i\in\F_{\ubar{n}}^{n_i}$ of an arithmetic circuit $\F_{n_i}^{\ell_i}\rightarrow\F_{n_i}$ of degree at most $n_i\le\poly(n)$, where $\ell_i=\begin{cases}\polylog(n_0)&i=1\\O(\log(\ell_{i-1}))&i>1\end{cases}$.

Let $\sft_i$ be a polynomially bounded shift function such that $\sft_i(n_i)=\ubar{n}=n_i^{\Omega(1)}$, and let $\set{C_n}$ be the corresponding hypothesized family of hard polynomials over $\F^{\sft}$. We use~\Cref{thm:ki:tarhsg} with $\set{C_n}$, and with $m=\ell_i$ and $\sigma=0$ and $d(n_i)=n_i$ and $\eps=1/r$ where $r$ is such that the size of $C_{n_i}$ is at most $n_i^r$. The network $G_{n_i}$ is defined over $\F_{\ubar{n}}$, and given $\vec{\Lambda}_i\in\F_{\ubar{n}}^{n_i}$ it computes a description of $\calG_{\vec{\Lambda}_i}\colon\F_{\ubar{n}}^{\ell_{i+1}}\rightarrow\F_{\ubar{n}}^{\ell_i}$. Denote
\mm{
\vec{\Lambda}_{i+1} \coloneqq (\Lambda\circ\calG_{\vec{\Lambda}})\in\F_{\ubar{n}}[x_1,...,x_{\ell_{i+1}}] \;\;\;.
}

Again, we claim that $\vec{\Lambda}_{i+1}$ is nonzero as a polynomial. Assuming otherwise, the network $R_{n_i}(\vec{\Lambda}_i)$ computes $C_{n_i}(\vec{\Lambda}_i)$ over $\F_{\ubar{n}}$ with error degree $(n_i\cdot d\cdot \log(s))<n_i^k$ and with size 
\mm{
s^{\eps}\cdot(n_i,d,m^{\mathrm{loglog}(s)},\log(q))^{k_0} = s^{\eps}\cdot n_i^{2c_0\cdot k_0}\cdot 2^{k_0\cdot\log(\ell_i)\cdot\mathrm{loglog}(s)} < n_i^{4c_0\cdot k_0}\;\;\;,
}
a contradiction. By~\Cref{thm:ki:tarhsg}, the size and degree of $\vec{\Lambda}_{i+1}$ are at most $n_i^{k'}$ for some constant $k'\in\N$ that depends on $\set{C_n}$,\footnote{The size and degree also depend on $\eps$, but we defined $\eps=1/r$ where $r$ is a function of the size $s$ of $\set{C_n}$. Thus, the size and degree indeed only depend on the family $\set{C_n}$.} and we define $n_{i+1}=n_i^{k'}$ and consider $\vec{\Lambda}_{i+1}$ as over $\F_{n_{i+1}}$.

\paragraph{Final step.}
After iteration $i=C-1$ we obtain a nonzero $\vec{\Lambda}_{C}\in\F_{\ubar{n}}[x_1,...,x_{\ell_{C}}]$, where $n_{C}=\poly(n)$ and $\ell_{C}=\log^{(C)}(n)$ and the degree of $\vec{\Lambda}_{C}$ is at most $d'=\poly(n)$. We enumerate over the elements in a grid $S^{\ell_{C}}\subseteq\F_{\ubar{n}}^{\ell_{C}}$ of width $|S|=d'+1$, and since $\vec{\Lambda}_{C}$ is nonzero and $\F_{\ubar{n}}$ is of size $\ubar{n}^{\omega(1)}=n_C^{\omega(1)}$, we are guaranteed to find a nonzero element.

\paragraph{A comment about uniformity.} Note that the algorithm considers constantly many shifts functions $\sft,\sft_1,...,\sft_{C-1}$, and that in the proof we assume that the reconstruction procedure $\set{R_n}$ with each of these shift functions $\sft_i$ fails on the corresponding input $\vec{\Lambda}_i$. Since there are only $C=O(1)$ shift functions, this amounts to considering hardness for $C=O(1)$ uniform algorithms (i.e., each algorithm is a pairing of $\set{R_n}$ with a shift function), and we can safely assume that the initial input length $n$ is large enough so that all of these $C$ algorithms fail on all inputs of length $n$ or more.

\paragraph{Implementation as a $\P$-uniform arithmetic network.}
Observe that the iterative procedure above repeats the following step: Given $\vec{\Lambda}_i$ (for $i\in\set{0,...,C-1}$), pad its description, feed the description to the arithmetic network from~\Cref{thm:ki:tarhsg}, which outputs a description of $\vec{\Lambda}_{i+1}$, and continue to the next iteration. After iteration $C$, the procedure evaluates the final circuit on $S^{\ell_C}$.

To see that the procedure is indeed implementable as a $\P$-uniform arithmetic network $\calN=\set{\calN_n}$, observe that there are only constantly many intermediary steps of padding and running arithmetic networks $\set{G_n}$, corresponding to the constantly many shifts $\sft_i$ and hard families $\set{C_n}$. The machines printing all of the $\set{C_n}$'s can be hard-wired into the machine that prints $\calN_n$, and so can the constants bounding the various polynomials (i.e., the sizes of $\set{C_n}$'s and the polynomial size blow-ups from each $\vec{\Lambda}_i$ to $\vec{\Lambda}_{i+1}$). The machine printing $\calN_n$ just adds padding gates and places the various $\set{G_n}$'s in the appropriate places in $\calN_n$.

The last step of $\calN_n$ is evaluating a circuit on $S^{\ell_C}$. Since $\F$ is feasible, the machine printing $\calN_n$ can print $\poly(n)$ constants, and can thus also hard-wire all of the points in $S^{\ell_C}$ into $\calN_n$. Thus, the only missing part is evaluating $\vec{\Lambda}_C$ on a point, which can be done by a $\P$-uniform network using~\Cref{prop:universal-arith-network}.
\eproof

\begin{remark} [degree bounds]
The hardness assumption in~\Cref{thm:main:finite} does not explicitly bound the degree of the randomized networks. (This is akin, for example, to the hardness assumption in~\cite{KI04} when working over finite fields.) Nevertheless, the degree of these networks is quasipolynomially bounded. Specifically, we can either assume hardness against randomized arithmetic networks with $p$\ts{th} root gates\footnote{For a definition and discussion, see the comments after~\Cref{thm:ki:tarhsg} as well as~\Cref{sec:ki:recon:pth}.} of degree quasipolynomial in $n$ and in the field's characteristic, similarly to~\Cref{thm:main:zero}; or assume hardness against randomized arithmetic circuits (without $p$\ts{th} root gates) of degree that is quasipolynomial in both $n$ and $q$. Both statements follow directly from the degree bounds on the reconstruction in~\Cref{thm:ki:tarhsg}.
\end{remark}

\section{Open problems}\label{sec:open}

We believe that our result provide compelling motivation for studying $\pit$ vs lower bounds in the \emph{uniform setting}, and more generally for studying lower bounds for \emph{uniform arithmetic circuits}. Let us therefore suggest several open problems that strike us as important and potentially tractable in this context:

\begin{problem} \label{prob:int:ksquare}
Relax the condition on the hard function in~\Cref{thm:int:main:zero} to allow functions computable in arbitrary polynomial size (i.e., for any $k>1$ allow an upper bound of $n^{c_k}$ where $c_k$ may be arbitrary).
\end{problem}

The reason for the upper bound of $n^{k^2}$ is our use of the generator of \textcite{GKSS22}, and in particular our implementation for the initial step in their reconstruction. Using a different implementation or a different generator (or reconstruction) might be a path forward for tackling~\Cref{prob:int:ksquare}.

\begin{problem} \label{prob:int:rand}
Relax the hardness assumption in~\Cref{thm:int:main} to refer to networks with $\pit$ gates, rather randomized networks with bounded error degree (see~\Cref{thm:main:finite}).
\end{problem}

The main source of trouble that currently prevents assuming hardness only for networks with $\pit$ gates is the use of Kaltofen's~\cite{Kaltofen1989FactorizationOP} classical factoring algorithm.

\begin{problem}
Relax the hardness assumption in~\Cref{thm:int:main:zero,thm:int:main} to refer only to networks of polynomial degree.
\end{problem}

As mentioned in~\Cref{sec:int:results,sec:pit}, our results require hardness for networks of quasipolynomial degree. Technically, this degree stems from repeated applications of a universal arithmetic circuit (\`{a} la~\cite{Raz10}) in our construction.
\begin{problem}
Strengthen the $\pit$ conclusion in~\Cref{thm:int:main} to have a $\pit$ algorithm running in strictly polynomial time.
\end{problem}

Indeed, the super-polynomial running time is reminiscent of running times obtained via bootstrapping techniques as in~\cite{AGS18,KST19}, and the technical reasons for it are also very similar.

\begin{problem}
Continue the study of lower bounds for uniform arithmetic circuits, by proving interesting lower bounds that go beyond what is known for non-uniform arithmetic circuits.
\end{problem}

As mentioned in~\Cref{sec:int}, lower bounds for uniform constant-free arithmetic circuits of polynomial size and bounded depth (against the permanent) are known (see~\cite{KP09}, and also~\cite{JS13,CKK14}). However, in the Boolean setting we have lower bounds for general circuits, with no depth restriction (e.g., Santhanam and Williams~\cite{SW13} proved lower bounds in $\P$ for general uniform circuits of size $n^k$, for any fixed $k\in\N$), raising the possibility that stronger lower bounds for general uniform arithmetic circuits might be provable.

\section*{Acknowledgements}

Part of this work was done while Robert Andrews was at the Institute for Advanced Study and was supported by NSF grant CCF-1900460 and the Erik Ellentuck Endowed Fellowship Fund. Deepanshu Kush is thankful for financial support from an Ontario Graduate Scholarship (OGS-International) award. Roei Tell is supported by the Natural Sciences and Engineering Research Council of Canada (NSERC) Discovery Grant RGPIN-2024-04490.

\begin{appendices}

\section{Uniformization of classical algorithms} \label{sec:uniform algos}

In this section, we prove that various well-known algorithms in algebraic computation can be carried out in a uniform manner. 
We start in \Cref{subsec:data structures}, where we show that uniform arithmetic networks are capable of translating between the standard and universal encodings of arithmetic circuits.
This affords us the flexibility to use either standard encodings or universal encodings of arithmetic circuits when designing algorithms.
In \Cref{subsec:network algos for algebra}, we design uniform arithmetic networks for a variety of standard problems encountered when manipulating arithmetic circuits, such as polynomial interpolation and decomposition into homogeneous parts.
\Cref{subsec:uniform univariate factorization} studies the problem of factoring univariate polynomials over finite fields, showing that the Cantor--Zassenhaus algorithm, as modified by \textcite{vzGathen84}, can be implemented as a uniform arithmetic network.
\Cref{subsec:kaltofen-arith-net} proves that Kaltofen's algorithm to factor arithmetic circuits \cite{Kaltofen1989FactorizationOP} can be carried out by a uniform arithmetic network.
Finally, we prove in \Cref{subsec:uniform depth reduction} that the depth reduction of \textcite{VSBR83} for arithmetic circuits can be applied to a $\log^c$-uniform family of arithmetic circuits in a manner that preserves $\log^c$-uniformity.

The notation of this appendix deviates slightly from the rest of the paper.
Throughout the appendix, we work with a fixed feasible sequence of fields $\F = \set{\F_n}_{n \in \N}$.
Every algorithm in \Cref{sec:uniform algos} receives $1^n$ as part of its input, so we adopt the convention that each algorithm operates over the corresponding field $\F_n$, suppressing the sequence of fields from the statement of the results.
Whenever further assumptions on the field $\F_n$ are necessary, these will be made explicit in the statement of the corresponding result.

Throughout \Cref{subsec:data structures,subsec:network algos for algebra}, we use $\F_n$ to refer to the $n$\ts{th} field in the sequence of fields.
The field $\F_n$ may be a finite field of growing size (and not necessarily of order $n$), or $\F_n$ may be a field of characteristic zero (in which case $\F_n$ may not actually depend on the index $n$).
In \Cref{subsec:uniform univariate factorization,subsec:kaltofen-arith-net}, we restrict our attention to finite fields, so we change notation to $\F_{q(n)}$ to more transparently reflect that the $n$\ts{th} field in the sequence $\F_{q(n)}$ is a finite field of order $q(n)$.

\subsection{Converting between universal and standard encodings} \label{subsec:data structures}

Many of our algorithms are uniform arithmetic networks that themselves operate on descriptions of arithmetic circuits.
It is sometimes more convenient to describe an algorithm that operates on a standard description of an arithmetic circuit, even though the network receives a universal description as input, or vice-versa.
The following lemmas show that uniform arithmetic networks can efficiently translate between the standard and universal encodings of arithmetic circuits.
As a result, we may work with whichever encoding is more convenient in context at the expense of an additive polynomial increase in the complexity of our algorithms.

We start by converting standard encodings of arithmetic circuits (\Cref{def:standard encoding}) to universal encodings (\Cref{def:universal encoding}).

\begin{lem}[Standard to universal]\label{lem:network-algo-std-uni}
    There is a deterministic, polynomial-time Turing machine $M$ that receives as input $(1^n, 1^d, 1^s)$ and outputs an arithmetic network that satisfies the following properties.
    \begin{enumerate}
        \item
            The network receives as input the standard encoding of an $n$-variate size-$s$ arithmetic circuit $C$ that computes a polynomial of degree at most $d$.
        \item
            The network outputs a universal encoding of length $\poly(n,s,d)$ of the circuit $C$.
        \item
            The network has size $\poly(n,s,d)$ and degree $O(1)$.
    \end{enumerate}
\end{lem}

\bproof
    The network first uses \Cref{thm:raz universal circuit} to print a circuit $\Psi(\vec{x}, \vec{y})$ that is universal for $n$-variate, degree-$d$, size-$s$ computation.
    The network then converts $C$ to a circuit $\Phi$ that computes the same polynomial, but whose underlying directed acyclic graph is the same as the universal circuit $\Psi$.
    This is done by applying the argument appearing in the proof of \cite[Proposition 2.8]{Raz10}, noting that all steps in this transformation can be carried out by a uniform arithmetic network of size $\poly(n,s,d)$ and degree $O(1)$.
    As Raz's universal circuit only has constants appearing on edges that feed into sum gates, we propagate any constants in $\Phi$ that feed into a product gate into the sum gates that take this product gate as input.
    This rewiring can likewise be performed by a uniform arithmetic network of size $\poly(n,s,d)$ and constant degree.
    The resulting labeling of the constants in the circuit corresponds to the universal encoding of $C$, which the network then outputs.
\eproof

We now go in the other direction, converting universal encodings of arithmetic circuits into standard encodings.

\begin{lem}[Universal to standard]\label{lem:network-algo-uni-std}
    There is a deterministic, polynomial-time Turing machine $M$ that receives as input $(1^n, 1^d, 1^s)$ and outputs an arithmetic network that satisfies the following properties.
    \begin{enumerate}
        \item
            The network receives as input the universal encoding $\vec{\Lambda}$ of an $n$-variate size-$s$ arithmetic circuit $C$ that computes a polynomial of degree at most $d$.
        \item
            The network outputs a standard encoding of length $\poly(n,s,d)$ of the circuit $C$.
        \item
            The network has size $\poly(n,s,d)$ and degree $1$.
    \end{enumerate}
\end{lem}

\bproof
    The network first invokes \Cref{thm:raz universal circuit} to construct a circuit $\Psi(\vec{x}, \vec{y})$ that is universal for $n$-variate, degree-$d$, size-$s$ computation.
    Because $\vec{\Lambda}$ is a universal encoding of $C$, we have that $C(\vec{x}) = \Psi(\vec{x}, \vec{\Lambda})$ as polynomials.
    Thus, to output a standard description of $C$, the network replaces the $\vec{y}$ inputs to $\Psi$ with the given vector $\vec{\Lambda}$ and then outputs a standard description of the resulting circuit.

    This network is the combination of \Cref{thm:raz universal circuit} and a simple modification of the inputs to the universal circuit $\Psi$.
    It is clear that this network has size $\poly(n,s,d)$.
    Because this network does not do any arithmetic computation on the input $\vec{\Lambda}$, it has degree $1$ as claimed.

    To obtain a polynomial-time Turing machine that prints this network, the part of the network that builds the circuit $\Psi$ can be printed by running the Turing machine of \Cref{thm:raz universal circuit}.
    The remaining part of the network that replaces the $\vec{y}$ variables in the input to $\Psi$ with the given vector $\vec{\Lambda}$ can clearly be printed by a polynomial-time Turing machine.
\eproof

\subsection{Algorithms that manipulate arithmetic circuits} \label{subsec:network algos for algebra}

This subsection describes uniform families of arithmetic networks (sometimes with PIT gates) that perform a variety of tasks on arithmetic circuits, such as finding a variable a circuit depends on, homogenizing a circuit, and computing the greatest common divisor of two circuits.
None of these algorithms are new.
Rather, the goal of this subsection is to verify that each of these algorithms can be implemented by a uniform family of arithmetic networks, and to provide bounds on the complexity parameters of the resulting networks.

We start by describing a uniform arithmetic network (with PIT gates) that receives as input an arithmetic circuit $C(x_1,\ldots,x_n)$ and finds a variable $x_i$ that $C$ depends on.
Of course, if the circuit $C$ computes a constant polynomial, then $C$ depends on no variables, in which case the arithmetic network correctly reports that $C$ is a constant polynomial.

\begin{lem}[Finding a variable a polynomial depends on] \label{lem:network algo find variable}
    There is a deterministic, polynomial-time Turing machine $M$ that receives as input $(1^n, 1^d, 1^s)$ and outputs an arithmetic network with PIT gates that satisifes the following properties.
    \begin{enumerate}
        \item
            The network receives as input the standard description of an arithmetic circuit $C$ of size $s$ over $\F_n$ computing a polynomial $f \in \F_n[x_1,\ldots,x_n]$ of degree at most $d$.
        \item
            The network outputs data $(b, \vec{a}) \in \zo \times \zo^{\log_2 n}$ satisfying the following.
            \begin{enumerate}
                \item 
                    The bit $b$ is set to $1$ if and only if $f$ is a constant polynomial.
                \item
                    If $b = 0$ (and hence $f$ is nonconstant), then $\vec{a}$ is the binary encoding of an integer $i \in [n]$ such that $f$ depends on the variable $x_i$.
            \end{enumerate}
        \item 
            The network has size $\poly(n,s,d)$ and degree $O(1)$.
    \end{enumerate}
\end{lem}

\bproof
    For each $i \in [n]$, the network iteratively checks if $f$ depends on the variable $x_i$.
    To do this, the network computes a description of a circuit that computes the polynomial
    \[
        f_i(\vec{x}, y) \coloneqq f(x_1, \ldots, x_{i-1}, y, x_{i+1}, \ldots, x_n),
    \]
    where $y$ is a fresh variable.
    A description of a circuit for $f_i$ can be obtained easily from a circuit that computes $f$ by replacing the variable $x_i$ with $y$ at the input gates.
    The network then uses its PIT gates to test if $f(\vec{x}) - f_i(\vec{x}, y) = 0$.
    If the PIT gate reports that $f(\vec{x}) - f_i(\vec{x}, y) \neq 0$, then $f$ depends on $x_i$, so the network outputs $(0, \bin(i))$, where $\bin(i) \in \zo^{\log_2 n}$ is the binary representation of $i$.
    Otherwise, the network proceeds to check if $f$ depends on the next variable $x_{i+1}$.
    If $f$ depends on none of the variables, then the network outputs $(1, 0^{\log_2 n})$ to indicate that $f$ is a constant polynomial.

    To bound the size of the resulting network, we observe that the network makes simple modifications to the given circuit for $f$ and then performs PIT on the modified circuits.
    For each $i \in [n]$, replacing the variable $x_i$ with $y$ in the input to the circuit for $f$ can be done by a network of size $O(s)$.
    This produces a standard description of the circuit for $f - f_i$, again using a network of size $O(s)$.
    We then convert this to a universal description using \Cref{lem:network-algo-std-uni}, incurring a $\poly(n,s,d)$ increase in the size of our network, and then feed this universal description into a PIT gate.
    The post-processing of the PIT gate's output can be done by a network of size $O(1)$.
    We perform this sequence of operations $n$ times, so the overall size of our network is bounded by $\poly(n,s,d)$.
    Each iteration can be implemented by a network of degree $O(1)$, so the degree of the overall network is likewise bounded by $O(1)$.

    Finally, the network can be printed by a polynomial-time Turing machine as follows.
    The components of the network that modify the circuit for $f$ into a circuit for $f - f_i$, as well as the components that post-process the outcome of the PIT operation, can clearly be printed in polynomial time.
    The part of the network that translates from standard encodings to universal encodings of arithmetic circuits can be printed by repeatedly running the Turing machine that prints the network of \Cref{lem:network-algo-std-uni}.
\eproof

Our next tool is polynomial interpolation.
Given a multivariate polynomial $f(x_1, \ldots, x_n, y)$, we can always regard $f$ as a univariate polynomial in $y$ whose coefficients are themselves polynomials in $x_1, \ldots, x_n$.
That is, there are polynomials $f_0, f_1, \ldots, f_d \in \F[\vec{x}]$ such that
\[
    f(\vec{x}, y) = \sum_{i=0}^d f_i(\vec{x}) y^i,
\]
where $d$ is the degree of $f$.
From a circuit that computes $f$, we can build a circuit of comparable size that computes these coefficients $f_0, \ldots, f_d$.
The following lemma verifies that this transformation can be carried out by a uniform family of arithmetic networks.
To avoid any constraints on the size of the underlying field $\F$, we carry out polynomial interpolation via a gate simulation argument, rather than by partially evaluating the circuit and then interpolating from these evaluations.
If the circuit computing $f$ has size $s$, the gate simulation results in circuits of size $O(s d^2)$ for the coefficients $f_i$.
This loses a factor of $d$ over the evaluation and interpolation approach, which produces circuits of size $O(sd)$ for the coefficients $f_i$.
This extra factor of $d$ will not meaningfully impact any of our applications.

\begin{lem}[Polynomial interpolation] \label{lem:network algo interpolation}
    There is a deterministic, polynomial-time Turing machine $M$ that receives as input $(1^n, 1^d, 1^s)$ and outputs an arithmetic network satisfying the following properties.
    \begin{enumerate}
        \item
            The network receives as input the description of an arithmetic circuit $C$ of size $s$ over $\F_n$ computing a polynomial $f \in \F_n[x_1,\ldots,x_n]$ of degree at most $d$, and the binary encoding of an integer $i \in [n]$.
        \item
            The network outputs the description a $(d+1)$-output arithmetic circuit $C'$ of size $O(s d^2)$ that does not depend on the variable $x_i$.
            Let $f_j(x_1, \ldots, x_{i-1}, x_{i+1},\ldots,x_n)$ be the polynomial computed by the $j$\ts{th} output gate of $C'$.
            Then the polynomial identity 
            \[
                f(x_1,\ldots,x_n) = \sum_{j=0}^d f_j(x_1,\ldots,x_{i-1},x_{i+1},\ldots,x_n) x_i^j
            \]
            holds.
        \item 
            The network has size $\poly(n,s,d)$ and degree $O(1)$.
    \end{enumerate}
\end{lem}

\bproof
    We use gate simulation to convert the given circuit computing $f$ into a circuit that computes the coefficient polynomials $f_0, f_1, \ldots, f_d$.
    We first describe the structure of the circuit that computes the coefficient polynomials, and then explain how this transformation can be carried out by a uniform family of arithmetic networks.

    Each gate $v$ of the input circuit $C$ will be split into $d+1$ gates $(v, 0), (v,1), \ldots, (v,d)$.
    If $f_v(\vec{x})$ is the polynomial computed by the gate $v$ in $C$, then the gate $(v, j)$ is intended to compute a polynomial $f_{v,j}(\vec{x})$ such that
    \[
        f_v(\vec{x}) = \sum_{j=0}^d f_{v,j}(\vec{x}) x_i^j
    \]
    and none of the polynomials $f_{v,j}$ depend on the variable $x_i$.
    The gates $(v,j)$ are wired depending on the type of the gate $v$ in the circuit $C$.
    \begin{itemize}
        \item
            If $v$ is an input gate in $C$ labeled by a field element $\alpha$ or a variable $x_k$ where $k \neq i$, then we set $(v, 0)$ to be an input gate labeled by $\alpha$ or $x_k$, respectively.
            For $j \ge 1$, we set $(v, j)$ to be an input gate labeled by zero.
        \item
            If $v$ is an input gate in $C$ labeled by the variable $x_i$, then we set $(v, 1)$ to be an input gate labeled by $1$.
            For $j \neq 1$, we set $(v, j)$ to be an input gate labeled by zero.
        \item
            If $v = u + w$ is an addition gate with children $u$ and $w$, then for each $j \in \set{0, 1, \ldots, d}$, we set $(v,j)$ to be an addition gate with children $(u,j)$ and $(w,j)$.
        \item
            If $v = u \times w$ is a multiplication gate with children $u$ and $w$, then for each $j \in \set{0, 1, \ldots, d}$, we set $(v,j)$ to be the output of the subcircuit given by
            \[
                (v,j) = \sum_{k=0}^j (u, k) \times (w, j-k).
            \]
    \end{itemize}
    Denote this new circuit by $C'$.
    Using induction, one can prove that each gate $(v,j)$ of $C'$ correctly computes the polynomial $f_{v,j}(\vec{x})$.
    Each gate of $C$ is copied into $d+1$ gates in $C'$, and each of these gates is implemented by a subcircuit of size at most $O(d)$, so the total number of gates in $C'$ is bounded by $O(s d^2)$.
    Furthermore, it follows from the definition of $C'$ that $x_i$ does not appear on any of the input gates of $C'$, so no output of $C'$ depends on the variable $x_i$.
    This establishes the existence of the claimed circuits that compute the coefficients $f(\vec{x})$ as a polynomial in $x_i$.

    It remains to prove that this transformation can be carried out by a uniform arithmetic network with the claimed bounds on size and degree.
    The circuit $C'$ is obtained by modifying the graph underlying the circuit $C$.
    Thanks to \Cref{lem:network-algo-uni-std}, we may assume without loss of generality that we are given a standard description of $C$.
    In this case, the only modifications to $C$ are done on the boolean part of the circuit's description.
    We can obtain a polynomial-time Turing machine that modifies the gates and wiring of $C$ into that of $C'$ by directly implementing the transformation as described above.
    This implies there is a uniform family of polynomial-size boolean circuits that take as input the boolean part of $C$ and output the boolean part of $C'$.
    These circuits are likewise arithmetic networks that operate on the boolean part of the circuit description, so we obtain a uniform family of arithmetic networks of size $\poly(n,s,d)$ and degree $O(1)$ that computes a description of $C'$ from the description of $C$ as desired.
\eproof

An easy application of polynomial interpolation allows us to take partial derivatives of arithmetic circuits with respect to one variable.

\begin{lem}[Partial derivatives] \label{lem:network algo partial derivative}
    There is a deterministic, polynomial-time Turing machine that receives as input $(1^n, 1^d, 1^s)$ and outputs an arithmetic network satisfying the following.
    \begin{enumerate}
        \item
            The network receives as input the description of an arithmetic circuit $C$ of size $s$ over $\F_n$ computing a polynomial $f \in \F_n[x_1,\ldots,x_n]$ of degree at most $d$, and the binary encoding of two integers $i \in [n]$ and $k \in [d]$.
        \item
            The network outputs a description of an arithmetic circuit $C'$ of size $O(s d^2)$.
            The arithmetic circuit $C'$ computes the polynomial
            \[
                \frac{\partial^k}{\partial x_i^k}(f(\vec{x})).
            \]
        \item 
            The network has size $\poly(n,s,d)$ and degree $O(1)$.
    \end{enumerate}
\end{lem}

\bproof
    Write $f(\vec{x})$ as a polynomial in $x_i$ as
    \[
        f(\vec{x}) = \sum_{j=0}^d f_j(x_1,\ldots,x_{i-1},x_{i+1},\ldots,x_n) x_i^j.
    \]
    Then the derivative $\frac{\partial^k}{\partial x_i^k}(f(\vec{x}))$ is given by
    \[
        \frac{\partial^k}{\partial x_i^k} = \sum_{j=k}^d j (j-1) \cdots (j-k+1) \cdot f_j(x_1,\ldots,x_{i-1},x_{i+1},\ldots,x_n) x_i^{j-k}.
    \]
    \Cref{lem:network algo interpolation} shows that there is a multi-output arithmetic circuit of size $O(s d^2)$ that computes all the $f_j$ above, and that this circuit can be computed by a uniform family of arithmetic networks.
    By taking an appropriate linear combination of the $f_j$ as above, we obtain an arithmetic circuit of size $O(s d^2)$ that computes $\frac{\partial^k}{\partial x_i^k}(f(\vec{x}))$ as desired.

    Given a description of the circuit computing the polynomials $f_j$, it is straightforward to add an extra addition gate on top and print the values of the falling factorials $j (j-1) \cdots (j - k + 1)$.
    This shows that the description of the circuit for $\frac{\partial^k}{\partial x_i^k}(f(\vec{x}))$ can be computed by a uniform family of arithmetic networks as desired.
\eproof

Next on our list of circuit transformations is homogenization.
Every multivariate polynomial $f(\vec{x})$ of degree at most $d$ can be written as a sum $f = f_0 + f_1 + \cdots + f_d$, where each $f_i$ is a polynomial consisting only of monomials of degree exactly $i$.
Such polynomials are called \emphdef{homogeneous}, and the polynomial $f_0, f_1, \ldots, f_d$ are called the \emphdef{homogeneous components} of $f$.
A gate simulation argument, very similar to the one used for interpolation, allows us to efficiently extract homogeneous components from arithmetic circuits.
We sometimes apply this transformation in a manner where we only take homogeneous components with respect to a subset of the variables, pushing the others into the ground field.
To accommodate this generality, we provide a list of variables as part of the input to the arithmetic network that performs homogenization, with the intended meaning that the network homogenizes the given circuit only with respect to these specified variables.
By homogenizing with respect to a single variable, we essentially recover polynomial interpolation (\Cref{lem:network algo interpolation}) as a special case of homogenization.

\begin{lem}[Decomposition into homogeneous parts] \label{lem:network algo homogeneous components}
    There is a deterministic, polynomial-time Turing machine $M$ that receives as input $(1^n, 1^d, 1^s)$ and outputs an arithmetic network satisfying the following properties.
    \begin{enumerate}
        \item
            The network receives as input the description of an arithmetic circuit $C$ of size $s$ over $\F_n$ computing a polynomial $f \in \F_n[x_1,\ldots,x_n]$ of degree at most $d$, as well as a set $S \subseteq [n]$ encoded as a bit vector in $\zo^n$.
        \item
            The network outputs the description of a $(d+1)$-output arithmetic circuit $C'$ of size bounded by $O(s d^2)$.
            Let $f_j(\vec{x})$ be the polynomial computed by the $j$\ts{th} output gate of $C'$.
            Then $f_j$ is a homogeneous polynomial of degree $j$ when viewed as a polynomial in the variables indexed by $S$, where the remaining variables are regarded as elements of the underlying field.
            Further, the polynomial identity 
            \[
                f(x_1,\ldots,x_n) = \sum_{j=0}^d f_j(\vec{x})
            \]
            holds.
            That is, the polynomial $f_j$ is the degree-$j$ homogeneous component of $f$ with respect to the variables indexed by $S$.
        \item 
            The network has size $\poly(n,s,d)$ and degree $O(1)$.
    \end{enumerate}
\end{lem}

\bproof
    To obtain a circuit that computes the homogeneous components of $f$, we apply a gate simulation argument similar to the one used in the proof of \Cref{lem:network algo interpolation}.
    We first describe the gate simulation argument used to obtain circuits for the homogeneous components of $f$, and then argue that there is a uniform family of arithmetic networks that implements this simulation.
    For the sake of completeness, we include the full details of the gate simulation.

    Each gate $v$ of the given circuit $C$ will be split into $d+1$ gates $(v, 0), (v,1), \ldots, (v, d+1)$.
    Letting $f_v(\vec{x})$ be the polynomial computed by gate $v$, the gate $(v, j)$ is intended to the degree-$j$ homogeneous component of $f_v$ with respect to the variables in $S$.
    The wiring of the gates $(v,j)$ depends on the label of the original gate $v$ in $C$.
    \begin{enumerate}
        \item 
            If $v$ is an input gate labeled by a field constant $\alpha$ or a variable $x_k$ where $k \notin S$, we set $(v, 0)$ to be an input gate labeled by $\alpha$ or $x_k$, respectively.
            For $j \ge 1$, we set $(v,j)$ to be an input gate labeled by zero.
        \item
            If $v$ is an input gate labeled by a variable $x_k$ where $k \in S$, then we set $(v,1)$ to be an input gate labeled by $x_k$.
            For $j \neq 1$, we set $(v, j)$ to be an input gate labeled by zero.
        \item
            If $v = u + w$ is an addition gate with children $u$ and $w$, then for each $j \in \set{0, 1, \ldots, d}$, we set $(v,j)$ to be an addition gate with children $(u,j)$ and $(w,j)$.
        \item
            If $v = u \times w$ is a product gate with children $u$ and $w$, then for each $j \in \set{0, 1, \ldots, d}$, we set $(v, j)$ to be the output of the subcircuit given by
            \[
                (v, j) = \sum_{k=0}^j (u, k) \times (w, j - k).
            \]
    \end{enumerate}
    Let $C'$ be the new circuit obtained by applying this gate simulation to $C$.
    By induction, one can show that $(v, j)$ computes the degree-$j$ homogeneous component of $f_v$ with respect to the variables indexed by $S$.
    Every gate of $C$ is copied into $d+1$ gates in $C'$, each of which can be implemented by a subcircuit of size $O(d)$, so there are at most $O(s d^2)$ gates in $C'$.
    This proves the existence of the claimed circuits for the homogeneous components of $f$ with respect to the variables indexed by $S$.

    To show that this transformation can be carried by a uniform family of arithmetic networks, we use the same argument as in the proof of \Cref{lem:network algo interpolation}.
    By directly implementing the description of $C'$ above, we obtain a polynomial-time Turing machine that computes a description of the boolean part of $C'$ from a description of the boolean part of $C$.
    This implies the existence of a uniform family of boolean circuits---and hence a uniform family of arithmetic networks---that carries out the same computation.
    Using \Cref{lem:network-algo-uni-std}, we may assume that we are given a standard description of $C$ as input, so we obtain the desired family of uniform arithmetic networks of size $\poly(n,s,d)$ and degree $O(1)$.
\eproof

Homogenization of arithmetic circuits is a useful subroutine in the reconstruction algorithm for our targeted hitting set generator based on the GKSS generator.
There, we sometimes find ourselves in the situation where the circuit $C$ to be homogenized is already known to be partially homogeneous.
That is, we know the circuit $C$ can be written as a composition $C = C_1 \circ C_2$ of two arithmetic circuits, where the circuit $C_2$ is already homogeneous.
With this extra information, we can homogenize $C$ in a more efficient manner than what \Cref{lem:network algo homogeneous components} provides.
Letting $s_1$ and $s_2$ denote the sizes of $C_1$ and $C_2$, respectively, \Cref{lem:network algo homogeneous components} constructs an equivalent homogeneous circuit of size $O((s_1 + s_2) d^2)$.
The following lemma uses the fact that $C_2$ is already homogeneous to improve the size of the resulting homogeneous circuit to $s_2 + O(s_1 d^2)$, which is important when bounding the complexity of the reconstruction procedure in \Cref{sec:gkss-reconstruction}.

\begin{lem}[{Partial homogenization \cite[Lemma 2.10]{GKSS22}}] \label{lem:network algo partial hom}
    There is a deterministic, polynomial-time Turing machine $M$ that receives as input $(1^n, 1^m, 1^d, 1^s, 1^t)$ and outputs an arithmetic network satisfying the following properties.
    \begin{enumerate}
        \item
            The network receives as input the description of a multi-output homogeneous arithmetic circuit $C$ of size $s$ over $\F_n$ with $m$ outputs computing polynomials $f_1,\ldots,f_m \in \F_n[x_1,\ldots,x_n]$, and the description of an $m$-input multi-output circuit $C'$ of size $t$.
        \item
            If $\deg(C) \cdot \deg(C') \le d$, then the network outputs the description of a homogeneous multi-output circuit $D$ of size $s + O(t d^2)$ that computes the homogeneous components of degree at most $d$ of the outputs of $C' \circ C$.
            If $\deg(C) \cdot \deg(C') > d$, the output of the network is meaningless.
        \item 
            The network has size $\poly(n,m,s,t,d)$ and degree $O(1)$.
    \end{enumerate}
\end{lem}

\bproof
    As in the proof of \cite[Lemma 2.10]{GKSS22}, we apply the standard gate simulation used to homogenize circuits (as done in \Cref{lem:network algo homogeneous components}), but we only simulate the gates appearing in the not-necessarily-homogeneous circuit $C'$.
    Like our previous applications of gate simulation, we first describe the structure of the new circuit, and then explain how this new circuit can be computed by a uniform family of arithmetic networks.

    We first create a new circuit $C''$ from $C'$ by applying gate simulation.
    For each gate $v$ of $C'$, we create $d+1$ new gates $(v, 0), (v, 1), \ldots, (v, d)$ in $C''$.
    Similar to prior simulations, if $f_v$ denotes the polynomial computed by gate $v$ in the composed circuit $C' \circ C$, then the gate $(v, j)$ is intended to compute the degree-$j$ homogeneous component of $f_v$ when viewed as part of the composed circuit $C'' \circ C$.
    We wire the gates $(v, j)$ as follows.
    \begin{enumerate}
        \item
            If $v$ is an input gate labeled by a field constant $\alpha$, we set $(v, 0)$ to be an input gate labeled by $\alpha$.
            For all $j \ge 1$, we set $(v, j)$ to be an input gate labeled by zero.
        \item
            If $v$ is an input gate labeled by the variable $y_i$, then we set $(v, \deg(f_i))$ to be an input gate labeled by $y_i$, where $f_i$ is the polynomial computed by the $i$\ts{th} output gate of the circuit $C$.
            For $j \neq \deg(f_i)$, we set $(v, j)$ to be an input gate labeled by zero.
        \item 
            If $v = u + w$ is an addition gate with children $u$ and $w$, then we set $(v, j) = (u, j) + (w, j)$ for each $j \in \set{0,1,\ldots,d}$.
        \item
            If $v = u \times w$ is a product gate with children $u$ and $w$, then for each $j \in \set{0,1,\ldots,d}$, we set $(v, j)$ to be the output of the subcircuit
            \[
                (v, j) = \sum_{k=0}^j (u, k) \times (w, j - k).
            \]
    \end{enumerate}
    Let $C''$ be the circuit obtained by applying this simulation to the gates of $C'$.
    It is clear that the size of $C''$ is bounded by $O(t d^2)$.
    Using induction, one can prove that the circuit $C'' \circ C$ is homogeneous, and that the outputs of $C'' \circ C$ correspond to the homogeneous components of the outputs of $C' \circ C$.
    Thus, the circuit $D \coloneqq C'' \circ C$ is a homogeneous circuit of size $s + O(t d^2)$ and computes the homogeneous components of the outputs of $C' \circ C$.
    This establishes the existence of a circuit with the desired behavior and complexity bounds.

    It remains to describe a uniform family of arithmetic networks that receive descriptions of $C$ and $C'$ as input and output a description of $D$.
    To do this, it suffices to describe a uniform family of arithmetic networks that compute a description of the circuit $C''$ that simulates $C'$.
    If we are given the degrees of the output polynomials of $f_1, \ldots, f_m$ of the circuit $C$, then the construction of $C''$ can be carried out by a uniform family of arithmetic networks in the same manner as in theproof of \Cref{lem:network algo homogeneous components}.
    Thus, our task amounts to computing the degrees of the polynomials computed by the circuit $C$ in a uniform manner.

    Recall that the circuit $C$ is homogeneous.
    This allows us to compute the degree of each gate of $C$ via dynamic programming, following a topological ordering of the gates of $C$.
    Starting at the input layer, each input gate of $C$ computes a polynomial either of degree zero or degree one, depending on whether the gate is labeled with a field element or a variable.
    In either case, we can easily determine the degree of the polynomial computed at the corresponding gate.
    If $v = u + w$ is an addition gate with children $u$ and $w$, then the fact that $C$ is homogeneous implies that $\deg(f_v) = \deg(f_u) = \deg(f_w)$, so from the degree of $u$ or $w$ we can immediately compute the degree of $v$.
    If $v = u \times w$ is a product gate, then we have $\deg(f_v) = \deg(f_u) + \deg(f_w)$, so we can efficiently compute the degree of $v$ from the degrees of $u$ and $w$.
    
    The preceding algorithm yields a polynomial-time Turing machine that takes the description of $C$ as input and computes the degrees of each gate of $C$.
    This, in turn, produces a uniform family of boolean circuits---and hence a uniform family of arithmetic networks---that take the description of $C$ as input and output the degrees of each gate of $C$.
\eproof

We now turn our attention away from syntactic operations on arithmetic circuits and towards operations of a more algebraic nature, which will be particularly useful in our description of Kaltofen's algorithm to factor arithmetic circuits.
We start by showing that uniform families of arithmetic networks with $\pit$ gates can solve linear systems, even when the coefficients of the linear system are multivariate polynomials encoded as arithmetic circuits.
This result is a natural combination of two algorithms. 
The first is a simple algorithm, due to \textcite{BvzGH82}, that computes a basis for the nullspace of a matrix by finding a basis of the column space and then expressing the remaining columns of the matrix in terms of this basis.
The second is a natural extension of this algorithm, appearing in the work of \textcite{KSS15}, that computes a single nonzero element of the nullspace of a matrix in the setting where the entries of the matrix are themselves multivariate polynomials encoded by arithmetic circuits.

\begin{lem}[{Solving linear systems \cite[Theorem 5]{BvzGH82}, \cite[Lemma 2.6]{KSS15}}] \label{lem:network algo linear systems}
    There is a deterministic, polynomial-time Turing machine $M$ that receives as input $(1^n, 1^d, 1^k, 1^\ell, 1^s)$ and outputs an arithmetic network with $\pit$ gates satisfying the following properties.
    \begin{enumerate}
        \item
            The network receives as input the description of a multi-output arithmetic circuit $C$ of size $s$ that computes $k\ell$ polynomials $f_{1,1}, \ldots, f_{k,\ell} \in \F_n[x_1,\ldots,x_n]$, each of degree at most $d$.
        \item
            Let $M \in \F_n[\vec{x}]^{k \times \ell}$ be the $k \times \ell$ matrix given by $M_{i,j} = f_{i,j}$.
            The network outputs an integer $m \in \set{0,1,\ldots,\ell}$, encoded in binary, that corresponds to the dimension of the space of solutions to the linear system $M \vec{v} = \vec{0}$.
            In addition, the network outputs a description of a multi-output arithmetic circuit $D$ of size $\poly(n,k,\ell,s,d)$ that computes $m \ell$ polynomials $g_{1,1}, \ldots, g_{m,\ell} \in \F_n[\vec{x}]$ such that the vectors $\set{(g_{i,1}, \ldots, g_{i,n}) : i \in [m]} \subseteq \F_n[\vec{x}]^\ell$ are a basis for the solution space of $M \vec{v} = \vec{0}$.
        \item 
            The network has size $\poly(n,k,\ell,s,d)$ and degree $O(1)$.
    \end{enumerate}
\end{lem}

\bproof
    As described prior to the statement of the lemma, we combine the algorithm of \textcite{BvzGH82} with the observation of \textcite{KSS15} that when the entries of the matrix are specified by arithmetic circuits, the necessary polynomial arithmetic can be carried out using $\pit$ gates.
    We first describe the algorithm, and then verify that the algorithm can be implemented as a uniform family of arithmetic networks with $\pit$ gates.

    We first find a maximal nonsingular submatrix of $M$.
    This is done iteratively: for each $i \in [\ell]$, suppose we have found a set of $i-1$ rows $R \subseteq [k]$ and a set of $i-1$ columns $C \subseteq [\ell]$ such that the submatrix $M_{R,C}$ is nonsingular.
    For each $r \in [k] \setminus R$ and $c \in [\ell] \setminus C$, we check if the submatrix $M_{R \cup \set{r}, C \cup \set{c}}$ is nonsingular by computing $\det(M_{R \cup \set{r}, C \cup \set{c}})$.
    If $\det(M_{R \cup \set{r}, C \cup \set{c}}) \neq 0$, we have found an invertible $i \times i$ submatrix, so we update $R \gets R \cup \set{r}$ and $C \gets C \cup \set{c}$ and proceed to the next iteration.
    Otherwise, we have found a maximal nonsingular submatrix, which we will use to construct a basis for the nullspace of $M$ in the next step.

    Fix subsets $R \subseteq [k]$ and $C \subseteq [\ell]$ with $|R| = |C|$ such that $M_{R, C}$ is a maximal nonsingular submatrix.
    If $C = [\ell]$, then the matrix $M$ has no nonzero vectors in its kernel, so our algorithm reports that the system $M \vec{v} = \vec{0}$ has a solution space of dimension zero.
    Otherwise, for each $j \in [\ell] \setminus C$, the $j$\ts{th} column of $M$ is in the span of the columns indexed by $C$, and the resulting linear relation between the columns of $M$ corresponds to an element of the nullspace of $M$.
    In more detail, let $\vec{v}^{(j)} \in \F_n[\vec{x}]^\ell$ be the vector 
    \[
        \vec{v}^{(j)}_i \coloneqq 
        \begin{cases}
            \det(M_{R,C\setminus\set{i}\cup\set{j}}) & \text{if $i \in C$,} \\
            -\det(M_{R,C}) & \text{if $i = j$,} \\
            0 & \text{otherwise,}
        \end{cases}
    \]
    where $M_{R, C \setminus \set{i} \cup \set{j}}$ is the matrix obtained by replacing the $i$\ts{th} column of $M_{R,C}$ with $M_{R,j}$.
    Cramer's rule implies that $M \vec{v}^{(j)} = \vec{0}$.
    Because only the vector $\vec{v}^{(j)}$ has a nonzero entry in the $j$\ts{th} coordinate, the vectors $\set{\vec{v}^{(j)} : j \in [\ell] \setminus C}$ are linearly independent.
    Since the nullity of $M$ is $\ell - |C|$ and we have found $\ell - |C|$ linearly independent vectors in the nullspace of $M$, these vectors form a basis for the nullspace of $M$.

    It remains to describe an implementation of this algorithm as a uniform family of arithmetic networks with $\pit$ gates.
    To find a maximal nonsingular submatrix, the network tries all $k \ell$ possible choices of which row and column to add to the current submatrix, selecting the lexicographically-first choice of row $r$ and column $c$ where $M_{R \cup \set{r}, C \cup \set{c}}$ is nonsingular, if such a choice exists.
    To test if $M_{R \cup \set{r}, C \cup \set{c}}$ is nonsingular, the network computes $\det(M_{R \cup \set{r}, C \cup \set{c}})$ and checks if this determinant is nonzero.
    The entries of $M$ are multivariate polynomials specified by arithmetic circuits, so to check if this determinant is nonzero, the networks passes a description of a circuit computing $\det(M_{R \cup \set{r}, C \cup \set{c}})$ to a $\pit$ gate.
    We compute a description of a circuit computing $\det(M_{R \cup \set{r}, C \cup \set{c}})$ by composing the circuits computing the entries of $M_{R \cup \set{r}, C \cup \set{c}}$ with the uniform arithmetic circuit for the determinant constructed by \textcite{Berkowitz84}.

    Once we have found a maximal nonsingular submatrix $M_{R, C}$, we report that $M \vec{v} = \vec{0}$ has a solution space of dimension $m = \ell - |C|$.
    In case that $m > 0$, a description of a multi-output circuit computing a basis for the nullspace of $M$ can be obtained from the explicit description of the vectors $\vec{v}^{(j)}$ above.
    This is again done by composing the uniform arithmetic circuits of \textcite{Berkowitz84} for the determinant with the circuits computing the relevant entries of $M$.

    It is clear from the above description that the resulting arithmetic network has size bounded by $\poly(n,k,\ell,s,d)$.
    Since the network only manipulates descriptions of arithmetic circuits and does no arithmetic computation itself, its degree is bounded by $O(1)$.
    Finally, it is evident from the uniform description of this arithmetic network that it can be printed by a polynomial-time Turing machine.
\eproof

Next, we consider the problem of divisibility testing: we are given two polynomials $f(\vec{x})$ and $g(\vec{x})$, represented as arithmetic circuits, and we want to test if $f$ divides $g$.
There is a simple reduction from divisibility testing to $\pit$ due to \textcite{Forbes15}, based on Strassen's method of eliminating divisions in arithmetic circuits \cite{Strassen73}.
As the following lemma shows, this reduction (and hence divisibility testing itself) can be implemented as a uniform family of arithmetic networks with $\pit$ gates.

\begin{lem}[{\cite[Section 7]{Forbes15}}] \label{lem:network algo divisibility}
    There is a deterministic, polynomial-time Turing machine that receives as input $(1^n, 1^d, 1^s)$ and outputs an arithmetic network with $\pit$ gates that satisfies the following properties.
    \begin{enumerate}
        \item 
            The network receives as input the descriptions of two arithmetic circuits $C_1$ and $C_2$ of size $s$ over $\F_n$, computing polynomials $f, g \in \F_n[x_1, \ldots, x_n]$ respectvely, both of degree at most $d$.
        \item
            If $|\F_n| > d$, the network outputs a bit indicating whether $f$ divides $g$.
            Otherwise, the output of the network is meaningless.
        \item
            The network has size $\poly(n,s,d)$ and degree $O(1)$.
    \end{enumerate}
\end{lem}

\bproof
    We follow Forbes's \cite{Forbes15} reduction of divisibility testing to $\pit$.
    We first present this reduction, and then explain how to implement it as a uniform family of arithmetic networks with $\pit$ gates.

    We begin by handling the special case where either $f$ or $g$ is zero.
    If $g$ is identically zero, then every polynomial divides $g$, so we report that $f$ divides $g$.
    If $g$ is nonzero and $f$ is zero, then $f$ does not divide $g$ and we output this accordingly.
    Otherwise, we proceed with the case where both $f$ and $g$ are nonzero.

    Following \textcite{Forbes15}, we execute Strassen's division elimination procedure \cite{Strassen73} on an arithmetic circuit that computes $g(\vec{x}) / f(\vec{x})$.
    This produces a circuit that computes a polynomial $h(\vec{x})$ which satisifes the identity $g(\vec{x}) - f(\vec{x}) h(\vec{x}) = 0$ if and only if $f$ divides $g$.
    Thus, to check if $f$ divides $g$, it suffices to construct a circuit that computes the polynomial $h(\vec{x})$ and then test the identity $g(\vec{x}) - f(\vec{x}) h(\vec{x}) = 0$.
    We proceed to construct a circuit for the polynomial $h(\vec{x})$.

    Because $f$ is nonzero and $|\F_n| > d \ge \deg(f)$, there is a point $\vec{\alpha} \in \F_n^n$ such that $f(\vec{\alpha}) \neq 0$.
    Translating this point to the origin, we may assume that $f(\vec{0}) \neq 0$.
    Since $f$ has a nonzero constant term, we can find an inverse of $f$ as a formal power series.
    Let $\hat{f}(\vec{x}) \coloneqq f(\vec{x}) - f(\vec{0})$ be the nonconstant part of $f(\vec{x})$.
    Then we have the equality of formal power series
    \[
        \frac{1}{f(\vec{x})} = \frac{1}{f(\vec{0})} \cdot \frac{1}{1 - \del{\frac{-\hat{f}(\vec{x})}{f(\vec{0})}}} = \frac{1}{f(\vec{0})} \sum_{i=0}^{\infty} \del{\frac{-\hat{f}(\vec{x})}{f(\vec{0})}}^{i}.
    \]
    Multiplying by $g(\vec{x})$, we have
    \[
        \frac{g(\vec{x})}{f(\vec{x})} = \frac{g(\vec{x})}{f(\vec{0})} \sum_{i=0}^{\infty} \del{\frac{-\hat{f}(\vec{x})}{f(\vec{0})}}^{i}.
    \]
    In the case where $f$ divides $g$, we know that the quotient $g/f$ is a polynomial of degree at most $d$, so the terms of degree at most $d$ on the right-hand side above will correspond to the quotient.
    Let $h(\vec{x})$ be the polynomial defined by
    \[
        h(\vec{x}) \coloneqq \mathrm{Hom}_{\le d}\sbr{\frac{g(\vec{x})}{f(\vec{0})} \sum_{i=0}^{\infty} \del{\frac{-\hat{f}(\vec{x})}{f(\vec{0})}}^{i}}.
    \]
    where $\mathrm{Hom}_{\le d}[\bullet]$ is the projection onto the homogeneous components of degree at most $d$.
    Because $\hat{f}(\vec{x})$ has no constant term, the polynomial $\hat{f}(\vec{x})^i$ is only supported on monomials of degree $i$ and higher.
    This implies that when $i > d$, the term $\hat{f}(\vec{x})^i$ does not contribute to the homogeneous components of degree at most $d$ of the power series above.
    This lets us write
    \[
        h(\vec{x}) = \mathrm{Hom}_{\le d}\sbr{\frac{g(\vec{x})}{f(\vec{0})} \sum_{i=0}^{d} \del{\frac{-\hat{f}(\vec{x})}{f(\vec{0})}}^{i}},
    \]
    so $h$ can be obtained from the homogeneous components of the finite sum above.
    When $f$ divides $g$, the preceding discussion implies that $h(\vec{x}) = g(\vec{x}) / f(\vec{x})$, so we have the identity $g(\vec{x}) - f(\vec{x}) h(\vec{x}) = 0$.
    On the other hand, if $f$ does not divide $g$, then we must have $g(\vec{x}) \neq f(\vec{x}) h(\vec{x})$, as otherwise $f$ would be a divisor of $g$, so we have the nonidentity $g(\vec{x}) - f(\vec{x}) h(\vec{x}) \neq 0$.
    Thus, to test if $f$ divides $g$, it suffices to test the identity $g(\vec{x}) - f(\vec{x}) h(\vec{x}) = 0$.

    We now explain how to carry out this reduction using an arithmetic network.
    The initial steps to check if either $f$ or $g$ is identically zero can be carried out by $\pit$ gates.
    In the case where $f$ and $g$ are both nonzero, we need to find a point $\vec{\alpha} \in \F_n^n$ where $f(\vec{\alpha}) \neq 0$.
    We do this by carrying out the search-to-decision reduction for $\pit$.

    Fix a subset $S \subseteq \F_n$ of size $d+1$.
    We will find such an $\vec{\alpha}$ in $S^n$ one coordinate at a time.
    For each $i \in [n]$, suppose we have determine the first $i-1$ coordinates of $\vec{\alpha}$.
    Because $|S| > \deg(f)$ and $f(\alpha_1, \ldots, \alpha_{i-1}, x_i, \ldots, x_n) \neq 0$, one of the polynomials $\set{f(\alpha_1, \ldots,\alpha_{i-1}, \beta, x_{i+1}, \ldots, x_n) : \beta \in S}$ must be nonzero.
    We use $\pit$ gates to determine which of these are nonzero, and we set $\alpha_i$ to be any one of the $\beta$ that results in a nonzero polynomial.
    Because this step only manipulates the circuit computing $f$ and evaluates $\pit$ gates on the result, this step can be carried out by an arithmetic network of size $\poly(n,s,d)$ and degree $O(1)$.
    Further, it is clear that a polynomial-time Turing machine can print this part of the network.

    With $\vec{\alpha}$ in hand, we replace $f(\vec{x})$ and $g(\vec{x})$ by $f'(\vec{x}) \coloneqq f(\vec{x} - \vec{\alpha})$ and $g'(\vec{x}) \coloneqq g(\vec{x} - \vec{\alpha})$, respectively.
    Because $f(\vec{x})$ divides $g(\vec{x})$ if and only if $f'(\vec{x})$ divides $g'(\vec{x})$, it suffices to test divisibility of $f'(\vec{x})$ and $g'(\vec{x})$.
    We know that $f'(\vec{0}) = f(\vec{\alpha}) \neq 0$.
    Let
    \[
        h'(\vec{x}) \coloneqq \mathrm{Hom}_{\le d}\sbr{\frac{g'(\vec{x})}{f'(\vec{0})} \sum_{i=0}^{d} \del{\frac{f'(\vec{0})-f'(\vec{x})}{f'(\vec{0})}}^{i}}
    \]
    as above.
    We can compute a description of a circuit computing $h'(\vec{x})$ by directly writing a circuit that computes the polynomial
    \[
        \frac{g'(\vec{x})}{f'(\vec{0})} \sum_{i=0}^{d} \del{\frac{f'(\vec{0})-f'(\vec{x})}{f'(\vec{0})}}^{i}
    \]
    and then applying \Cref{lem:network algo homogeneous components} to extract the homogeneous components of degree up to $d$.
    This results in a description of a circuit of size $O(s d^3)$ that computes $h$, and this description can be computed by a network of size $\poly(n,s,d)$ and degree $O(1)$.

    Once we have a description of a circuit that computes $h'(\vec{x})$, we compute a description of a circuit that computes $g'(\vec{x}) - f'(\vec{x}) h'(\vec{x})$.
    We then use a $\pit$ gate to test if this circuit computes the zero polynomial, and we report that $f$ divides $g$ exactly when $g' - f' h'$ is the zero polynomial.
    The correctness of this algorithm follows from the preceding discussion.
    This last step can likewise be carried out by a uniform network of size $\poly(n,s,d)$ and degree $O(1)$ as desired.
\eproof

Our next problem of interest is computing the greatest common divisor of polynomials.
The following lemma shows that the greatest common divisor of two arithmetic circuits $C_1$ and $C_2$ can be computed by an arithmetic circuit $D$ of comparable size, and that a description of $D$ can be computed efficiently by a uniform family of arithmetic networks.

\begin{lem}[Greatest common divisor of arithmetic circuits] \label{lem:network algo gcd}
    There is a deterministic, polynomial-time Turing machine that receives as input $(1^n, 1^d, 1^s)$ and outputs an arithmetic network with $\pit$ gates that satisfies the following properties.
    \begin{enumerate}
        \item
            The network receives as input the descriptions of two arithmetic circuits $C_1$ and $C_2$ of size $s$ over $\F_n$, computing polynomials $f, g \in \F_n(x_1,\ldots,x_n)[y]$, respectively, both of degree at most $d$.
        \item
            The network outputs a description of an arithmetic circuit $D$ of size $\poly(s, d)$.
            The arithmetic circuit $D$ computes the polynomial $\gcd(f, g) \in \F_n(\vec{x})[y]$.
        \item 
            The network has size $\poly(n,s,d)$ and degree $\poly(d)$.
    \end{enumerate}
\end{lem}

\bproof
    Use the network of \Cref{lem:network algo interpolation} to obtain descriptions of arithmetic circuits of size $O(s d^2)$ that compute the coefficients of $f$ and $g$ with respect to $y$.
    To compute the GCD of $f$ and $g$, we execute the Euclidean algorithm on $f$ and $g$, using $\pit$ gates to both implement division with remainder\footnote{Strictly speaking, the use of $\pit$ gates is not necessary to implement division with remainder. Given two polynomials $f$ and $g$, let $f = qg + r$ with $\deg(r) < \deg(g)$. The coefficients of $q$ and $r$ are rational functions of the coefficients of $f$ and $g$, and so can be computed by arithmetic circuits without the use of branching or $\pit$ gates.} and to detect when to stop.
    The Euclidean algorithm runs for at most $d$ iterations.
    In each iteration, we have descriptions of arithmetic circuits that compute the coefficients of two polynomials $r_{i-1}, r_i \in \F_n(\vec{x})[y]$, and we compute $r_{i+1} \coloneqq r_i \bmod r_{i-1}$, the remainder when $r_i$ is divided by $r_{i-1}$.

    Because $r_{i-1}$ and $r_i$ have degree at most $d$, the coefficients of $r_{i+1}$ can be computed from the coefficients of $r_{i-1}$ and $r_i$ using $\poly(d)$ arithmetic operations.
    This implies that the circuit complexity of the coefficients of $r_{i+1}$ increases by an additive $\poly(d)$ factor in each step of the Euclidean algorithm, so throughout the execution of the algorithm we only need to operate on arithmetic circuits of size $\poly(s, d)$.
    This yields the claimed bound of $\poly(n,s,d)$ on the size of the arithmetic network.
    The bound of $\poly(d)$ on the degree of the network follows from the fact that we only perform division with remainder on degree-$d$ polynomials, which can be implemented by a network of degree $\poly(d)$.
    The uniformity of the network follows from the uniformity of the Euclidean algorithm.
\eproof

For polynomials $f, g, h \in \F[x]$, we have the identity
\[
    \gcd(f, g, h) = \gcd(\gcd(f,g), h),
\]
so we can repeatedly apply \Cref{lem:network algo gcd} to compute the GCD of more than two polynomials.
To compute the GCD of $k$ polynomials, each of degree $d$, we can apply \Cref{lem:network algo gcd} a total of $O(\log k)$ times in parallel, resulting in a network of degree bounded by $d^{O(\log k)}$.
Later in \Cref{subsec:uniform univariate factorization}, we will want to compute the GCD of $O(d)$ polynomials of degree $d$.
Taking GCD's one-by-one results in a network of degree $d^{O(\log d)}$.
We can improve this degree bound to $\poly(d)$ by using an algorithm that is designed to compute the GCD of many polynomials at once.
Our application in \Cref{subsec:uniform univariate factorization} will only require us to compute the GCD of polynomials whose coefficients are field elements, not rational functions computed by small arithmetic circuits, so we specialize the statement of \Cref{lem:network algo multigcd} below to this setting.
We use an algorithm due to Steve Cook as reported by \textcite[Theorem 3.1]{vzGathen84}.
The recent GCD algorithm of \textcite{AW24} would work equally well for our purposes here.

\begin{lem}[Greatest common divisor of many polynomials {\cite[Theorem 3.1]{vzGathen84}}] \label{lem:network algo multigcd}
    There is a deterministic, polynomial-time Turing machine that receives as input $(1^n, 1^d)$ and outputs an arithmetic network that satisfies the following properties.
    \begin{enumerate}
        \item 
            The network receives as input the coefficients of $n$ univariate polynomials $f_1, \ldots, f_n \in \F_n[x]$, each of degree at most $d$.
        \item
            The network outputs the coefficients of $\gcd(f_1, \ldots, f_n) \in \F_n[x]$.
        \item
            The network has size and degree $\poly(n,d)$.
    \end{enumerate}
\end{lem}

\bproof
    As shown by \textcite[Theorem 3.1]{vzGathen84}, we can compute $\gcd(f_1, \ldots, f_n)$ by solving a collection of linear systems.
    For each $k \in \set{0,1,\ldots,d}$, we form a system of linear equations that expresses the existence of polynomials $g_1, \ldots, g_n$ of degree less than $d$ such that $\sum_{i=1}^n f_i(x) g_i(x)$ is a monic polynomial of degree $k$.
    Write $f_i(x) = \sum_{j=0}^d f_{i,j} x^j$ and let $g_i(x) = \sum_{j=0}^{d-1} g_{i,j} x^j$.
    Consider the system of linear equations
    \[
        \sum_{i=1}^n \sum_{j=0}^{\ell} g_{i,j} f_{i, \ell - j} = 
        \begin{cases}
            1 & \text{if $\ell = k$,} \\
            0 & \text{if $k < \ell \le 2d$,}
        \end{cases}
    \]
    where the $f_{i,j}$ are known and the $g_{i,j}$ are unknown.
    This system has a solution exactly when there are $g_1, \ldots, g_n \in \F_n[x]$ such that $\sum_{i=1}^n f_i g_i$ is monic of egree $k$.

    To compute the GCD, we solve all of these systems in parallel for $k \in \set{0, 1, \ldots, d}$ using the network of \Cref{lem:network algo linear systems}.
    Letting $g_1(x), \ldots, g_n(x)$ be the solution of this system, we output $\sum_{i=1}^n f_i(x) g_i(x)$ as $\gcd(f_1, \ldots, f_n)$.
    The correctness of this algorithm follows from \cite[Theorem 3.1]{vzGathen84}.
    The network of \Cref{lem:network algo linear systems} has size $\poly(n,d)$ and degree $O(1)$, so by \Cref{prop:concat-arith-networks}, the parallel invocation of $d$ copies of this network likewise has size $\poly(n,d)$ and degree $O(1)$.
    \Cref{lem:network algo linear systems} provides us a network with $\pit$ gates, but since the coefficients of the polynomials $f_1, \ldots, f_n$ are provided directly as input (and are not represented as arithmetic circuits), we can emulate the network of \Cref{lem:network algo linear systems} without using $\pit$ gates.
\eproof

In addition to the GCD, we can also compute the B\'{e}zout coefficients of two polynomials efficiently by means of uniform arithmetic networks.
Recall that for two univariate polynomials $f, g \in \F_n[x]$, their \emphdef{B\'{e}zout coefficients} are two polynomials $a, b \in \F_n[x]$ satisfying $\deg(a) < \deg(g)$, $\deg(b) < \deg(f)$, and 
\[
    af + bg = \gcd(f, g).
\]
The B\'{e}zout coefficients are a useful tool in computer algebra.
For our purposes, we will make use of them in \Cref{lem:network algo hensel lifting} to implement Hensel lifting as a uniform family of arithmetic networks.
The interested reader can find more information about the B\'{e}zout coefficients in the book of \textcite{vzGG13}.

\begin{lem}[B\'{e}zout coefficients] \label{lem:network algo bezout coefficients}
    There is a deterministic, polynomial-time Turing machine that receives as input $(1^n, 1^d)$ and outputs an arithmetic network satisfying the following properties.
    \begin{enumerate}
        \item
            The network receives as input the coefficients of two univariate polynomials $f, g \in \F_n[x]$, both of degree at most $d$.
        \item
            The network outputs the coefficients of two polynomials $a, b \in \F_n[x]$ such that 
            \begin{enumerate}
                \item $\deg(a) < \deg(g)$,
                \item $\deg(b) < \deg(f)$, and
                \item $a f + b g = \gcd(f,g)$.
            \end{enumerate}
        \item 
            The network has size and degree $\poly(d)$.
    \end{enumerate}
\end{lem}

\bproof
    Execute the extended Euclidean algorithm on the polynomials $f$ and $g$.
    The desired polynomials $a$ and $b$ are computed as part of the output of the extended Euclidean algorithm.
    The extended Euclidean algorithm runs for at most $d$ steps, and in each step computes polynomial division with remainder, which can be implemented by an arithmetic network of size and degree $\poly(d)$, so the extended Euclidean algorithm can likewise be implemented by an arithmetic network of size and degree $\poly(d)$.
    The uniformity of the network follows from the uniformity of the extended Euclidean algorithm.
\eproof

We conclude this section by showing that Hensel lifting can be implemented by a uniform family of arithmetic networks.
\emphdef{Hensel lifting} is an iterative procedure that, for our purposes, starts with an approximate factorization of a polynomial $f$ and improves the accuracy of this factorization.
More specifically, for an ideal $I \subseteq \F[\vec{x}]$, if we have a factorization
\[
    f(\vec{x}) = g_0(\vec{x}) h_0(\vec{x}) \bmod{I},
\]
then Hensel lifting computes, for each $k \in \N$, a pair of polynomials $g_k$ and $h_k$ such that
\[
    f(\vec{x}) = g_k(\vec{x}) h_k(\vec{x}) \bmod{I^{2^k}}.
\]
For us, Hensel lifting will be used as a subroutine in Kaltofen's algorithm \cite{Kaltofen1989FactorizationOP} to factor arithmetic circuits.
We will not need the full details of the Hensel lifting procedure, so we do not cover them all here.
The interested reader can find more details on Hensel lifting in the work of \textcite{KSS15} or the book of \textcite{vzGG13}.

\begin{lem}[{Hensel lifting \cite[Lemma 3.6]{KSS15}}] \label{lem:network algo hensel lifting}
    There is a deterministic, polynomial-time Turing machine that receives as input $(1^n, 1^d, 1^s, 1^{2^k})$ and outputs an arithmetic network satisfying the following properties.
    \begin{enumerate}
        \item
            The network receives as input the description of an arithmetic circuit $C$ of size $s$ that computes a polynomial $f(\vec{x}, y, z) \in \F_n[x_1,\ldots,x_n][y,z]$ of degree at most $d$.
            It also receives the coefficients of polynomials $g_0, h_0, a_0, b_0 \in \F_n[y]$ that satisfy the following conditions.
            \begin{enumerate}
                \item 
                    The polynomial $g_0(y)$ is an irreducible monic factor of $f(\vec{0}, y, 0)$.
                \item
                    The polynomials $g_0(y)$ and $h_0(y)$ satisfy the identity $f(\vec{0}, y, 0) = g_0(y) \cdot h_0(y)$.
                \item
                    The polynomials $a_0(y)$ and $b_0(y)$ satisfy the identity $a_0(y) g_0(y) + b_0(y) h_0(y) = 1$.
                    Furthermore, we have $\deg(a_0) < \deg(h_0)$ and $\deg(b_0) < \deg(g_0)$.
            \end{enumerate}
        \item
            The network outputs the description of a multi-output arithmetic circuit $D$ over $\F_n$.
            The circuit $D$ has size $\poly(n,s,d,2^k)$, degree $\poly(d, 2^k)$, and outputs $\poly(d, 2^k)$ many polynomials.
            The outputs of $D$ are the the coefficients of two polynomials $g_k(y, z), h_k(y, z) \in \F_n[\vec{x}][y, z]$ of individual degree at most $\max(d, 2^k)$ that satisfy the following conditions.
            \begin{enumerate}
                \item 
                    Letting $\langle z^{2^k} \rangle \subseteq \F_n(\vec{x})[y,z]$ denote the ideal generated by $z^{2^k}$, the polynomials $g_k$ and $h_k$ satisfy the equality
                    \[
                        f(\vec{x}, y, z) = g_k(y, z) \cdot h_k(y, z) \bmod \langle z^{2^k} \rangle.
                    \]
                \item
                    The polynomial $g_k(y, z)$ satisfies $g_k(y, 0) = g_0(y)$.
                \item
                    The polynomial $g_k(y, z)$ is monic with respect to $y$.
            \end{enumerate}
        \item 
            The network has size $\poly(n,s,d,2^k)$ and degree $O(1)$.
    \end{enumerate}
\end{lem}

\bproof
    We construct the circuit $C$ iteratively.
    For $i \in \set{0,1,\ldots,k-1}$, suppose we have constructed a circuit $C_i$ of size $\poly(n, s, d, 2^i)$ and degree $\poly(d, 2^i)$ that computes the coefficients of polynomials $g_i, h_i, a_i, b_i \in \F_n[\vec{x}][y, z]$ such that
    \begin{align*}
        f(\vec{x}, y, z) &= g_i(y, z) \cdot h_i(y, z) \bmod \langle z^{2^i} \rangle, \\
        a_i(y, z) g_i(y, z) + b_i(y, z) h_i(y, z) &= 1 \bmod \langle z^{2^i} \rangle,
    \end{align*}
    and $g_i(y, z)$ is monic with respect to $y$.
    From the coefficients of $g_i$, $h_i$, $a_i$, and $b_i$, we will compute the coefficients of polynomials $g_{i+1}, h_{i+1}, a_{i+1}, b_{i+1} \in \F_n[\vec{x}][y, z]$ such that the same equalities hold modulo the ideal $\langle z^{2^{i+1}} \rangle$.
    That is, we will have
    \begin{align*}
        f(\vec{x}, y, z) &= g_{i+1}(y, z) \cdot h_{i+1}(y, z) \bmod \langle z^{2^{i+1}} \rangle, \\
        a_{i+1}(y, z) g_{i+1}(y, z) + b_{i+1}(y, z) h_{i+1}(y, z) &= 1 \bmod \langle z^{2^{i+1}} \rangle,
    \end{align*}
    and $g_{i+1}$ is likewise monic with respect to $y$.

    To compute the coefficients of $g_{i+1}, h_{i+1}, a_{i+1}$, and $b_{i+1}$, we implement a single step of Hensel lifting as described in, e.g., \cite[Lemma 3.4]{KSS15}.
    For our purposes, the precise details of a Hensel lifting step are not important; it is enough to know that a single step consists of a constant number of arithmetic operations and one division with remainder, and that there is a uniform description of a lifting step.
    This implies that we can construct the circuit $C_{i+1}$ by adding an additional $\poly(d, 2^k)$ gates to the circuit $C_i$.
    To bound the degree of the final circuit $C_k$, we use the fact that each output of $C_k$ has degree bounded by $d 2^k$, so we can homogenize $C_k$ using \cref{lem:network algo homogeneous components} to obtain a circuit whose degree is likewise bounded by $d 2^k$.

    The bound on the size of the resulting arithmetic network is clear.
    The network itself does not do any nontrivial arithmetic computation; it only adds gates to the given circuit $C$.
    This implies that the degree of the network is bounded by $O(1)$.
    The uniformity of the arithmetic network follows from the uniformity of the Hensel lifting algorithm.
\eproof

\subsection{Univariate factorization over finite fields} \label{subsec:uniform univariate factorization}

An important step in factoring polynomials given by arithmetic circuits is factorization of univariate polynomials.
For our purposes (in particular, our targeted hitting set generator based on the Kabanets--Impagliazzo generator in \Cref{sec:ki}), it will be enough to have an algorithm to factor univariate polynomials over finite fields.
A randomized polynomial-time algorithm to factor polynomials over finite fields was first described by \textcite{Berlekamp70}.
Later work by \textcite{vzGathen84} showed that the randomized algorithm of \textcite{CantorZassenhaus81} admits an efficient parallel implementation.

The algorithm of \cite{vzGathen84} is essentially formalized as a uniform family of low-depth arithmetic networks, so we can easily conclude that univariate factorization over finite fields can be solved by uniform arithmetic networks.
However, our notions of the degree and error degree of a randomized arithmetic network do not appear in \cite{vzGathen84}, so we must be careful to bound these quantities.
We state and prove this result only for factoring squarefree polynomials over finite fields.
To factor general polynomials, we need the additional ability to compute $p$\ts{th} roots, where $p$ is the characteristic of the underlying field, and this necessitates either the use of $p$\ts{th} root gates in our network or operations of large degree.
For more on factorization over finite fields, we refer the reader to \cite[Chapter 14]{vzGG13}.

\begin{theorem}[Univariate factorization over finite fields {\cite[Theorem 4.1]{vzGathen84}}] \label{thm:network algo univariate factorization}
    Let $\F = \set{\F_{q(n)}}_{n \in \N}$ be a feasible family of finite fields, where $\F_{q(n)}$ is a finite field of order $q(n)$ and characteristic $p(n)$.
    There is a deterministic, polynomial-time Turing machine $M_1$ that receives as input $(1^n, 1^d)$ and outputs a randomized arithmetic network satisfying the following properties.
    \begin{enumerate}
        \item 
            The network receives as input the coefficients of a degree-$d$ polynomial $f \in \F_{q(n)}[x]$.
        \item
            If the input polynomial $f$ is squarefree, then the network outputs the complete factorization of $f$ into irreducible polynomials.
            Otherwise, the output of the network is meaningless.
        \item
            The network has size $\poly(d, \log q(n))$, degree $d^{O(\log(q(n)))}$, and error degree $1$.
    \end{enumerate}
    In addition, there is another deterministic, polynomial-time Turing machine $M_2$ that receives as input $(1^n, 1^d)$ and outputs a randomized arithmetic network satisfying the following properties.
    \begin{enumerate}
        \item 
            The network receives as input the coefficients of an irreducible polynomial $t \in \F_{p(n)}[x]$ such that $\F_{q(n)}[x] = \F_{p(n)}[x]/\langle t(x) \rangle$.
            The network also receives as input the coefficients of a degree-$d$ polynomial $f \in \F_{q(n)}[x]$, each represented as a tuple of elements from $\F_{p(n)}$.
        \item
            If the input polynomial $f$ is squarefree, then the network outputs the complete factorization of $f$ into irreducible polynomials.
            Otherwise, the output of the network is meaningless.
        \item
            The network has size $\poly(d, p)$, degree $\poly(d, p)$, and error degree $1$.
    \end{enumerate}
\end{theorem}

\bproof
    We first describe von zur Gathen's modified version of the Cantor--Zassenhaus algorithm \cite[Theorem 4.1]{vzGathen84}, bounding the size and degree of the resulting arithmetic network along the way, and then bound the error degree of the algorithm.
    As our concern is primarily on the complexity of the algorithm and not its correctness, we address issues of correctness only at a level of detail that suffices to establish the claimed bounds on the complexity of the algorithm.
    For simplicity, we write $q \coloneqq q(n)$ and $p \coloneqq p(n)$ for the order and characteristic of $\F_{q(n)}$, respectively.
    Let $k \in \N$ be a natural number such that $q = p^k$.
    Throughout, let $R \coloneqq \F_q[x] / \langle f \rangle$ be the quotient of $\F_q[x]$ by the ideal generated by $f$.
    We assume, without loss of generality, that our input polynomial $f$ is monic.

    We begin by recalling the algorithm of \textcite{vzGathen84}.
    \begin{enumerate}
        \item 
            In the case where we do not have a representation of $\F_q$ as an extension of $\F_p$, we start by computing the matrix $Q \in \F_q^{d \times d}$ corresponding to the linear map $R \to R$ given by $a \mapsto a^q$.
            Since $\set{1, x, x^2, \ldots, x^{d-1}}$ is a basis of $R$ over $\F_q$, it suffices to compute representatives of $x^{iq}$ in $R$ for $i \in \set{0,1,\ldots,d-1}$.
            The representative of $x^{iq}$, written in the basis $\set{1, x, \ldots, x^{d-1}}$, is precisely the remainder when $x^{iq}$ is divided by $f(x)$.

            To compute the remainder of $x^{iq}$ on division with $f(x)$, we compute $x^{iq}$ via repeated squaring in the quotient ring $R = \F_q[x]/\langle f \rangle$.
            Given a polynomial $g(x)$ of degree at most $d$ such that $g(x) = x^k \bmod f(x)$, we can compute a representative of $x^{2k}$ by computing $g(x)^2 \bmod f(x)$.
            The polynomial $g(x)^2$ has degree at most $2d$, so this polynomial division with remainder can be carried out with an arithmetic network of size and degree $\poly(d)$.
            There are $O(\log(qd))$ squaring operations, so computing representatives of the $x^{iq}$ can be done with size $\poly(d, \log(q))$ and degree $d^{O(\log(d) + \log(q))}$.

            In the case where we are given a representation of $\F_q$ as an extension of $\F_p$, we instead let $Q \in \F_q^{d \times d}$ be the matrix corresponding to the linear map $R \to R$ given by $a \mapsto a^p$.
            As above, to compute this matrix, we find representations of the polynomials $\set{x^{i p} : i \in \set{0,1,\ldots,d-1}}$ in the basis $\set{1,x,\ldots,x^{d-1}}$ of $R$.
            We do this using polynomial division with remainder, but instead of using repeated squaring in $R$, we simply write down the polynomial $x^{ip}$ and then find its remainder on division with $f(x)$.
            Each division with remainder can be carried out by an arithmetic network of size $\poly(d, p)$ and degree $\poly(d, p)$, and there are $d$ such networks run in parallel, so this step can likewise be implemented in size $\poly(d, p)$ and degree $\poly(d, p)$.
        \item
            We next compute the dimension of and a basis for the nullspace of $Q - I_d$, where $I_d \in \F_q^{d \times d}$ is the identity matrix.
            Let $r$ be the dimension of this nullspace, and let $g_1, \ldots, g_r \in \F_q[x]$ be polynomials such that $\set{g_i \bmod f : i \in [r]}$ is a basis for the nullspace of $Q - I_d$.
            The value of $r$ corresponds to the number of irreducible factors of $f$.
            If $r = 1$, then we report that $f$ is irreducible and the algorithm terminates.
            Otherwise, we proceed with factorization of $f$.
            (Here, we used the assumption that $f$ is squarefree: if $f$ is squarefree and has only one irreducible factor, then $f$ itself must be this irreducible factor, hence $f$ is irreducible.)

            We can compute $r$ and a basis for the nullspace of $Q - I_d$ using the network of \Cref{lem:network algo linear systems}.
            In this case, the entries of the matrix $Q - I_d$ are field elements, so the basis produced by \Cref{lem:network algo linear systems} will likewise consist of vectors over the field $\F_q$.
            In particular, this invocation of \Cref{lem:network algo linear systems} does not require the use of $\pit$ gates.
        \item
            In the case where $r \ge 2$, let $m \coloneqq 5 \log_{9/5}(r)$.
            Let $v_{1,1}, \ldots, v_{m,r}$ be chosen independently at random from $\F_q$ or $\F_p$, depending on whether we have a representation of $\F_q$ as an extension of $\F_p$.
            For $i \in [m]$, define $h_i(x) \coloneqq \sum_{j=1}^r v_{i,j} g_j(x)$.
            The coefficients of each $h_i$ can be computed by a network of size $O(rd) \le O(d^2)$ and degree $2$.
        \item
            For each $i \in [m]$, we compute a polynomial $c_i(x)$ that is likely to be a proper factor of $f(x)$.
            If we are not given a representation of $\F_q$ as an extension of $\F_p$, we compute
            \[
                c_i(x) \coloneqq 
                \begin{cases}
                    \gcd(f(x), h_i(x)^{\frac{q-1}{2}} - 1) & \text{if $p$ is odd,} \\
                    \gcd(f(x), \sum_{j=0}^{\log_2(q)-1} h_i(x)^{2^j}) & \text{otherwise.}
                \end{cases}
            \]
            If instead we have a representation of $\F_q$ as an extension of $\F_p$, we compute
            \[
                c_i(x) \coloneqq 
                \begin{cases}
                    \gcd(f(x), h_i(x)^{\frac{p-1}{2}} - 1) & \text{if $p$ is odd,} \\
                    \gcd(f(x), h_i(x)) & \text{otherwise.}
                \end{cases}
            \]
            Some of the polynomials $c_i$ will be proper factors of $f$, which can be thought of as a partial factorization of $f$.
            We improve this to a finer factorization of $f$ as follows.
            For every $I \subseteq \zo \times [m]$, we compute
            \[
                s_I(x) \coloneqq \gcd\del{\set{c_i(x) : (0, i) \in I} \cup \set{\frac{f(x)}{c_i(x)} : (1, i) \in I}}.
            \]
            We then compute the set
            \[
                T \coloneqq \set{I \subseteq \zo \times [m] : \text{$s_I \neq 1$ and for all $J \supseteq I$, either $s_J = 1$ or $s_J = s_I$}},
            \]
            which corresponds to the minimal choices of $I$ that produce the finest nontrivial factorization of $f$.

            If $|T| = r$, then the polynomials $\set{s_I : I \in T}$ are the irreducible factors of $f$.
            In this case, the network outputs the coefficients of these polynomials and terminates.
            Otherwise, we have $|T| \neq r$, in which case the network fails to factor $f$ into irreducibles and reports this failure accordingly.

            To compute the polynomial $c_i(x)$ when we don't have a representation of $\F_q$ as an extension of $\F_p$, we first compute
            \[
                g_i(x) \coloneqq \begin{cases}
                    h_i(x)^{\frac{q-1}{2}} - 1 \bmod f(x) & \text{if $q$ is odd,} \\
                    \sum_{j=0}^{\log_2(q)-1} h_i(x)^{2^j} \bmod f(x) & \text{otherwise}
                \end{cases}
            \]
            in a manner similar to how we computed $x^{iq} \bmod f(x)$ in step (1).
            In both situations, we compute the necessary powers of $h_i$ via repeated squaring, performing division with remainder after each squaring operation to ensure we always square a polynomial of degree bounded by $d$.
            This requires $O(\log(q))$ squaring operations, each of which can be implemented by a network of size and degree bounded by $\poly(d)$, so each $g_i$ can be computed by a network of size $\poly(d, \log(q))$ and degree $d^{O(\log q)}$.
            Once we have computed $g_i$, we can compute $c_i(x) = \gcd(f(x), g_i(x))$ using \Cref{lem:network algo gcd}, which is implemented as a network of size and degree bounded by $\poly(d)$.

            If we instead have a representation of $\F_q$ as an extension of $\F_p$, we directly compute the coefficients of $h_i(x)^{\frac{p-1}{2}}-1$ (if necessary) and then apply the network of \Cref{lem:network algo gcd} to compute the polynomial $c_i(x)$.
            This can be implemented by a network of size and degree $\poly(d, p)$.

            Once the $c_i(x)$'s have been computed, we can compute each $s_I(x)$ using \Cref{lem:network algo multigcd}.
            There are at most $2^{2m} \le \poly(r)$ choices of $I \subseteq \zo \times [m]$, and for each such $I$, the polynomial $s_I(x)$ can be computed by direct application of the network of \Cref{lem:network algo multigcd}.
            This computes each $s_I(x)$ using a network of size $\poly(r, d) \le \poly(d)$ and degree $\poly(r,d) \le \poly(d)$.
            In all, we have computed each $s_I(x)$ using a network of size $\poly(d, \log q)$ and degree $d^{O(\log q)}$ when $\F_q$ is not presented as an extension of $\F_p$, and size $\poly(d, p)$ and degree $\poly(d, p)$ when $\F_q$ is given as an extension of $\F_p$.
            By \Cref{prop:concat-arith-networks} we can compute all the $s_I(x)$ in parallel using a network of size $\poly(d)$ and degree $d^{O(\log q)}$ in the former case, and size $\poly(d, p)$ and degree $\poly(d, p)$ in the latter case.
            
            The set $T$ can be computed directly by a network of $\poly(d)$ size and degree $O(1)$: for $I, J \subseteq \zo \times [m]$ with $I \subseteq J$, if either $s_I(x) = 1$ or $s_J(x) \neq 1$ and $s_J(x) \neq s_I(x)$, then we know that $I \notin T$.
            Thus, we can check if $I \in T$ by checking all choices of $J \supseteq I$, of which there are at most $2^{2m} \le \poly(r) \le \poly(d)$.
            This clearly costs $\poly(d)$ size.
            Since the computation of $T$ is done by the boolean part of the network, this computation only incurs degree $O(1)$.
    \end{enumerate}

    The preceding description bounds the size and degree of the resulting arithmetic network.
    The uniformity of the arithmetic network follows from the uniform description of the algorithm above.
    To bound the error degree of the network, we need to understand when the polynomials $s_I(x)$ appearing in step (4) fail to provide a complete factorization of $f$ into irreducible factors.

    To do this, we follow von zur Gathen's analysis of his algorithm.
    Let $f = f_1 \cdots f_t$ be the factorization of $f$ into irreducible factors.
    (We assume that the given polynomial $f$ is squarefree, so each irreducible factor occurs with multiplicity $1$.)
    Recall that $R = \F_q[x]/\langle f \rangle$ and define $R' \coloneqq \F_q[x]/\langle f_1 \rangle \times \cdots \times \F_q[x]/\langle f_t \rangle$.
    By the Chinese remainder theorem, there is an isomorphism of rings $\varphi : R \to R'$.

    We first consider the case when $p$ is odd and we are not given a representation of $\F_q$ as an extension of $\F_p$.
    For a polynomial $h \in \F_q[x]$, we say that $h$ \emphdef{separates} the factors $f_i$ and $f_j$ if exactly one of $f_i$ and $f_j$ divides $h^{\frac{q-1}{2}}-1$.
    Let $\overline{h}$ be the image of $h$ in $R$ and let $\alpha(h) = (u_1, \ldots, u_r) \in \F_q^r$. 
    We know from \cite{vzGathen84} that $h$ separates $f_i$ and $f_j$ exactly when one of $u_i$ and $u_j$ satisfies $u^{\frac{q-1}{2}} - 1 = 0$.
    There are exactly $\frac{q^2-1}{2}$ such pairs $(u_i, u_j)$, so the probability that a single randomly-chosen $h$ separates $f_i$ and $f_j$ is
    \[
        \frac{q^2 - 1}{2q^2} \ge \frac{4}{9},
    \]
    since $q \ge 3$.
    This implies that the probability that \emph{none} of the $h_k$ separate the factors $f_i$ and $f_j$ is bounded by $(5/9)^m$.
    By a union bound, the probability that some pair of factors $f_i$ and $f_j$ is not separated is bounded by
    \[
        \binom{r}{2} \del{\frac{5}{9}}^m \le r^2 \del{\frac{5}{9}}^{\log_{9/5}(r)} = \frac{1}{r^3} \le \frac{1}{8}.
    \]
    The last inequality above uses the fact that the algorithm only errs when $f$ is reducible, in which case $r \ge 2$.
    If every pair of factors is separated, then the algorithm correctly returns a factorization of $f$ into irreducible factors, so the error probability of the algorithm is bounded by $1/8$.
    This bounds the error degree of the algorithm by $q/8$.

    If $p$ is odd and we are given a representation of $\F_q$ as an extension of $\F_p$, then the preceding analysis still applies, replacing $q$ with $p$ everywhere.
    Since our random elements are drawn from $\F_p$ and not $\F_q$, we instead obtain a bound of $p/8$ on the error degree of the algorithm.

    If $p$ is even and we are not given a representation of $\F_q$ as an extension of $\F_p$, we instead say that $h$ \emphdef{separates} the factors $f_i$ and $f_j$ if exactly one of $f_i$ and $f_j$ divides $\sum_{j=0}^{\log_2(q)-1} h_i(x)^{2^j}$.
    \textcite{vzGathen84} shows that a randomly-chosen $h$ separates $f_i$ and $f_j$ with probability exactly $1/2$, so the same analysis as in the case of odd $p$ applies.
    If we are also given a representation of $\F_q$ as an extension of $\F_p$, then \textcite{vzGathen84} likewise shows that for a randomly-chosen $h$, with probability $1/2$ exactly one of $f_i$ and $f_j$ will divide $h(x)$.
    Again, the same analysis as in the case of odd $p$ applies.

    For the improved bound of $1$ on the error degree, we run $\log_8(q)$ or $\log_8(p)$ copies of the preceding algorithm in parallel, depending on whether or not $\F_q$ is presented as an extension of $\F_p$.
    Each independent run fails with probability $1/8$, so the probability that all copies fail is bounded by either $(1/8)^{\log_8(q)} = 1/q$ or $(1/8)^{\log_8(p)} = 1/p$.
    All successful runs of the algorithm will output the same set of irreducible factors of $f$, possibly up to permutation, so we may output the results of the first successful run of the algorithm.
    A run of the algorithm succeeds if the set $T$ computed in step (4) has size $r$, so we can determine which runs of the algorithm are successful and output the first such run.
    Hence the algorithm fails either with probability $1/q$ or $1/p$ as appropriate, so the error degree is bounded by $1$.
\eproof

\subsection{Factorization of arithmetic circuits over finite fields} \label{subsec:kaltofen-arith-net}

Our goal in this subsection is to argue that Kaltofen's~\cite{Kaltofen1989FactorizationOP} algorithm to factor arithmetic circuits over finite fields can itself be implemented as a uniform family of randomized arithmetic networks.

\subsubsection{The resultant and discriminant}

Before presenting Kaltofen's algorithm, we need to recall the notions of the resultant and the discriminant.
These are well-known tools that are useful in the study of polynomial factorization, and we will make use of them to argue correctness of our presentation of Kaltofen's algorithm.
For more on the resultant, we refer the reader to \cite[Chapter 6]{vzGG13}.

\begin{definition}[Resultant] \label{def:resultant}
    Let $\F$ be a field and let $f, g \in \F[x]$ be univariate polynomials.
    Write $f(x) = \sum_{i=0}^n f_i x^i$ and $g(x) = \sum_{i=0}^m g_i x^i$, with $f_n \neq 0$ and $g_m \neq 0$.
    The \emphdef{resultant} of $f$ and $g$, denoted $\res(f, g)$, is given by the determinant
    \[
        \res(f, g) \coloneqq \det
        \begin{pmatrix}
            f_n & & & & g_m & & & & & \\
            f_{n-1} & f_n & & & g_{m-1} & g_m & & & & \\
            \vdots & \vdots & \ddots & & \vdots & \vdots & \ddots & & & \\
            \vdots & \vdots & & f_n & g_1 & \vdots & & \ddots & & \\
            \vdots & \vdots & & f_{n-1} & g_0 & \vdots & & & \ddots & \\
            \vdots & \vdots & & \vdots & & g_0 & & & & g_m \\
            f_0 & \vdots & & \vdots & & & \ddots & & & \vdots \\
            & f_0 & & \vdots & & & & \ddots & & \vdots \\
            & & \ddots & \vdots & & & & & \ddots & \vdots \\
            & & & f_0 & & & & & & g_0
        \end{pmatrix},
    \]
    where the above matrix is of size $(n+m) \times (n+m)$, the first parallelogram consists of $m$ columns, and the second parallelogram consists of $n$ columns.
\end{definition}

The utility of the resultant arises from the fact that it detects when two polynomials share a common factor.

\begin{lem}[{see, e.g., \cite[Corollary 6.17]{vzGG13}}] \label{lem:resultant common factor}
    Let $\F$ be a field and let $f, g \in \F[x]$ be univariate polynomials.
    Then $\res(f, g) = 0$ if and only if $f$ and $g$ have a common factor.
\end{lem}

The resultant of two polynomials $f$ and $g$ can always be written as $s(x) f(x) + t(x) g(x)$, where $s, t \in \F[x]$ are univariate polynomials satisfying $\deg(s) < \deg(g)$ and $\deg(t) < \deg(f)$.
We will not need the fact that this linear combination can be found in low degree, so we do not quote that here.

\begin{lem}[{see, e.g., \cite[Corollary 6.21]{vzGG13}}] \label{lem:resultant linear combination}
    Let $\F$ be a field and let $f, g \in \F[x]$ be univariate polynomials.
    There are polynomials $s, t \in \F[x]$ such that $s(x) f(x) + t(x) g(x) = \res(f,g)$.
\end{lem}

By specializing $g$ to $\frac{\partial f}{\partial x}$ in the resultant, we obtain the discriminant of a single polynomial.

\begin{definition}[Discriminant] \label{def:discriminant}
    Let $\F$ be a field and let $f \in \F[x]$ be a univariate polynomial.
    The \emphdef{discriminant} of $f$, denoted by $\disc(f)$, is the resultant $\res(f, \frac{\partial f}{\partial x})$.
\end{definition}

Just as the resultant checks if two polynomials share a common factor, the discriminant checks if a single polynomial has repeated factors.

\begin{lem}[{see, e.g., \cite[Corollaries 6.17 and 14.25]{vzGG13}}] \label{lem:discriminant squarefree}
    Let $\F$ be a field and let $f \in \F[x]$ be a univariate polynomial.
    Suppose $f$ is nonconstant.
    Then $\disc(f) \neq 0$ if and only if $f$ is squarefree.
\end{lem}

\subsubsection{Kaltofen's algorithm}

We now present Kaltofen's algorithm to factor arithmetic circuits over finite fields, closely following the implementation appearing in the work of \textcite{KSS15}.

\begin{proposition}\label{prop:kaltofen arith network}
    Let $\F = \set{\F_{q(n)}}_{n \in \N}$ be a feasible family of finite fields, where the $n$\ts{th} field has order $q(n)$.
    There is a deterministic, polynomial-time Turing machine that receives as input $(1^n, 1^d, 1^s)$ and outputs a randomized arithmetic network satisfying the following properties.
    \begin{enumerate}
        \item 
            The network receives as input the description of an arithmetic circuit $C$ of size $s$ over $\F_{q(n)}$ that computes a polynomial $f \in \F_{q(n)}[x_1,\ldots,x_n]$ of degree at most $d$.
        \item
            The network outputs the descriptions of arithmetic circuits $C_1, \ldots, C_t$ together with natural numbers $e_1, \ldots, e_t \in \N$ and $\beta_1, \ldots, \beta_t \in \N$ encoded in binary.
            For each $i \in [t]$, let $g_i \in \F_{q(n)}[x_1, \ldots, x_n]$ be the polynomial computed by the circuit $C_i$.
            The output of the network satisfies the following properties.
            \begin{enumerate}
                \item
                    Each circuit $C_i$ is of size $\poly(n,s,d)$.
                \item
                    For each $i \in [t]$, the polynomial $g_i$ is either irreducible or a $p^{\beta_i}$-th power of an irreducible polynomial, where $p = \ch(\F_{q(n)})$.
                \item
                    For each $i \in [t]$, the natural number $e_i$ is not divisible by $p$.
                \item
                    For $i \neq j$, the polynomials $g_i$ and $g_j$ are coprime.
                \item
                    The polynomial identity $f(\vec{x}) = \prod_{i=1}^t g_i(\vec{x})^{e_i}$ holds.
            \end{enumerate}
        \item
            The network has size $\poly(n,s,d,\log(q))$, degree $n^{O(1)} d^{O(\log(q))}$, and error degree $\poly(d)$.
            If we are additionally provided with a representation of $\F_{q(n)}$ as an extension of $\F_{p(n)}$, where $p(n)$ is the characteristic of $\F_{q(n)}$, then the size of the network becomes $\poly(n,s,d,p)$ and the degree becomes $\poly(n,d,p)$.
    \end{enumerate}
\end{proposition}

The remainder of this subsubsection consists of the proof of \Cref{prop:kaltofen arith network}.
We make use of several subroutines for manipulating polynomials and arithmetic circuits, all of which can be implemented as uniform randomized arithmetic networks in a straightforward manner.
The curious reader may find the details of these implementations in \Cref{subsec:network algos for algebra}.
For simplicity, we abbreviate $q = q(n)$ throughout, so we work over the finite field $\F_q$ of order $q$.

We first check if $f(\vec{x})$ is a constant polynomial.
We do this by evaluating $f$ on two random points using \Cref{prop:universal-arith-network}, declaring $f$ to be a constant polynomial exactly when these evaluations match.
If we determine that $f$ is a constant polynomial, we output a single circuit $C_1$ that computes the corresponding constant single gate labeled by either of the evaluations of $f$ we have computed.
We also output a multiplicity of $e_1 = 1$ for this factor of $f(\vec{x})$.
Otherwise, we know that $f(\vec{x})$ is nonconstant and we proceed to factor $f$.
This step can be implemented by a network of size $\poly(n,s,d)$, degree $\poly(n,d)$, and error degree $d$, and this network can be printed by invoking the corresponding Turing machine from \Cref{prop:universal-arith-network}.

\paragraph{Preparations for Factorization}

Before factoring $f(\vec{x})$, we need some preparatory steps to ensure $f$ satisfies several convenient properties.
In particular, we will ensure that $f$ is not a $p$\ts{th} power, that $f$ is monic in a known variable, and that $f$ is squarefree.

\begin{enumerate}
    \item
        First, we ensure that $f$ is not a $p$\ts{th} power.
        Suppose that $f$ is a $p^e$-th power, but not a $p^{e+1}$-th power.
        Then there is a variable $x_i$ such that $f$ depends on $x_i$ and some monomial of $f$ is divisible by $x_i$, but is not divisible by $x_i^{p^{e+1}}$.
        Write $f$ as a polynomial in $x_i$ as
        \[
            f(\vec{x}) = \sum_{j=0}^d f_j(x_1,\ldots,x_{i-1}, x_{i+1}, \ldots, x_n) x_i^j.
        \]
        Because $f$ depends on $x_i$, at least one of the $f_j$ is nonzero.
        By the assumption that $f$ is a $p^e$-th power, we know that $f_j(x_1,\ldots,x_{i-1}, x_{i+1}, \ldots, x_n) \neq 0$ if and only if $j = p^e k$ for some $k \in \N$.
        Consider the polynomial
        \[
            f^*(\vec{x}) = \sum_{j=0}^{d/p^e} f_{p^e j}(x_1,\ldots,x_{i-1}, x_{i+1}, \ldots, x_n) x_i^j.
        \]
        That is, whenever $x_i^{p^e j}$ appears in a monomial of $f$, we replace it with $x_i^{j}$.
        By our choice of the variable $x_i$, there is some $j$ with $p \nmid j$ such that $f_{p^e j}(x_1, \ldots, x_{i-1}, x_{i+1}, \ldots, x_n) \neq 0$, and because $f$ depends on the variable $x_i$, there exists such a $j$ that is nonzero.
        This implies that $x_i$ appears in $f^*$ with an exponent that is not divisible by $p$, so $f^*$ is not a $p$\ts{th} power.

        We claim that the irreducible factors of $f$ are in one-to-one correspondence with the irreducible factors of $f^*$.
        \begin{itemize}
            \item 
                In one direction, if $h(x_1,\ldots,x_n)$ is an irreducible factor of $f$, then because $f$ is a $p^e$-th power, it follows that $h(x_1,\ldots,x_n)^{p^e}$ divides $f(x_1,\ldots,x_n)$.
                We can write $h^{p^e}$ as
                \[
                    h(x_1,\ldots,x_n)^{p^e} = h^*(x_1,\ldots,x_{i-1},x_i^{p^e},x_{i+1},\ldots,x_n)
                \]
                for some polynomial $h^* \in \F_q[\vec{x}]$.
                Because $h^*(x_1,\ldots,x_{i-1},x_i^{p^e},x_{i+1},\ldots,x_n)$ divides the polynomial $f^*(x_1,\ldots,x_{i-1},x_i^{p^e},x_{i+1},\ldots,x_n)$, it follows that $h^*(\vec{x})$ divides $f^*(\vec{x})$.
            \item
                In the other direction, if $h^*(\vec{x})$ is an irreducible factor of $f^*(\vec{x})$, then it is clear that the polynomial $h(x_1,\ldots,x_{i-1},x_i^{p^e},x_{i+1},\ldots,x_n)$ divides 
                \[
                    f^*(x_1,\ldots,x_{i-1},x_i^{p^e},x_{i+1},\ldots,x_n) = f(\vec{x}).
                \]
        \end{itemize}
        In particular, given an irreducible factor of $f^*$, we can recover a $p^e$-th power of the corresponding irreducible factor of $f$.
        Thus, to factor $f$ (up to $p^e$-th powers), it suffices to instead factorize $f^*$.
        (For details, see \textcite[Section 3.1]{KSS15}.)

        We now argue that the description of an arithmetic circuit computing $f^*$ can be efficiently computed from the description of a circuit computing $f$.
        We first find a natural number $e \in \N$ such that $f$ is a $p^e$-th power but not a $p^{e+1}$-th power.
        For each variable $x_i$, we can write $f$ as a polynomial in $x_i$ whose coefficients are polynomials in the remaining variables, i.e.,
        \[
            f(\vec{x}) = \sum_{j=0}^d f_{i,j}(x_1,\ldots,x_{i-1},x_{i+1},\ldots,x_n) x_i^j.
        \]
        Using \Cref{lem:network algo interpolation}, we can compute descriptions of circuits that compute the polynomials $f_{i,j}$ above.
        For each pair $(i,j) \in [n] \times \set{0,1,\ldots,d}$, we evaluate $f_{i,j}$ on a random point using \Cref{prop:universal-arith-network} to determine if the polynomial $f_{i,j}$ is zero or nonzero.
        We then compute, for each $i \in [n]$, the largest natural number $e_i \in \N$ such that for all $j$, we have $f_{i,j} \neq 0$ if and only if $p^{e_i} \mid j$.
        The desired exponent $e$ is given by $e \coloneqq \min_{i \in [n]} e_i$.

        Let $i^\star \coloneqq \argmin_{i \in [n]} e_i$.
        By our choice of $i^\star$, there is a nonzero $j$ such that $p^{e+1} \nmid j$ and $f_{i^\star, j} \neq 0$.
        We can take $f^*(\vec{x})$ to be the polynomial
        \[
            f^*(\vec{x}) = \sum_{j=0}^{d/p^e} f_{i^\star, p^e j}(x_1,\ldots,x_{i^\star-1},x_{i^\star+1},\ldots,x_n) x_{i^\star}^{j}.
        \]
        Our choice of $i^\star$ ensures that $f^*$ is not a $p$\ts{th} power.
        We already have computed descriptions of circuits for the polynomials $\set{f_{i^\star, p^e j} : j \in \set{0,1,\ldots,d/p^e}}$, so we can efficiently compute a description of a circuit that computes $f^*(\vec{x})$.
        As described above, we can recover the factorization of $f$ from the factorization of $f^*$, so we will proceed to factor $f^*$.
        We also record the integer $e$, which will correspond to the output $\beta_j$ for any factor $g_j$ of $f$ we compute by factoring $f^*$.

        We now analyze the cost of this step.
        We invoke \Cref{lem:network algo interpolation} in parallel $n$ times, and each invocation is implemented by a network of size $\poly(n,s,d)$ and degree $O(1)$, so computing the coefficients $f_{i,j}$ can be done by a network of size $\poly(n,s,d)$ and degree $O(1)$.
        Evaluating the coefficients $f_{i,j}$ on random points is done in parallel using \Cref{prop:universal-arith-network}, which is implemented by a network of size $\poly(n,s,d)$ and degree $\poly(n,d)$.
        Computing $e$ and $i^\star$ can be done by a boolean subcircuit of size $\poly(n,d)$.
        Once $i^\star$ has been computed, a description of a circuit for $f^*$ can be computed in a straightforward manner using a network of size $\poly(n,s,d)$ and degree $O(1)$.
        In total, this step can be implemented by a network of size $\poly(n,s,d)$ and degree $\poly(n,d)$.

        The only source of error in this step is when a polynomial $f_{i,j}$ is nonzero, but evaluates to zero on a random point.
        Since each $f_{i,j}$ has degree at most $d$, when we sample the coordinates of the evaluation point at random from a set $S \subseteq \F_q$, we err with probability at most $d/|S|$.
        Although there are $n(d+1)$ such polynomials $f_{i,j}$, we can avoid paying a union bound in the error degree.
        Suppose $i^\star \in [n]$ is the index of a variable $x_{i^\star}$ such that $f$ depends on $x_{i^\star}$ and some monomial of $f$ is divisible by $x_{i^\star}$, but is not divisible by $x_{i^\star}^{p^{e+1}}$.
        Then there is some $j^\star \in [d]$ such that $p \nmid j$ and $f_{i^\star, p^e j^\star} \neq 0$.
        As long as the evaluation of $f_{i^\star, j^\star}$ is nonzero, the algorithm will correctly modify $f$ into a non-$p$\ts{th}-power $f^*$, so the error degree of this step is bounded by $d$.

        From here on, we use $f(\vec{x})$ to denote the polynomial $f^*(\vec{x})$ obtained above for the sake of notational ease.
        
    \item
        To make $f$ monic, we introduce a new variable $y$ and perform a random linear change of variables.
        Let $\vec{\alpha} \in \F_q^n$ be chosen at random and define
        \[
            \hat{f}(\vec{x}, y) \coloneqq f(\vec{x} + y \cdot \vec{\alpha}) = f(x_1 + y \alpha_1, \ldots, x_n + y \alpha_n).
        \]
        Let $f_d(\vec{x})$ be the top-degree homogeneous component of $f(\vec{x})$.
        A standard calculation shows that if $f_d(\vec{\alpha}) \neq 0$, then the polynomial $\hat{f}(\vec{x}, y)$ is monic in the fresh variable $y$.

        Given $\vec{\alpha}$, we can compute a description of a circuit computing $\hat{f}$ from the description of the circuit computing $f$.
        This is done by adding new gates computing the new input $\vec{x} + y \vec{\alpha}$, which can be implemented by an arithmetic network of size $\poly(n,s)$ and degree $O(1)$.

        Because the irreducible factors of $f$ and $\hat{f}$ are in one-to-one correspondence, in order to factor $f$, it suffices to factor $\hat{f}$.
        After factoring $\hat{f}$, we can recover the factors of $f$ by applying the change of variables $(\vec{x}, y) \mapsto (\vec{x} - y \vec{\alpha}, y)$, which can be performed by an arithmetic network of size $\poly(n,s)$ and degree $O(1)$.

        This step of the algorithm errs if the point $\vec{\alpha}$ is chosen so that $f_d(\vec{\alpha}) = 0$.
        Because $f_d(\vec{x})$ is a polynomial of degree $d$, if we choose the coordinates of $\vec{\alpha}$ at random from a subset $S \subseteq \F_q$, then \Cref{lem:sz} implies that $f_d(\vec{\alpha}) = 0$ with probability at most $d / |S|$.
        This implies that the error degree of this step is bounded by $d$.

        In the following steps, we write $f(\vec{x}, y)$ for $\hat{f}(\vec{x},y)$ to avoid cluttered notation.
    \item
        Finally, we need to ensure that $f(\vec{x}, y)$ is squarefree.
        It is a standard fact that the squarefree part of $f$ can be written as
        \[
            \frac{f(\vec{x}, y)}{\gcd\del{f, \frac{\partial f}{\partial y}}},
        \]
        assuming that $\frac{\partial f}{\partial y} \neq 0$.
        Because $f$ is nonzero, not a $p$\ts{th} power, and depends on the variable $y$, it follows that $\frac{\partial f}{\partial y} \neq 0$.
        The description of a circuit for $\frac{\partial f}{\partial y}$ can be computed from the description of a circuit for $f$ by using \Cref{lem:network algo partial derivative}.
        A description of a circuit computing $\gcd(f, \frac{\partial f}{\partial y})$ can then be computed using \Cref{lem:network algo gcd}.

        The applications of \Cref{lem:network algo partial derivative} and \Cref{lem:network algo gcd} can be implemented by arithmetic networks of size $\poly(n,s,d)$ and degree $\poly(d)$.
        The corresponding arithmetic networks use no randomness, so they have error degree $0$.
\end{enumerate}

In total, the preparatory steps above can be implemented by an arithmetic network of size $\poly(n,s,d)$.
Using \Cref{prop:error-deg-comp}, we can bound the degree and error degree of the resulting network by $\poly(d)$ and $O(d)$, respectively.
The description of the preparatory steps is a uniform algorithm, so the resulting arithmetic network can indeed be printed by a polynomial-time Turing machine.

\paragraph{Reduction to Bivariate Factoring}

We are now in the situation where the polynomial $f(\vec{x}, y)$ to be factored is not a $p$\ts{th} power, is monic in $y$, and is squarefree.
We now reduce the task of factoring the multivariate polynomial $f(\vec{x}, y)$ in $\F_q[\vec{x},y]$ to factoring a bivariate polynomial, albeit over the larger field $\F_q(x_1,\ldots,x_n)$.

The discriminant will help us find a good reduction to bivariate factorization.
We know that the polynomial $f(\vec{x}, y)$ to be factored is squarefree, and we would like to find a substitution $\vec{x} \mapsto \vec{\alpha}$ so that the resulting polynomial $f(\vec{\alpha}, y)$ remains squarefree.
Because $f(\vec{x}, y)$ is squarefree and is nonconstant as a polynomial in $y$, viewing $f$ as an element of $\F_q(\vec{x})[y]$, \Cref{lem:discriminant squarefree} implies that the discriminant of $f$ is a nonzero polynomial in $\F_q[\vec{x}]$.
To guarantee that $f(\vec{\alpha}, y)$ is squarefree, it suffices to find a point $\vec{\alpha}$ where the discriminant of $f$ evaluates to a nonzero value, as this implies (again, by \Cref{lem:discriminant squarefree}) that the polynomial $f(\vec{\alpha}, y)$ is squarefree.

Let $\Delta(\vec{x}) \coloneqq \disc_y(f(\vec{x}, y))$ be the discriminant of $f$ with respect to $y$.
We know that $\Delta(\vec{x})$ is a nonzero polynomial and that $\Delta(\vec{x})$ has degree bounded by $2d^2$.
If we sample a point $\vec{\alpha} \in \F_q^n$ whose coordinates are chosen uniformly at random from a set $S \subseteq \F_q$, then \Cref{lem:sz} the probability that $\Delta(\vec{\alpha}) = 0$ is at most $2d^2 / |S|$, so we can find such a point $\vec{\alpha}$ with error degree $2d^2$.
Applying the change of variables $\vec{x} \mapsto \vec{x} - \vec{\alpha}$, we may assume that $f$ has nonzero discriminant at the origin, which will be convenient for us later.

Consider the polynomial
\[
	\overline{f}(\vec{x}, y, z) \coloneqq f(x_1 z, \ldots, x_n z, y) \in \F(x_1,\ldots,x_n)[y,z].
\]
One can show that the irreducible factors of $\overline{f}$ are in one-to-one correspondence with those of $f$ (see \cite[Section 3.4]{KSS15} for details), so to factor $f(\vec{x}, y)$, it suffices to factor $\overline{f}(\vec{x}, y, z)$ as a polynomial in the ring $\F(\vec{x})[y, z]$.
Given a factor $\overline{g}(\vec{x}, y, z)$ of $\overline{f}$, the corresponding factor of $f(\vec{x}, y)$ is simply $\overline{g}(\vec{x}, y, 1)$.
We can compute a description of $\overline{f}$ from a description of $f$ using a network of size $\poly(n,s,d)$ and degree $O(1)$.

In total, the reduction to bivariate factoring can be implemented by an arithmetic network of size $\poly(n,s,d)$, degree $O(1)$, and error degree $2 d^2$.
It is clear that this network can be printed by a polynomial-time Turing machine.

\paragraph{Univariate factorization}

The first step in Kaltofen's algorithm is to factor the univariate polynomial $\overline{f}(\vec{0}, y, 0)$ in $\F_q[y]$.
Because we know that $f(\vec{x}, y)$ has nonzero discriminant at the origin and $f(\vec{x}, y)$ is monic in $y$, it follows that the discriminant of $f(\vec{0}, y) = \overline{f}(\vec{0}, y, 0)$ is nonzero, which implies that $\overline{f}(\vec{0}, y, 0)$ is squarefree.
This allows us to factor $\overline{f}(\vec{0}, y, 0)$ in $\F_q[y]$ by applying the network of \Cref{thm:network algo univariate factorization}, which has size $\poly(d, \log q)$, degree $d^{O(\log(q))}$, and error degree $1$.
If we are also given a representation of $\F_{q(n)}$ as an extension of $\F_{p(n)}$, then \Cref{thm:network algo univariate factorization} instead provides us with a network of size $\poly(d, p)$ and degree $\poly(d, p)$.

This is the only part of the arithmetic network whose size depends on $\log(q)$ or $p$.
Because a univariate polynomial of degree $d$ can be computed by an arithmetic circuit of size $O(d)$, the factors produced in this step can be represented by circuits of size $O(d)$.
The remaining part of the network that implements \Cref{prop:kaltofen arith network} has size $\poly(n,s,d)$, so it follows that the circuits $C_1, \ldots, C_t$ produced for the factors of $f$ likewise have size bounded by $\poly(n,s,d)$ as claimed.

\paragraph{Hensel lifting}

We now perform Hensel lifting to lift the factorization of $\overline{f}(\vec{0}, y, 0)$ to an approximate factorization of $\overline{f}(\vec{x}, y, z)$ in $\F_q(\vec{x})[y,z]$.
Let $g_0(y) \in \F_q[y]$ be an irreducible factor of $\overline{f}(\vec{0}, y, 0)$, and let $h_0(y) \in \F_q[y]$ be a polynomial such that
\[
	\overline{f}(\vec{0}, y, 0) = g_0(y) \cdot h_0(y).
\]
We know that $\overline{f}(\vec{0}, y, 0)$ is squarefree, so $\gcd(g_0(y), h_0(y)) = 1$.
Because $g_0$ and $h_0$ are coprime, there are polynomials $a_0, b_0 \in \F_q[y]$ of small degree such that $a_0 g_0 + b_0 h_0 = 1$.
Moreover, from the coefficients of $g_0$ and $h_0$, we can compute the coefficients of $a_0$ and $b_0$ using the network of \Cref{lem:network algo bezout coefficients}, which has size and degree $\poly(d)$.

We now use Hensel lifting to lift the factorization $\overline{f}(\vec{0}, y, 0) = g_0(y) \cdot h_0(y)$ to an approximate factorization of $\overline{f}(\vec{x}, y, z)$.
We apply the network of \Cref{lem:network algo hensel lifting} with $k = 2 \log_2(d) + 2$, providing as input the description of the circuit that computes $\overline{f}(\vec{x}, y, z)$ as well as the coefficients of $g_0(y)$, $h_0(y)$, $a_0(y)$, and $b_0(y)$.
This network outputs the description of an arithmetic circuit $D$ of size $\poly(n,s,d)$ and degree $\poly(d)$.
The outputs of $D$ are the coefficients of two polynomials $g_k(y, z), h_k(y, z) \in \F[\vec{x}][y,z]$ such that
\begin{align*}
    \overline{f}(\vec{x}, y, z) &= g_k(y, z) \cdot h_k(y, z) \bmod \langle z^{4 d^2} \rangle \\
    g_k(y, 0) &= g_0(y),
\end{align*}
and $g_k(y,z)$ is monic with respect to $y$.
The network that computes $D$ has size $\poly(n,s,d)$ and degree $O(1)$.

\paragraph{Factor reconstruction}

From Hensel lifting, we have computed an approximate factorization
\[
    \overline{f}(\vec{x}, y, z) = g_k(y, z) \cdot h_k(y, z) \bmod \langle z^{4d^2} \rangle.
\]
Our goal is to use this to recover a genuine factor of $\overline{f}(\vec{x}, y, z)$.
At this point, our presentation of Kaltofen's algorithm deviates slightly from prior work.
Normally, we would use this approximate factorization to set up a system of linear equations where any nontrivial solution of the linear system will provide us with a nontrivial factor $\overline{g}(\vec{x}, y, z)$ of $\overline{f}(\vec{x}, y, z)$.
We would then recursively continue factoring $\overline{g}$ and $\overline{f}/\overline{g}$.
In the worst case, this involves repeating the entire factoring algorithm $d-1$ times in sequence, which (using \Cref{prop:error-deg-comp}) would result in an arithmetic network of degree bounded by $d^{O(d \log(d) + d \log(q))}$.
By slightly modifying this reconstruction step, we can avoid re-running the algorithm $d-1$ times in sequence, saving a factor of $d$ in the exponent of the resulting degree bound.

Because $g_k$ has individual degree at most $\max(d, 2^k) = 4 d^2$, we can write $g_k(y, z)$ as
\[
	g_k(y, z) = \sum_{i, j = 1}^{4d^2} g_{i,j}(\vec{x}) y^i z^j.
\]
Consider the system of linear equations over $\F_q(\vec{x})$ in the variables $a_{i,j}$ and $b_{i,j}$ given by
\[
	\sum_{i < d, j \le d} a_{i,j} y^i z^j = g_k(y,z) \sum_{i \le 4d^2, j \le 4d^2} b_{i,j} y^i z^j \bmod \langle z^{4 d^2} \rangle.
\]
It is a standard fact (see, e.g., \cite[Claims 3.8 and 3.10]{KSS15}) that this is system has a nontrivial solution if and only if $\overline{f}(\vec{x}, y, z)$ is reducible, and that if $(a(y,z), b(y,z))$ is a nontrivial solution, then $a(y,z)$ and $\overline{f}(\vec{x}, y, z)$ have nontrivial GCD in $\F_q(\vec{x})[y,z]$.
We will need more control over the factor that we obtain from this GCD---in particular, we want to guarantee that we obtain an irreducible factor of $\overline{f}(\vec{x}, y, z)$.

Suppose $\overline{f}(\vec{x}, y, z)$ is reducible $\F_q(\vec{x})[y,z]$.
Let $r(\vec{x}, y, z)$ be an irreducible factor of $\overline{f}(\vec{x}, y, z)$ such that $g_0(y)$ divides $r(\vec{0}, y, 0)$.
Because $\overline{f}(\vec{0}, y, 0)$ is squarefree, the polynomial $g_0(y)$ is a factor of $\overline{f}(\vec{0}, y, 0)$ of multiplicity 1, so there is in fact a unique irreducible factor $r(\vec{x}, y, z)$ of $\overline{f}(\vec{x}, y, z)$ such that $g_0(y)$ divides $r(\vec{0}, y, 0)$.
We will argue that for every nontrivial solution $(a(y,z), b(y,z))$ of the above system of linear equations, we have that $r(\vec{x}, y, z)$ divides $a(y,z)$.

We first show that there is a solution to the linear system above whose left-hand side is precisely the irreducible factor $r(\vec{x}, y, z)$.
Let $s_0(y)$ be a polynomial such that
\[
    r(\vec{0}, y, 0) = g_0(y) \cdot s_0(y).
\]
Although we do not have $r(\vec{0}, y, 0)$ nor the factorization above, for the sake of analysis, we will consider what happens when we apply Hensel lifting to the above factorization.
If we apply Hensel lifting for $k = 2 \log_2(d) + 2$ steps, we obtain an approximate factorization
\[
    r(\vec{x}, y, z) = \tilde{g}_k(y, z) \cdot s_k(y, z) \bmod \langle z^{4 d^2} \rangle.
\]
This, together with degree bounds on $\tilde{g}_k$ and $s_k$ implied by \Cref{lem:network algo hensel lifting}, implies that $(r, s_k)$ solves the linear system above if we replace $g_k$ with $\tilde{g}_k$, i.e., that $r = \tilde{g}_k s_k \bmod \langle z^{4 d^2} \rangle$.
As we will see, the pair $(r, s_k)$ also solves the original linear system of interest.
Writing $\overline{f}(\vec{x}, y, z) = r(\vec{x}, y, z) \cdot t(\vec{x}, y, z)$, this yields an approximate factorization of $\overline{f}(\vec{x}, y, z)$ as
\[
    \overline{f}(\vec{x}, y, z) = \tilde{g}_k(y, z) \cdot s_k(y, z) \cdot t(\vec{x}, y, z) \bmod \langle z^{4 d^2} \rangle.
\]
Because $\tilde{g}_k(y, z)$ is monic with respect to $y$ and satisfies $\tilde{g}_k(y, 0) = g_0(y)$, uniqueness of Hensel lifting (\cite[Lemma 3.4]{KSS15}) implies that
\[
    \tilde{g}_k(y, z) = g_k(y, z) \bmod \langle z^{4 d^2} \rangle.
\]
This lets us write
\[
    r(\vec{x}, y, z) = g_k(y, z) \cdot s_k(y, z) \bmod \langle z^{4 d^2} \rangle.
\]
Because $\overline{f}$ is monic with respect to $y$ and $r$ is a proper factor of $\overline{f}$, it must be the case that $r$ has individual degree less than $d$ with respect to $y$.
\Cref{lem:network algo hensel lifting} implies that $s_k(y,z)$ is a polynomial of degree at most $4d^2$ with respect to $y$ and $z$.
Thus $(r, s_k)$ is a solution to the linear system $a = g_k b \bmod \langle z^{4 d^2} \rangle$ as claimed.

Of course, there may be other solutions to this linear system whose left-hand side is not the polynomial $r(\vec{x}, y, z)$.
Let $(a(y,z), b(y,z))$ be one such solution.
We will show that $r(\vec{x}, y, z)$ divides $a(y, z)$.
Let 
\[
    \rho(\vec{x}, z) \coloneqq \res_y(r(\vec{x}, y, z), a(y, z))
\]
be the resultant of $r(\vec{x}, y, z)$ and $a(y, z)$ with respect to $y$.
Importantly, the resultant $\rho(\vec{x}, z)$ does not depend on the variable $y$.
From \Cref{lem:resultant linear combination}, there are polynomials $u, v \in \F_q(\vec{x})[y, z]$ such that
\[
    u(\vec{x}, y, z) r(\vec{x}, y, z) + v(\vec{x}, y, z) a(y, z) = \rho(\vec{x}, z).
\]
This lets us write
\begin{align*}
    \rho(\vec{x}, z) &= u(\vec{x}, y, z) r(\vec{x}, y, z) + v(\vec{x}, y, z) a(y, z) \bmod \langle z^{4 d^2} \rangle \\
    &= u(\vec{x}, y, z) g_k(y, z) s_k(y, z) + v(\vec{x}, y, z) g_k(y, z) b(y, z) \bmod \langle z^{4 d^2} \rangle \\
    &= g_k(y, z) \del{u(\vec{x}, y, z) s_k(y, z) + v(\vec{x}, y, z) b(y, z)} \bmod \langle z^{4 d^2} \rangle.
\end{align*}
Recall that the polynomial $g_k(y, z)$ is monic with respect to $y$.
This implies that any multiple of $g_k(y, z)$ that is nonzero modulo $\langle z^{4 d^2} \rangle$ must depend on the variable $y$.
The left-hand side $\rho(\vec{x}, z)$ above does not depend on $y$, so it must be the case that 
\[
    \rho(\vec{x}, z) = 0 \bmod \langle z^{4d^2} \rangle.
\]
The resultant $\rho(\vec{x}, z)$ is the determinant of a matrix of size $2d \times 2d$.
Each entry of this matrix is a polynomial of degree at most $d$, so we have $\deg(\rho(\vec{x}, z)) \le 2 d^2 < 4d^2$.
Because $\rho(\vec{x}, z)$ has degree less than $4 d^2$, the fact that $\rho(\vec{x}, z) = 0 \bmod \langle z^{4 d^2} \rangle$ implies that $\rho(\vec{x}, z) = 0$ identically.
This means that $r(\vec{x}, y, z)$ and $a(y, z)$ have a nontrivial common factor over $\F_q(\vec{x}, z)[y]$.
Because $r(\vec{x}, y, z)$ is an irreducible factor of $\overline{f}(\vec{x}, y, z)$ and $\overline{f}$ is monic in $y$, it must be the case that $r$ is also monic in $y$.
Gauss's Lemma (\Cref{lem:gauss}) then implies that $r$ is also irreducible in $\F_q(\vec{x}, z)[y]$.
Since $r$ is irreducible in $\F_q(\vec{x}, z)[y]$ and $r$ and $a$ have a nontrivial common factor in $\F_q(\vec{x}, z)[y]$, it must be the case that $r$ divides $a$, as claimed.

We now describe how to find a nontrivial factor of $\overline{f}(\vec{x}, y, z)$.
For every choice of $d_y, d_z \in [d]$, we solve the linear system
\[
	\sum_{i < d_y, j \le d_z} a_{i,j} y^i z^j = g_k(y,z) \sum_{i \le 4d^2, j \le 4d^2} b_{i,j} y^i z^j \bmod \langle z^{4 d^2} \rangle
\]
using the arithmetic network of \Cref{lem:network algo linear systems}, simulating the $\pit$ gates by instead evaluating an input to a $\pit$ gate on a randomly-chosen point.
The network of \Cref{lem:network algo linear systems} has size $\poly(n, d, s)$ and degree $O(1)$, and we invoke this network $O(d^2)$ times in parallel, so we can solve all of these systems with a network of size $\poly(n, d, s)$ and degree $O(1)$.
In our setting, each $\pit$ gate in the network of \Cref{lem:network algo linear systems} tests a polynomial of degree at most $\poly(d)$, and we perform at most $\poly(d)$ such tests, so the error degree of the resulting network is bounded by $\poly(d)$.

If there is no nontrivial solution to any of these linear systems, then the polynomial $\overline{f}(\vec{x}, y, z)$ is irreducible (\cite[Claim 3.8]{KSS15}).
In this case, we report that $f(\vec{x})$ is irreducible and the algorithm terminates.
On the other hand, if at least one of these systems has a nontrivial solution, let $d_y, d_z \in [d]$ be chosen such that the corresponding linear system has a solution and the sum $d_y + d_z$ is minimized, and let $(a(y,z), b(y,z))$ be a nontrivial solution to the corresponding linear system.
(In other words, we find a solution $a(y,z)$ of minimal degree.)
Since $r(\vec{x}, y, z)$ is one such solution of degree at most $D \coloneqq \deg(r)$, we know that we have a solution $a(y,z)$ of degree at most $D$.
By the analysis above, we know that $r$ divides $a$, so it must be the case that $r(\vec{x}, y, z) = a(y,z)$, i.e., that $a(y,z)$ is an irreducible factor of $\overline{f}(\vec{x}, y, z)$.
We report that $a(y,z)$ is an irreducible factor of $\overline{f}(\vec{x}, y, z)$, which in turn yields a nontrivial factor of the original polynomial $f(\vec{x})$.
This step can be implemented by an arithmetic network of size $\poly(n,s,d)$, degree $\poly(n,d)$, and error degree $\poly(d)$.

\paragraph{Complete factorization}

To obtain a complete factorization of $\overline{f}(\vec{x}, y, z)$ into irreducibles, we run the Hensel lifting and factor reconstruction steps in parallel on all the irreducible factors of $\overline{f}(\vec{0}, y, 0)$.
This produces a collection of arithmetic circuits that compute the irreducible factors of $\overline{f}(\vec{x}, y, z)$.
There may be multiple circuits that compute the same irreducible factor of $\overline{f}(\vec{x}, y, z)$ in this list.
To prune these duplicates, we evaluate all circuits on a randomly-chosen point, declaring two circuits to be duplicates if they have the same evaluation at this point.
For each unique irreducible factor, we output the lexicographically-first circuit that computes it.
These evaluations can be carried out by the network of \Cref{prop:universal-arith-network}, which has size $\poly(n,s,d)$ and degree $\poly(n, d)$.
The polynomials we are testing have degree at most $d$, so by \Cref{lem:sz}, the error degree of this procedure is bounded by $d$.
In total, this step results in a network of size $\poly(n,s,d)$ and, by \Cref{prop:concat-arith-networks}, degree $\poly(n,d)$ and error degree $\poly(d)$.

\paragraph{Computing factor multiplicities}

We have now computed descriptions of a collection of arithmetic circuits $C_1, \ldots, C_t$, where each $C_i$ computes a polynomial $g_i$ that is either irreducible or a $p^{\beta_i}$-th power of an irreducible polynomial.
The last step in the proof of \Cref{prop:kaltofen arith network} is to determine the multiplicities of each $g_i(\vec{x})$ as a factor of $f(\vec{x})$.
For each $i \in [t]$ and $j \in [d]$, we test if $g_i(\vec{x})^j$ divides $f(\vec{x})$ using the network of \Cref{lem:network algo divisibility}.
We set the multiplicity of $g_i$ as a factor of $f$ to be the largest $j$ for which $g_i(\vec{x})^j$ divides $f(\vec{x})$.
The network of \Cref{lem:network algo divisibility} is invoked $O(d^2)$ times, each of which costs size $\poly(n,s,d)$ and degree $O(1)$, so the total cost of this step is likewise $\poly(n,s,d)$ size.
Because we do not have access to $\pit$ gates, we have to simulate the $\pit$ gates in the network of \Cref{lem:network algo divisibility} by evaluating circuits on random points using \Cref{prop:universal-arith-network}.
Every polynomial that is tested in the network of \Cref{lem:network algo divisibility} has degree bounded by $2d$, so this simulation can be implemented in size $\poly(n,s,d)$, degree $\poly(n,d)$, and error degree $O(d)$.

\medskip
\noindent
This completes the proof of \Cref{prop:kaltofen arith network}.

\subsubsection{An application to root-finding}

As a direct application of \Cref{prop:kaltofen arith network}, we can solve the root-finding problem for arithmetic circuits over $\F_q$ efficiently using a randomized arithmetic network.
The root-finding problem is the following: given a polynomial $f(\vec{x}, y)$, find all polynomials $g(\vec{x})$ such that $f(\vec{x}, g(\vec{x})) = 0$.
This is directly related to factorization, as \Cref{lem:gauss} implies that such roots are precisely the factors of $f(\vec{x}, y)$ of the form $y - g(\vec{x})$.
Since \Cref{prop:kaltofen arith network} allows us to factor an arithmetic circuit, we have a solution for the root-finding problem as an immediate corollary.
We state this as a separate corollary, as this is result we ultimately use in our reconstruction algorithm in \Cref{sec:ki-reconstruction}.

\begin{cor}\label{cor:kaltofen arith network}
    Let $\F = \set{\F_{q(n)}}_{n \in \N}$ be a feasible family of finite fields, where the $n$\ts{th} field has order $q(n)$.
    There is a deterministic, polynomial-time Turing machine that receives as input $(1^n, 1^d, 1^s)$ and outputs a randomized arithmetic network satisfying the following properties.
    \begin{enumerate}
        \item
            The network receives as input the description of an arithmetic circuit $C$ of size $s$ over $\F_{q(n)}$ that computes a polynomial $f(\vec{x}, y) \in \F_q[x_1, \ldots, x_n]$ of degree at most $d$.
        \item
            The network outputs the descriptions of arithmetic circuits $C_1, \ldots, C_t$ for some $t \le d$, together with natural numbers $\beta_1, \ldots, \beta_t \in \N$ encoded in binary.
            For each $i \in [t]$, let $g_i \in \F_{q(n)}[x_1, \ldots, x_n]$ be the polynomial computed by the circuit $C_i$.
            The output of the network satisfies the following properties.
            \begin{enumerate}
                \item
                    Each circuit $C_i$ is of size $\poly(n,s,d)$.
                \item
                    For each $i \in [t]$, the polynomial $g_i$ has degree at most $d$.
                \item
                    For every polynomial $g(\vec{x})$ such that $f(\vec{x}, g(\vec{x})) = 0$, there is an index $i^\star \in [t]$ such that $g_{i^\star}(\vec{x}) = g(\vec{x})^{p^{\beta_{i^\star}}}$.
            \end{enumerate}
        \item
            The network has size $\poly(n,s,d,\log(q))$, degree $n^{O(1)} d^{O(\log(q))}$, and error degree $\poly(d)$.
            If we are additionally provided with a representation of $\F_{q(n)}$ as an extension of $\F_{p(n)}$, where $p(n)$ is the characteristic of $\F_{q(n)}$, then the size of the network becomes $\poly(n,s,d,p)$ and the degree becomes $\poly(n,d,p)$.
    \end{enumerate}
\end{cor}

\bproof
    By \Cref{cor:gauss}, if $C(\vec{z}, g(\vec{z})) = 0$, then the polynomial $y - g(\vec{z})$ is an irreducible factor of $C(\vec{z}, y)$.
    Using the arithmetic network of \Cref{prop:kaltofen arith network}, we can compute the standard descriptions of a set of circuits $C_1, \ldots, C_t$ and natural numbers $e_1, \ldots, e_t \in \N$ and $\beta_1, \ldots, \beta_t \in \N$ such that the circuit $C_i$ computes a $p^{\beta_i}$-th power of an irreducible factor of $C(\vec{z}, y)$.
    In particular, one of the circuits $C_i$ must compute $(y - g(\vec{z}))^{p^{\beta_i}} = y^{p^{\beta_i}} - g(\vec{z})^{p^{\beta_i}}$.
    Subtracting $y^{p^{\beta_i}}$ and negating the output yields a circuit that computes $g(\vec{z})^{p^{\beta_i}}$.
    Since we do not know which circuit $C_i$ computes the desired irreducible factor, we simultaneously modify all circuits output by Kaltofen's algorithm and output these modified circuits.
\eproof

\subsection{Uniformity-preserving depth reduction} \label{subsec:uniform depth reduction}

In this final section, we show that the depth reduction of \textcite{VSBR83} can be carried out in a uniformity-preserving way.
Specifically, we show that if $\set{C_n}_{n \in \N}$ is a $\log^c$-uniform family of arithmetic circuits that compute polynomials of low degree, then there is a family of \emph{low-depth} arithmetic circuits $\set{C'_n}_{n \in \N}$ that computes the same family of polynomials and is $\log^{c'}$-uniform for some constant $c' = c + O(1)$.
To do this, we first need a lemma that says circuits can be homogenized in a uniformity-preserving manner.

\begin{lem}[uniformity-preserving homogenization] \label{lem:uniform homogenization}
    Let $\set{C_n}$ be a $\log^c$-uniform arithmetic circuit family of size $T(n)$ that computes a family of polynomials $\set{f_n}$.
    For any polynomial-time computable $d : \N \to \N$, there is a $\log^{c+O(1)}$-uniform family of homogeneous multi-output arithmetic circuits $\set{C'_n}$ of size $T'(n) = O(T(n) d(n)^2)$ such that $C'_n$ computes the homogeneous components of the polynomial $f_n$ up to degree $d(n)$.
\end{lem}

\bproof
    We follow the standard gate-simulation argument that arithmetic circuits can be homogenized efficiently (as seen in the proof of \Cref{lem:network algo homogeneous components}), verifying that this argument can be carried out in a uniformity-preserving manner.

    Each gate $u$ in $C_n$ will correspond to a collection of $d+1$ gates $(u,0), \ldots, (u,d)$ in $C'_n$.
    The intention is that the gate $(u,a)$ in $C'_n$ computes the degree-$a$ homogeneous component of the polynomial computed by the gate $u$ in $C_n$.
    The wiring of the gates in $C'_n$ depends on the type of the gate $u$.
    \begin{itemize}
        \item 
            If $u$ is an input gate labeled by a field constant $\alpha \in \F$, then we set $(u, 0)$ to be an input gate labeled by $\alpha$.
            For $i \in \set{1,2,\ldots,d}$, we set $(u,i)$ to be an input gate labeled by zero.
        \item
            If $u$ is an input gate labeled by the variable $x_i$, we set $(u,1)$ to be an input gate labeled by $x_i$.
            For $i \in \set{0,2,3,\ldots,d}$, we set $(u,i)$ to be an input gate labeled by zero.
        \item
            If $u$ is an addition gate in $C_n$ with children $v$ and $w$, then for each $i \in \set{0,1,\ldots,d}$, we set $(u,i)$ to be an addition gate with children $(v,i)$ and $(w,i)$.
        \item
            If $u$ is a multiplication gate in $C_n$ with children $v$ and $w$, then for each $a \in \set{0,1,\ldots,d}$, we set
            \[
                (u,a) = \sum_{b = 0}^a (v,b) \times (w,a-b),
            \]
            adding extra gates as necessary to implement the sum of products above as an alternating circuit of fan-in two.
    \end{itemize}
    We set the output gates of $C'_n$ to be $(u,0), (u,1), \ldots, (u,d)$, where $u$ is the output gate of $C_n$ and $d$ is the degree of the polynomial computed by $C_n$.

    The following induction argument shows that the gate $(u,a)$ computes precisely the degree-$a$ homogeneous component of the polynomial computed by the gate $u$.
    \begin{itemize}
        \item 
            If $u$ is an input gate, then it immediately follows from the definition of the gate $(u,a)$ that $(u,a)$ computes the degree-$a$ homogeneous component of $u$.
        \item
            If $u$ an addition gate with children $v$ and $w$, then the degree-$a$ component of $u$ is the sum of the degree-$a$ components of $v$ and $w$.
            By induction, the gates $(v,a)$ and $(w,a)$ compute the degree-$a$ components of $v$ and $w$, respectively, so $(u,a)$ correctly computes the degree-$a$ component of $u$.
        \item
            If $u$ is instead a product gate with children $v$ and $w$, then the degree-$a$ component of $u$ is given by $\sum_{b=0}^a v_b w_{a-b}$, where $v_b$ is the degree-$b$ component of $v$ and likewise $w_{a-b}$ is the degree-$(a-b)$ component of $w$.
            By induction, the gates $(v,b)$ and $(w,a-b)$ correctly compute these polynomials, so $(u,a)$ correctly computes the degree-$a$ component of $u$.
    \end{itemize}
    
    Because each gate of the circuit $C'_n$ computes a homogeneous polynomial, the circuit $C'_n$ is homogeneous.
    To bound the size of $C'_n$, observe that $C'_n$ contains $d(n)+1$ copies of every gate in $C_n$, and that each copy $(u,a)$ of a gate from $C_n$ uses at most $O(d(n))$ additional gates in its implementation.
    This results in the claimed bound of $O(T(n) d(n)^2)$ on the number of gates of $C'_n$.

    It remains to bound the uniformity of the circuit family $\set{C'_n}$.
    To do this, we slightly modify the construction of $C'_n$ described above.
    Whenever $u$ is a multiplication gate in $C_n$ with children $v$ and $w$, we implement each convolution $(u,a) = \sum_{b=0}^a (v, b) \times (w, a-b)$ as a balanced full binary tree of depth $\log d$ by implementing a subcircuit that computes $(u,a)$ as
    \[
        (u, a) = \sum_{b=0}^d (v, b) \times (w, a-b),
    \]
    with the convention that $(w, a-b)$ is labeled by the constant zero whenever $a < b$.
    If $u$ is a product gate in $C_n$, we label the gates in the subcircuit computing $(u,a)$ as $(u, a, x)$, where $x \in \set{0,1}^{\le 2 \log d}$ is a bitstring of length at most $2 \log d$.
    The intended meaning is that $(u,a,x)$ is either the sum or product of $(u,a,x0)$ and $(u,a,x1)$, depending on the parity of the length of $x$.
    Building $C'_n$ in this manner allows us to more efficiently decide the gate adjacency relation.
    We can quickly decide if the gates in the $i$-th layer compute (1) $(u,a)$ where $u$ is an addition gate, (2) $(u,a)$ where $u$ is a product gate, or (3) a subgate used in the computation of $(u,a)$.

    We now describe a $\P$-uniform family of formulas $\set{\Phi'_n}$ that decides the gate adjacency relation in the new circuit family $C'_n$.
    If the original circuit $C_n$ has depth $\Delta(n)$, then $C'_n$ has depth $\Delta'(n) \coloneqq O(\Delta(n) \log d)$.
    Given $i \in [\Delta'(n)]$ and gate names $\hat{u}, \hat{v}, \hat{w} \in [T'(n)]$, we first inspect the parity of $i$ to determine the type of gates $\hat{u}, \hat{v}$, and $\hat{w}$ should be.
    If any of $\hat{u}$, $\hat{v}$, or $\hat{w}$ are not of the correct type (which we can inspect from the binary encoding of the gate names), then $\Phi'_n(i,\hat{u}, \hat{w}, \hat{v}) = 0$.
    Otherwise, we branch based on the type of gate $\hat{u}$.
    \begin{enumerate}
        \item
            Suppose $\hat{u} = (u,a)$, $\hat{v} = (v,b)$, and $\hat{w} = (w, c)$, where $u$ is an addition gate at layer $i'$ of the circuit $C_n$.
            In this case, we set $\Phi'_n(i,\hat{u},\hat{v},\hat{w}) = 1$ iff $a = b = c$ and $\Phi_n(i', u, v, w) = 1$.
        \item
            Suppose instead that $\hat{u} = (u,a,x)$ is a gate inside the multiplication tree used to compute $(u,a)$.
            If $x$ is not of maximal length, then we output 1 iff $\hat{v} = (u,a,x0)$ and $\hat{w} = (u,a,x1)$.
            Otherwise, if $x$ is of maximal length, we output 1 iff $\hat{v} = (v,b)$ and $\hat{w} = (w,a-b)$ for the correct values of $b$ and $a-b$ (determined by the bits of $x$) and if $\Phi(i', u, v, w) = 1$.
    \end{enumerate}
    It is clear that if the formulas $\set{\Phi_n}$ are $\P$-uniform, then the resulting formulas $\set{\Phi'_n}$ describing the circuit family $\set{C'_n}$ are also $\P$-uniform.
    Furthermore, if the formulas $\set{\Phi_n}$ have size $(\log T(n))^c$, then the formulas $\Phi'_n$ have size $(\log T(n))^c (\log T'(n))^{O(1)}$, so the formulas $\set{\Phi'_n}$ are $\log^{c + O(1)}$-uniform.
\eproof

With \Cref{lem:uniform homogenization}, we now proceed to show that the depth reduction of \textcite{VSBR83} preserves $\log^c$-uniformity.

\begin{proposition}[Uniformity-preserving depth reduction]\label{prop:uniform-depth-red}
    Let $\set{C_n}$ be a $\log^c$-uniform arithmetic circuit family of size $T(n)$ and degree $d(n)$ that computes a family of polynomials $\set{f_n}$. 
    Then, there is a $\log^{c+O(1)}$-uniform family of arithmetic circuits $\set{C'_n}$ of size $T'(n)=\poly(T(n), d(n))$ and depth $\Delta(n)=O(\log T(n) \cdot \log d(n))$ that computes the same family of polynomials $\set{f_n}$.
\end{proposition}

\bproof
    At a high level, the circuit family $\set{C'_n}$ will be obtained by applying the depth reduction of \textcite{VSBR83} to the circuit family $\set{C_n}$.
    The depth reduction guarantees that the circuit family $\set{C'_n}$ has size $T'(n) = \poly(T(n))$ and depth $\Delta(n) = O(\log T(n) \log d(n))$.
    We need to verify that this depth reduction can be performed in a uniformity-preserving manner.

    By applying \Cref{lem:uniform homogenization}, we obtain a $\log^{c+O(1)}$-uniform family of arithmetic circuits $\set{C_n^{\text{hom}}}$ such that $C_n^{\text{hom}}$ is a homogeneous arithmetic circuit of size $O(T(n) d(n)^2)$ that computes the homogeneous components of the polynomial $f_n$.
    We apply the depth reduction to the circuit family $\set{C_n^{\text{hom}}}$.

    We now present the depth reduction, following the construction and proof given by \textcite[\textsection 5.3]{Saptharishi-survey}.

    We first define gate quotients, which are intermediate polynomials that will be used in wiring the depth-reduced circuit.
    For a pair of gates $u$ and $v$ in $C_n^{\text{hom}}$, we define the \emph{gate quotient} $[u : v](\vec{x})$ to be the polynomial defined inductively as follows.
    \begin{enumerate}
        \item
            If $u$ and $v$ are the same game, then we define $[u : v](\vec{x}) \coloneqq 1$.
        \item 
            If $u$ and $v$ are different gates, then the definition of $[u : v](\vec{x})$ depends on the gate type of $u$.
            \begin{enumerate}
                \item
                    If $u$ is a leaf, then we define $[u : v](\vec{x}) \coloneqq 0$.
                \item
                    If $u$ is an addition gate with children $u_\ell$ and $u_r$, then we define $[u : v](\vec{x}) \coloneqq [u_\ell : v](\vec{x}) + [u_r : v](\vec{x})$.
                \item
                    If $u$ is a multiplication gate with children $u_\ell$ and $u_r$ where $\deg(u_\ell) \le \deg(u_r)$, then we define $[u : v](\vec{x}) \coloneqq [u_\ell](\vec{x}) \times [u_r : v](\vec{x})$, where $[u_\ell](\vec{x})$ is the polynomial computed by the gate $u_\ell$.%
            \end{enumerate}
    \end{enumerate}
    Because the circuit $C_n^{\text{hom}}$ is homogeneous, each gate quotient $[u:v](\vec{x})$ is a homogeneous polynomial,  %
    and if $[u : v](\vec{x}) \neq 0$, then $\deg([u:v](\vec{x})) = \deg(u) - \deg(v)$.

    For a parameter $m \in \N$, we define the \emph{frontier at degree $m$}, denoted by $\mathcal{F}_m$, to be the set of gates of degree $\ge m$ whose children have degree strictly less than $m$.
    Formally, we have
    \[
        \mathcal{F}_m \coloneqq \set{v : \deg(v) \ge m, \deg(v_\ell) < m, \deg(v_r) < m},    
    \]
    where $v_\ell$ and $v_r$ are the left and right children, respectively, of $v$.
    Note that if $v$ appears in a frontier $\mathcal{F}_m$, then because the circuit $C_n^{\text{hom}}$ is homogeneous, it must be the case that $v$ is a multiplication gate.
    By induction, one can establish the following identities.

    \begin{subclaim}[{\cite[Lemma 5.12]{Saptharishi-survey}}] \label{subclaim:frontier sum}
        Let $u$ and $v$ be gates in $C_n^{\text{hom}}$ and let $m \in \N$ be a parameter such that $\deg(u) \ge m > \deg(v)$.
        Then
        \begin{align*}
            [u](\vec{x}) &= \sum_{w \in \mathcal{F}_m} [u : w](\vec{x}) \cdot [w](\vec{x}) \\
            [u:v](\vec{x}) &= \sum_{w \in \mathcal{F}_m} [u:w](\vec{x}) \cdot [w:v](\vec{x}).
        \end{align*}
    \end{subclaim}
    For our purposes, we only need the identities of \Cref{subclaim:frontier sum} to argue about the correctness of our depth-reduced circuit.
    As our main concern is the uniformity of the depth reduction procedure, not its correctness, we encourage readers who are interested in further details about the correctness of the depth reduction to consult \textcite[Section 5.3]{Saptharishi-survey}.

    We now construct the depth reduced circuit $\set{C_n'}$ in $\log_2 d$ rounds, where in the $i$-th round we compute polynomials of the form $[u](\vec{x})$ or $[u:v](\vec{x})$ of degree $t$, where $2^{i-1} < t \le 2^i$.
    For $i = 1$, each linear polynomial of the form $[u](\vec{x})$ or $[u:v](\vec{x})$ can be computed directly by an arithmetic circuit of size $O(s)$ and depth $O(\log s)$.
    
    For $i \ge 2$, we compute these polynomials as follows.
    \begin{itemize}
        \item 
            We first describe how to compute polynomials of the form $[u](\vec{x})$ for a gate $u$.
            Let $m = \deg(u)/2$.
            Then it follows from \Cref{subclaim:frontier sum} that
            \[
                [u](\vec{x}) = \sum_{w \in \mathcal{F}_m} [u:w](\vec{x}) \cdot [w](\vec{x}) = \sum_{w \in \mathcal{F}_m} [u:w](\vec{x}) \cdot [w_\ell](\vec{x}) \cdot [w_r](\vec{x}),
            \]
            where $w_\ell$ and $w_r$ are the children of gate $w$.
            By definition of the frontier $\mathcal{F}_m$, each polynomial on the right-hand side above has degree at most $\deg(u)/2$, and so has been computed in a previous round of the depth reduction.
            Thus, we can compute $[u](\vec{x})$ by adding an additional $O(T)$ gates and $O(\log T)$ depth.

        \item
            Next, we describe how to compute a polynomial of the form $[u:v](\vec{x})$ for gates $u$ and $v$.
            Let $m = (\deg(u) + \deg(v))/2$.
            After two applications of \Cref{subclaim:frontier sum}, we obtain the identity
            \[
                [u : v](\vec{x}) = \sum_{w \in \mathcal{F}_m} \sum_{p \in \mathcal{F}(w_\ell)} [u:w](\vec{x}) \cdot [w_\ell : p](\vec{x}) \cdot [p_\ell](\vec{x}) \cdot [p_r](\vec{x}) \cdot [w_r : v](\vec{x}),
            \]
            where the gates $w_\ell$ and $w_r$ are the children of the gate $w$ with $\deg(w_\ell) \le \deg(w_r)$, the gates $p_\ell$ and $p_r$ are likewise the children of $p$, and the set $\mathcal{F}(w_\ell)$ is the frontier at degree $\deg(w_\ell)/2$.
            As a consequence of the choice of the frontiers in the above sum, each of the five terms in the innermost product are of degree at most $(\deg(u) - \deg(v))/2 = \deg([u:v])/2$. 
            By induction, each of these terms has already been computed in a previous round of the construction.
            This allows us to compute the gate quotient $[u:v](\vec{x})$ by adding $O(T^2)$ gates and $O(\log T)$ depth to the circuit.
    \end{itemize}

    In all, the preceding construction produces a circuit of size $T' = \poly(T, d)$ and depth $O(\log T \cdot \log d)$ that computes the same outputs as the circuit $C_{n}^{\text{hom}}$.
    We obtain the depth-reduced circuit for $f_n$ by summing over the gates that correspond to the outputs of $C_{n}^{\text{hom}}$, each of which computes one of the homogeneous components of $f_n$.

    It remains to describe a $\P$-uniform family of formulas $\set{\Phi_n'}$ that decides the gate adjacency relation in the depth-reduced circuit family $\set{C_n'}$.
    We describe a high-level implementation; the low-level details to properly implement each operation as a full binary tree are similar to those appearing in the proof of \Cref{lem:uniform homogenization}, so we leave these out for the sake of clarity.
    Let $\set{\Phi_n}$ be the $\P$-uniform family of formulas that decide the gate adjacency relation in the circuit family $\set{C_n^{\text{hom}}}$.
    Recall that $C_n'$ has depth $\Delta'(n) = O(\log T(n) \log d(n))$ and size $T'(n) = \poly(T(n), d(n))$.
    Given $i \in [\Delta'(n)]$ and gate names $x,y,z \in [T'(n)]$, we first inspect the parity of $i$ to determine the intended gate types of $x,y,z$.
    We can determine the types of $x$, $y$, and $z$ from their binary expansion.
    If any gate is not of the correct type, we output a 0.
    Otherwise, we branch based on the type of gate $x$.

    We first describe a formula $\Psi_n$ that takes two gates $x, y \in [T'(n)]$ as input and returns 1 if and only if $x$ is a parent of $y$.
    \begin{enumerate}
        \item
            Suppose $x$ is a gate meant to compute $[u]$ for some gate $u$ in the circuit $C_n^{\text{hom}}$.
            To check if $y$ is a child of $x$, we need to check that $y$ is of the form $[u : w]$, $[w_\ell]$ or $[w_r]$ for some gate $w \in \mathcal{F}_{\deg(u)/2}$.
            This is straightforward to implement as an OR over all such gate names.
            (Because the underlying circuit is obtained from the homogenization procedure of \Cref{lem:uniform homogenization}, we can assume without loss of generality that each gate is also labeled by its degree.
            This makes it easy to determine which gates are in the frontier $\mathcal{F}_{\deg(u)/2}$.)
        \item
            Suppose $x$ is a gate meant to compute $[u:v]$ for some gate $u$ in the circuit $C_n^{\text{hom}}$.
            To check if $y$ is a child of $x$, we need to check that $y$ is one of the five terms appearing in the sum used to compute $[u:v]$.
            As in the previous case, this is easy to implement as an OR over the appropriate gates, using the fact that each gate is labeled by its degree.
    \end{enumerate}
    Given $\Psi_n$, it is easy to construct $\Phi'_n$: define
    \[
        \Phi'_n(i,x,y,z) \coloneqq \Psi_n(i,x,y) \land \Psi_n(i,x,z).
    \]
    If the formulas $\set{\Phi_n}$ are $\P$-uniform, then the same is true for $\set{\Phi'_n}$.
    Furthermore, if the formulas $\set{\Phi_n}$ have size $(\log T(n))^c$, then the formulas $\Phi'_n$ have size $(\log T(n))^c (\log T'(n))^{O(1)}$, so the formulas $\Phi'_n$ are $\log^{c+O(1)}$-uniform as claimed.

    The preceding description of $C_n'$ and the formulas $\Phi'_n$ allow the circuit $C_n'$ to have arbitrary fan-in.
    We can replace $C'_n$ with an equivalent circuit where every gate has fan-in two and the gate types alternate by replacing each large fan-in gate in $C_n'$ with a full alternating binary tree of gates.
    The details of the wiring for this circuit and the appropriate modifications to the formula $\Phi'_n$ are the same as in \Cref{lem:uniform homogenization}.
\eproof

\section{An alternative proof for a special case of~\Cref{thm:int:main}} \label{apdx:lowchar}

In this appendix we describe what is essentially an alternative proof for a special case of~\Cref{thm:int:main}. The special case when the field's characteristic is bounded by a fixed polynomial, say $n$,\footnote{This can be generalized to $n^c$ for any fixed constant $c>1$, at the cost of increasing the constant $k$ in the hypothesis.} and when there is a polynomial-time algorithm that finds a representation of the field (i.e., a suitable irreducible polynomial). In the alternative proof, the conclusion we get differs from the one in~\Cref{thm:int:main} in two aspects: the deduced PIT algorithm runs in polynomial time (rather than in time $n^{\log^{(c)}(n)}$), but this algorithm is a Boolean computation (rather than an arithmetic network). We view these differences as relatively minor, and the main reason for presenting the alternative proof is that we find the proof itself interesting.

\paragraph{The main observation.} The proof is based on the observation that in low-characteristic fields, the hardness assumption in~\Cref{thm:int:main} actually implies hardness for Boolean circuits. Hence, we can use results from Boolean hardness-to-randomness to deduce derandomization of a Boolean class that contains $\pit$.

\paragraph{The alternative proof.}
Recall that the hardness assumption in~\Cref{thm:int:main} is determined by the complexity of the reconstruction procedure from~\Cref{thm:ki:tarhsg}. The reconstruction procedure is a network of degree larger than $q=|\F|$, and with this relatively high degree, one may hope that the network will be able to perform not only arithmetic computation, but also simulate Boolean circuits. Indeed, when such simulation is possible, the hardness hypothesis implies hardness for Boolean circuits.

Thus, we now ask if it is possible to simulate any Boolean circuit computing a function $\F^n\rightarrow\F$ (when the elements are represented as bit-strings) by an arithmetic circuit of similar complexity and degree larger than $q$. We stress that this question refers to functional computation, rather than to computing a polynomial syntactically. As far as we are aware, such simulation is not known in general, but it is known over fields of low characteristics, and it yields the following result:

\begin{proposition} [simulating Boolean computation by degree-$q$ arithmetic procedures, over fields of low characteristic]
Let $\F=\set{\F_n}_{n\in\N}$ be a field family where $\F_n$ is of characteristic at most $p(n)\le\poly(n)$, and there is a $\poly(n)$-time algorithm that maps $1^n$ to a representation of $\F_n$. Then, any $\P$-uniform Boolean $\mathcal{NC}$ circuit family of size $n^k$ computing a function $\F_n^n\rightarrow\F_n$ can be (functionally) simulated by a $\P$-uniform arithmetic circuit family of size $O(n^{k+1}\cdot p^3)$, degree $p\cdot 2^{\polylog(n)}$, and $p^{th}$ root gates.
\end{proposition}

\begin{proof}[Proof sketch]
For ease of notation, we refer to the $n^{th}$ field $\F_n$ as $\F_q=\F_{p^r}$ (i.e., the field of size $p^r$). Let $C_n\colon\F_q^n\rightarrow\F_q$ be a Boolean $\mathcal{NC}$ circuit. The simulating arithmetic circuit $A$ works as follows:
\begin{enumerate}
	\item Given input elements in $\F_q$, use the procedure in~\cite[Proof of Lemma 2]{lsz82} to transform each element into its representation as an $r$-dimensional vector over $\F_p$. (This requires a trace-dual basis for the standard basis of $\F_q$, and an explicit formula for the trace-dual basis is given in~\cite[Exercise 2.40]{ld94}.) 
	
	The arithmetic complexity of this step is essentially that of computing the field trace function on each element, and computing the trace reduces to computing $r$-many $p^{th}$ roots and adding them.
	\item Transform each $\F_p$ element into its binary representation, by brute-force. This can be done using $O(p^3)$ operations per element $x$: for each $i\in[p]$ we map $x$ to $\sum_{c\in \F_p:c_i=1}\de_c(x)$, where $c_i$ is the $i^{th}$ bit in the binary representation of $c$ and $\de_c(x)=\prod_{a\in\F_p\setminus\set{c}}\frac{x-a}{c-a}$. (This is the only part in the proof using the fact that the characteristic is bounded.) 
	\item Simulate the $\mathcal{NC}$ circuit $C$ on the bits, using arithmetic $+,\times$ operations. Since there are at most $\polylog(n)$ layers of $C$ to simulate, the degree will be at most $2^{\polylog(n)}$.
\end{enumerate}

The size of the arithmetic circuit above is $O\left(n\cdot(p^3+r\cdot n\cdot\log(p))+n^k\right)$, and its degree (excluding $p^{th}$ root gates) is at most $p\cdot 2^{\polylog(n)}$. The circuit $A$ is $\P$-uniform because $C$ is $\P$-uniform, and due to our assumption that there is a $\poly(n)$-time algorithm finding a representation of the field (which means that a trace-dual basis can be efficiently hard-wired into $A$).
\end{proof}

Thus, in the special case we are considering, the hardness assumption in~\Cref{thm:int:main} (even when using the refined degree bounds in~\Cref{thm:ki:tarhsg}) implies that there is a Boolean function computable in fixed polylogarithmic depth that is hard for Boolean circuits of fixed polynomial size $n^{5k}$ and larger polylogarithmic depth on all but finitely many inputs. Specifically, the function treats each block of bits in its input as a field element, and simulates the depth-reduced circuit computing the hard polynomial $\set{C_n}$ using Boolean operations. The hard function is computable by circuits of depth $O(\log^2(n)\cdot\polylog(q))$, where the $\polylog(q)$ term comes from simulating field arithmetic; and the depth of the Boolean circuits of size $n^{5k}$ that fail to compute the hard function is determined by the degree bound $2^{\polylog(n)}$ in our hardness assumption, and thus may be an arbitrarily large polylogarithm.

Using~\cite[Theorem 1.5]{CT21}, it follows that $\text{uniform-}\mathcal{BPNC} \subseteq \text{uniform-}\mathcal{NC}$, where the uniformity here refers to logspace-uniformity. Finally,~\textcite{mrk88} proved that $\pit$ can be solved in $\text{uniform-}\mathcal{BPNC}$, and hence it can also be solved in $\text{uniform-}\mathcal{NC}\subseteq\P$.

\end{appendices}
\normalem
\printbibliography

\end{document}